\documentclass[11pt]{article}

\usepackage{geometry}
\usepackage{amsmath}
\usepackage{amssymb}
\usepackage{amsthm}
\usepackage{graphicx}
\usepackage{subcaption}
\usepackage[round]{natbib}
\usepackage{hyperref}
\hypersetup{
	colorlinks=true,
	linkcolor=blue,
	citecolor=blue,
	urlcolor=blue,
}
\usepackage{xcolor}
\usepackage{enumerate}
\usepackage{bm}
\usepackage{booktabs}
\usepackage{caption}
\usepackage{mathtools}
\usepackage{tikz}

\theoremstyle{plain}

\newtheorem{example}{Example}
\newtheorem{theorem}{Theorem}
\newtheorem{proposition}{Proposition}
\newtheorem{lemma}{Lemma}
\newtheorem{corollary}{Corollary}
\newtheorem{assumption}{Assumption}
\newtheorem{remark}{Remark}

\def\bp{\mathbf{p}}

\def\by{\mathbf{y}}
\def\bY{\mathbf{Y}}
\def\0{\mbox{\bf{0}}}
\def\bs{\mathbf{s}}
\def\bS{\mathbf{S}}
\def\bI{\mathbf{I}}

\def\bz{\mathbf{z}}

\def\bM{\mathbf{M}}

\def\bX{\mathbf{X}}
\def\bx{\mathbf{x}}

\def\bv{\mathbf{v}}%
\def\bA{\mathbf{A}}

\def\btheta{\boldsymbol{\theta}}
\def\bTheta{\boldsymbol{\Theta}}
\def\bupsilon{\boldsymbol{\upsilon}}
\def\bUpsilon{\boldsymbol{\Upsilon}}
\def\bpi{\boldsymbol{\pi}}
\def\bphi{\boldsymbol{\phi}}

\def\balpha{\boldsymbol{\alpha}}
\def\bbeta{\boldsymbol{\beta}}
\def\bSigma{\boldsymbol{\Sigma}}
\def\bvartheta{\boldsymbol{\vartheta}}

\title{
	Asymptotic Properties of the Maximum Likelihood Estimator for Markov-Switching Observation-Driven Models
}
\author{
	\normalsize{Frederik Bjerg Krabbe}\footnote{
		Email: \href{mailto:frederik.krabbe@econ.au.dk}{frederik.krabbe@econ.au.dk}.
	} \\[1ex]
	\small{Department of Economics and Business Economics} \\[-1ex]
	\small{Aarhus University} \\[-1ex]
	\small{Universitetsbyen 51} \\[-1ex]
	\small{8000 Aarhus C} \\[-1ex]
	\small{Denmark}
}
\date{}

\allowdisplaybreaks

\begin{document}
	
\maketitle

\begin{abstract}

A Markov-switching observation-driven model is a stochastic process $((S_t,Y_t))_{t \in \mathbb{Z}}$ where $(S_t)_{t \in \mathbb{Z}}$ is an unobserved Markov chain with finite state space and $(Y_t)_{t \in \mathbb{Z}}$ is an observed stochastic process such that the conditional distribution of $Y_t$ given $(Y_\tau)_{\tau \leq t-1}$ and $(S_\tau)_{\tau \leq t}$ depends only on $(Y_\tau)_{\tau \leq t-1}$, through time-varying parameters, and $S_t$. In this paper, we prove consistency and asymptotic normality of the maximum likelihood estimator for such a model. As a special case, we also give conditions under which the maximum likelihood estimator for the widely applied Markov-switching generalised autoregressive conditional heteroscedasticity model introduced by \cite*{HaasMittnikPaolella2004b} is consistent and asymptotically normal.
\vspace{1.5ex} \\ \noindent
\textbf{Keywords:} State Space Models, Hidden Markov Models, Maximum Likelihood Estimation, Consistency, Asymptotic Normality, Markov-Switching Generalised Autoregressive Conditional Heteroscedasticity Models.
\vspace{1.5ex} \\ \noindent
\textbf{JEL Classifications:}  C12, C13, C22, C32, C58.

\end{abstract}

\section{Introduction} \label{Introduction}

State space models and their extensions are ubiquitous in economics and finance. A state space model is a stochastic process $((S_t,Y_t))_{t \in \mathbb{Z}}$ where $(S_t)_{t \in \mathbb{Z}}$ is an unobserved Markov process taking values in $\mathbb{S}$ and $(Y_t)_{t \in \mathbb{Z}}$ is an observed stochastic process taking values in $\mathcal{Y}$ such that the conditional distribution of $Y_t$ given $\mathbf{Y}_{-\infty}^{t-1}$ where $\bY_{i}^{j} := (Y_{i},...,Y_{j})$ and $\bS_{-\infty}^{t}$ where $\bS_{i}^{j} := (S_i,...,S_j)$ depends only on $S_t$. It is called a hidden Markov model when $\mathbb{S}$ is finite, say, $\mathbb{S} = \{1,...,J\}$.

The arguably most well-known extension of the state space model is the $p$'th-order autoregressive state space model in which the conditional distribution of $Y_t$ given $\mathbf{Y}_{-\infty}^{t-1}$ and $\bS_{-\infty}^{t}$ depends on both $\mathbf{Y}_{t-p}^{t-1}$ and $S_t$. An example of an autoregressive state space model is the seminal Markov-switching autoregressive model, in which $\mathbb{S}$ is finite, introduced by \cite{Hamilton1989} to model economic growth. Another example is the Markov-switching autoregressive conditional heteroscedasticity (ARCH) model, in which $\mathbb{S}$ is also finite, introduced independently by \cite{Cai1994} and \cite{HamiltonSusmel1994} to model financial returns. Markov-switching models have since their introduction been popular in economics and finance because they can capture regimes, high persistence via both the persistence in $Y$ and that in $S$, skewness, excess kurtosis, and non-linearity that economic and financial time series can exhibit.\footnote{See, for instance, \cite{Hamilton2010} and \cite{AngTimmermann2012} for more examples of Markov-switching models in economics and finance, respectively.}

Another extension of the state space model that has gained popularity recently is the observation-driven state space model in which the conditional distribution of $Y_t$ given $\mathbf{Y}_{-\infty}^{t-1}$ and $\bS_{-\infty}^{t}$ depends on both $\mathbf{Y}_{-\infty}^{t-1}$ through time-varying parameters and $S_t$. A main reason for allowing $Y_t$ to depend on infinitely many lags of $Y$ through time-varying parameters, instead of only $p$ lags of $Y$, is to be able to \textit{parsimoniously} capture the high persistence that economic and financial time series can exhibit, such as that in the volatility of daily stock returns. It is well known that a high-order ARCH model is typically necessary to capture this persistence, whereas a low-order generalised ARCH (GARCH) model is typically sufficient to capture it, which can be explained by the fact that the GARCH model admits an infinite-order ARCH representation (see, for instance, \cite{FrancqZakoian2019}). The same holds for the Markov-switching ARCH model and the Markov-switching GARCH model introduced by \cite{HaasMittnikPaolella2004b}, which is an example of an observation-driven state space model.\footnote{Note that there exist two types of Markov-switching GARCH models, namely the one considered by \cite{FrancqRoussignolZakoian2001}, in which the conditional distribution of $Y_t$ given $\mathbf{Y}_{-\infty}^{t-1}$ and $\bS_{-\infty}^{t}$ depends on both $\mathbf{Y}_{-\infty}^{t-1}$ and $\bS_{-\infty}^{t}$, and the one by \cite{HaasMittnikPaolella2004b}, in which it depends only on $\mathbf{Y}_{-\infty}^{t-1}$ and $S_t$. See Section \ref{anothertype} for details.} Indeed, as illustrated in Section \ref{sec:app}, a high-order Markov-switching ARCH model is needed to capture the high persistence in the volatility of daily stock returns. The Markov-switching GARCH model by \cite{HaasMittnikPaolella2004b}, however, can parsimoniously capture it, a possible explanation being that the Markov-switching GARCH model by \cite{HaasMittnikPaolella2004b} admits an infinite-order Markov-switching ARCH representation (see Section \ref{Model}). The same reasoning applies to many other observation-driven state space models, such as the score-driven state space model, in which $\mathbb{S}$ is not finite, introduced by \cite{MonachePetrellaVenditti2021} to model the relationship between annual price-dividend ratios, expected stock returns, and expected dividend growth.\footnote{For examples of observation-driven state space models in economics, see, for instance, \cite{MonachePetrellaVenditti2016} and \cite{AngeliniGorgi2018}, in which $\mathbb{S}$ is not finite. For more examples of observation-driven state space models in finance, see, for instance, \cite{HaasMittnikPaolella2004a}, \cite{Ardia2009}, \cite{BrodaHaasKrausePaolellaSteude2013}, \cite{ArdiaBluteauBoudtCatania2018}, \cite{HaasLiu2018}, \cite{BernardiCatania2019}, and \cite{Walden2019}, in which $\mathbb{S}$ is finite, and \cite{BuccheriBormettiCorsiLillo2021} and \cite{BuccheriCorsi2021}, in which $\mathbb{S}$ is not finite.}

Statistical inference for state space models and their extensions – including estimation, which is typically done by maximum likelihood estimation – is therefore of significant practical importance. The asymptotic properties of the maximum likelihood estimator (MLE) for observation-driven state space models have, however, attracted almost no attention in the literature despite the increasing popularity of such models. In this paper, we prove both consistency (Theorem \ref{TheoremConsistency}) and asymptotic normality (Theorem \ref{TheoremAsymptoticNormality}) of the MLE for an observation-driven state space model where $\mathbb{S}$ is finite, called a Markov-switching observation-driven model. To the best of our knowledge, the only other result of this kind in the literature is \cite{KandjiMisko2024}, who gave conditions under which the MLE for the widely applied Markov-switching GARCH model by \cite{HaasMittnikPaolella2004b} is consistent. As a special case of our results and to demonstrate how to apply them, we also give conditions under which the MLE for the Markov-switching GARCH model by \cite{HaasMittnikPaolella2004b} is both consistent (Theorem \ref{theo:c}) and asymptotically normal (Theorem \ref{theo:an}). The conditions under which it is consistent are almost identical to those in \cite{KandjiMisko2024}, so the contribution of this consistency result is primarily the new method of proof, as detailed in Section \ref{Examples2}. The asymptotic normality result is, however, new to the literature.

In contrast, the asymptotic properties of the MLE for both standard observation-driven models and autoregressive state space models have attracted much attention in the literature. For standard observation-driven models, in which $Y_t$ depends only on $\mathbf{Y}_{-\infty}^{t-1}$, \cite{BerkesHorvathKokoszka2003} proved consistency and asymptotic normality of the quasi-MLE (QMLE) for the GARCH model, and \cite{FrancqZakoian2004} both relaxed some of their assumptions and extended their results to the autoregressive moving average GARCH model. In a seminal paper, \cite{StraumannMikosch2006} developed a framework for proving consistency and asymptotic normality of the QMLE for general conditional heteroscedasticity models using theory for stochastic difference equations, a framework that has since become the standard for studying the asymptotic properties of the MLE for other observation-driven models as well. Indeed, \cite{BlasquesGorgiKoopmanWintenberger2018} proved consistency of the MLE for a general observation-driven model, and \cite{BlasquesVanBrummelenKoopmanLucas2022} proved consistency and asymptotic normality of the MLE for the score-driven model, both adopting the same proof strategy as \cite{StraumannMikosch2006}. Our results extend those for standard observation-driven models by allowing $Y_t$ to also depend on $S_t$.

For autoregressive state space models, in which $Y_t$ depends only on $\mathbf{Y}_{t-p}^{t-1}$ and $S_t$, \cite{FrancqRoussignol1998}, \cite{KrishnamurthyRyden1998}, and \cite{FrancqRoussignolZakoian2001} proved consistency of the MLE when $\mathbb{S}$ is finite. In another seminal paper, \cite{DoucMoulinesRyden2004} proved consistency and asymptotic normality of the MLE when $\mathbb{S}$ is compact and not necessarily finite, results that can be used to give conditions under which the MLEs for both the Markov-switching autoregressive model and the Markov-switching ARCH model are consistent and asymptotically normal. More recently, \cite{KasaharaShimotsu2019} relaxed some of their assumptions.\footnote{Consistency of the MLE for hidden Markov models was proved by \cite{Leroux1992}, and its asymptotic normality by \cite{BickelRitovRyden1998}. Local consistency and asymptotic normality of the MLE for state space models where $\mathbb{S}$ is compact were proved by \cite{JensenPetersen1999}, and global consistency of the MLE for general state space models was proved by \cite{DoucMoulinesOlssonVanHandel2011}. See \cite{DoucMoulinesOlssonVanHandel2011} for more references on the asymptotic properties of the MLE for state space models.} Our results extend those for autoregressive state space models by allowing $Y_t$ to also depend on $\mathbf{Y}_{-\infty}^{t-p-1}$. Unlike the results for both standard observation-driven models and autoregressive state space models, our results can be used to give conditions under which the MLE for the Markov-switching GARCH model, in which $Y_t$ depends on both $\mathbf{Y}_{-\infty}^{t-1}$ and $S_t$, is consistent and asymptotically normal.

Since $Y_t$ depends on both $\mathbf{Y}_{-\infty}^{t-1}$ through the time-varying parameters and $S_t$, where neither $\mathbf{Y}_{-\infty}^{0}$ nor $S_t$ is observed, the feasible likelihood function – and thus the MLE – for the Markov-switching observation-driven model depends on initialisations of both the time-varying parameters and the filter. In contrast, the feasible likelihood function for standard observation-driven models depends only on initialisations of the time-varying parameters, since $Y_t$ depends only on $\mathbf{Y}_{-\infty}^{t-1}$ through the time-varying parameters, where $\mathbf{Y}_{-\infty}^{0}$ is unobserved; and that for autoregressive state space models depends only on an initialisation of the filter, since $Y_t$ depends only on $\mathbf{Y}_{t-p}^{t-1}$ and $S_t$, where $\mathbf{Y}_{t-p}^{t-1}$ is observed but $S_t$ is not.\footnote{Note that if, as in \cite{DoucMoulinesRyden2004} and \cite{KasaharaShimotsu2019}, the index of the first observation is set equal to $1-p$, then $\mathbf{Y}_{t-p}^{t-1}$ is observed for all $t \in \mathbb{N}$.} To prove that the MLE for the Markov-switching observation-driven model is consistent and asymptotically normal, it is therefore crucial to prove that both the time-varying parameters and the filter forget these initialisations asymptotically.

We first prove that the time-varying parameters forget their initialisations asymptotically (Lemma \ref{LemmaInvertibilityX}). Since they are non-Markov, we adopt the same proof strategy as \cite{StraumannMikosch2006} and use the theory for stationary and ergodic stochastic difference equations developed by \cite{Bougerol1993}.

Using the fact that the time-varying parameters forget their initialisations asymptotically, we then prove that the filter forgets its initialisation asymptotically (Lemma \ref{LemmaInvertibilityFilter}), but, to do so, we adopt a different proof strategy from \cite{DoucMoulinesRyden2004} and \cite{KasaharaShimotsu2019}. For the autoregressive state space model, \cite{DoucMoulinesRyden2004} and \cite{KasaharaShimotsu2019} first proved that, for all $m,n \in \mathbb{Z}$ such that $m \leq n$, $(S_t)_{m \leq t \leq n}$ given $(Y_\tau)_{m-p+1 \leq \tau \leq n}$ is a non-homogeneous Markov chain and then, using theory for Markov chains, that it forgets its initialisation asymptotically, which suffices to conclude that the filter forgets its initialisation asymptotically. For the Markov-switching observation-driven model, we likewise prove, and use, that, for all $n \in \mathbb{Z}$, $(S_t)_{t \leq n}$ given $(Y_\tau)_{\tau \leq n}$ is a non-homogeneous Markov chain, and, using theory for Markov chains, one can also prove that it forgets its initialisation asymptotically if it is initialised. If the filter did not depend on the initialisations of the time-varying parameters, this would suffice to conclude that the filter forgets its initialisation asymptotically, as for the autoregressive state space model, but since it does, this is not the case. To prove that the filter forgets its initialisation asymptotically, we therefore instead work directly with the filter, which satisfies a stochastic difference equation, and use a generalisation of the theory for asymptotically stationary and ergodic stochastic difference equations that \cite{StraumannMikosch2006} developed to prove that the derivatives of their time-varying parameters forget their initialisations asymptotically (see \cite{Krabbe2025}).\footnote{It is worth mentioning that, to prove that the filter forgets its initialisation asymptotically for the autoregressive state space model, one can also work directly with the filter and use the theory for stationary and ergodic stochastic difference equations developed by \cite{Bougerol1993}.}

The rest of the paper is organised as follows. Section \ref{Model} introduces the Markov-switching observation-driven model and gives examples of it. Section \ref{CAN} studies the asymptotic properties of the MLE for the Markov-switching observation-driven model, and Section \ref{Examples2} studies both the asymptotic and finite-sample properties of the MLE for the Markov-switching GARCH model by \cite{HaasMittnikPaolella2004b}, the latter in a small Monte Carlo simulation study. Section \ref{sec:app} presents an application of the Markov-switching GARCH model by \cite{HaasMittnikPaolella2004b}, illustrating that it can parsimoniously capture the high persistence in the volatility of daily stock returns, in contrast to the Markov-switching ARCH model. Finally, Section \ref{Conclusion} concludes. The Appendix collects the proofs, and the Supplementary Appendix collects the auxiliary lemmas used in the proofs, denoted Lemma \ref{LemmaInvertibility1}, Lemma \ref{LemmaInvertibilityA}, and so on.

A word on the notation. $\mathbb{N}_0 := \{0\} \cup \mathbb{N}$. For a scalar $x \in (0,\infty)$, $\log^{+} x := \max(\log x,0)$ and $\log^{-} x := -\min(\log x,0)$. For a vector $\bx \in \mathbb{R}^n$, $[\bx]_i$ denotes the $i$'th element of $\bx$, $|| \bx ||_{p} := (\sum_{i=1}^{n} | [\bx]_i |^p)^{1/p}, p \geq 1$ is the $p$-norm of $\bx$, and $\odot$ denotes the Hadamard product. Similarly, for a matrix $\bX \in \mathbb{R}^{n \times m}$, $[\bX]_{ij}$ denotes the $(i,j)$'th element of $\bX$, $|| \bX ||_{p,p} := (\sum_{i=1}^{n} \sum_{j=1}^{m} | [\bX]_{ij} |^p)^{1/p}, p \geq 1$ is the $p$-norm of $\bX$, and $\otimes$ denotes the Kronecker product. Unless otherwise stated, for a function $f : \mathbb{R}^n \rightarrow \mathbb{R}$, $\nabla_{\bx} f(\bx)$ is the first-order derivative of $f$ with respect to (w.r.t.) $\bx$, and $\nabla_{\bx^{\prime}} f(\bx) := (\nabla_{\bx} f(\bx))^{\prime}$. Moreover, $\nabla_{\bx \bx} f(\bx)$ is the second-order derivative of $f$ w.r.t. $\bx$. If $f : \mathbb{R} \times \mathbb{R}^{n} \rightarrow \mathbb{R}$ and $G : \mathbb{R}^n \rightarrow \mathbb{R}$, then $\partial_g f(G(\bx),\bx)$ and $\partial_{\bx} f(G(\bx),\bx)$ are the first-order derivatives of $f$ w.r.t. the first and second arguments, respectively, such that $\nabla_{\bx} f(G(\bx),\bx) = \partial_g f(G(\bx),\bx) \nabla_{\bx} G(\bx) + \partial_{\bx} f(G(\bx),\bx)$, with similar notation for $\partial_{gg} f(G(\bx),\bx)$, $\partial_{\bx g} f(G(\bx),\bx)$, and $\partial_{\bx \bx} f(G(\bx),\bx)$. $\overset{\mathbb{P}}{\rightarrow}$ denotes convergence in probability, and $\overset{a.s.}{\rightarrow}$ denotes almost sure convergence. A sequence of random matrices $(\bX_t)_{t \in \mathbb{N}}$ is said to converge to zero exponentially fast almost surely (e.a.s.) if there exists a $\gamma > 1$ such that $\gamma^t || \bX_t ||_{p,p} \overset{a.s.}{\rightarrow} 0$ as $t \rightarrow \infty$. $\overset{e.a.s.}{\rightarrow}$ denotes such convergence, and $\overset{d}{\rightarrow}$ denotes convergence in distribution. Finally, $C \in (0,\infty)$ denotes a generic constant that may change from line to line.

\section{The Markov-Switching Observation-Driven Model} \label{Model}

We first introduce the Markov-switching observation-driven model and then give examples of it. A Markov-switching observation-driven model is a stochastic process $((S_t,Y_t))_{t \in \mathbb{Z}}$ where $(S_t)_{t \in \mathbb{Z}}$ is an unobserved homogeneous Markov chain taking values in $\{1,...,J\}$ with transition probabilities
\begin{equation*}
	p_{ij} := \mathbb{P}(S_{t} = j \mid S_{t-1} = i), \quad i,j \in \{1,...,J\}.
\end{equation*}
Moreover, $(Y_t)_{t \in \mathbb{Z}}$ is an observed stochastic process taking values in $\mathcal{Y} \subseteq \mathbb{R}$ such that the conditional distribution of $Y_t$ given $\bY_{-\infty}^{t-1}$ and $\bS_{-\infty}^{t}$ depends only on $\bY_{-\infty}^{t-1}$ and $S_t$ as follows
\begin{equation}
	Y_t \mid (\bY_{-\infty}^{t-1},S_t = j) \sim \mathcal{D}_j (X_{j,t},\bupsilon_j), \label{eq:MSODM1}
\end{equation}
where $\mathcal{D}_j$ is a distribution, $X_{j,t}$ is a time-varying parameter taking values in a complete set $\mathcal{X}_j \subseteq \mathbb{R}$ given by
\begin{equation}
	X_{j,t+1} = \phi_{j} (Y_{t},X_{j,t};\bupsilon_{j}), \label{eq:MSODM2}
\end{equation}
and $\bupsilon_{j}$ is a vector of constant parameters taking values in a set $\bUpsilon_{j} \subseteq \mathbb{R}^{d_{j}}$. Figure \ref{fig:1} illustrates the dependence structure of the Markov-switching observation-driven model.

\begin{figure}[htbp]
	
	\centering
	
	\begin{tikzpicture}[minimum size = 1.25cm]
		
		\node[left] at (-3,-1.5) {$\dots$};
		\node[draw, circle] (S0) at (0,-1.5) {$S_{t-1}$};
		\node[draw, circle] (S1) at (3,-1.5) {$S_t$};
		\node[right] at (6,-1.5) {$\dots$};
		
		\node[left] at (-3,1.5) {$\dots$};
		\node[draw, circle] (Y0) at (0,1.5) {$Y_{t-1}$};
		\node[draw, circle] (Y1) at (3,1.5) {$Y_t$};
		\node[right] at (6,1.5) {$\dots$};
		
		\draw[->] (-3,-1.5) -- (S0);
		\draw[->] (S0) -- (S1);
		\draw[->] (S1) -- (6,-1.5);
		
		\draw[->] (-3,1.5) -- (Y0);
		\draw[->] (-3,1.5) to [out=30, in=150] (Y1);
		\draw[->] (-3,1.5) to[out=30, in=150] (6,1.5);
		\draw[->] (Y0) -- (Y1);
		\draw[->] (Y0) to[out=30, in=150] (6,1.5);
		\draw[->] (Y1) -- (6,1.5);
		
		\draw[->] (S0) -- (Y0);
		\draw[->] (S1) -- (Y1);
		
	\end{tikzpicture}
	
	\captionsetup{font=footnotesize}
	
	\caption{The dependence structure of the Markov-switching observation-driven model in Equations \eqref{eq:MSODM1} and \eqref{eq:MSODM2}.}
	
	\label{fig:1}
	
\end{figure}

Note that the first-order Markov-switching autoregressive model and the hidden Markov model are both special cases of the Markov-switching observation-driven model. Indeed, the Markov-switching observation-driven model reduces to the first-order Markov-switching autoregressive model if $\phi_{j} (Y_{t},X_{j,t};\bupsilon_{j}) = \phi_{j} (Y_{t};\bupsilon_{j})$ for all $j \in \{1,...,J\}$ and to the hidden Markov model if $X_{j,t} \equiv X_{j}$ for all $j \in \{1,...,J\}$.\footnote{All results in this paper continue to hold in the case where $X_{j,t+1} = \phi_{j} (Y_{t},...,Y_{t-p+1},X_{j,t};\bupsilon_{j})$ – in which case the Markov-switching observation-driven model reduces to the $p$'th-order Markov-switching autoregressive model if $\phi_{j} (Y_{t},...,Y_{t-p+1},X_{j,t};\bupsilon_{j}) = \phi_{j} (Y_{t},...,Y_{t-p+1};\bupsilon_{j})$ for all $j \in \{1,...,J\}$ – but, to ease the exposition, we consider the case where $X_{j,t+1} = \phi_{j} (Y_{t},X_{j,t};\bupsilon_{j})$.}

\subsection{Examples}

\begin{example} \label{ex:1}
	
	An example of a Markov-switching observation-driven model is
	\begin{equation}
		Y_t = \sum_{j=1}^{J} 1_{\{S_t = j\}} (X_{j,t} + \sigma_{j} \varepsilon_t), \label{eq:ARINF}
	\end{equation}
	where $(\varepsilon_t)_{t \in \mathbb{Z}}$ is a sequence of independent standard normal distributed random variables, independent of $(S_t)_{t \in \mathbb{Z}}$, and
	\begin{equation*}
		X_{j,t} = \omega_{j} + \alpha_{j} Y_{t-1} + \beta_{j} X_{j,t-1},
	\end{equation*}
	where $\omega_{j} \in \mathbb{R}$, $\alpha_{j} \in \mathbb{R}$, $\beta_{j} \in \mathbb{R}$, and $\sigma_{j}^2 > 0$. Here, $\mathcal{Y} = \mathbb{R}$, $\mathcal{D}_{j}$ is the normal distribution with mean $X_{j,t}$ and variance $\sigma_{j}^{2}$, $\mathcal{X}_{j} = \mathbb{R}$, $\phi_{j} (y,x_j;\bupsilon_j) = [\bupsilon_j]_1 + [\bupsilon_j]_2 y + [\bupsilon_j]_3 x_j$, $\bupsilon_{j} = (\omega_{j},\alpha_{j},\beta_{j},\sigma_{j}^{2})^{\prime}$, and $\bUpsilon_{j} = \mathbb{R} \times \mathbb{R} \times \mathbb{R} \times (0,\infty)$.
	
	The Markov-switching observation-driven model in Equation \eqref{eq:ARINF} admits an infinite-order Markov-switching autoregressive representation and is thus a parsimonious alternative to a high-order Markov-switching autoregressive model. Indeed, it can be shown that if $(Y_t)_{t \in \mathbb{Z}}$ is stationary and ergodic with $\mathbb{E} [\log^{+}|Y_t|] < \infty$ and $|\beta_j| < 1$ for all $j \in \{1,...,J\}$, then
	\begin{equation*}
		X_{j,t} = \frac{\omega_j}{1-\beta_j} + \sum_{i=1}^{\infty} \alpha_j \beta_{j}^{i-1} Y_{t-i}
	\end{equation*}
	for all $j \in \{1,...,J\}$. The Markov-switching observation-driven model in Equation \eqref{eq:ARINF} can thus be written as
	\begin{equation*}
		Y_t = a_{S_t} + \sum_{i=1}^{\infty} b_{S_t}^{(i)} Y_{t-i} + \sigma_{S_t} \varepsilon_t,
	\end{equation*}
	where
	\begin{equation*}
		a_{S_t} := \frac{\omega_{S_t}}{1-\beta_{S_t}} \quad \text{and} \quad b_{S_t}^{(i)} := \alpha_{S_t} \beta_{S_t}^{i-1},
	\end{equation*}
	which is an infinite-order Markov-switching autoregressive model. Potential applications of the Markov-switching observation-driven model in Equation \eqref{eq:ARINF} are thus economic growth and monthly, quarterly, and yearly financial returns, which the Markov-switching autoregressive model has been used to model.
	
\end{example}

\begin{example} \label{ex:2}
	
	The Markov-switching GARCH model by \citet{HaasMittnikPaolella2004b}, which can be used to model daily financial returns, is given by
	\begin{equation*}
		Y_t = \sum_{j=1}^{J} 1_{\{S_t = j\}} \sqrt{X_{j,t}} \varepsilon_t,
	\end{equation*}
	where $(\varepsilon_t)_{t \in \mathbb{Z}}$ is a sequence of independent standard normal distributed random variables, independent of $(S_t)_{t \in \mathbb{Z}}$, and
	\begin{equation*}
		X_{j,t} = \omega_{j} + \alpha_{j} Y_{t-1}^2 + \beta_{j} X_{j,t-1},
	\end{equation*}
	where $\omega_{j} > 0$, $\alpha_{j} \geq 0$, and $\beta_{j} \geq 0$.\footnote{The Markov-switching GARCH model reduces to the mixture GARCH model by \citet{HaasMittnikPaolella2004a} if $(S_t)_{t \in \mathbb{Z}}$ is an independent and identically distributed (i.i.d.) chain, that is, if $p_{1j} = \cdots = p_{Jj}$ for all $j \in \{1,...,J\}$.} This is also an example of a Markov-switching observation-driven model where $\mathcal{Y} = \mathbb{R}$, $\mathcal{D}_{j}$ is the normal distribution with mean zero and variance $X_{j,t}$, $\mathcal{X}_j = \left[ 0, \infty \right)$, $\phi_{j} (y,x_j;\bupsilon_j) = [\bupsilon_j]_1 + [\bupsilon_j]_2 y^2 + [\bupsilon_j]_3 x_j$, $\bupsilon_{j} = (\omega_{j},\alpha_{j},\beta_{j})^{\prime}$, and $\bUpsilon_{j} = (0,\infty) \times [0,\infty) \times [0,\infty)$. 
	
	Like the GARCH model admits an infinite-order ARCH representation (see Section 2.3.2 in \cite{FrancqZakoian2019}), the Markov-switching GARCH model admits an infinite-order Markov-switching ARCH representation. Indeed, if $(Y_t)_{t \in \mathbb{Z}}$ is stationary and ergodic with $\mathbb{E} [\log^{+}|Y_t|] < \infty$ and $\beta_j < 1$, then
	\begin{equation*}
		X_{j,t} = \frac{\omega_j}{1-\beta_j} + \sum_{i=1}^{\infty} \alpha_j \beta_{j}^{i-1} Y_{t-i}^2,
	\end{equation*}
	so the Markov-switching GARCH model can be written as
	\begin{equation*}
		Y_t = \sqrt{a_{S_t} + \sum_{i=1}^{\infty} b_{S_t}^{(i)} Y_{t-i}^2} \varepsilon_t,
	\end{equation*}
	where
	\begin{equation*}
		a_{S_t} := \frac{\omega_{S_t}}{1-\beta_{S_t}} \quad \text{and} \quad b_{S_t}^{(i)} := \alpha_{S_t} \beta_{S_t}^{i-1},
	\end{equation*}
	which is an infinite-order Markov-switching ARCH model. The Markov-switching GARCH model is thus a parsimonious alternative to a high-order Markov-switching ARCH model, like the GARCH model is a parsimonious alternative to a high-order ARCH model.
	
	The Markov-switching GARCH model is not the only Markov-switching conditional heteroscedasticity model that the Markov-switching observation-driven model nests. Consider, for instance, the Markov-switching asymmetric power GARCH model given by
	\begin{equation*}
		Y_t = \sum_{j=1}^{J} 1_{\{S_t = j\}} X_{j,t}^{\frac{1}{\delta}} \varepsilon_t, \quad \delta > 0,
	\end{equation*}
	where $(\varepsilon_t)_{t \in \mathbb{Z}}$ is as above and
	\begin{equation*}
		X_{j,t} = \omega_{j} + \alpha_{j} \left( \left| Y_{t-1} \right| - \xi_{j} Y_{t-1} \right)^{\delta}  + \beta_{j} X_{j,t-1},
	\end{equation*}
	where $\omega_{j} > 0$, $\alpha_{j} \geq 0$, $|\xi_{j}| \leq 1$, and $\beta_{j} \geq 0$, which reduces to the Markov-switching Glosten-Jagannathan-Runkle GARCH model if $\delta = 2$, the Markov-switching threshold GARCH model if $\delta = 1$, and the standard Markov-switching GARCH model above if $\delta = 2$ and $\xi_j = 0$. This is also an example of a Markov-switching observation-driven model where $\mathcal{Y} = \mathbb{R}$, $\mathcal{D}_{j}$ is the normal distribution with mean zero and variance $X_{j,t}^{\frac{2}{\delta}}$, $\mathcal{X}_j = \left[ 0, \infty \right)$, $\phi_{j} (y,x_j;\bupsilon_j) = [\bupsilon_j]_1 + [\bupsilon_j]_2 (|y| - [\bupsilon_j]_3 y)^{\delta} + [\bupsilon_j]_4 x_j$, $\bupsilon_{j} = (\omega_{j},\alpha_{j},\xi_j,\beta_{j})^{\prime}$, and $\bUpsilon_{j} = (0,\infty) \times [0,\infty) \times [-1,1] \times [0,\infty)$.

\end{example}

\begin{example} \label{ex:3}
	
	Let $y \mapsto F(y;x,\tilde{\bupsilon})$ be a cumulative distribution function (cdf) with support $\mathcal{N} \subseteq [0,\infty)$, for instance, $\mathcal{N} = \{0,1\}$, $\mathcal{N} = \mathbb{N}$, or $\mathcal{N} = [0,1]$, indexed by the mean $x$ and a vector of parameters $\tilde{\bupsilon}$, and assume that
	\begin{equation*}
		x \leq x^{*} \quad \Rightarrow \quad F^{-} (u;x,\tilde{\bupsilon}) \leq F^{-} (u;x^{*},\tilde{\bupsilon})
	\end{equation*}
	for all $u \in (0,1)$ where $F^{-} (u;x,\tilde{\bupsilon}) = \inf \{ y \in \mathcal{N} : F(y;x,\tilde{\bupsilon}) \geq u\}$.
	
	The (present-regime dependent) Markov-switching positive linear conditional mean model considered by \citet{AknoucheFrancq2022} is given by
	\begin{equation*}
		Y_t = \sum_{j=1}^{J} 1_{\{S_t = j\}} F^{-}_{j} (U_t;X_{j,t},\tilde{\bupsilon}_{j}),
	\end{equation*}
	where $(U_t)_{t \in \mathbb{Z}}$ is a sequence of independent uniform distributed random variables on $[0,1]$, independent of $(S_t)_{t \in \mathbb{Z}}$, and
	\begin{equation*}
		X_{j,t} = \omega_{j} + \alpha_{j} Y_{t-1} + \beta_{j} X_{j,t-1},
	\end{equation*}
	where $\omega_{j} > 0$, $\alpha_{j} \geq 0$, $\beta_{j} \geq 0$, and $\tilde{\bupsilon}_{j} \in \tilde{\bUpsilon}_{j}$. This is another example of a Markov-switching observation-driven model where $\mathcal{Y} = \mathcal{N}$, $F_{j}$ is the cdf of $\mathcal{D}_{j}$, $\mathcal{X}_j = \left[ 0, \infty \right)$, $\phi_{j} (y,x_j;\bupsilon_j) = [\bupsilon_j]_1 + [\bupsilon_j]_2 y + [\bupsilon_j]_3 x_j$, $\bupsilon_{j} = (\omega_{j},\alpha_{j},\beta_{j},\tilde{\bupsilon}_{j})^{\prime}$, and $\bUpsilon_{j} = (0,\infty) \times [0,\infty) \times [0,\infty) \times \tilde{\bUpsilon}_{j}$. Special cases of the Markov-switching positive linear conditional mean model include the Markov-switching autoregressive conditional duration (ACD) model, which can be used to model the time between trades, and the Markov-switching integer-valued GARCH (INGARCH) model, which can be used to model the number of trades.
	
\end{example}

\begin{example} \label{ex:4}
	
	A random variable $Y$ with support $[-\pi,\pi]$ is said to be von Mises distributed with location $x \in [-\pi,\pi]$ and concentration $\tilde{\upsilon} > 0$ if the probability density function (pdf) of $Y$ is given by
	\begin{equation*}
		f(y;x,\tilde{\upsilon}) = \frac{1}{2 \pi I_0(\tilde{\upsilon})} \exp(\tilde{\upsilon} \cos(y - x)), \quad y \in [-\pi,\pi],
	\end{equation*}
	where $I_0(\cdot)$ is the modified Bessel function of the first kind of order 0.
	
	Yet another example of a Markov-switching observation-driven model is the Markov-switching score-driven model
	\begin{equation*}
		Y_t = \sum_{j=1}^{J} 1_{\{S_t = j\}} (( X_{j,t} + \varepsilon_{j,t} ) \textup{mod} \, (2 \pi) - \pi),
	\end{equation*}
	where, for all $j \in \{1,...,J\}$, $(\varepsilon_{j,t})_{t \in \mathbb{Z}}$ is a sequence of independent von Mises distributed random variables with location $0$ and concentration $\tilde{\upsilon}_j > 0$, which is independent of $(\varepsilon_{i,t})_{t \in \mathbb{Z}}$ for all $i \in \{1,...,J\}$ such that $i \neq j$ and $(S_t)_{t \in \mathbb{Z}}$, and
	\begin{equation*}
		X_{j,t} = \omega_{j} + \alpha_{j} \tilde{\upsilon}_{j} \sin \left(Y_{t-1} - X_{j,t-1}\right) + \beta_{j} X_{j,t-1},
	\end{equation*}
	where $\omega_{j} \in \mathbb{R}$, $\alpha_{j} \in \mathbb{R}$, and $\beta_{j} \in \mathbb{R}$. The model is similar to the regime-switching score-driven model introduced by \cite{HarveyPalumbo2023} to model wind direction, which is a generalisation of the score-driven model introduced by \cite{HarveyHurnPalumboThiele2024} with the same purpose. The Markov-switching score-driven model is a Markov-switching observation-driven model where $\mathcal{Y} = [-\pi,\pi]$, $\mathcal{D}_{j}$ is the von Mises distribution with location $X_{j,t}$ and concentration $\tilde{\upsilon}_{j}$, $\mathcal{X}_j = \mathbb{R}$, $\phi_{j} (y,x_j;\bupsilon_j) = [\bupsilon_j]_1 + [\bupsilon_j]_2 [\bupsilon_j]_4 \sin (y-x_j) + [\bupsilon_j]_3 x_j$, $\bupsilon_{j} = (\omega_{j},\alpha_{j},\beta_{j},\tilde{\upsilon}_{j})^{\prime}$, and $\bUpsilon_{j} = \mathbb{R} \times \mathbb{R} \times \mathbb{R} \times (0,\infty)$.
	
\end{example}

\subsection{Another Type of Markov-Switching Observation-Driven Model} \label{anothertype}

It is important to mention that another type of Markov-switching observation-driven model is obtained if $(Y_t)_{t \in \mathbb{Z}}$ is an observed stochastic process such that the conditional distribution of $Y_t$ given $\bY_{-\infty}^{t-1}$ and $\bS_{-\infty}^{t}$ depends on both $\bY_{-\infty}^{t-1}$ and $\bS_{-\infty}^{t}$ as follows
\begin{equation}
	Y_t \mid (\bY_{-\infty}^{t-1},\bS_{-\infty}^{t}) \sim \mathcal{D}_{S_t} (X_{S_t,t},\bupsilon_{S_t}), \label{eq:MSODM3}
\end{equation}
where $\mathcal{D}_{S_t}$ is a distribution, $X_{S_t,t}$ is a time-varying parameter given by
\begin{equation}
	X_{S_{t+1},t+1} = \phi_{S_{t+1}} (Y_{t},X_{S_t,t};\bupsilon_{S_{t+1}}), \label{eq:MSODM4}
\end{equation}
and $\bupsilon_{S_t}$ is a vector of constant parameters, such as in the Markov-switching GARCH model considered by \cite{FrancqRoussignolZakoian2001}. Figure \ref{fig:2} illustrates the dependence structure of such a Markov-switching observation-driven model.

\begin{figure}[htbp]
	
	\centering
	
	\begin{tikzpicture}[minimum size = 1.25cm]
		
		\node[left] at (-3,-1.5) {$\dots$};
		\node[draw, circle] (S0) at (0,-1.5) {$S_{t-1}$};
		\node[draw, circle] (S1) at (3,-1.5) {$S_t$};
		\node[right] at (6,-1.5) {$\dots$};
		
		\node[left] at (-3,1.5) {$\dots$};
		\node[draw, circle] (Y0) at (0,1.5) {$Y_{t-1}$};
		\node[draw, circle] (Y1) at (3,1.5) {$Y_t$};
		\node[right] at (6,1.5) {$\dots$};
		
		\draw[->] (-3,-1.5) -- (S0);
		\draw[->] (S0) -- (S1);
		\draw[->] (S1) -- (6,-1.5);
		
		\draw[->] (-3,1.5) -- (Y0);
		\draw[->] (-3,1.5) to [out=30, in=150] (Y1);
		\draw[->] (-3,1.5) to[out=30, in=150] (6,1.5);
		\draw[->] (Y0) -- (Y1);
		\draw[->] (Y0) to[out=30, in=150] (6,1.5);
		\draw[->] (Y1) -- (6,1.5);
		
		\draw[->] (-3,-1.5) -- (Y0);
		\draw[->] (-3,-1.5) -- (Y1);
		\draw[->] (-3,-1.5) -- (6,1.5);
		\draw[->] (S0) -- (Y0);
		\draw[->] (S0) -- (Y1);
		\draw[->] (S0) -- (6,1.5);
		\draw[->] (S1) -- (Y1);
		\draw[->] (S1) -- (6,1.5);
		
	\end{tikzpicture}
	
	\captionsetup{font=footnotesize}
	
	\caption{The dependence structure of the Markov-switching observation-driven model in Equations \eqref{eq:MSODM3} and \eqref{eq:MSODM4}.}
	
	\label{fig:2}
	
\end{figure}

In contrast to the Markov-switching observation-driven model in Equations \eqref{eq:MSODM1} and \eqref{eq:MSODM2}, however, the Markov-switching observation-driven model in Equations \eqref{eq:MSODM3} and \eqref{eq:MSODM4} suffers from the so-called path-dependence problem, which limits its usefulness. We stress that we consider the Markov-switching observation-driven model in Equations \eqref{eq:MSODM1} and \eqref{eq:MSODM2} and \textit{not} the Markov-switching observation-driven model in Equations \eqref{eq:MSODM3} and \eqref{eq:MSODM4}.

\section{Asymptotic Properties of the Maximum Likelihood Estimator} \label{CAN}

We now study the asymptotic properties of the MLE for the Markov-switching observation-driven model. Assume that a sample $(Y_t)_{t=1}^{T}$ from the Markov-switching observation-driven model $((S_t,Y_t))_{t \in \mathbb{Z}}$ with $\btheta = \btheta_0$ is available. Here,
\begin{equation*}
	\btheta := (p_{ij}, i = 1,...,J, j = 1,...,J-1, \bupsilon_j^{\prime}, j = 1,...,J)^{\prime}
\end{equation*}
is the parameter vector since $\sum_{j=1}^{J} p_{ij} = 1$ for all $i \in \{1,...,J\}$ and 
\begin{equation*}
	\bTheta \subset \left\lbrace \btheta \in \mathbb{R}^d : p_{ij} > 0, i = 1,...,J, j = 1,...,J-1, \sum_{j=1}^{J-1} p_{ij} < 1, i = 1,...,J, \bupsilon_j \in \bUpsilon_{j}, j = 1,...,J \right\rbrace
\end{equation*}
is the parameter space where $d := J(J-1) + \sum_{j=1}^{J} d_j$. Then, the MLE $\hat{\btheta}_T$ of the true parameter $\btheta_0$ is given by
\begin{equation*}
	\hat{\btheta}_T = \underset{\btheta \in \bTheta}{\arg \max} \, \hat{L}_T (\btheta),
\end{equation*}
where $\hat{L}_T (\btheta)$ is the \textit{initialised} and thus non-stationary average log-likelihood function given by
\begin{equation*}
	\hat{L}_T (\btheta) = \frac{1}{T} \sum_{t=1}^{T} \log \hat{f} (Y_t;\btheta)
\end{equation*}
with
\begin{equation*}
	\hat{f} (Y_t;\btheta) = \sum_{j=1}^{J} \hat{\pi}_{j,t \mid t-1} (\btheta) f_{j} (Y_t;\hat{X}_{j,t} (\bupsilon_j),\bupsilon_j).
\end{equation*}
Here, for each $j \in \{1,...,J\}$, $(\hat{X}_{j,t} (\bupsilon_j))_{t \in \mathbb{N}}$ is the \textit{initialised} time-varying parameter given by
\begin{equation}
	\hat{X}_{j,t+1} (\bupsilon_j) = \phi_j (Y_{t},\hat{X}_{j,t} (\bupsilon_j);\bupsilon_{j}) \label{eq:tvp}
\end{equation}
for some initialisation $\hat{X}_{j,1} (\bupsilon_j) \in \mathcal{X}_{j}$. Moreover, $f_j$ is the pdf of $\mathcal{D}_j$. Finally, $(\hat{\bpi}_{t \mid t-1} (\btheta))_{t \in \mathbb{N}}$, with $\hat{\bpi}_{t \mid s} (\btheta) := (\hat{\pi}_{1,t \mid s} (\btheta),...,\hat{\pi}_{J,t \mid s} (\btheta))^{\prime}$ and $\hat{\pi}_{j,t \mid s} (\btheta) := \hat{\mathbb{P}}_{\btheta} (S_t = j \mid \bY_{1}^{s}), j \in \{1,...,J\}$, is the \textit{initialised} predictor given by
\begin{equation}
	\begin{aligned}
		\hat{\bpi}_{t+1 \mid t} (\btheta) &= \mathbf{P}^{\prime} \hat{\bpi}_{t \mid t} (\btheta) \\
		\hat{\bpi}_{t \mid t} (\btheta) &= \hat{\mathbf{F}}_t (\hat{\bpi}_{t \mid t-1} (\btheta);\btheta) \hat{\bpi}_{t \mid t-1} (\btheta)
	\end{aligned} \label{eq:p&f}
\end{equation}
for some initialisation $\hat{\bpi}_{0 \mid 0} (\btheta) \in \mathcal{S}$ with $\mathcal{S} := \{\bx \in \mathbb{R}^{J} : x_j \geq 0, j=1,...,J, \sum_{j=1}^{J} x_j = 1 \}$ where $\mathbf{P}$ is the transition probability matrix given by
\begin{equation*}
	\mathbf{P}
	:=
	\begin{bmatrix}
		p_{11} & \cdots & p_{1J} \\
		\vdots & \ddots & \vdots \\
		p_{J1} & \cdots & p_{JJ} \\
	\end{bmatrix}
\end{equation*}
and $\hat{\mathbf{F}}_t (\bs;\btheta), \bs \in \mathcal{S}$ is a diagonal matrix with generic element
\begin{equation*}
	[\hat{\mathbf{F}}_t (\bs;\btheta)]_{ii} := \frac{f_{i} (Y_t;\hat{X}_{i,t} (\bupsilon_i),\bupsilon_{i})}{\sum_{k=1}^{J} s_k f_{k} (Y_t;\hat{X}_{k,t} (\bupsilon_k),\bupsilon_{k})}, \quad i \in \{1,...,J\}.
\end{equation*}

\subsection{Consistency} \label{Consistency}

Consistency follows from classical arguments if the \textit{initialised} and thus non-stationary average log-likelihood function $\hat{L}_T (\btheta)$ converges uniformly a.s. to a function, say, $L(\btheta)$, which is uniquely maximised at $\btheta_0$. 

We show that $\hat{L}_T (\btheta)$ converges uniformly a.s. to $L(\btheta)$ in two steps, the first step being the main difficulty and hence focus. In the first step, we show that the difference between $\hat{L}_T (\btheta)$ and $L_T (\btheta)$ converges uniformly a.s. to zero where $L_T (\btheta)$ is the \textit{non-initialised}, stationary, and ergodic average log-likelihood function given by 
\begin{equation*}
	L_T (\btheta) = \frac{1}{T} \sum_{t=1}^{T} \log f (Y_t;\btheta)
\end{equation*}
with
\begin{equation*}
	f (Y_t;\btheta) = \sum_{j=1}^{J} \pi_{j,t \mid t-1} (\btheta) f_{j} (Y_t;X_{j,t} (\bupsilon_j),\bupsilon_j).
\end{equation*}
Here, for each $j \in \{1,...,J\}$, $(X_{j,t} (\bupsilon_j))_{t \in \mathbb{Z}}$ is the \textit{non-initialised} time-varying parameter given in Lemma \ref{LemmaInvertibilityX} below, and $(\bpi_{t \mid t-1} (\btheta))_{t \in \mathbb{Z}}$, with $\bpi_{t \mid s} (\btheta) := (\pi_{1,t \mid s} (\btheta),...,\pi_{J,t \mid s} (\btheta))^{\prime}$ and $\pi_{j,t \mid s} (\btheta) := \mathbb{P}_{\btheta} (S_t = j \mid \bY_{-\infty}^{s}), j \in \{1,...,J\}$, is the \textit{non-initialised} predictor given in Corollary \ref{CorollaryInvertibilityPrediction} below. We show this by showing that the time-varying parameters and the predictor forget their initialisations asymptotically, that is, that, for each $j \in \{1,...,J\}$, the difference between $(\hat{X}_{j,t} (\bupsilon_j))_{t \in \mathbb{N}}$ and $(X_{j,t} (\bupsilon_j))_{t \in \mathbb{Z}}$ and the one between $(\hat{\bpi}_{t \mid t-1} (\btheta))_{t \in \mathbb{N}}$ and $(\bpi_{t \mid t-1} (\btheta))_{t \in \mathbb{Z}}$ converge uniformly e.a.s. to zero. In the second step, we then show that $L_T (\btheta)$ converges uniformly a.s. to
\begin{equation*}
	L(\btheta) = \mathbb{E} [\log f (Y_t;\btheta)]
\end{equation*}
by applying a uniform law of large numbers. Finally, we show that $L(\btheta)$ is uniquely maximised at $\btheta_0$ by using standard arguments.

We assume the following where, for each $j \in \{1,...,J\}$,
\begin{equation*}
	\bvartheta_j := (x_j,\bupsilon_j^{\prime})^{\prime}.
\end{equation*}
\begin{assumption} \label{AssumptionY}
	$((S_t,Y_t))_{t \in \mathbb{Z}}$ is stationary and ergodic.
\end{assumption}
\begin{assumption} \label{AssumptionTheta1}
	$\bTheta$ is compact.
\end{assumption}
\begin{assumption} \label{AssumptionF1}
	For  each $j \in \{1,...,J\}$,
	\begin{enumerate} [(i)]
		\item $\bvartheta_j \mapsto f_{j} (y;x_j,\bupsilon_j)$ is continuous for all $y \in \mathcal{Y}$ and
		\item $x_j \mapsto f_{j} (y;x_j,\bupsilon_j)$ is differentiable for all $y \in \mathcal{Y}$ and $\bupsilon_j \in \bUpsilon_{j}$.
	\end{enumerate}
\end{assumption}
\begin{assumption} \label{AssumptionPhi1}
	For each $j \in \{1,...,J\}$,
	\begin{enumerate} [(i)]
		\item $\bvartheta_j \mapsto \phi_{j} (y,x_j;\bupsilon_j)$ is continuous for all $y \in \mathcal{Y}$ and
		\item $x_j \mapsto \phi_{j} (y,x_j;\bupsilon_j)$ is differentiable for all $y \in \mathcal{Y}$ and $\bupsilon_j \in \bUpsilon_{j}$.
	\end{enumerate}
\end{assumption}
\noindent
Assumption \ref{AssumptionY} is a standard assumption and similar to Assumption 2 in \cite{DoucMoulinesRyden2004} and Assumption 2(d) in \cite{KasaharaShimotsu2019}. Assumption \ref{AssumptionY} is rather weak. \cite{Liu2006} gave conditions under which Assumption \ref{AssumptionY} is satisfied for the Markov-switching GARCH model in Example \ref{ex:2} using theory for stochastic difference equations, which can also be used to give conditions under which it is satisfied for the Markov-switching observation-driven model in Example \ref{ex:1}, and \cite{AknoucheFrancq2022} gave conditions under which Assumption \ref{AssumptionY} is satisfied for the Markov-switching positive linear conditional mean model in Example \ref{ex:3}. Assumption \ref{AssumptionTheta1} is a standard condition and also assumed in \cite{DoucMoulinesRyden2004} and \cite{KasaharaShimotsu2019}. It implies that there exists a $p_{-} > 0$ such that
\begin{equation}
	p_{ij} \geq p_{-} \label{eq:lb}
\end{equation}
for all $i,j \in \{1,...,J\}$ since $p_{ij} > 0$ for all $i,j \in \{1,...,J\}$, which is equivalent to Assumption 1(a) in \cite{DoucMoulinesRyden2004} but stronger than Assumption 1(d) in \cite{KasaharaShimotsu2019} when the state space of the Markov chain is finite; \cite{KasaharaShimotsu2019} only assumed that there exist a $h \in \mathbb{N}_0$ and a $p_{-} > 0$ such that $\mathbb{P}(S_{t+h} = j \mid S_{t-1} = i) \geq p_{-}$ for all $i,j \in \{1,...,J\}$. As in \cite{DoucMoulinesRyden2004}, Equation \eqref{eq:lb}, which implies that $(S_t)_{t \in \mathbb{Z}}$ is geometrically ergodic, is crucial when showing that the predictor forgets its initialisation asymptotically. Assumptions \ref{AssumptionF1} and \ref{AssumptionPhi1} are standard regularity conditions. Condition (i) in Assumption \ref{AssumptionF1} is similar to Assumption 4 in \cite{DoucMoulinesRyden2004} and Assumption 6(a) in \cite{KasaharaShimotsu2019}, Condition (ii) in Assumption \ref{AssumptionF1} is an extra regularity condition, which is used to show that the predictor forgets its initialisation asymptotically, and Conditions (i) and (ii) in Assumption \ref{AssumptionPhi1} are similar to the conditions in \cite{BlasquesGorgiKoopmanWintenberger2018}.

The following lemma, which follows from an application of Theorem 3.1 in \citet{Bougerol1993}, gives conditions under which, for each $j \in \{1,...,J\}$, the difference between $(\hat{X}_{j,t} (\bupsilon_j))_{t \in \mathbb{N}}$ and $(X_{j,t} (\bupsilon_j))_{t \in \mathbb{Z}}$ converges uniformly e.a.s. to zero. For each $j \in \{1,...,J\}$, let
\begin{equation*}
	\Lambda_{j,t} (\bupsilon_j) := \sup_{x_j \in \mathcal{X}_{j}} \left| \nabla_{x_j} \phi_{j} (Y_t,x_j;\bupsilon_j) \right|.
\end{equation*}
\begin{lemma} \label{LemmaInvertibilityX}
	Assume that Assumptions \ref{AssumptionY}, \ref{AssumptionTheta1}, and \ref{AssumptionPhi1} hold. Moreover, assume that, for each $j \in \{1,...,J\}$,
	\begin{enumerate} [(i)]
		\item there exists an $x_j \in \mathcal{X}_{j}$ such that $\mathbb{E} [\log^{+} \sup_{\bupsilon_{j} \in \bUpsilon_{j}} | \phi_{j} (Y_t,x_j;\bupsilon_{j}) - x_j | ] < \infty$,
		\item $\mathbb{E} [\log^{+} \sup_{\bupsilon_{j} \in \bUpsilon_{j}} \Lambda_{j,t} (\bupsilon_{j})] < \infty$, and
		\item $-\infty \leq \mathbb{E} [\log \sup_{\bupsilon_{j} \in \bUpsilon_{j}} \Lambda_{j,t} (\bupsilon_{j})] < 0$.
	\end{enumerate}
	Then, for each $j \in \{1,...,J\}$, $(X_{j,t} (\bupsilon_j))_{t \in \mathbb{Z}}$, given by
	\begin{equation*}
		X_{j,t+1} (\bupsilon_j) = \phi_j (Y_t,X_{j,t} (\bupsilon_j);\bupsilon_{j}),
	\end{equation*}
	is stationary and ergodic for all $\bupsilon_j \in \bUpsilon_{j}$ and
	\begin{equation*}
		\sup_{\bupsilon_{j} \in \bUpsilon_{j}} \left| \hat{X}_{j,t} (\bupsilon_j) - X_{j,t} (\bupsilon_j) \right| \overset{e.a.s.}{\rightarrow} 0 \quad \text{as} \quad t \rightarrow \infty
	\end{equation*}
	for any initialisation $\hat{X}_{j,1} (\bupsilon_j) \in \mathcal{X}_{j}$.
\end{lemma}
\noindent
Conditions (i) and (ii) are standard moment conditions and Condition (iii) is a standard contraction condition. Indeed, by Jensen's inequality, Condition (iii) holds if, for each $j \in \{1,...,J\}$,
\begin{equation*}
	\mathbb{E} \left[ \sup_{\bupsilon_{j} \in \bUpsilon_{j}} \Lambda_{j,t} (\bupsilon_{j}) \right] < 1,
\end{equation*}
which is a contraction condition in expectation. Conditions (i)-(iii) are similar to the conditions in \cite{BlasquesGorgiKoopmanWintenberger2018}.

In contrast to $(\hat{X}_{1,t} (\bupsilon_1))_{t \in \mathbb{N}},...,(\hat{X}_{J,t} (\bupsilon_J))_{t \in \mathbb{N}}$, Theorem 3.1 in \citet{Bougerol1993} cannot be applied to $(\hat{\bpi}_{t \mid t-1} (\btheta))_{t \in \mathbb{N}}$ because the latter is a function of the former, which are non-stationary. The following lemma, which shows that the difference between the \textit{initialised} filter $(\hat{\bpi}_{t \mid t} (\btheta))_{t \in \mathbb{N}_0}$ and the \textit{non-initialised} filter $(\bpi_{t \mid t} (\btheta))_{t \in \mathbb{Z}}$ converges uniformly e.a.s. to zero, follows instead from an application of Propositions 1 and 3 in \citet{Krabbe2025}, which together generalise Theorem 2.10 in \cite{StraumannMikosch2006}.\footnote{Propositions 1 and 3 in \citet{Krabbe2025} together generalise Theorem 2.10 in \cite{StraumannMikosch2006}, which is about stochastic difference equations taking values in a real or complex separable Banach space, to stochastic difference equations taking values in a complete subspace of a real or complex separable Banach space.}
\begin{lemma} \label{LemmaInvertibilityFilter}
	Assume that Assumptions \ref{AssumptionY}-\ref{AssumptionPhi1} and the conditions in Lemma \ref{LemmaInvertibilityX} hold. Moreover, assume that for each $j \in \{1,...,J\}$, there exists an $m_j > 0$ such that
	\begin{equation*}
		\mathbb{E} \left[ \sup_{\bupsilon_j \in \bUpsilon_j} \sup_{x_j \in \mathcal{X}_{j}} \left| \nabla_{x_j} \log f_{j} (Y_t;x_j,\bupsilon_j) \right|^{m_j} \right] < \infty. 
	\end{equation*} 
	Then, $(\bpi_{t \mid t} (\btheta))_{t \in \mathbb{Z}}$, given by
	\begin{equation*}
		\bpi_{t \mid t} (\btheta) = \mathbf{F}_t (\mathbf{P}^{\prime} \bpi_{t-1 \mid t-1} (\btheta);\btheta) \mathbf{P}^{\prime} \bpi_{t-1 \mid t-1} (\btheta)
	\end{equation*}
	where $\mathbf{F}_t (\bs;\btheta), \bs \in \mathcal{S}$ is a diagonal matrix with generic element
	\begin{equation*}
		[\mathbf{F}_t (\bs;\btheta)]_{ii} := \frac{f_{i} (Y_t;X_{i,t} (\bupsilon_i),\bupsilon_{i})}{\sum_{k=1}^{J} s_k f_{k} (Y_t;X_{k,t} (\bupsilon_k),\bupsilon_{k})}, \quad i \in \{1,...,J\},
	\end{equation*}
	is stationary and ergodic for all $\btheta \in \bTheta$ and
	\begin{equation*}
		\sup_{\btheta \in \bTheta} \left| \left| \hat{\bpi}_{t \mid t} (\btheta) - \bpi_{t \mid t} (\btheta) \right| \right|_{2} \overset{e.a.s.}{\rightarrow} 0 \quad \text{as} \quad t \rightarrow \infty
	\end{equation*}
	for any initialisation $\hat{\bpi}_{0 \mid 0} (\btheta) \in \mathcal{S}$.
\end{lemma}
\noindent
The next corollary, showing that the difference between the \textit{initialised} predictor $(\hat{\bpi}_{t \mid t-1} (\btheta))_{t \in \mathbb{N}}$ and the \textit{non-initialised} predictor $(\bpi_{t \mid t-1} (\btheta))_{t \in \mathbb{Z}}$ converges uniformly e.a.s. to zero, is a direct consequence of Lemma \ref{LemmaInvertibilityFilter}.
\begin{corollary} \label{CorollaryInvertibilityPrediction}
	Under the assumptions in Lemma \ref{LemmaInvertibilityFilter}, $(\bpi_{t \mid t-1} (\btheta))_{t \in \mathbb{Z}}$, given by
	\begin{equation*}
		\bpi_{t+1 \mid t} (\btheta) = \mathbf{P}^{\prime} \bpi_{t \mid t} (\btheta),
	\end{equation*}
	is stationary and ergodic for all $\btheta \in \bTheta$ and
	\begin{equation*}
		\sup_{\btheta \in \bTheta} \left| \left| \hat{\bpi}_{t \mid t-1} (\btheta) - \bpi_{t \mid t-1} (\btheta) \right| \right|_{2} \overset{e.a.s.}{\rightarrow} 0 \quad \text{as} \quad t \rightarrow \infty
	\end{equation*}
	for any initialisation $\hat{\bpi}_{0 \mid 0} (\btheta) \in \mathcal{S}$.
\end{corollary}
\begin{remark} \label{remark}
	Note that $\bpi_{t \mid t-1} (\btheta) \in \mathcal{S}_{-}$ for all $t \in \mathbb{Z}$ with $\mathcal{S}_{-} := \{\bx \in \mathbb{R}^{J} : x_j \geq p_{-}, j=1,...,J, \sum_{j=1}^{J} x_j = 1 \}$. Indeed, for all $j \in \{1,...,J\}$,
	\begin{equation}
		\pi_{j,t \mid t-1} (\btheta) = \sum_{i=1}^{J} p_{ij} \pi_{i,t-1 \mid t-1} (\btheta) \geq p_{-} \sum_{i=1}^{J} \pi_{i,t-1 \mid t-1} (\btheta) = p_{-} \label{eq:remark}
	\end{equation}
	for all $t \in \mathbb{Z}$, where the inequality follows from Equation \eqref{eq:lb}.
\end{remark}
\noindent
Lemma \ref{LemmaInvertibilityFilter} shows that the filter forgets its initialisation exponentially fast, which is intuitive since $(S_t)_{t \in \mathbb{Z}}$ is geometrically ergodic and therefore itself forgets its initialisation exponentially fast if it is initialised. For the autoregressive state space model, \cite{DoucMoulinesRyden2004} and \cite{KasaharaShimotsu2019} showed indirectly that the filter forgets its initialisation exponentially fast as follows. First, they showed that, for all $m,n \in \mathbb{Z}$ such that $m \leq n$, $(S_t)_{m \leq t \leq n}$ given $(Y_\tau)_{m-p+1 \leq \tau \leq n}$ is a non-homogeneous Markov chain (Lemma 1 in \cite{DoucMoulinesRyden2004} and Lemma 6 in \cite{KasaharaShimotsu2019}), and then, using theory for Markov chains, that it forgets its initialisation exponentially fast (Corollary 1 in \cite{DoucMoulinesRyden2004} and Lemma 1 in \cite{KasaharaShimotsu2019}), and this suffices to conclude that the filter forgets its initialisation exponentially fast. For the Markov-switching observation-driven model, we likewise show, and use, that, for all $n \in \mathbb{Z}$, $(S_t)_{t \leq n}$ given $(Y_\tau)_{\tau \leq n}$ is a non-homogeneous Markov chain (Lemma \ref{LemmaInvertibility1}), and, using theory for Markov chains, one can also show that it forgets its initialisation exponentially fast if it is initialised. If the filter were not a function of the initialisations of the time-varying parameters, this would suffice to conclude that the filter forgets its initialisation exponentially fast, as for the autoregressive state space model, but since it is, this is not the case. This is why we instead show directly that the filter forgets its initialisation exponentially fast using theory for stochastic difference equations.

Moreover, we assume the following.
\begin{assumption} \label{AssumptionF2}
	For each $j \in \{1,...,J\}$,
	\begin{equation*}
		\mathbb{E} \left[ \sup_{\bupsilon_j \in \bUpsilon_j} | \log f_{j} (Y_t;X_{j,t}(\bupsilon_j),\bupsilon_j)| \right] < \infty.
	\end{equation*}
\end{assumption}
\noindent
Assumption \ref{AssumptionF2} is a standard moment condition, which is used to show that $L_T (\btheta)$ converges uniformly a.s. to $L(\btheta)$ via a uniform law of large numbers. It is similar to Assumption 3 in \cite{DoucMoulinesRyden2004} and Assumption 4 in \cite{KasaharaShimotsu2019}.

Finally, we assume the following as in \citet{FrancqRoussignol1998} and \cite{FrancqRoussignolZakoian2001} where $f^{(m)} (\by ; \btheta), \by \in \mathcal{Y}^{m}$ denotes the conditional pdf of $\bY_{t-m+1}^{t}$ given $\bY_{-\infty}^{t-m}$.
\begin{assumption} \label{AssumptionIdentification}
	There exists an $m \in \mathbb{N}$ such that
	\begin{align*}
		f^{(m)} (\bY_{t-m+1}^{t} ; \btheta) &= f^{(m)} (\bY_{t-m+1}^{t} ; \btheta_0) \quad a.s.
		\shortintertext{implies that}
		\btheta &= \btheta_0.
	\end{align*}
\end{assumption}
\noindent
Assumption \ref{AssumptionIdentification} is a standard identification condition, which implies that $L(\btheta)$ is uniquely maximised at $\btheta_0$. Note that slightly different identification conditions are used in \cite{DoucMoulinesRyden2004} and \cite{KasaharaShimotsu2019}.

Theorem \ref{TheoremConsistency} gives conditions under which the MLE is consistent.
\begin{theorem} \label{TheoremConsistency}
	Assume that Assumptions \ref{AssumptionY}-\ref{AssumptionIdentification} and the conditions in Lemmas \ref{LemmaInvertibilityX} and \ref{LemmaInvertibilityFilter} hold. Then,
	\begin{equation*}
		\hat{\btheta}_T \overset{a.s.}{\rightarrow} \btheta_0 \quad \text{as} \quad T \rightarrow \infty.
	\end{equation*}
\end{theorem}

\subsection{Asymptotic Normality} \label{AsymptoticNormality}

Asymptotic normality follows from classical arguments like consistency if the \textit{initialised} and thus non-stationary score vector $\sqrt{T} \nabla_{\btheta} \hat{L}_T (\btheta_0)$ and observed Fisher information matrix $(-\nabla_{\btheta \btheta} \hat{L}_T (\btheta))$ converge in distribution to $\mathcal{N} (\mathbf{0},\bI(\btheta_0))$ and uniformly a.s. to $\bI(\btheta)$, respectively, where
\begin{equation*}
	\bI (\btheta) = - \mathbb{E} [ \nabla_{\btheta \btheta} \log f (Y_t;\btheta) ]
\end{equation*}
is the Fisher information matrix. 

We show both in two steps, the first steps being the main difficulties and hence focus. To show the former, we first show that the difference between $\sqrt{T} \nabla_{\btheta} \hat{L}_T (\btheta)$ and the \textit{non-initialised}, stationary, and ergodic score vector $\sqrt{T} \nabla_{\btheta} L_T (\btheta)$ converges uniformly a.s. to zero. This, we do by showing that the first-order derivatives of the time-varying parameters and the predictor forget their initialisations asymptotically, that is, that, for each $j \in \{1,...,J\}$, the difference between $(\nabla_{\bupsilon_j} \hat{X}_{j,t} (\bupsilon_j))_{t \in \mathbb{N}}$ and $(\nabla_{\bupsilon_j} X_{j,t} (\bupsilon_j))_{t \in \mathbb{Z}}$ and the one between $(\nabla_{\btheta} \hat{\bpi}_{t \mid t-1} (\btheta))_{t \in \mathbb{N}}$ and $(\nabla_{\btheta} \bpi_{t \mid t-1} (\btheta))_{t \in \mathbb{Z}}$ converge uniformly e.a.s. to zero. We then show that $\sqrt{T} \nabla_{\btheta} L_T (\btheta_0)$ converges in distribution to $\mathcal{N} (\mathbf{0},\bI(\btheta_0))$ by applying a central limit theorem. 

To show the latter, we first show that the difference between $(-\nabla_{\btheta \btheta} \hat{L}_T (\btheta))$ and the \textit{non-initialised}, stationary, and ergodic observed Fisher information matrix $(-\nabla_{\btheta \btheta} L_T (\btheta))$ converges uniformly a.s. to zero. We do this by showing that also the second-order derivatives of the time-varying parameters and the predictor forget their initialisations asymptotically, that is, that also, for each $j \in \{1,...,J\}$, the difference between $(\nabla_{\bupsilon_j \bupsilon_j} \hat{X}_{j,t} (\bupsilon_j))_{t \in \mathbb{N}}$ and $(\nabla_{\bupsilon_j \bupsilon_j} X_{j,t} (\bupsilon_j))_{t \in \mathbb{Z}}$ and the one between $(\nabla_{\btheta \btheta} \hat{\bpi}_{t \mid t-1} (\btheta))_{t \in \mathbb{N}}$ and $(\nabla_{\btheta \btheta} \bpi_{t \mid t-1} (\btheta))_{t \in \mathbb{Z}}$ converge uniformly e.a.s. to zero. We then show that $(-\nabla_{\btheta \btheta} L_T (\btheta))$ converges uniformly a.s. to $\bI(\btheta)$ by applying a uniform law of large numbers.

In addition to Assumptions \ref{AssumptionY}-\ref{AssumptionIdentification}, we assume the following.
\begin{assumption} \label{AssumptionTheta2}
	$\btheta_0 \in \textup{int}(\bTheta)$.
\end{assumption}
\begin{assumption} \label{AssumptionF3}
	For  each $j \in \{1,...,J\}$,
	\begin{enumerate} [(i)]
		\item $\bvartheta_j \mapsto f_{j} (y;x_j,\bupsilon_j)$ is twice continuously differentiable for all $y \in \mathcal{Y}$ and
		\item $x_j \mapsto \nabla_{\bvartheta_j \bvartheta_j} f_{j} (y;x_j,\bupsilon_j)$ is differentiable for all $y \in \mathcal{Y}$ and $\bupsilon_j \in \bUpsilon_{j}$.
	\end{enumerate}
\end{assumption}
\begin{assumption} \label{AssumptionPhi2}
	For each $j \in \{1,...,J\}$,
	\begin{enumerate} [(i)]
		\item $\bvartheta_j \mapsto \phi_{j} (y,x_j;\bupsilon_j)$ is twice continuously differentiable for all $y \in \mathcal{Y}$.
	\end{enumerate}
\end{assumption}
\noindent
Assumption \ref{AssumptionTheta2} is a standard condition and also assumed in \cite{DoucMoulinesRyden2004} and \cite{KasaharaShimotsu2019}. It implies that there exists an $\varepsilon > 0$ such that $\btheta_0 \in \textup{int}(\bar{\bTheta})$ where
\begin{equation*}
	\bar{\bTheta} := \left\lbrace \btheta \in \bTheta : \left| \left| \btheta - \btheta_0 \right| \right|_2 \leq \varepsilon \right\rbrace \subset \textup{int}(\bTheta).
\end{equation*}
Assumptions \ref{AssumptionF3} and \ref{AssumptionPhi2} are standard regularity conditions. Condition (i) in Assumption \ref{AssumptionF3} is similar to Assumption 6 in \cite{DoucMoulinesRyden2004} and Assumption 7(a) in \cite{KasaharaShimotsu2019}, and Condition (ii) in Assumption \ref{AssumptionF3} is an extra regularity condition, which is used to show that the second-order derivative of the predictor forgets its initialisation asymptotically.

The next two lemmas, which follow from applications of Theorem 2.10 in \citet{StraumannMikosch2006}, give conditions under which, for each $j \in \{1,...,J\}$, the difference between $( \nabla_{\bupsilon_j} \hat{X}_{j,t} (\bupsilon_j) )_{t \in \mathbb{N}}$ and $( \nabla_{\bupsilon_j} X_{j,t} (\bupsilon_j) )_{t \in \mathbb{Z}}$ and the one between $( \nabla_{\bupsilon_j \bupsilon_j} \hat{X}_{j,t} (\bupsilon_j) )_{t \in \mathbb{N}}$ and $( \nabla_{\bupsilon_j \bupsilon_j} X_{j,t} (\bupsilon_j) )_{t \in \mathbb{Z}}$ converge uniformly e.a.s. to zero.
\begin{lemma} \label{LemmaInvertibilityDX}
	Assume that Assumptions \ref{AssumptionY}, \ref{AssumptionTheta1}, \ref{AssumptionPhi1}, \ref{AssumptionPhi2}, and the conditions in Lemma \ref{LemmaInvertibilityX} hold. Moreover, assume that, for each $j \in \{1,...,J\}$,
	\begin{enumerate} [(i)]
		\item $\mathbb{E} [\log^{+} \sup_{\bupsilon_{j} \in \bar{\bUpsilon}_{j}} || \partial_{\bupsilon_j} \phi_{j} (Y_t,X_{j,t} (\bupsilon_{j});\bupsilon_{j}) ||_{2} ] < \infty$.
	\end{enumerate}
	Then, for each $j \in \{1,...,J\}$, $( \nabla_{\bupsilon_j} X_{j,t} (\bupsilon_j) )_{t \in \mathbb{Z}}$ is stationary and ergodic for all $\bupsilon_j \in \bar{\bUpsilon}_j$. Finally, assume that, for each $j \in \{1,...,J\}$,
	\begin{enumerate} [(i)]
		\item $\mathbb{E} [ \log^{+} \sup_{\bupsilon_{j} \in \bar{\bUpsilon}_{j}} || \nabla_{\bupsilon_j} X_{j,t} (\bupsilon_j) ||_{2} ] < \infty$ and
		\item $\sup_{\bupsilon_{j} \in \bar{\bUpsilon}_{j}} || \partial_{\bvartheta_j} \phi_{j} (Y_t,\hat{X}_{j,t} (\bupsilon_{j});\bupsilon_{j}) - \partial_{\bvartheta_j} \phi_{j} (Y_t,X_{j,t} (\bupsilon_{j});\bupsilon_{j}) ||_{2} \overset{e.a.s.}{\rightarrow} 0$ as $t \rightarrow \infty$.
	\end{enumerate}
	Then, for each $j \in \{1,...,J\}$,
	\begin{equation*}
		\sup_{\bupsilon_{j} \in \bar{\bUpsilon}_{j}} || \nabla_{\bupsilon_j} \hat{X}_{j,t} (\bupsilon_j) - \nabla_{\bupsilon_j} X_{j,t} (\bupsilon_j) ||_{2} \overset{e.a.s.}{\rightarrow} 0 \quad \text{as} \quad t \rightarrow \infty
	\end{equation*}
	for any initialisation $\nabla_{\bupsilon_j} \hat{X}_{j,1} (\bupsilon_j) \in \mathbb{R}^{d_j}$.
\end{lemma}
\begin{lemma} \label{LemmaInvertibilityDDX}
	Assume that Assumptions \ref{AssumptionY}, \ref{AssumptionTheta1}, \ref{AssumptionPhi1}, \ref{AssumptionPhi2}, and the conditions in Lemmas \ref{LemmaInvertibilityX} and \ref{LemmaInvertibilityDX} hold. Moreover, assume that, for each $j \in \{1,...,J\}$,
	\begin{enumerate} [(i)]
		\item $\mathbb{E} [\log^{+} \sup_{\bupsilon_{j} \in \bar{\bUpsilon}_{j}} || \partial_{\bvartheta_j \bvartheta_j} \phi_{j} (Y_t,X_{j,t} (\bupsilon_{j});\bupsilon_{j}) ||_{2,2} ] < \infty$.
	\end{enumerate}
	Then, for each $j \in \{1,...,J\}$, $( \nabla_{\bupsilon_j \bupsilon_j} X_{j,t} (\bupsilon_j) )_{t \in \mathbb{Z}}$ is stationary and ergodic for all $\bupsilon_j \in \bar{\bUpsilon}_j$. Finally, assume that, for each $j \in \{1,...,J\}$,
	\begin{enumerate} [(i)]
		\item $\mathbb{E} [ \log^{+} \sup_{\bupsilon_{j} \in \bar{\bUpsilon}_{j}} || \nabla_{\bupsilon_j \bupsilon_j} X_{j,t} (\bupsilon_j) ||_{2,2} ] < \infty$ and
		\item $\sup_{\bupsilon_{j} \in \bar{\bUpsilon}_{j}} || \partial_{\bvartheta_j \bvartheta_j} \phi_{j} (Y_t,\hat{X}_{j,t} (\bupsilon_{j});\bupsilon_{j}) - \partial_{\bvartheta_j \bvartheta_j} \phi_{j} (Y_t,X_{j,t} (\bupsilon_{j});\bupsilon_{j}) ||_{2,2} \overset{e.a.s.}{\rightarrow} 0$ as $t \rightarrow \infty$.
	\end{enumerate}
	Then, for each $j \in \{1,...,J\}$,
	\begin{equation*}
		\sup_{\bupsilon_{j} \in \bar{\bUpsilon}_{j}} || \nabla_{\bupsilon_j \bupsilon_j} \hat{X}_{j,t} (\bupsilon_j) - \nabla_{\bupsilon_j \bupsilon_j} X_{j,t} (\bupsilon_j) ||_{2,2} \overset{e.a.s.}{\rightarrow} 0 \quad \text{as} \quad t \rightarrow \infty
	\end{equation*}
	for any initialisation $\nabla_{\bupsilon_j \bupsilon_j} \hat{X}_{j,1} (\bupsilon_j) \in \mathbb{R}^{d_j \times d_j}$.
\end{lemma}

We also assume the following.
\begin{assumption} \label{AssumptionF4}
	For each $j \in \{1,...,J\}$,	
	\begin{enumerate} [(i)]
		\item $\mathbb{E} [ \sup\nolimits_{\bupsilon_j \in \bar{\bUpsilon}_{j}} | \partial_{x_j} \log f_{j} (Y_t;X_{j,t} (\bupsilon_j),\bupsilon_j) |^2 || \nabla_{\bupsilon_j} X_{j,t} (\bupsilon_j) ||_{2}^2 ] < \infty$,
		\item $\mathbb{E} [ \sup\nolimits_{\bupsilon_j \in \bar{\bUpsilon}_{j}} || \partial_{\bupsilon_j} \log f_{j} (Y_t;X_{j,t} (\bupsilon_j),\bupsilon_j) ||_{2}^{2} ] < \infty$,
		\item $\mathbb{E} [ \sup\nolimits_{\bupsilon_j \in \bar{\bUpsilon}_{j}} | \partial_{x_j x_j} \log f_{j} (Y_t;{X}_{j,t} (\bupsilon_j),\bupsilon_j) | || \nabla_{\bupsilon_j} X_{j,t} (\bupsilon_j) ||_{2}^2 ] < \infty$,
		\item $\mathbb{E} [ \sup\nolimits_{\bupsilon_j \in \bar{\bUpsilon}_{j}} | \partial_{x_j} \log f_{j} (Y_t;{X}_{j,t} (\bupsilon_j),\bupsilon_j) | || \nabla_{\bupsilon_j \bupsilon_j} X_{j,t} (\bupsilon_j) ||_{2,2} ] < \infty$,
		\item $\mathbb{E} [ \sup\nolimits_{\bupsilon_j \in \bar{\bUpsilon}_{j}} || \partial_{\bupsilon_j x_j} \log f_{j} (Y_t;{X}_{j,t} (\bupsilon_j),\bupsilon_j) ||_{2} || \nabla_{\bupsilon_j} X_{j,t} (\bupsilon_j) ||_{2} ] < \infty$, and
		\item $\mathbb{E} [ \sup\nolimits_{\bupsilon_j \in \bar{\bUpsilon}_{j}} || \partial_{\bupsilon_j \bupsilon_j} \log f_{j} (Y_t;{X}_{j,t} (\bupsilon_j),\bupsilon_j) ||_{2,2} ] < \infty$.
	\end{enumerate}
\end{assumption}
\noindent
Assumption \ref{AssumptionF4} is standard moment conditions, which are used to show that $\sqrt{T} \nabla_{\btheta} L_T (\btheta_0)$ converges in distribution to $\mathcal{N} (\mathbf{0},\bI(\btheta_0))$ via a central limit theorem and that $(-\nabla_{\btheta \btheta} L_T (\btheta))$ converges uniformly a.s. to $\bI(\btheta)$ via a uniform law of large numbers. Assumption \ref{AssumptionF4} implies that, for each $j \in \{1,...,J\}$,
\begin{equation*}
	\mathbb{E} [ \sup\nolimits_{\bupsilon_j \in \bar{\bUpsilon}_j} || \nabla_{\bupsilon_j} \log f_{j} (Y_t;X_{j,t}(\bupsilon_j),\bupsilon_j) ||_{2}^{2} ] < \infty
\end{equation*}
and
\begin{equation*}
	\mathbb{E} [ \sup\nolimits_{\bupsilon_j \in \bar{\bUpsilon}_j} || \nabla_{\bupsilon_j \bupsilon_j} \log f_{j} (Y_t;X_{j,t}(\bupsilon_j),\bupsilon_j) ||_{2,2} ] < \infty
\end{equation*}
and is thus similar to Assumption 7(b) in \cite{DoucMoulinesRyden2004} and Assumption 7(c) in \cite{KasaharaShimotsu2019}.

In the next two lemmas, for a vector $\bx \in \mathbb{R}^n$ given by $\bx = (x_1,...,x_n)^{\prime}$ and a function $\mathbf{f} : \mathbb{R}^n \rightarrow \mathbb{R}^m$ given by $\mathbf{f} (\bx) = (f_1(\bx),...,f_m(\bx))^{\prime}$, $\nabla_{\bx} \mathbf{f}(\bx)$ is given by
\begin{equation*}
	[\nabla_{\bx} \mathbf{f}(\bx)]_{(j-1)n+i} = \frac{\partial f_j (\bx)}{\partial x_i}, \quad (i,j) \in \{1,...,n\} \times \{1,...,m\},
\end{equation*}
and $\nabla_{\bx \bx} \mathbf{f}(\bx)$ is given by
\begin{equation*}
	[\nabla_{\bx \bx} \mathbf{f}(\bx)]_{(j-1)n^2+(i-1)n+k} = \frac{\partial^2 f_j (\bx)}{\partial x_i \partial x_k}, \quad (i,j,k) \in \{1,...,n\} \times \{1,...,m\} \times \{1,...,n\}.
\end{equation*}
With this notation, the next two lemmas, which also follow from applications of Theorem 2.10 in \citet{StraumannMikosch2006}, give conditions under which the difference between $( \nabla_{\btheta} \hat{\bpi}_{t \mid t-1} (\btheta) )_{t \in \mathbb{N}}$ and $( \nabla_{\btheta} \bpi_{t \mid t-1} (\btheta) )_{t \in \mathbb{Z}}$ and the one between $( \nabla_{\btheta \btheta} \hat{\bpi}_{t \mid t-1} (\btheta) )_{t \in \mathbb{N}}$ and $( \nabla_{\btheta \btheta} \bpi_{t \mid t-1} (\btheta) )_{t \in \mathbb{Z}}$ converge uniformly e.a.s. to zero.
\begin{lemma} \label{LemmaInvertibilityDPrediction}
	Assume that Assumptions \ref{AssumptionY}-\ref{AssumptionPhi1}, \ref{AssumptionF3}-\ref{AssumptionF4}, and the conditions in Lemmas \ref{LemmaInvertibilityX}-\ref{LemmaInvertibilityDX} hold. Moreover, assume that for each $j \in \{1,...,J\}$, there exists an $m_j > 0$ such that
	\begin{equation*}
		\mathbb{E} \left[ \sup_{\bupsilon_j \in \bar{\bUpsilon}_j} \sup_{x_j \in \mathcal{X}_{j}} \left| \left| \partial_{\bvartheta_j x_j} \log f_{j} (Y_t;x_j,\bupsilon_j) \right| \right|_{2}^{m_j} \right] < \infty.
	\end{equation*}
	Then, $( \nabla_{\btheta} \bpi_{t \mid t-1} (\btheta) )_{t \in \mathbb{Z}}$ is stationary and ergodic for all $\btheta \in \bar{\bTheta}$ with $\mathbb{E} [ \sup_{\btheta \in \bar{\bTheta}} || \nabla_{\btheta} \bpi_{t \mid t-1} (\btheta) ||_{2}^{2} ] < \infty$ and
	\begin{equation*}
		\sup_{\btheta \in \bar{\bTheta}} | | \nabla_{\btheta} \hat{\bpi}_{t \mid t-1} (\btheta) - \nabla_{\btheta} \bpi_{t \mid t-1} (\btheta) | |_{2} \overset{e.a.s.}{\rightarrow} 0 \quad \text{as} \quad t \rightarrow \infty
	\end{equation*}
	for any initialisation $\nabla_{\btheta} \hat{\bpi}_{0 \mid 0} (\btheta) \in \mathbb{R}^{Jd}$.
\end{lemma}
\begin{lemma} \label{LemmaInvertibilityDDPrediction}
	Assume that Assumptions \ref{AssumptionY}-\ref{AssumptionPhi1}, \ref{AssumptionF3}-\ref{AssumptionF4}, and the conditions in Lemmas \ref{LemmaInvertibilityX}-\ref{LemmaInvertibilityDPrediction} hold. Moreover, assume that for each $j \in \{1,...,J\}$, there exists an $m_j > 0$ such that
	\begin{equation*}
		\mathbb{E} \left[ \sup_{\bupsilon_j \in \bar{\bUpsilon}_j} \sup_{x_j \in \mathcal{X}_{j}} \left| \left| \partial_{\bvartheta_j \bvartheta_j x_j} \log f_{j} (Y_t;x_j,\bupsilon_j) \right| \right|_{2,2}^{m_j} \right] < \infty.
	\end{equation*}
	Then, $( \nabla_{\btheta \btheta} \bpi_{t \mid t-1} (\btheta) )_{t \in \mathbb{Z}}$ is stationary and ergodic for all $\btheta \in \bar{\bTheta}$ with $\mathbb{E} [ \sup_{\btheta \in \bar{\bTheta}} || \nabla_{\btheta \btheta} \bpi_{t \mid t-1} (\btheta) ||_{2} ] < \infty$ and
	\begin{equation*}
		\sup_{\btheta \in \bar{\bTheta}} | | \nabla_{\btheta \btheta} \hat{\bpi}_{t \mid t-1} (\btheta) - \nabla_{\btheta \btheta} \bpi_{t \mid t-1} (\btheta) | |_{2} \overset{e.a.s.}{\rightarrow} 0 \quad \text{as} \quad t \rightarrow \infty
	\end{equation*}
	for any initialisation $\nabla_{\btheta \btheta} \hat{\bpi}_{0 \mid 0} (\btheta) \in \mathbb{R}^{Jd^2}$.
\end{lemma}

Theorem \ref{TheoremAsymptoticNormality} gives conditions under which the MLE is also asymptotically normal.
\begin{theorem} \label{TheoremAsymptoticNormality}
	Assume that Assumptions \ref{AssumptionY}-\ref{AssumptionF4} and the conditions in Lemmas \ref{LemmaInvertibilityX}-\ref{LemmaInvertibilityDDPrediction} hold and that the Fisher information matrix $\bI (\btheta_0)$ is invertible. Then,
	\begin{equation*}
		\sqrt{T} (\hat{\btheta}_T - \btheta_0) \overset{d}{\rightarrow} \mathcal{N} (\mathbf{0},\bI(\btheta_0)^{-1}) \quad \text{as} \quad T \rightarrow \infty.
	\end{equation*}
\end{theorem}
\noindent
\cite{DoucMoulinesRyden2004} also assumed that the Fisher information matrix $\bI (\btheta_0)$ is invertible; an assumption which seems to be missing in \cite{KasaharaShimotsu2019}.

Finally, let
\begin{equation*}
	\hat{\bI}_T (\btheta) = \frac{1}{T} \sum_{t=1}^{T} \nabla_{\btheta} \log \hat{f} (Y_t;\btheta) \nabla_{\btheta^{\prime}} \log \hat{f} (Y_t;\btheta),
\end{equation*}
with
\begin{equation*}
	\nabla_{\btheta} \log \hat{f} (Y_t;\btheta) = \frac{\sum_{j=1}^{J} (\nabla_{\btheta} \hat{\pi}_{j,t \mid t-1} (\btheta) f_{j} (Y_t;\hat{X}_{j,t} (\bupsilon_j),\bupsilon_j) + \hat{\pi}_{j,t \mid t-1} (\btheta) \nabla_{\btheta}  f_{j} (Y_t;\hat{X}_{j,t} (\bupsilon_j),\bupsilon_j))}{\sum_{j=1}^{J} \hat{\pi}_{j,t \mid t-1} (\btheta) f_{j} (Y_t;\hat{X}_{j,t} (\bupsilon_j),\bupsilon_j)},
\end{equation*}
and
\begin{align*}
	\nabla_{[\btheta]_m} \hat{\pi}_{j,t \mid t-1} (\btheta) &= \sum_{i=1}^{J} (\nabla_{[\btheta]_m} p_{ij} \hat{\pi}_{i,t-1 \mid t-1} (\btheta) + p_{ij} \nabla_{[\btheta]_m} \hat{\pi}_{i,t-1 \mid t-1} (\btheta)), \ j \in \{1,...,J-1\}, \\
	\nabla_{[\btheta]_m} \hat{\pi}_{J,t \mid t-1} (\btheta) &= - \sum_{j=1}^{J-1} \nabla_{[\btheta]_m} \hat{\pi}_{j,t \mid t-1} (\btheta), \\
	\nabla_{[\btheta]_m}  f_{j} (Y_t;\hat{X}_{j,t} (\bupsilon_j),\bupsilon_j) &= \partial_{x_j}  f_{j} (Y_t;\hat{X}_{j,t} (\bupsilon_j),\bupsilon_j) \nabla_{[\btheta]_m} \hat{X}_{j,t} (\bupsilon_j) + \partial_{[\btheta]_m}  f_{j} (Y_t;\hat{X}_{j,t} (\bupsilon_j),\bupsilon_j),
\end{align*}
where $\hat{\pi}_{j,t \mid t-1} (\btheta)$ and $\hat{\pi}_{j,t-1 \mid t-1} (\btheta)$ are given in Equation \eqref{eq:p&f}, $\hat{X}_{j,t} (\bupsilon_j)$ is given in Equation \eqref{eq:tvp}, $\nabla_{[\btheta]_m} \hat{\pi}_{j,t-1 \mid t-1} (\btheta), j \in \{1,...,J-1\}$ is given in the proof of Lemma \ref{LemmaInvertibilityDFilter}, $\nabla_{[\btheta]_m} \hat{\pi}_{J,t-1 \mid t-1} (\btheta) = - \sum_{j=1}^{J-1} \nabla_{[\btheta]_m} \hat{\pi}_{j,t-1 \mid t-1} (\btheta)$, and $\nabla_{[\btheta]_m} \hat{X}_{j,t} (\bupsilon_j)$ is given in the proof of Lemma \ref{LemmaInvertibilityDX}. Proposition \ref{prop:varcovarmatrix} shows that, under the assumptions in Theorem \ref{TheoremAsymptoticNormality}, $\hat{\bI}_T (\hat{\btheta}_T)$ is a consistent estimator of the Fisher information matrix $\bI (\btheta_0)$.
\begin{proposition} \label{prop:varcovarmatrix}
	Under the assumptions in Theorem \ref{TheoremAsymptoticNormality},
	\begin{equation*}
		\hat{\bI}_T (\hat{\btheta}_T) \overset{a.s.}{\rightarrow} \bI (\btheta_0) \quad \text{as} \quad T \rightarrow \infty.
	\end{equation*}
\end{proposition}
\noindent
Using similar arguments, it can also be shown that, under the assumptions in Theorem \ref{TheoremAsymptoticNormality}, the observed Fisher information matrix
\begin{equation*}
	- \nabla_{\btheta \btheta} \hat{L}_T (\hat{\btheta}_T) = - \frac{1}{T} \sum_{t=1}^{T} \nabla_{\btheta \btheta} \log \hat{f} (Y_t;\hat{\btheta}_T)
\end{equation*}
is a consistent estimator of the Fisher information matrix. The observed Fisher information matrix, however, is much more difficult to compute than $\hat{\bI}_T (\hat{\btheta}_T)$, since it requires the second-order derivatives $\nabla_{\btheta \btheta} \log \hat{f} (Y_t;\hat{\btheta}_T)$, whereas $\hat{\bI}_T (\hat{\btheta}_T)$ only requires the first-order derivatives $\nabla_{\btheta} \log \hat{f} (Y_t;\hat{\btheta}_T)$. This explains the focus on $\hat{\bI}_T (\hat{\btheta}_T)$.

\section{The Markov-Switching GARCH Model} \label{Examples2}

In this section, we study both the asymptotic and finite-sample properties of the MLE for the Markov-switching GARCH model by \cite{HaasMittnikPaolella2004b}, the latter in a small Monte Carlo simulation study.\footnote{Recall that the Markov-switching GARCH model reduces to the mixture GARCH model by \citet{HaasMittnikPaolella2004a} if $(S_t)_{t \in \mathbb{Z}}$ is an i.i.d. chain.}

\subsection{Asymptotic Properties}

Recall that the Markov-switching GARCH model by \citet{HaasMittnikPaolella2004b} is given by
\begin{equation*}
	Y_t = \sum_{j=1}^{J} 1_{\{S_t = j\}} \sqrt{X_{j,t}} \varepsilon_t,
\end{equation*}
where $(S_t)_{t \in \mathbb{Z}}$ is a stationary, irreducible, and aperiodic – and thus ergodic – Markov chain with transition probabilities $p_{ij,0} \in (0,1), i,j \in \{1,...,J\}$,
\begin{equation*}
	X_{j,t+1} = \omega_{j,0} + \alpha_{j,0} Y_{t}^2 + \beta_{j,0} X_{j,t}, \quad j \in \{1,...,J\},
\end{equation*}
where $\omega_{j,0} > 0$, $\alpha_{j,0} \geq 0$, and $\beta_{j,0} \geq 0$, $(\varepsilon_t)_{t \in \mathbb{Z}}$ is a sequence of independent normal distributed random variables with zero mean and unit variance, and $(S_t)_{t \in \mathbb{Z}}$ and $(\varepsilon_t)_{t \in \mathbb{Z}}$ are independent.

\cite{Liu2006} gave conditions under which the model is stationary and ergodic. In the following, $\bM_0$ is a $J^2 \times J^2$ matrix given by
\begin{equation*}
	\left[ \bM_0 \right]_{ij} := p_{ji,0} \left( \balpha_0 \mathbf{e}_i^{\prime} + \bbeta_0 \right), \quad i,j \in \{1,...,J\}, 
\end{equation*}
where $\balpha_0 := (\alpha_{1,0},...,\alpha_{J,0})^{\prime}$, $\bbeta_0 := \textup{diag}(\beta_{1,0},...,\beta_{J,0})$, and $\mathbf{e}_i$ is the $i$'th unit vector in $\mathbb{R}^J$.
\begin{theorem} [\cite{Liu2006}] \label{theo:se}
	$((S_t,Y_t))_{t \in \mathbb{Z}}$ is stationary and ergodic with $\mathbb{E} [Y_t^2] < \infty$ if and only if $\rho(\bM_0) < 1$ where $\rho(\bM_0)$ is the spectral radius of $\bM_0$. 
\end{theorem}
\noindent
Note that $\rho(\bM_0) < 1$ does not imply that $\alpha_{j,0} + \beta_{j,0} < 1$ for all $j \in \{1,...,J\}$ (see \cite{Liu2006}).

We now give conditions under which the MLE is consistent and asymptotically normal. The MLE is consistent under the following assumptions.
\begin{assumption} \label{ass:one}
	$\btheta_0 \in \bTheta$.
\end{assumption}
\begin{assumption} \label{ass:se}
	$\rho(\bM_0) < 1$.
\end{assumption}
\begin{assumption} \label{ass:compact}
	$\bTheta$ is compact.
\end{assumption}
\begin{assumption} \label{ass:inv}
	For all $\btheta \in \bTheta$, $\beta_j < 1$ for all $j \in \{1,...,J\}$. 
\end{assumption}
\begin{assumption} \label{ass:identification}
	For all $\btheta \in \bTheta$, there exists an $m \in \mathbb{N}$ such that $f^{(m)} (\bY_{t-m+1}^{t} ; \btheta) = f^{(m)} (\bY_{t-m+1}^{t} ; \btheta_0)$ a.s. implies that $\btheta = \btheta_0$.
\end{assumption}
\begin{theorem} \label{theo:c}
	If Assumptions \ref{ass:one}-\ref{ass:identification} hold, then
	\begin{equation*}
		\hat{\btheta}_T \overset{a.s.}{\rightarrow} \btheta_0 \quad \text{as} \quad T \rightarrow \infty.
	\end{equation*}
\end{theorem}
\noindent
Assumptions \ref{ass:one}-\ref{ass:identification} are identical to the assumptions in \cite{KandjiMisko2024}, except that \cite{KandjiMisko2024} assumed that Assumption \ref{ass:identification} holds with $m = 1$.

Although the above result is almost identical to the one in \cite{KandjiMisko2024}, the proof of it is different. To prove consistency of the MLE, \cite{KandjiMisko2024} used one representation of the likelihood function, which is similar to the one used in \cite{FrancqRoussignol1998} and \cite{FrancqRoussignolZakoian2001} (see also Equation (12.25) in \cite{FrancqZakoian2019}). We, instead, use another representation of the likelihood function, which is similar to the one used in \cite{DoucMoulinesRyden2004} and \cite{KasaharaShimotsu2019} (see also Equation (12.30) in \cite{FrancqZakoian2019}). The advantage of the latter is that it can also be used to prove asymptotic normality of the MLE, which we now do.

To give conditions under which the MLE is also asymptotically normal, we need the following result where, for all $s \in \mathbb{N}$, $\bSigma_0^{(\otimes s)}$ is a $J^{s+1} \times J^{s+1}$ matrix given by
\begin{equation*}
	\left[ \bSigma_0^{(\otimes s)} \right]_{ij} := p_{ji,0} \mathbb{E} \left[ \bA_{it,0}^{(\otimes s)} \right], \quad i,j \in \{1,...,J\},
\end{equation*}
with 
\begin{equation*}
	\bA_{it,0} := \varepsilon_t^2 \balpha_0 \mathbf{e}_i^{\prime} + \bbeta_0.
\end{equation*}
\begin{theorem} [\cite{Liu2006}] \label{theo:m}
	If $\rho(\bM_0) < 1$ and $\rho(\bSigma_0^{(\otimes s)}) < 1$, then $\mathbb{E} [Y_t^{2s}] < \infty$.
\end{theorem}

In addition to the above assumptions, the MLE is asymptotically normal under the following assumptions. 
\begin{assumption} \label{ass:int}
	$\btheta_0 \in \textup{int}(\bTheta)$.
\end{assumption}
\begin{assumption} \label{ass:m}
	$\rho(\bSigma_0^{(\otimes 3)}) < 1$.
\end{assumption}
\begin{theorem} \label{theo:an}
	If Assumptions \ref{ass:one}-\ref{ass:m} hold and $\bI(\btheta_0)$ is invertible, then
	\begin{equation*}
		\sqrt{T} (\hat{\btheta}_T - \btheta_0) \overset{d}{\rightarrow} \mathcal{N} (\mathbf{0},\bI(\btheta_0)^{-1}) \quad \text{as} \quad T \rightarrow \infty,
	\end{equation*}
	and
	\begin{equation*}
		\hat{\bI}_T (\hat{\btheta}_T) \overset{a.s.}{\rightarrow} \bI (\btheta_0) \quad \text{as} \quad T \rightarrow \infty.
	\end{equation*}
\end{theorem}
\begin{remark}
	An inspection of the proof of Theorem \ref{theo:an} shows that we only need to assume that there exists an $\varepsilon > 0$ such that $\mathbb{E} [|Y_t|^{4+\varepsilon}] < \infty$. However, in the literature, there exist, to the best of our knowledge, no conditions under which this is directly the case. We therefore assume that $\rho(\bSigma_0^{(\otimes 3)}) < 1$, which implies that $\mathbb{E} [Y_t^{6}] < \infty$, a stronger assumption than the assumption that there exists an $\varepsilon > 0$ such that $\mathbb{E} [|Y_t|^{4+\varepsilon}] < \infty$.
\end{remark}

\subsection{Finite-Sample Properties}

To study the finite-sample properties of the MLE for the Markov-switching GARCH model by \cite{HaasMittnikPaolella2004b}, we perform a small Monte Carlo simulation study.

Table \ref{tab:MC1} reports the estimated means and standard deviations of the MLEs of a two-state Markov-switching GARCH model with true parameters $\omega_{1,0} = 0.003$, $\alpha_{1,0} = 0.010$, $\beta_{1,0} = 0.979$, $\omega_{2,0} = 0.119$, $\alpha_{2,0} = 0.110$, $\beta_{2,0} = 0.866$, $p_{11,0} = 0.976$, and $p_{21,0} = 0.061$, which are set equal to the estimated parameters in the application below. Table \ref{tab:MC2} reports the estimated $95$ percent coverage probabilities of the true parameters of the model, that is, the estimated frequency with which the $95$ percent confidence intervals constructed from Theorem \ref{theo:c} contain the true parameters. The means and standard deviations of the MLEs and the $95$ percent coverage probabilities of the true parameters are estimated from $5000$ replications, where a replication consists of first simulating $T$ observations from the model and then estimating the model from the $T$ observations; the R package MSGARCH by \cite{ArdiaBluteauBoudtCataniaTrottier2019} is used for both purposes. The finite-sample properties of the MLEs of the model are good. Indeed, the estimated means of the MLEs approach the true parameters, the estimated standard deviations of the MLEs approach zero, and, although the estimated $95$ percent coverage probabilities of the true parameters are below $95$ percent when $T = 1250$, they approach $95$ percent as $T$ approaches $5000$, with good coverage when $T = 5000$.
\begin{table}[htbp]
	\centering
	\begin{subtable}{\linewidth}
		\centering
		\begin{tabular}{ccccccccc} 
			\toprule
			& $\omega_1$ & $\alpha_1$ & $\beta_1$ & $\omega_2$ & $\alpha_2$ & $\beta_2$ & $p_{11}$ & $p_{21}$ \\
			\midrule
			$T = 1250$ & $0.039$ & $0.023$ & $0.893$ & $0.327$ & $0.126$ & $0.750$ & $0.975$ & $0.057$ \\
			$T = 2500$ & $0.010$ & $0.013$ & $0.961$ & $0.192$ & $0.118$ & $0.818$ & $0.976$ & $0.058$ \\
			$T = 5000$ & $0.004$ & $0.010$ & $0.976$ & $0.148$ & $0.115$ & $0.844$ & $0.976$ & $0.058$ \\
			\bottomrule
		\end{tabular}
		\captionsetup{font=footnotesize}
		\caption{The estimated means of the MLEs of the model.}
		\label{tab:MC1a}
	\end{subtable}
	\begin{subtable}{\linewidth}
		\vspace{0.1cm}
		\centering
		\begin{tabular}{ccccccccc} 
			\toprule
			& $\omega_1$ & $\alpha_1$ & $\beta_1$ & $\omega_2$ & $\alpha_2$ & $\beta_2$ & $p_{11}$ & $p_{21}$ \\
			\midrule
			$T = 1250$ & $0.096$ & $0.037$ & $0.216$ & $0.460$ & $0.071$ & $0.200$ & $0.047$ & $0.059$ \\
			$T = 2500$ & $0.039$ & $0.021$ & $0.099$ & $0.197$ & $0.040$ & $0.101$ & $0.027$ & $0.032$ \\
			$T = 5000$ & $0.012$ & $0.004$ & $0.031$ & $0.066$ & $0.027$ & $0.047$ & $0.006$ & $0.013$ \\
			\bottomrule
		\end{tabular}
		\captionsetup{font=footnotesize}
		\caption{The estimated standard deviations of the MLEs of the model.}
		\label{tab:MC1b}
	\end{subtable}
	\captionsetup{font=footnotesize}
	\caption{The estimated means and standard deviations of the MLEs of a Markov-switching GARCH model with true parameters $\omega_{1,0} = 0.003$, $\alpha_{1,0} = 0.010$, $\beta_{1,0} = 0.979$, $\omega_{2,0} = 0.119$, $\alpha_{2,0} = 0.110$, $\beta_{2,0} = 0.866$, $p_{11,0} = 0.976$, and $p_{21,0} = 0.061$.}
	\label{tab:MC1}
\end{table}
\begin{table}[htbp]
	\centering
	\begin{tabular}{ccccccccc} 
		\toprule
		& $\omega_1$ & $\alpha_1$ & $\beta_1$ & $\omega_2$ & $\alpha_2$ & $\beta_2$ & $p_{11}$ & $p_{21}$ \\
		\midrule
		$T = 1250$ & $0.897$ & $0.876$ & $0.850$ & $0.934$ & $0.942$ & $0.897$ & $0.801$ & $0.789$ \\
		$T = 2500$ & $0.942$ & $0.914$ & $0.928$ & $0.973$ & $0.947$ & $0.934$ & $0.888$ & $0.880$ \\
		$T = 5000$ & $0.964$ & $0.941$ & $0.953$ & $0.982$ & $0.951$ & $0.950$ & $0.924$ & $0.924$ \\
		\bottomrule
	\end{tabular}
	\captionsetup{font=footnotesize}
	\caption{The estimated $95$ percent coverage probabilities of the true parameters of a Markov-switching GARCH model with true parameters $\omega_{1,0} = 0.003$, $\alpha_{1,0} = 0.010$, $\beta_{1,0} = 0.979$, $\omega_{2,0} = 0.119$, $\alpha_{2,0} = 0.110$, $\beta_{2,0} = 0.866$, $p_{11,0} = 0.976$, and $p_{21,0} = 0.061$.}
	\label{tab:MC2}
\end{table}

\section{Application} \label{sec:app}

As an illustration, we fit a two-state Markov-switching GARCH model, two-state Markov-switching ARCH models of orders 1 through 10, and a standard GARCH model to the daily log-returns on the Financial Times Stock Exchange (FTSE) 100 index from January 5, 2000 to December 31, 2024. The data, plotted in Figure \ref{fig:APP1}, are downloaded from Yahoo Finance using the R package quantmod and consist of 6313 observations after removing missing observations. The Markov-switching GARCH model is estimated using the R package MSGARCH by \cite{ArdiaBluteauBoudtCataniaTrottier2019}.
\begin{figure}[htbp]
	\centering
	\includegraphics[width=0.60\textwidth,height=0.25\textheight]{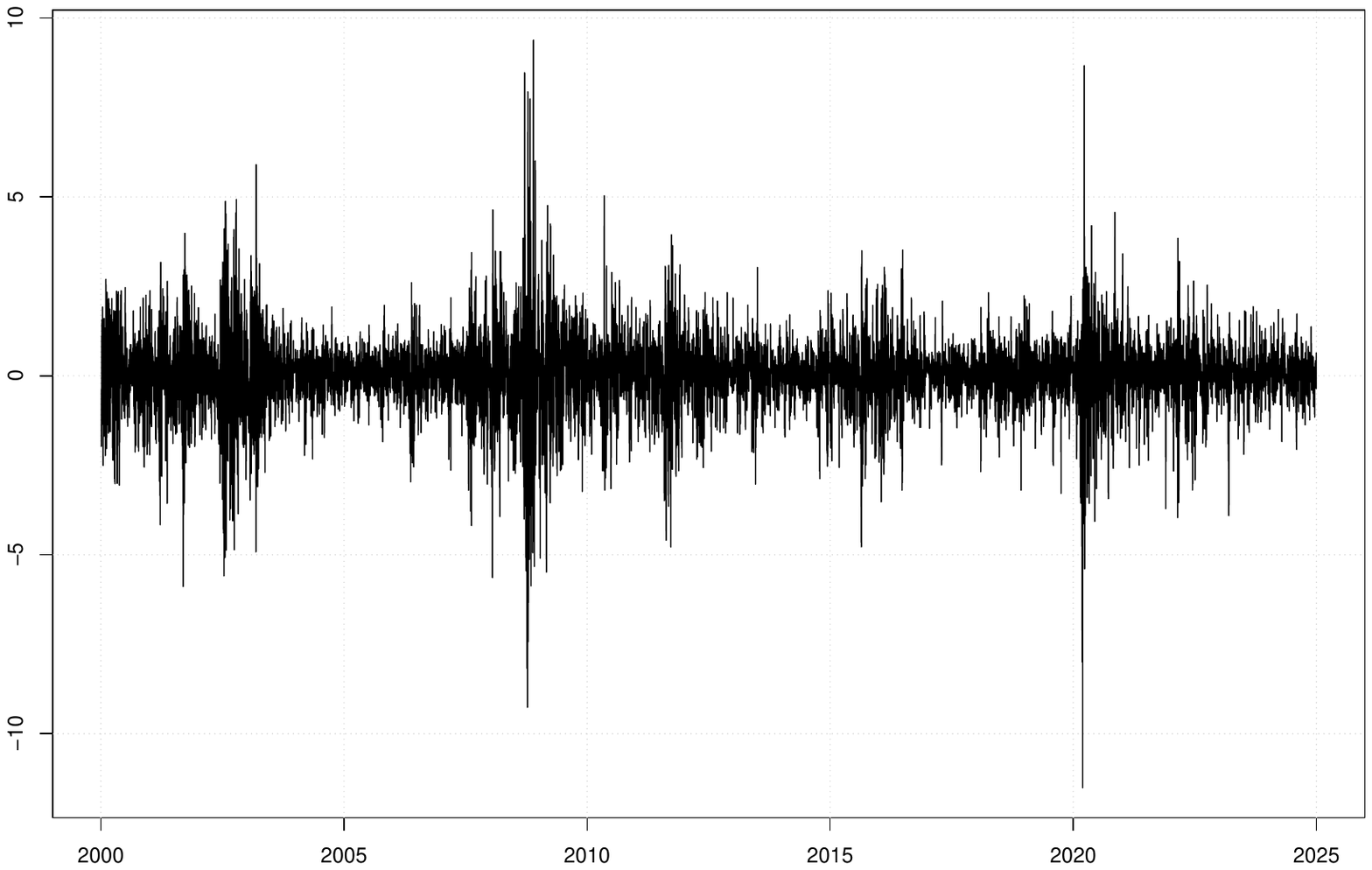}
	\captionsetup{font=footnotesize}
	\caption{The daily log-returns on the Financial Times Stock Exchange 100 index from January 5, 2000 to December 31, 2024.}
	\label{fig:APP1}
\end{figure}

Table \ref{tab:APP2} reports the Akaike information criterion (AIC) and the Bayesian information criterion (BIC) for the models. The Markov-switching GARCH model fits the data better than both the Markov-switching ARCH models and the standard GARCH model. Among the Markov-switching ARCH models, according to AIC, the model of order 10 fits the data best, and, according to BIC, the model of order 8 fits the data best. The fact that a high-order Markov-switching ARCH model is needed to fit the data is a possible explanation for why the Markov-switching GARCH model fits the data best: as shown in Example \ref{ex:2}, the Markov-switching GARCH model admits an infinite-order Markov-switching ARCH representation and can therefore parsimoniously capture the high-order Markov-switching ARCH dynamics that are present in the data.
\begin{table}[htbp]
	\centering
	\resizebox{\textwidth}{!}{%
		\begin{tabular}{ccccccccccccc} 
			\toprule
			& GARCH & \multicolumn{10}{c}{MS-ARCH($p$)} & MS-GARCH \\
			\cmidrule(lr){3-12}
			& & $p=1$ & $p = 2$ & $p = 3$ & $p = 4$ & $p = 5$ & $p = 6$ & $p = 7$ & $p = 8$ & $p = 9$ & $p = 10$ & \\
			\midrule
			AIC & $17072$ & $17318$ & $17245$ & $17135$ & $17090$ & $17047$ & $17044$ & $17015$ & $16989$ & $16988$ & $16979$ & $16886$ \\
			BIC & $17092$ & $17359$ & $17299$ & $17203$ & $17171$ & $17141$ & $17152$ & $17137$ & $17124$ & $17136$ & $17141$ & $16940$ \\
			\bottomrule
		\end{tabular}%
	}
	\captionsetup{font=footnotesize}
	\caption{The Akaike information criterion and the Bayesian information criterion for the Markov-switching GARCH model, the Markov-switching ARCH models of orders 1 through 10, and the standard GARCH model.}
	\label{tab:APP2}
\end{table}

Table \ref{tab:APP3} reports the estimated parameters for the Markov-switching GARCH model together with their asymptotic standard errors (in parentheses). We note that the first state of the Markov-switching GARCH model is a persistent low-volatility state, while the second is a persistent high-volatility state, since $\hat{\omega}_1 < \hat{\omega}_2$, $\hat{\alpha}_1 + \hat{\beta}_1 \approx \hat{\alpha}_2 + \hat{\beta}_2$, $\hat{p}_{11} \gg \hat{p}_{12}$, and $\hat{p}_{21} \ll \hat{p}_{22}$. Note also that $\hat{\alpha}_1 < \hat{\alpha}_2$ and $\hat{\beta}_1 > \hat{\beta}_2$, which – according to \cite{HaasMittnikPaolella2004b} – may indicate a tendency to overreact to news in the high-volatility state, possibly due to a prevailing panic-like mood.
\begin{table}[htbp]
	\centering
	\begin{tabular}{cccccccc} 
		\toprule
		$\omega_1$ & $\alpha_1$ & $\beta_1$ & $\omega_2$ & $\alpha_2$ & $\beta_2$ & $p_{11}$ & $p_{21}$ \\
		\midrule
		$\underset{(0.001)}{0.003}$ & $\underset{(0.002)}{0.010}$ & $\underset{(0.003)}{0.979}$ & $\underset{(0.031)}{0.119}$ & $\underset{(0.017)}{0.110}$ & $\underset{(0.027)}{0.866}$ & $\underset{(0.005)}{0.976}$ & $\underset{(0.012)}{0.061}$ \\
		\bottomrule
	\end{tabular}
	\captionsetup{font=footnotesize}
	\caption{The estimated parameters for the Markov-switching GARCH model together with their asymptotic standard errors (in parentheses).}
	\label{tab:APP3}
\end{table}

Finally, Figure \ref{fig:APP4} reports the estimated smoothed probabilities for the high-volatility state of the Markov-switching GARCH model (in black), together with the local and global decoded states (in red), from January 3, 2006 to December 31, 2009; 0 denotes the low-volatility state and 1 denotes the high-volatility state. The local decoded states are the most likely states at each point in time, whereas the global decoded states are the most likely sequence of states, which is obtained via the Viterbi algorithm (\cite{Viterbi1967}). We note that before and after the Great Recession, the decoded states are mostly the low-volatility state, whereas during the Great Recession, the decoded states are mostly the high-volatility state, indicating that the Markov-switching GARCH model is able to correctly distinguish periods with financial turmoil from periods without.
\begin{figure}[htbp]
	\centering
	\begin{subfigure}[b]{0.45\textwidth}
		\centering
		\includegraphics[width=\textwidth,height=0.25\textheight]{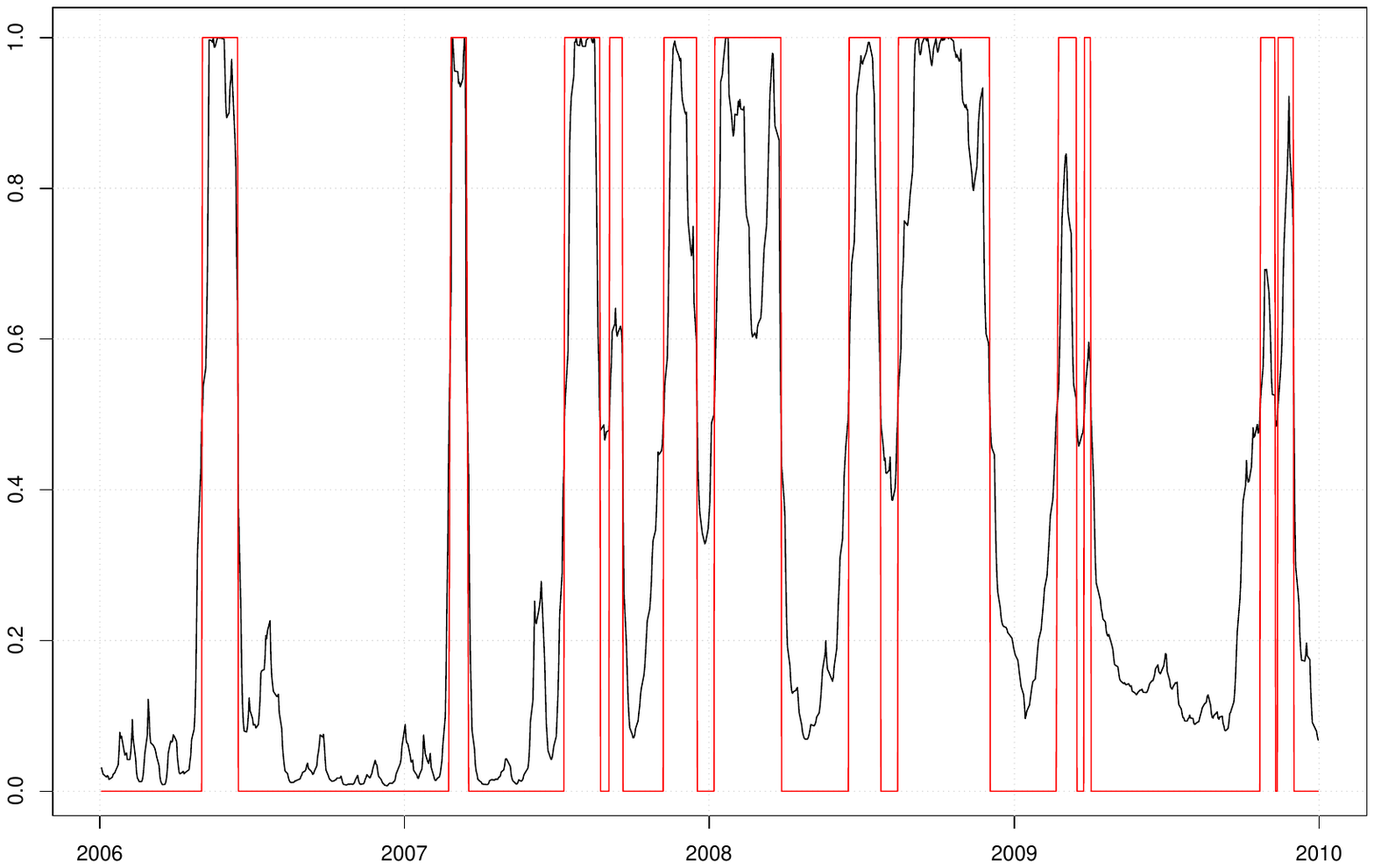}
		\caption{The local decoded states.}
	\end{subfigure}
	\hspace{0.01\textwidth}
	\begin{subfigure}[b]{0.45\textwidth}
		\centering
		\includegraphics[width=\textwidth,height=0.25\textheight]{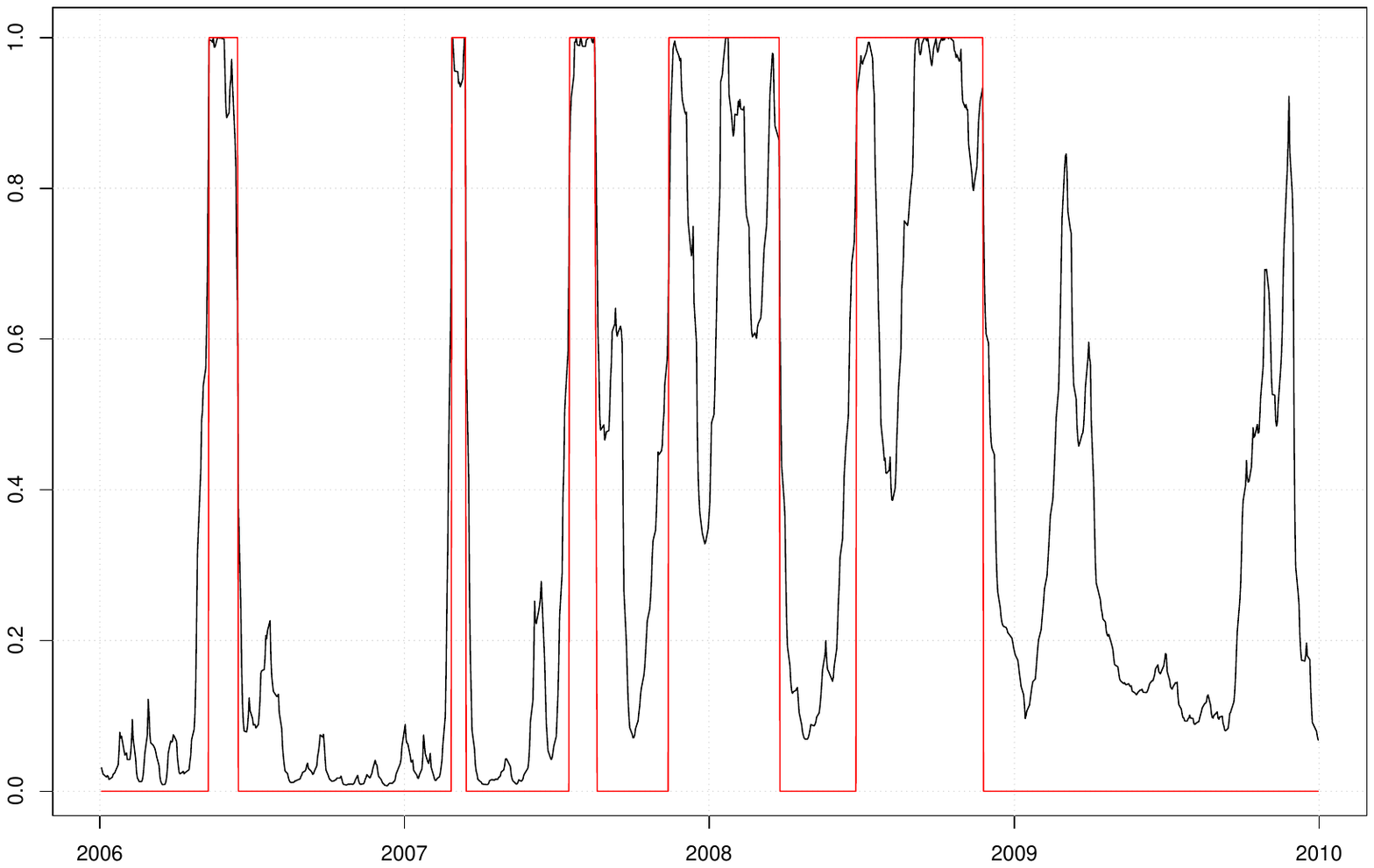}
		\caption{The global decoded states.}
	\end{subfigure}
	\captionsetup{font=footnotesize}
	\caption{The estimated smoothed probabilities for the high-volatility state of the Markov-switching GARCH model (in black), together with the local and global decoded states (in red), from January 3, 2006 to December 31, 2009; 0 denotes the low-volatility state and 1 denotes the high-volatility state.}
	\label{fig:APP4}
\end{figure}

\section{Conclusion} \label{Conclusion}

State space models and their extensions are ubiquitous in economics and finance, so statistical inference for them – including estimation, which is typically done by maximum likelihood estimation – is of significant practical importance.

In this paper, we proved both consistency and asymptotic normality of the MLE for a Markov-switching observation-driven model, that is, an observation-driven state space model where $\mathbb{S}$ is finite. The MLE for the Markov-switching observation-driven model depends on initialisations of the time-varying parameters and the filter, so, to prove that it is consistent and asymptotically normal, it is crucial to prove that the time-varying parameters and the filter forget their initialisations asymptotically. We first proved this for the time-varying parameters using theory for stationary and ergodic stochastic difference equations. Using the fact that the time-varying parameters forget their initialisations asymptotically, we then proved it for the filter using theory for asymptotically stationary and ergodic stochastic difference equations, which is needed since the filter depends on the initialisations of the time-varying parameters. We also applied the results to give conditions under which the MLE for the Markov-switching GARCH model by \cite{HaasMittnikPaolella2004b} is both consistent and asymptotically normal, the latter being new to the literature. Finally, we illustrated that the Markov-switching GARCH model by \cite{HaasMittnikPaolella2004b} can parsimoniously capture the high persistence in the volatility of daily stock returns, in contrast to the Markov-switching ARCH model.

An interesting extension of this paper is to generalise the results to an observation-driven state space model where $\mathbb{S}$ is not necessarily finite, to cover more of the examples in the introduction. We conjecture that the results can be extended to an observation-driven state space model where $\mathbb{S}$ is compact by combining the arguments in this paper with those in \cite{DoucMoulinesRyden2004} and \cite{KasaharaShimotsu2019} that allow $\mathbb{S}$ to be compact. We, however, leave this for future research.

\appendix

\section*{Appendix}

\section{Proofs}

Recall that the auxiliary lemmas used in the proofs, denoted Lemma \ref{LemmaInvertibility1}, Lemma \ref{LemmaInvertibilityA}, and so on, are collected in the Supplementary Appendix.

\subsection{Proof of Lemma \ref{LemmaInvertibilityX}}

For each $j \in \{1,...,J\}$, $(X_{j,t} (\bupsilon_{j}))_{t \in \mathbb{Z}}$ is a stochastic process taking values in $\mathcal{X}_{j}$ given by
\begin{equation*}
	X_{j,t+1} (\bupsilon_{j}) = \phi_{j,t} (X_{j,t} (\bupsilon_{j});\bupsilon_{j})
\end{equation*}
with
\begin{equation*}
	\phi_{j,t} (X_{j,t} (\bupsilon_{j});\bupsilon_{j}) = \phi_{j} (Y_t,X_{j,t} (\bupsilon_{j});\bupsilon_{j}).
\end{equation*}
Note that
\begin{equation*}
	\Lambda (\phi_{j,t};\bupsilon_j) = \sup_{\underset{x_j \neq y_j}{x_j,y_j \in \mathcal{X}_{j}}} \frac{ | \phi_{j,t} (x_j;\bupsilon_j) - \phi_{j,t} (y_j;\bupsilon_j)| }{ |x_j - y_j| } \leq \sup_{x_j \in \mathcal{X}_{j}} \left| \nabla_{x_j} \phi_{j,t} (x_j;\bupsilon_j) \right| = \Lambda_{j,t} (\bupsilon_j).
\end{equation*}
The result thus follows from Theorem 3.1 in \cite{Bougerol1993} since, for each $j \in \{1,...,J\}$,
\begin{enumerate} [(i)]
	\item there exists an $x_j \in \mathcal{X}_{j}$ such that $\mathbb{E} [\log^{+} \sup_{\bupsilon_{j} \in \bUpsilon_{j}} | \phi_{j,t} (x_j;\bupsilon_{j}) - x_j | ] < \infty$,
	\item $\mathbb{E} [\log^{+} \sup_{\bupsilon_{j} \in \bUpsilon_{j}} \Lambda_{j,t} (\bupsilon_{j}) ] < \infty$, and
	\item $-\infty \leq \mathbb{E} [\log \sup_{\bupsilon_{j} \in \bUpsilon_{j}} \Lambda_{j,t} (\bupsilon_{j})] < 0$
\end{enumerate}
by assumption.

\subsection{Proof of Lemma \ref{LemmaInvertibilityFilter}}

The proof consists of two parts. The first part of the proof is about the stochastic process $(\bpi_{t \mid t} (\btheta))_{t \in \mathbb{Z}}$, and the second part of the proof is about the stochastic process $(\hat{\bpi}_{t \mid t} (\btheta))_{t \in \mathbb{N}_0}$.

$(\bpi_{t \mid t} (\btheta))_{t \in \mathbb{Z}}$ is a stochastic process taking values in $\mathcal{S}$ given by
\begin{equation*}
	\bpi_{t \mid t} (\btheta) = \bphi_t (\bpi_{t-1 \mid t-1} (\btheta);\btheta)
\end{equation*} 
with 
\begin{equation*}
	\bphi_t (\bpi_{t-1 \mid t-1} (\btheta);\btheta) = \mathbf{F}_t (\mathbf{P}^{\prime} \bpi_{t-1 \mid t-1} (\btheta);\btheta) \mathbf{P}^{\prime} \bpi_{t-1 \mid t-1} (\btheta).
\end{equation*}
Let
\begin{equation*}
	\Lambda (\bphi_t^{(r)};\btheta) = \sup_{\underset{\bx \neq \by}{\bx,\by \in \mathcal{S}}} \frac{|| \bphi_t^{(r)} (\bx;\btheta) - \bphi_t^{(r)} (\by;\btheta) ||_{2}}{|| \bx - \by ||_{2}},
\end{equation*}
where
\begin{equation*}
	\bphi_t^{(r)} (\, \cdot \,;\btheta) = \bphi_t \circ \cdots \circ \bphi_{t-r+1} (\, \cdot \,;\btheta), \quad r \in \mathbb{N}.
\end{equation*}
The fact that $(\bpi_{t \mid t} (\btheta))_{t \in \mathbb{Z}}$ is stationary and ergodic for all $\btheta \in \bTheta$ follows from Proposition 1 in \cite{Krabbe2025}, which is a special case of Theorem 3.1 in \cite{Bougerol1993}, if
\begin{enumerate} [(i)]
	\item there exists an $\bs \in \mathcal{S}$ such that $\mathbb{E} [\log^{+} \sup_{\btheta \in \bTheta} ||\bphi_t (\bs;\btheta) - \bs ||_{2}] < \infty$,
	\item $\mathbb{E} [\log^{+} \sup_{\btheta \in \bTheta} \Lambda (\bphi_t;\btheta)] < \infty$, and
	\item there exists an $r \in \mathbb{N}$ such that $-\infty \leq \mathbb{E} [\log \sup_{\btheta \in \bTheta} \Lambda (\bphi_t^{(r)};\btheta)] < 0$.
\end{enumerate}
Condition (i) is trivial. For Conditions (ii) and (iii), note that, if $t \leq u$, then
\begin{align*}
	\pi_{j,t \mid u} (\btheta) &= \mathbb{P}_{\btheta} (S_t = j \mid \bY_{-\infty}^{u}) \\
	&= \sum_{i = 1}^{J} \mathbb{P}_{\btheta} (S_t = j , S_{t-1} = i \mid \bY_{-\infty}^{u}) \\
	&= \sum_{i = 1}^{J} \mathbb{P}_{\btheta} (S_t = j \mid S_{t-1} = i , \bY_{-\infty}^{u}) \mathbb{P}_{\btheta} (S_{t-1} = i \mid \bY_{-\infty}^{u}) \\
	&= \sum_{i = 1}^{J} \frac{\mathbb{P}_{\btheta} (\bY_{t}^{u} \in \textup{d}\by_{t}^{u} \mid \bY_{-\infty}^{t-1}, S_t = j) p_{ij}}{\sum_{k=1}^{J} \mathbb{P}_{\btheta} (\bY_{t}^{u} \in \textup{d}\by_{t}^{u} \mid \bY_{-\infty}^{t-1}, S_t = k) p_{ik} } \mathbb{P}_{\btheta} (S_{t-1} = i \mid \bY_{-\infty}^{u}) \\
	&= \sum_{i = 1}^{J} \frac{\mathbb{P}_{\btheta} (\bY_{t}^{u} \in \textup{d}\by_{t}^{u} \mid \bY_{-\infty}^{t-1}, S_t = j) p_{ij}}{\sum_{k=1}^{J} \mathbb{P}_{\btheta} (\bY_{t}^{u} \in \textup{d}\by_{t}^{u} \mid \bY_{-\infty}^{t-1}, S_t = k) p_{ik} } \pi_{i,t-1 \mid u} (\btheta)
\end{align*}
for all $j \in \{1,...,J\}$ where the fourth equality follows from Lemma \ref{LemmaInvertibility1}. Thus,
\begin{equation}
	\bpi_{t \mid u} (\btheta) = \bM_{t \mid u}^{\prime}(\btheta) \bpi_{t-1 \mid u} (\btheta), \label{eq:cc1}
\end{equation}
where $\bM_{t \mid u} (\btheta)$ is a row stochastic matrix, that is, a square matrix with non-negative elements where each row sum to one, with generic element
\begin{equation*}
	[\bM_{t \mid u}(\btheta)]_{ij} := \frac{\mathbb{P}_{\btheta} (\bY_{t}^{u} \in \textup{d}\by_{t}^{u} \mid \bY_{-\infty}^{t-1}, S_t = j) p_{ij}}{\sum_{k=1}^{J} \mathbb{P}_{\btheta} (\bY_{t}^{u} \in \textup{d}\by_{t}^{u} \mid \bY_{-\infty}^{t-1}, S_t = k) p_{ik}}, \quad i,j \in \{1,...,J\}.
\end{equation*}
Note that, since $p_{-} \leq p_{ij} \leq 1$ for all $i,j \in \{1,...,J\}$,
\begin{equation}
	[\bM_{t \mid u}(\btheta)]_{ij} \geq p_{-} [\bv_{t \mid u}(\btheta)]_{j} \label{eq:cc2}
\end{equation}
for all $i,j \in \{1,...,J\}$, where $\bv_{t \mid u}(\btheta)$ is a stochastic vector, that is, a vector with non-negative elements that sum to one, with generic element
\begin{equation*}
	[\bv_{t \mid u}(\btheta)]_{j} := \frac{\mathbb{P}_{\btheta} (\bY_{t}^{u} \in \textup{d}\by_{t}^{u} \mid \bY_{-\infty}^{t-1}, S_t = j)}{\sum_{k=1}^{J} \mathbb{P}_{\btheta} (\bY_{t}^{u} \in \textup{d}\by_{t}^{u} \mid \bY_{-\infty}^{t-1}, S_t = k)}, \quad j \in \{1,...,J\}.
\end{equation*}
Equation \eqref{eq:cc1} implies that
\begin{equation*}
	\bpi_{t \mid t} (\btheta) = \bM_{t \mid t}^{\prime} (\btheta) \cdots \bM_{t-r+1 \mid t}^{\prime} (\btheta) \bpi_{t-r \mid t} (\btheta)
\end{equation*}
for all $r \in \mathbb{N}$ where
\begin{align*}
	[\bpi_{t-r \mid t} (\btheta)]_j
	&= \mathbb{P}_{\btheta} (S_{t-r} = j \mid \bY_{-\infty}^{t}) \\
	&= \frac{\mathbb{P}_{\btheta} (\bY_{t-r+1}^{t} \in \textup{d}\by_{t-r+1}^{t},S_{t-r} = j \mid \bY_{-\infty}^{t-r})}{ \mathbb{P}_{\btheta} (\bY_{t-r+1}^{t} \in \textup{d}\by_{t-r+1}^{t} \mid \bY_{-\infty}^{t-r})} \\
	&= \frac{\mathbb{P}_{\btheta} (\bY_{t-r+1}^{t} \in \textup{d}\by_{t-r+1}^{t},S_{t-r} = j \mid \bY_{-\infty}^{t-r})}{\sum_{k=1}^{J} \mathbb{P}_{\btheta} (\bY_{t-r+1}^{t} \in \textup{d}\by_{t-r+1}^{t}, S_{t-r} = k \mid \bY_{-\infty}^{t-r})} \\
	&= \frac{\mathbb{P}_{\btheta} (\bY_{t-r+1}^{t} \in \textup{d}\by_{t-r+1}^{t} \mid \bY_{-\infty}^{t-r},S_{t-r} = j) \mathbb{P}_{\btheta} (S_{t-r} = j \mid \bY_{-\infty}^{t-r})}{\sum_{k=1}^{J} \mathbb{P}_{\btheta} (\bY_{t-r+1}^{t} \in \textup{d}\by_{t-r+1}^{t} \mid \bY_{-\infty}^{t-r},S_{t-r} = k) \mathbb{P}_{\btheta} (S_{t-r} = k \mid \bY_{-\infty}^{t-r})} \\
	&= \frac{\mathbb{P}_{\btheta}(\bY_{t-r+1}^{t} \in \textup{d}\by_{t-r+1}^{t} \mid \bY_{-\infty}^{t-r}, S_{t-r} = j) [\bpi_{t-r \mid t-r} (\btheta)]_{j} }{\sum_{k=1}^{J} \mathbb{P}_{\btheta}(\bY_{t-r+1}^{t} \in \textup{d}\by_{t-r+1}^{t} \mid \bY_{-\infty}^{t-r},S_{t-r} = k) [\bpi_{t-r \mid t-r} (\btheta)]_{k}}, \quad j \in \{1,...,J\}.
\end{align*}
We thus have that
\begin{equation}
	\Lambda (\bphi_t^{(r)};\btheta) = \sup_{\underset{\bx \neq \by}{\bx,\by \in \mathcal{S}}} \frac{|| \bM_{t \mid t}^{\prime} (\btheta) \cdots \bM_{t-r+1 \mid t}^{\prime} (\btheta) (\frac{\bp_t \odot \bx }{\bp_t^{\prime} \bx} - \frac{\bp_t \odot \by }{\bp_t^{\prime} \by}) ||_{2}}{|| \bx - \by ||_{2}} \label{eq:cc3}
\end{equation}
for all $r \in \mathbb{N}$ where $\bp_t := (p_{1,t},...,p_{J,t})^{\prime}$ with $p_{i,t} := \mathbb{P}_{\btheta}(\bY_{t-r+1}^{t} \in \textup{d}\by_{t-r+1}^{t} \mid \bY_{-\infty}^{t-r}, S_{t-r} = i), i \in \{1,...,J\}$. First, $||\bz||_2 \leq ||\bz||_1$ for all $\bz \in \mathbb{R}^J$ implies that
\begin{equation}
	\begin{aligned}
		& \left| \left| \bM_{t \mid t}^{\prime} (\btheta) \cdots \bM_{t-r+1 \mid t}^{\prime} (\btheta) \left(\frac{\bp_t \odot \bx }{\bp_t^{\prime} \bx} - \frac{\bp_t \odot \by }{\bp_t^{\prime} \by}\right) \right| \right|_{2} \\
		&\leq \left| \left| \bM_{t \mid t}^{\prime} (\btheta) \cdots \bM_{t-r+1 \mid t}^{\prime} (\btheta) \left(\frac{\bp_t \odot \bx }{\bp_t^{\prime} \bx} - \frac{\bp_t \odot \by }{\bp_t^{\prime} \by}\right) \right| \right|_{1}
	\end{aligned} \label{eq:cc4}
\end{equation}
for all $r \in \mathbb{N}$ and $\bx, \by \in \mathcal{S}$. Lemma \ref{LemmaInvertibilityA} together with Equation \eqref{eq:cc2} then imply that there exists an $\alpha \in (0,1)$ such that
\begin{equation}
	\left| \left| \bM_{t \mid t}^{\prime} (\btheta) \cdots \bM_{t-r+1 \mid t}^{\prime} (\btheta) \left(\frac{\bp_t \odot \bx }{\bp_t^{\prime} \bx} - \frac{\bp_t \odot \by }{\bp_t^{\prime} \by}\right) \right| \right|_{1} \leq \alpha^{r} \left| \left| \frac{\bp_t \odot \bx }{\bp_t^{\prime} \bx} - \frac{\bp_t \odot \by }{\bp_t^{\prime} \by} \right| \right|_{1} \label{eq:cc5}
\end{equation}
for all $r \in \mathbb{N}$ and $\bx, \by \in \mathcal{S}$. Note that
\begin{align*}
	\left| \left| \frac{\bp_t \odot \bx }{\bp_t^{\prime} \bx} - \frac{\bp_t \odot \by }{\bp_t^{\prime} \by} \right| \right|_{1} &= \left| \left| \frac{\bp_t \odot \bx }{\bp_t^{\prime} \bx} - \frac{\bp_t \odot \by }{\bp_t^{\prime} \bx} + \frac{\bp_t \odot \by }{\bp_t^{\prime} \bx} - \frac{\bp_t \odot \by }{\bp_t^{\prime} \by} \right| \right|_{1} \\
	&\leq \left| \left| \frac{\bp_t \odot (\bx-\by) }{\bp_t^{\prime} \bx} \right| \right|_{1} + \left| \left| \frac{(\bp_t \odot \by) (\bp_t^{\prime} (\bx-\by)) }{(\bp_t^{\prime} \bx)(\bp_t^{\prime} \by)} \right| \right|_{1} \\
	&\leq \frac{\max_{i \in \{1,...,J\}} p_{i,t}}{\bp_t^{\prime} \bx} \left| \left| \bx-\by \right| \right|_{1} + \frac{\left| \left| \bp_t \odot \by \right| \right|_{1} \max_{i \in \{1,...,J\}} p_{i,t}}{(\bp_t^{\prime} \bx)(\bp_t^{\prime} \by)} \left| \left| \bx-\by \right| \right|_{1} \\
	&= 2 \frac{\max_{i \in \{1,...,J\}} p_{i,t}}{\bp_t^{\prime} \bx} \left| \left| \bx-\by \right| \right|_{1} \\
	&\leq 2 \frac{\max_{i \in \{1,...,J\}} p_{i,t}}{\min_{i \in \{1,...,J\}} p_{i,t}} \left| \left| \bx-\by \right| \right|_{1} \\
	&\leq \frac{2}{p_{-}} \left| \left| \bx-\by \right| \right|_{1},
\end{align*}
where the second equality follows since $\left| \left| \bp_t \odot \by \right| \right|_{1} = \bp_t^{\prime} \by$ and the fourth inequality follows since $p_{i,t} \in [p_{-} \sum_{j=1}^{J} \tilde{p}_{j,t},\sum_{j=1}^{J} \tilde{p}_{j,t}]$ for all $i \in \{1,...,J\}$ with $\tilde{p}_{j,t} := \mathbb{P}_{\btheta}(\bY_{t-r+1}^{t} \in \textup{d}\by_{t-r+1}^{t} \mid \bY_{-\infty}^{t-r}, S_{t-r+1} = j), j \in \{1,...,J\}$. Indeed,
\begin{align*}
	p_{i,t} &= \mathbb{P}_{\btheta}(\bY_{t-r+1}^{t} \in \textup{d}\by_{t-r+1}^{t} \mid \bY_{-\infty}^{t-r}, S_{t-r} = i) \\
	&= \sum_{j=1}^{J} \mathbb{P}_{\btheta}(\bY_{t-r+1}^{t} \in \textup{d}\by_{t-r+1}^{t} , S_{t-r+1} = j \mid \bY_{-\infty}^{t-r}, S_{t-r} = i) \\
	&= \sum_{j=1}^{J} \mathbb{P}_{\btheta}(\bY_{t-r+1}^{t} \in \textup{d}\by_{t-r+1}^{t} \mid \bY_{-\infty}^{t-r}, S_{t-r+1} = j) \mathbb{P}_{\btheta}(S_{t-r+1} = j \mid S_{t-r} = i) \\
	&\in \left[ p_{-} \sum_{j=1}^{J} \tilde{p}_{j,t},\sum_{j=1}^{J} \tilde{p}_{j,t} \right]
\end{align*}
for all $i \in \{1,...,J\}$ since $p_{-} \leq p_{ij} \leq 1$ for all $i,j \in \{1,...,J\}$. Hence, continuing from Equation \eqref{eq:cc5},
\begin{equation}
	\left| \left| \bM_{t \mid t}^{\prime} (\btheta) \cdots \bM_{t-r+1 \mid t}^{\prime} (\btheta) \left(\frac{\bp_t \odot \bx }{\bp_t^{\prime} \bx} - \frac{\bp_t \odot \by }{\bp_t^{\prime} \by}\right) \right| \right|_{1} \leq \alpha^{r} \frac{2}{p_{-}} \left| \left| \bx-\by \right| \right|_{1} \label{eq:cc6}
\end{equation}
for all $r \in \mathbb{N}$ and $\bx, \by \in \mathcal{S}$. Finally, $||\bz||_1 \leq \sqrt{J} ||\bz||_2$ for all $\bz \in \mathbb{R}^J$ implies that
\begin{equation}
	\alpha^{r} \frac{2}{p_{-}} \left| \left| \bx-\by \right| \right|_{1} \leq \alpha^{r} \frac{2\sqrt{J}}{p_{-}} \left| \left| \bx-\by \right| \right|_{2} \label{eq:cc7}
\end{equation}
for all $r \in \mathbb{N}$ and $\bx, \by \in \mathcal{S}$. Hence, continuing from Equation \eqref{eq:cc3} and using Equations \eqref{eq:cc4}, \eqref{eq:cc6}, and \eqref{eq:cc7}, we have that
\begin{equation}
	\sup_{\btheta \in \bTheta} \Lambda (\bphi_t^{(r)};\btheta) \leq \alpha^{r} \frac{2\sqrt{J}}{p_{-}} \label{eq:cc8}
\end{equation}
for all $r \in \mathbb{N}$. Conditions (ii) and (iii) follow straightforwardly from Equation \eqref{eq:cc8}, which concludes the first part of the proof.

Moving on to the second part of the proof, $(\hat{\bpi}_{t \mid t} (\btheta))_{t \in \mathbb{N}_0}$ is a stochastic process taking values in $\mathcal{S}$ given by
\begin{equation*}
	\hat{\bpi}_{t \mid t} (\btheta) = \hat{\bphi}_t (\hat{\bpi}_{t-1 \mid t-1} (\btheta);\btheta)
\end{equation*}
with
\begin{equation*}
	\hat{\bphi}_t (\hat{\bpi}_{t-1 \mid t-1} (\btheta);\btheta) = \hat{\mathbf{F}}_t (\mathbf{P}^{\prime} \hat{\bpi}_{t-1 \mid t-1} (\btheta);\btheta) \mathbf{P}^{\prime} \hat{\bpi}_{t-1 \mid t-1} (\btheta),
\end{equation*}
where $\hat{\bpi}_{0 \mid 0} (\btheta)$ is given. The fact that also $\sup_{\btheta \in \bTheta} || \hat{\bpi}_{t \mid t} (\btheta) - \bpi_{t \mid t} (\btheta) ||_{2} \overset{e.a.s.}{\rightarrow} 0$ as $t \rightarrow \infty$ follows from Proposition 3 in \cite{Krabbe2025}, which is a generalisation of Theorem 2.10 in \cite{StraumannMikosch2006}, if
\begin{enumerate} [(i)]
	\item $\mathbb{E} [\log^{+} \sup_{\btheta \in \bTheta} || \bpi_{t \mid t} (\btheta) ||_{2} ] < \infty$,
	\item there exists an $\bs \in \mathcal{S}$ such that $\sup_{\btheta \in \bTheta} || \hat{\bphi}_t(\bs;\btheta) - \bphi_t(\bs;\btheta) ||_{2} \overset{e.a.s.}{\rightarrow} 0$ as $t \rightarrow \infty$, and
	\item $\sup_{\btheta \in \bTheta} \Lambda (\hat{\bphi}_t - \bphi_t;\btheta) \overset{e.a.s.}{\rightarrow} 0$ as $t \rightarrow \infty$.
\end{enumerate}
Condition (i) is trivial. For Condition (ii), we have that
\begin{equation}
	\begin{aligned}
		| [\hat{\bphi}_t(\bs;\btheta) - \bphi_t(\bs;\btheta)]_{j} | &= \left| \frac{\mathbf{p}_{\bullet j}^{\prime} \bs f_{j}(Y_t;\hat{X}_{j,t} (\bupsilon_j),\bupsilon_j)}{\bs^{\prime} \mathbf{P} \mathbf{f}(Y_t;\hat{\bX}_{t} (\bupsilon),\bupsilon)} - \frac{\mathbf{p}_{\bullet j}^{\prime} \bs f_{j}(Y_t;X_{j,t} (\bupsilon_j),\bupsilon_j)}{\bs^{\prime} \mathbf{P} \mathbf{f}(Y_t;\bX_{t} (\bupsilon),\bupsilon)} \right| \\ 
		&\leq \left| \frac{f_{j}(Y_t;\hat{X}_{j,t} (\bupsilon_j),\bupsilon_j) }{\bs^{\prime} \mathbf{P} \mathbf{f}(Y_t;\hat{\bX}_{t} (\bupsilon),\bupsilon)} - \frac{f_{j}(Y_t;X_{j,t} (\bupsilon_j),\bupsilon_j)}{\bs^{\prime} \mathbf{P} \mathbf{f}(Y_t;\bX_{t} (\bupsilon),\bupsilon)} \right|
	\end{aligned} \label{eq:cii1}
\end{equation}
for all $j \in \{1,...,J\}$ where $\mathbf{p}_{\bullet j} := (p_{1j},...,p_{Jj})^{\prime}$, $\mathbf{f}(Y_t;\hat{\bX}_{t} (\bupsilon),\bupsilon) := (f_{1}(Y_t;\hat{X}_{1,t} (\bupsilon_1),\bupsilon_1),...,f_{J}(Y_t;\hat{X}_{J,t} (\bupsilon_J),\bupsilon_J))^{\prime}$, $\mathbf{f}(Y_t;\bX_{t} (\bupsilon),\bupsilon) := (f_{1}(Y_t;X_{1,t} (\bupsilon_1),\bupsilon_1),...,f_{J}(Y_t;X_{J,t} (\bupsilon_J),\bupsilon_J))^{\prime}$, and the inequality follows since $p_{ij} \leq 1$ for all $i,j \in \{1,...,J\}$. Let $m_{t}^{\btheta} : \mathcal{X}_{1} \times \cdots \times \mathcal{X}_{J} \rightarrow \mathbb{R}$ be given by
\begin{equation*}
	m_{t}^{\btheta} (\bx) = \frac{f_{j}(Y_t;x_{j},\bupsilon_j)}{\bs^{\prime} \mathbf{P} \mathbf{f}(Y_t;\bx,\bupsilon)}.
\end{equation*}
Note that $m_{t}^{\btheta} (\hat{\bX}_{t} (\bupsilon)) = \frac{f_{j}(Y_t;\hat{X}_{j,t} (\bupsilon_j),\bupsilon_j) }{\bs^{\prime} \mathbf{P} \mathbf{f}(Y_t;\hat{\bX}_{t} (\bupsilon),\bupsilon)}$ and $m_{t}^{\btheta} (\bX_{t} (\bupsilon)) = \frac{f_{j}(Y_t;X_{j,t} (\bupsilon_j),\bupsilon_j)}{\bs^{\prime} \mathbf{P} \mathbf{f}(Y_t;\bX_{t} (\bupsilon),\bupsilon)}$. Thus, by the mean value theorem and the Cauchy–Schwarz inequality, there exists an $\bar{\bX}_t (\boldsymbol{\upsilon}) \in \mathcal{X}_{1} \times \cdots \times \mathcal{X}_{J}$ such that
\begin{equation}
	\left| \frac{f_{j}(Y_t;\hat{X}_{j,t} (\bupsilon_j),\bupsilon_j) }{\bs^{\prime} \mathbf{P} \mathbf{f}(Y_t;\hat{\bX}_{t} (\bupsilon),\bupsilon)} - \frac{f_{j}(Y_t;X_{j,t} (\bupsilon_j),\bupsilon_j)}{\bs^{\prime} \mathbf{P} \mathbf{f}(Y_t;\bX_{t} (\bupsilon),\bupsilon)} \right| \leq || \nabla_{\bx} m_{t}^{\btheta} (\bar{\bX}_t (\boldsymbol{\upsilon})) ||_{2} || \hat{\bX}_t (\boldsymbol{\upsilon}) - \bX_t (\boldsymbol{\upsilon}) ||_{2}. \label{eq:term1}
\end{equation}
If $m \neq j$, then
\begin{equation}
	\begin{aligned}
		| [\nabla_{\bx} m_{t}^{\btheta} (\bar{\bX}_t (\boldsymbol{\upsilon}))]_{m} | &\leq \frac{f_{j}(Y_t;\bar{X}_{j,t}(\bupsilon_j),\bupsilon_j) \mathbf{p}_{\bullet m}^{\prime} \bs | \nabla_{x_m} f_{m}(Y_t;\bar{X}_{m,t}(\bupsilon_m),\bupsilon_m) |}{(\bs^{\prime} \mathbf{P} \mathbf{f}(Y_t;\bar{\bX}_t (\bupsilon),\bupsilon))^2}  \\
		&\leq \frac{1}{p_{-}^2} \frac{f_{j}(Y_t;\bar{X}_{j,t}(\bupsilon_j),\bupsilon_j) | \nabla_{x_m} f_{m}(Y_t;\bar{X}_{m,t}(\bupsilon_m),\bupsilon_m) | }{(\mathbf{1}^{\prime} \mathbf{f}(Y_t;\bar{\bX}_t (\bupsilon),\bupsilon))^2}  \\
		&\leq \frac{1}{p_{-}^2} | \nabla_{x_m} \log f_{m}(Y_t;\bar{X}_{m,t}(\bupsilon_m),\bupsilon_m) |,
	\end{aligned} \label{eq:term2}
\end{equation}
and if $m = j$, then
\begin{equation}
	\begin{aligned}
		| [\nabla_{\bx} m_{t}^{\btheta} (\bar{\bX}_t (\boldsymbol{\upsilon}))]_{m} | &\leq \frac{| \nabla_{x_m} f_{m}(Y_t;\bar{X}_{m,t}(\bupsilon_m),\bupsilon_m) |}{\bs^{\prime} \mathbf{P} \mathbf{f}(Y_t;\bar{\bX}_t (\bupsilon),\bupsilon)} \\
		&+ \frac{f_{m}(Y_t;\bar{X}_{m,t}(\bupsilon_m),\bupsilon_m)\mathbf{p}_{\bullet m}^{\prime} \bs | \nabla_{x_m} f_{m}(Y_t;\bar{X}_{m,t}(\bupsilon_m),\bupsilon_m) | }{(\bs^{\prime} \mathbf{P} \mathbf{f}(Y_t;\bar{\bX}_t (\bupsilon),\bupsilon))^2} \\
		&\leq \frac{1}{p_{-}} \frac{| \nabla_{x_m} f_{m}(Y_t;\bar{X}_{m,t}(\bupsilon_m),\bupsilon_m) | }{\mathbf{1}^{\prime} \mathbf{f}(Y_t;\bar{\bX}_t (\bupsilon),\bupsilon)} \\
		&+ \frac{1}{p_{-}^2} \frac{f_{m}(Y_t;\bar{X}_{m,t}(\bupsilon_m),\bupsilon_m) | \nabla_{x_m} f_{m}(Y_t;\bar{X}_{m,t}(\bupsilon_m),\bupsilon_m) | }{(\mathbf{1}^{\prime} \mathbf{f}(Y_t;\bar{\bX}_t (\bupsilon),\bupsilon))^2} \\
		&\leq \left(\frac{1}{p_{-}} + \frac{1}{p_{-}^2}\right) | \nabla_{x_m} \log f_{m}(Y_t;\bar{X}_{m,t}(\bupsilon_m),\bupsilon_m) |,
	\end{aligned} \label{eq:term3}
\end{equation}
where $\mathbf{1} := (1,...,1)^{\prime}$ is a $J \times 1$ vector of ones and the second inequality in both Equations \eqref{eq:term2} and \eqref{eq:term3} follows since $p_{-} \leq p_{ij} \leq 1$ for all $i,j \in \{1,...,J\}$. Thus, continuing from Equation \eqref{eq:term1},
\begin{equation}
	\begin{aligned}
		&\left| \frac{f_{j}(Y_t;\hat{X}_{j,t} (\bupsilon_j),\bupsilon_j) }{\bs^{\prime} \mathbf{P} \mathbf{f}(Y_t;\hat{\bX}_{t} (\bupsilon),\bupsilon)} - \frac{f_{j}(Y_t;X_{j,t} (\bupsilon_j),\bupsilon_j)}{\bs^{\prime} \mathbf{P} \mathbf{f}(Y_t;\bX_{t} (\bupsilon),\bupsilon)} \right| \\
		&\leq C \sum_{m,n=1}^{J} | \nabla_{x_m} \log f_{m}(Y_t;\bar{X}_{m,t}(\bupsilon_m),\bupsilon_m) | | \hat{X}_{n,t} (\bupsilon_n) - X_{n,t} (\bupsilon_n) |.
	\end{aligned}
	\label{eq:term4}
\end{equation}
Hence, continuing from Equation \eqref{eq:cii1} and using Equation \eqref{eq:term4}, we have that
\begin{equation}
	\begin{aligned}
		&\sup_{\btheta \in \bTheta} | [\hat{\bphi}_t(\bs;\btheta) - \bphi_t(\bs;\btheta)]_{j} | \\
		&\leq C \sum_{m,n=1}^{J} \sup_{\bupsilon_m \in \bUpsilon_m} \sup_{x_m \in \mathcal{X}_{m}} | \nabla_{x_m} \log f_{m}(Y_t;x_m,\bupsilon_m) | \sup_{\bupsilon_n \in \bUpsilon_n} | \hat{X}_{n,t} (\bupsilon_n) - X_{n,t} (\bupsilon_n) |
	\end{aligned} \label{eq:cii2}
\end{equation}
for all $j \in \{1,...,J\}$. Condition (ii) thus follows from Lemmas 2.1 and 2.2 in \citet{StraumannMikosch2006} since $(Y_t)_{t \in \mathbb{Z}}$ is stationary by Assumption \ref{AssumptionY}, for each $j \in \{1,...,J\}$, there exists an $m_j > 0$ such that $\mathbb{E} \left[ \sup_{\bupsilon_j \in \bUpsilon_j} \sup_{x_j \in \mathcal{X}_{j}} \left| \nabla_{x_j} \log f_{j} (Y_t;x_j,\bupsilon_j) \right|^{m_j} \right] < \infty$ by assumption, and, for each $j \in \{1,...,J\}$, $\sup_{\bupsilon_{j} \in \bUpsilon_{j}} | \hat{X}_{j,t} (\bupsilon_j) - X_{j,t} (\bupsilon_j) | \overset{e.a.s.}{\rightarrow} 0$ as $t \rightarrow \infty$ by Lemma \ref{LemmaInvertibilityX}. For Condition (iii), note that, by the mean value theorem and the Cauchy–Schwarz inequality,
\begin{equation*}
	| [(\hat{\bphi}_t(\bx;\btheta) - \bphi_t(\bx;\btheta)) - (\hat{\bphi}_t(\by;\btheta) - \bphi_t(\by;\btheta))]_{j} | \leq C \sup_{\bs \in \mathcal{S}} || \nabla_{\bs} [\hat{\bphi}_t(\bs;\btheta)]_{j} - \nabla_{\bs} [\bphi_t(\bs;\btheta)]_{j} ||_{2} ||\bx - \by ||_{2}
\end{equation*}
for all $j \in \{1,...,J\}$.\footnote{Note that if $ f : \mathcal{S} \rightarrow \mathbb{R} $ where $s_J = 1-\sum_{j=1}^{J-1}s_j$, then $[\nabla_{\bs}^{R} f (\bs)]_j := \nabla_{s_j} f (\bs) - \nabla_{s_J} f (\bs), j \in \{1,...,J-1\}$ satisfies $$ || \nabla_{\bs}^{R} f (\bs) ||_2 \leq C || \nabla_{\bs}^{U} f (\bs) ||_2, $$ where $[\nabla_{\bs}^{U} f (\bs)]_j := \nabla_{s_j} f (\bs), j \in \{1,...,J\}$.} We have that
\begin{equation}
	\begin{aligned}
		| [ \nabla_{\bs} [\hat{\bphi}_t(\bs;\btheta)]_{j} - \nabla_{\bs} [\bphi_t(\bs;\btheta)]_{j} ]_{r} | &\leq \left| \frac{ p_{rj} f_{j}(Y_t;\hat{X}_{j,t} (\bupsilon_j),\bupsilon_j)}{\bs^{\prime} \mathbf{P} \mathbf{f}(Y_t;\hat{\bX}_{t} (\bupsilon),\bupsilon)} - \frac{ p_{rj} f_{j}(Y_t;X_{j,t} (\bupsilon_j),\bupsilon_j)}{\bs^{\prime} \mathbf{P} \mathbf{f}(Y_t;\bX_{t} (\bupsilon),\bupsilon)} \right| \\
		&+ \left| \frac{(\mathbf{p}_{\bullet j}^{\prime} \bs f_{j}(Y_t;\hat{X}_{j,t} (\bupsilon_j),\bupsilon_j)) (\mathbf{p}_{r \bullet}^{\prime} \mathbf{f}(Y_t;\hat{\bX}_{t} (\bupsilon),\bupsilon))}{(\bs^{\prime} \mathbf{P} \mathbf{f}(Y_t;\hat{\bX}_{t} (\bupsilon),\bupsilon))^2} \right. \\
		& \left. \quad \ - \frac{(\mathbf{p}_{\bullet j}^{\prime} \bs f_{j}(Y_t;{X}_{j,t} (\bupsilon_j),\bupsilon_j)) (\mathbf{p}_{r \bullet}^{\prime} \mathbf{f}(Y_t;{\bX}_{t} (\bupsilon),\bupsilon))}{(\bs^{\prime} \mathbf{P} \mathbf{f}(Y_t;{\bX}_{t} (\bupsilon),\bupsilon))^2} \right| \\
		&\leq \left| \frac{ p_{rj} f_{j}(Y_t;\hat{X}_{j,t} (\bupsilon_j),\bupsilon_j)}{\bs^{\prime} \mathbf{P} \mathbf{f}(Y_t;\hat{\bX}_{t} (\bupsilon),\bupsilon)} - \frac{ p_{rj} f_{j}(Y_t;X_{j,t} (\bupsilon_j),\bupsilon_j)}{\bs^{\prime} \mathbf{P} \mathbf{f}(Y_t;\bX_{t} (\bupsilon),\bupsilon)} \right| \\
		&+ \frac{1}{p_{-}} \left| \frac{\mathbf{p}_{\bullet j}^{\prime} \bs f_{j}(Y_t;\hat{X}_{j,t} (\bupsilon_j),\bupsilon_j)}{\bs^{\prime} \mathbf{P} \mathbf{f}(Y_t;\hat{\bX}_{t} (\bupsilon),\bupsilon)} -  \frac{\mathbf{p}_{\bullet j}^{\prime} \bs f_{j}(Y_t;X_{j,t} (\bupsilon_j),\bupsilon_j)}{\bs^{\prime} \mathbf{P} \mathbf{f}(Y_t;\bX_{t} (\bupsilon),\bupsilon)} \right| \\
		&+ \frac{1}{p_{-}} \left| \frac{\mathbf{p}_{r \bullet}^{\prime} \mathbf{f}(Y_t;\hat{\bX}_{t} (\bupsilon),\bupsilon) }{\bs^{\prime} \mathbf{P} \mathbf{f}(Y_t;\hat{\bX}_{t} (\bupsilon),\bupsilon)} - \frac{\mathbf{p}_{r \bullet}^{\prime} \mathbf{f}(Y_t;\bX_{t} (\bupsilon),\bupsilon)}{\bs^{\prime} \mathbf{P} \mathbf{f}(Y_t;\bX_{t} (\bupsilon),\bupsilon)} \right| \\
		&\leq C \sum_{k=1}^{J} \left| \frac{f_{k}(Y_t;\hat{X}_{k,t} (\bupsilon_k),\bupsilon_k) }{\bs^{\prime} \mathbf{P} \mathbf{f}(Y_t;\hat{\bX}_{t} (\bupsilon),\bupsilon)} - \frac{f_{k}(Y_t;X_{k,t} (\bupsilon_k),\bupsilon_k)}{\bs^{\prime} \mathbf{P} \mathbf{f}(Y_t;\bX_{t} (\bupsilon),\bupsilon)} \right| \label{eq:ciii}
	\end{aligned}
\end{equation}
for all $j \in \{1,...,J\}$ and $r \in \{1,...,J\}$ where $\mathbf{p}_{i \bullet} := (p_{i1},...,p_{iJ})^{\prime}$, the first inequality follows from the triangle inequality, the second inequality follows since $\hat{a}\hat{b} - ab = (\hat{a}-a)\hat{b} + a(\hat{b}-b)$ and $p_{-} \leq p_{ij} \leq 1$ for all $i,j \in \{1,...,J\}$, and the third inequality follows since $p_{ij} \leq 1$ for all $i,j \in \{1,...,J\}$. The upper bound in Equation \eqref{eq:ciii} is nothing but a constant times the sum of the upper bounds in Equation \eqref{eq:cii1}. Hence, by using the same arguments as the ones used to go from Equation \eqref{eq:cii1} to Equation \eqref{eq:cii2}, we have that
\begin{align*}
	&\sup_{\btheta \in \bTheta} \sup_{\bs \in \mathcal{S}} | [ \nabla_{\bs} [\hat{\bphi}_t(\bs;\btheta)]_{j} - \nabla_{\bs} [\bphi_t(\bs;\btheta)]_{j} ]_{r} | \\
	&\leq C \sum_{m,n=1}^{J} \sup_{\bupsilon_m \in \bUpsilon_m} \sup_{x_m \in \mathcal{X}_{m}} | \nabla_{x_m} \log f_{m}(Y_t;x_m,\bupsilon_m) | \sup_{\bupsilon_n \in \bUpsilon_n} | \hat{X}_{n,t} (\bupsilon_n) - X_{n,t} (\bupsilon_n) |
\end{align*}
for all $j \in \{1,...,J\}$ and $r \in \{1,...,J\}$. Condition (iii) thus follows by using the same arguments as the ones used to show Condition (ii), which concludes the second part of the proof.

\subsection{Proof of Theorem \ref{TheoremConsistency}}

The result follows from Lemmas 3.1 and 4.1 in \citet{PotscherPrucha1997} if
\begin{enumerate}[(i)]
	\item $\bTheta$ is compact, 
	\item $\sup_{\btheta \in \bTheta} | \hat{L}_T (\btheta) - L(\btheta) | \overset{a.s.}{\rightarrow} 0$ as $T \rightarrow \infty$,
	\item $\btheta \mapsto L(\btheta)$ is continuous, and 
	\item $L(\btheta) \leq L(\btheta_0)$ for all $\btheta \in \bTheta$ with equality if and only if  $\btheta = \btheta_0$.
\end{enumerate}
Condition (i) follows from Assumption \ref{AssumptionTheta1}. Note that
\begin{equation*}
	\sup_{\btheta \in \bTheta} | \hat{L}_T (\btheta) - L(\btheta) | \leq \sup_{\btheta \in \bTheta} | \hat{L}_T (\btheta) - L_T (\btheta) | + \sup_{\btheta \in \bTheta} | L_T (\btheta) - L(\btheta) |.
\end{equation*}
Condition (ii) thus follows since $\sup_{\btheta \in \bTheta} | \hat{L}_T (\btheta) - L_T (\btheta) | \overset{a.s.}{\rightarrow} 0$ as $T \rightarrow \infty$ by Lemma \ref{LemmaConsistencyA} and $\sup_{\btheta \in \bTheta} | L_T (\btheta) - L(\btheta) | \overset{a.s.}{\rightarrow} 0$ as $T \rightarrow \infty$ by Lemma \ref{LemmaConsistencyB}, and Condition (iii) is a by-product of the uniform law of large numbers used in the proof of Lemma \ref{LemmaConsistencyB}. Finally, Condition (iv) follows from Lemma \ref{LemmaConsistencyC}.

\subsection{Proof of Lemma \ref{LemmaInvertibilityDX}}

For each $j \in \{1,...,J\}$, $(\nabla_{\bupsilon_{j}} X_{j,t} (\bupsilon_{j}))_{t \in \mathbb{Z}}$ is a stochastic process taking values in $\mathbb{R}^{d_j}$ given by
\begin{equation*}
	\nabla_{\bupsilon_{j}} X_{j,t+1} (\bupsilon_{j}) = \bold{\phi}_{j,t}^{x} (\nabla_{\bupsilon_{j}} X_{j,t} (\bupsilon_{j});\bupsilon_{j})
\end{equation*}
with
\begin{equation*}
	\bold{\phi}_{j,t}^{x} (\nabla_{\bupsilon_{j}} X_{j,t} (\bupsilon_{j});\bupsilon_{j}) = \partial_{x_j} \phi_{j} (Y_t,X_{j,t} (\bupsilon_{j});\bupsilon_{j}) \nabla_{\bupsilon_{j}} X_{j,t} (\bupsilon_{j}) + \partial_{\bupsilon_j} \phi_{j} (Y_t,X_{j,t} (\bupsilon_{j});\bupsilon_{j}).
\end{equation*}
The first conclusion thus follows from Theorem 2.10 in \cite{StraumannMikosch2006} since, for each $j \in \{1,...,J\}$,
\begin{enumerate} [(i)]
	\item $\mathbb{E} [\log^{+} \sup_{\bupsilon_{j} \in \bar{\bUpsilon}_{j}} || \partial_{\bupsilon_j} \phi_{j} (Y_t,X_{j,t} (\bupsilon_{j});\bupsilon_{j}) ||_{2} ] < \infty$,
	\item $\mathbb{E} [\log^{+} \sup_{\bupsilon_{j} \in \bar{\bUpsilon}_{j}} | \partial_{x_j} \phi_{j} (Y_t,X_{j,t} (\bupsilon_{j});\bupsilon_{j}) | ] < \infty$, and
	\item $-\infty \leq \mathbb{E} [\log \sup_{\bupsilon_{j} \in \bar{\bUpsilon}_{j}} | \partial_{x_j} \phi_{j} (Y_t,X_{j,t} (\bupsilon_{j});\bupsilon_{j}) | ] < 0$
\end{enumerate}
by assumption.\footnote{Conditions (ii) and (iii) follow from Conditions (ii) and (iii) in Lemma \ref{LemmaInvertibilityX}.}

Moreover, for each $j \in \{1,...,J\}$, $(\nabla_{\bupsilon_{j}} \hat{X}_{j,t} (\bupsilon_{j}))_{t \in \mathbb{N}}$ is a stochastic process taking values in $\mathbb{R}^{d_j}$ given by
\begin{equation*}
	\nabla_{\bupsilon_{j}} \hat{X}_{j,t+1} (\bupsilon_{j}) = \hat{\bold{\phi}}_{j,t}^{x} (\nabla_{\bupsilon_{j}} \hat{X}_{j,t} (\bupsilon_{j});\bupsilon_{j})
\end{equation*}
with
\begin{equation*}
	\hat{\bold{\phi}}_{j,t}^{x} (\nabla_{\bupsilon_{j}} \hat{X}_{j,t} (\bupsilon_{j});\bupsilon_{j}) = \partial_{x_j} \phi_{j} (Y_t,\hat{X}_{j,t} (\bupsilon_{j});\bupsilon_{j}) \nabla_{\bupsilon_{j}} \hat{X}_{j,t} (\bupsilon_{j}) + \partial_{\bupsilon_j} \phi_{j} (Y_t,\hat{X}_{j,t} (\bupsilon_{j});\bupsilon_{j}),
\end{equation*}
where $\nabla_{\bupsilon_{j}} \hat{X}_{j,1} (\bupsilon_{j})$ is given. The second conclusion thus follows from Theorem 2.10 in \cite{StraumannMikosch2006} as well since, for each $j \in \{1,...,J\}$,
\begin{enumerate} [(i)]
	\item $\mathbb{E} [\log^{+} \sup_{\bupsilon_{j} \in \bar{\bUpsilon}_{j}} || \nabla_{\bupsilon_{j}} X_{j,t} (\bupsilon_{j}) ||_{2} ] < \infty$,
	\item $\sup_{\bupsilon_{j} \in \bar{\bUpsilon}_{j}} || \partial_{\bupsilon_j} \phi_{j} (Y_t,\hat{X}_{j,t} (\bupsilon_{j});\bupsilon_{j}) - \partial_{\bupsilon_j} \phi_{j} (Y_t,X_{j,t} (\bupsilon_{j});\bupsilon_{j}) ||_{2} \overset{e.a.s.}{\rightarrow} 0$ as $t \rightarrow \infty$, and 
	\item $\sup_{\bupsilon_{j} \in \bar{\bUpsilon}_{j}} | \partial_{x_j} \phi_{j} (Y_t,\hat{X}_{j,t} (\bupsilon_{j});\bupsilon_{j}) - \partial_{x_j} \phi_{j} (Y_t,X_{j,t} (\bupsilon_{j});\bupsilon_{j}) | \overset{e.a.s.}{\rightarrow} 0$ as $t \rightarrow \infty$
\end{enumerate}
by assumption.

\subsection{Proof of Lemma \ref{LemmaInvertibilityDDX}}

For each $j \in \{1,...,J\}$, $(\nabla_{\bupsilon_{j} \bupsilon_{j}} X_{j,t} (\bupsilon_{j}))_{t \in \mathbb{Z}}$ is a stochastic process taking values in $\mathbb{R}^{d_j \times d_j}$ given by
\begin{equation*}
	\nabla_{\bupsilon_{j} \bupsilon_{j}} X_{j,t+1} (\bupsilon_{j}) = \bold{\phi}_{j,t}^{xx} (\nabla_{\bupsilon_{j} \bupsilon_{j}} X_{j,t} (\bupsilon_{j});\bupsilon_{j})
\end{equation*}
with
\begin{align*}
	&\bold{\phi}_{j,t}^{xx} (\nabla_{\bupsilon_{j} \bupsilon_{j}} X_{j,t} (\bupsilon_{j});\bupsilon_{j}) \\
	&= \partial_{x_j} \phi_{j} (Y_t,X_{j,t} (\bupsilon_{j});\bupsilon_{j}) \nabla_{\bupsilon_{j} \bupsilon_{j}} X_{j,t} (\bupsilon_{j}) + \partial_{x_j x_j} \phi_{j} (Y_t,X_{j,t} (\bupsilon_{j});\bupsilon_{j}) \nabla_{\bupsilon_{j}} X_{j,t} (\bupsilon_{j}) \nabla_{\bupsilon_{j}^{\prime}} X_{j,t} (\bupsilon_{j}) \\
	&+ \nabla_{\bupsilon_{j}} X_{j,t} (\bupsilon_{j}) \partial_{x_j \bupsilon_j} \phi_{j} (Y_t,X_{j,t} (\bupsilon_{j});\bupsilon_{j}) + \partial_{\bupsilon_j x_j} \phi_{j} (Y_t,X_{j,t} (\bupsilon_{j});\bupsilon_{j}) \nabla_{\bupsilon_{j}^{\prime}} X_{j,t} (\bupsilon_{j}) \\
	&+ \partial_{\bupsilon_j \bupsilon_j} \phi_{j} (Y_t,X_{j,t} (\bupsilon_{j});\bupsilon_{j}).
\end{align*}
The first conclusion thus follows from Theorem 2.10 in \cite{StraumannMikosch2006} since, for each $j \in \{1,...,J\}$,
\begin{enumerate} [(i)]
	\item $\mathbb{E} [\log^{+} \sup_{\bupsilon_{j} \in \bar{\bUpsilon}_{j}} || \partial_{\bvartheta_j \bvartheta_j} \phi_{j} (Y_t,X_{j,t} (\bupsilon_{j});\bupsilon_{j}) ||_{2,2} ] < \infty$,
	\item $\mathbb{E} [\log^{+} \sup_{\bupsilon_{j} \in \bar{\bUpsilon}_{j}} | \partial_{x_j} \phi_{j} (Y_t,X_{j,t} (\bupsilon_{j});\bupsilon_{j}) | ] < \infty$, and
	\item $-\infty \leq \mathbb{E} [\log \sup_{\bupsilon_{j} \in \bar{\bUpsilon}_{j}} | \partial_{x_j} \phi_{j} (Y_t,X_{j,t} (\bupsilon_{j});\bupsilon_{j}) | ] < 0$
\end{enumerate}
by assumption.

Moreover, for each $j \in \{1,...,J\}$, $(\nabla_{\bupsilon_{j} \bupsilon_{j}} \hat{X}_{j,t} (\bupsilon_{j}))_{t \in \mathbb{N}}$ is a stochastic process taking values in $\mathbb{R}^{d_j \times d_j}$ given by
\begin{equation*}
	\nabla_{\bupsilon_{j} \bupsilon_{j}} \hat{X}_{j,t+1} (\bupsilon_{j}) = \hat{\bold{\phi}}_{j,t}^{xx} (\nabla_{\bupsilon_{j} \bupsilon_{j}} \hat{X}_{j,t} (\bupsilon_{j});\bupsilon_{j})
\end{equation*}
with
\begin{align*}
	&\hat{\bold{\phi}}_{j,t}^{xx} (\nabla_{\bupsilon_{j} \bupsilon_{j}} \hat{X}_{j,t} (\bupsilon_{j});\bupsilon_{j}) \\
	&= \partial_{x_j} \phi_{j} (Y_t,\hat{X}_{j,t} (\bupsilon_{j});\bupsilon_{j}) \nabla_{\bupsilon_{j} \bupsilon_{j}} \hat{X}_{j,t} (\bupsilon_{j}) + \partial_{x_j x_j} \phi_{j} (Y_t,\hat{X}_{j,t} (\bupsilon_{j});\bupsilon_{j}) \nabla_{\bupsilon_{j}} \hat{X}_{j,t} (\bupsilon_{j}) \nabla_{\bupsilon_{j}^{\prime}} \hat{X}_{j,t} (\bupsilon_{j}) \\
	&+ \nabla_{\bupsilon_{j}} \hat{X}_{j,t} (\bupsilon_{j}) \partial_{x_j \bupsilon_j} \phi_{j} (Y_t,\hat{X}_{j,t} (\bupsilon_{j});\bupsilon_{j}) + \partial_{\bupsilon_j x_j} \phi_{j} (Y_t,\hat{X}_{j,t} (\bupsilon_{j});\bupsilon_{j}) \nabla_{\bupsilon_{j}^{\prime}} \hat{X}_{j,t} (\bupsilon_{j}) \\
	&+ \partial_{\bupsilon_j \bupsilon_j} \phi_{j} (Y_t,\hat{X}_{j,t} (\bupsilon_{j});\bupsilon_{j}),
\end{align*}
where $\nabla_{\bupsilon_{j} \bupsilon_{j}} \hat{X}_{j,1} (\bupsilon_{j})$ is given. The second conclusion thus follows from Theorem 2.10 in \cite{StraumannMikosch2006} as well since, for each $j \in \{1,...,J\}$,
\begin{enumerate} [(i)]
	\item $\mathbb{E} [ \log^{+} \sup_{\bupsilon_{j} \in \bar{\bUpsilon}_{j}} || \nabla_{\bupsilon_{j} \bupsilon_{j}} X_{j,t} (\bupsilon_{j}) ||_{2,2} ] < \infty$,
	\item $\sup_{\bupsilon_{j} \in \bar{\bUpsilon}_{j}} || \partial_{\bvartheta_j \bvartheta_j} \phi_{j} (Y_t,\hat{X}_{j,t} (\bupsilon_{j});\bupsilon_{j}) - \partial_{\bvartheta_j \bvartheta_j} \phi_{j} (Y_t,X_{j,t} (\bupsilon_{j});\bupsilon_{j}) ||_{2,2} \overset{e.a.s.}{\rightarrow} 0$ as $t \rightarrow \infty$, and 
	\item $\sup_{\bupsilon_{j} \in \bar{\bUpsilon}_{j}} | \partial_{x_j} \phi_{j} (Y_t,\hat{X}_{j,t} (\bupsilon_{j});\bupsilon_{j}) - \partial_{x_j} \phi_{j} (Y_t,X_{j,t} (\bupsilon_{j});\bupsilon_{j}) | \overset{e.a.s.}{\rightarrow} 0$ as $t \rightarrow \infty$
\end{enumerate}
by assumption.\footnote{Condition (iii) follows from Condition (ii) in Lemma \ref{LemmaInvertibilityDX}.}

\subsection{Proof of Lemma \ref{LemmaInvertibilityDPrediction}}

The conclusion follows from Lemmas \ref{LemmaInvertibilityFilter} and \ref{LemmaInvertibilityDFilter}.

\subsection{Proof of Lemma \ref{LemmaInvertibilityDDPrediction}}

The conclusion follows from Lemmas \ref{LemmaInvertibilityFilter}, \ref{LemmaInvertibilityDFilter}, and \ref{LemmaInvertibilityDDFilter}.

\subsection{Proof of Theorem \ref{TheoremAsymptoticNormality}}

The result follows from Lemma 8.1 in \citet{PotscherPrucha1997} if
\begin{enumerate}[(i)]
	\item $\btheta_0 \in \textup{int}(\bTheta)$,
	\item $\hat{\btheta}_T \overset{\mathbb{P}}{\rightarrow} \btheta_0$ as $T \rightarrow \infty$,
	\item $\btheta \mapsto \hat{L}_T (\btheta)$ is twice continuously differentiable a.s.,
	\item $\sqrt{T} \nabla_{\btheta} \hat{L}_T (\btheta_0) \overset{d}{\rightarrow} \mathcal{N} (\mathbf{0},\bI(\btheta_0))$ as $T \rightarrow \infty$,
	\item $\sup_{\btheta \in \bar{\bTheta}} || \nabla_{\btheta \btheta} \hat{L}_T (\btheta) - (-\bI(\btheta)) ||_{2,2} \overset{\mathbb{P}}{\rightarrow} 0 $ as $T \rightarrow \infty$, and
	\item $\bI(\btheta_0)$ is invertible.
\end{enumerate}
First, Condition (i) follows from Assumption \ref{AssumptionTheta2}, Condition (ii) from Theorem \ref{TheoremConsistency}, and Condition (iii) follows from Assumptions \ref{AssumptionF3} and \ref{AssumptionPhi2}. Moreover, Condition (iv) follows since $\sup_{\btheta \in \bar{\bTheta}} || \sqrt{T} \nabla_{\btheta} \hat{L}_T (\btheta) - \sqrt{T} \nabla_{\btheta} L_T (\btheta) ||_{2} \overset{a.s.}{\rightarrow} 0$ as $T \rightarrow \infty$ by Lemma \ref{LemmaAsymptoticNormalityA}, $\sqrt{T} \nabla_{\btheta} L_T (\btheta_0) \overset{d}{\rightarrow} \mathcal{N} \left(\mathbf{0},\mathbb{E} \left[ \nabla_{\btheta} \log f (Y_t;\btheta_0) \nabla_{\btheta^{\prime}} \log f (Y_t;\btheta_0) \right] \right)$ as $T \rightarrow \infty$ by Lemma \ref{LemmaAsymptoticNormalityB}, and $\bI(\btheta_0) = \mathbb{E} \left[ \nabla_{\btheta} \log f (Y_t;\btheta_0) \nabla_{\btheta^{\prime}} \log f (Y_t;\btheta_0) \right]$ by Lemma \ref{LemmaAsymptoticNormalityC}. Moving on to Condition (v), note that
\begin{equation*}
	\sup_{\btheta \in \bar{\bTheta}} || \nabla_{\btheta \btheta} \hat{L}_T (\btheta) - (-\bI(\btheta)) ||_{2,2} \leq \sup_{\btheta \in \bar{\bTheta}} || \nabla_{\btheta \btheta} \hat{L}_T (\btheta) - \nabla_{\btheta \btheta} L_T (\btheta) ||_{2,2} + \sup_{\btheta \in \bar{\bTheta}} || \nabla_{\btheta \btheta} L_T (\btheta) - (-\bI(\btheta)) ||_{2,2}.
\end{equation*}
Condition (v) thus follows since $\sup_{\btheta \in \bar{\bTheta}} || \nabla_{\btheta \btheta} \hat{L}_T (\btheta) - \nabla_{\btheta \btheta} L_T (\btheta) ||_{2,2} \overset{a.s.}{\rightarrow} 0$ as $T \rightarrow \infty$ by Lemma \ref{LemmaAsymptoticNormalityD} and $\sup_{\btheta \in \bar{\bTheta}} || \nabla_{\btheta \btheta} L_T (\btheta) - (-\bI(\btheta)) ||_{2,2} \overset{a.s.}{\rightarrow} 0$ as $T \rightarrow \infty$ by Lemma \ref{LemmaAsymptoticNormalityE}. Finally, Condition (vi) is true by assumption.

\subsection{Proof of Proposition \ref{prop:varcovarmatrix}}

We have that
\begin{equation}
	\begin{aligned}
		&\left| \left| \hat{\bI}_T (\hat{\btheta}_T) - \bI (\btheta_0) \right| \right|_{2,2} \\
		&= \left| \left| \frac{1}{T} \sum_{t=1}^{T} \nabla_{\btheta} \log \hat{f} (Y_t;\hat{\btheta}_T) \nabla_{\btheta^{\prime}} \log 	\hat{f} (Y_t;\hat{\btheta}_T) - \mathbb{E} \left[ \nabla_{\btheta} \log f (Y_t;\btheta_0) \nabla_{\btheta^{\prime}} \log f (Y_t;\btheta_0) \right] \right| \right|_{2,2} \\
		&\leq \sup_{\btheta \in \bar{\bTheta}} \left| \left| \frac{1}{T} \sum_{t=1}^{T} \nabla_{\btheta} \log \hat{f} (Y_t;\btheta) 	\nabla_{\btheta^{\prime}} \log \hat{f} (Y_t;\btheta) - \frac{1}{T} \sum_{t=1}^{T} \nabla_{\btheta} \log f (Y_t;\btheta) \nabla_{\btheta^{\prime}} \log f (Y_t;\btheta) \right| \right|_{2,2} \\
		&+ \sup_{\btheta \in \bar{\bTheta}} \left| \left| \frac{1}{T} \sum_{t=1}^{T} \nabla_{\btheta} \log f (Y_t;\btheta) 	\nabla_{\btheta^{\prime}} \log f (Y_t;\btheta) - \mathbb{E} \left[ \nabla_{\btheta} \log f (Y_t;\btheta) \nabla_{\btheta^{\prime}} \log f (Y_t;\btheta) \right] \right| \right|_{2,2} \\
		&+ \left| \left| \mathbb{E} \left[ \nabla_{\btheta} \log f (Y_t;\hat{\btheta}_T) \nabla_{\btheta^{\prime}} \log f 	(Y_t;\hat{\btheta}_T) \right] - \mathbb{E} \left[ \nabla_{\btheta} \log f (Y_t;\btheta_0) \nabla_{\btheta^{\prime}} \log f (Y_t;\btheta_0) \right] \right| \right|_{2,2} \\
		&\leq \frac{1}{T} \sum_{t=1}^{T} \left( \sup_{\btheta \in \bar{\bTheta}} \left| \left| \nabla_{\btheta} \log \hat{f} (Y_t;\btheta) - 	\nabla_{\btheta} \log f (Y_t;\btheta) \right| \right|_2^2 \right. \\
		& \quad \quad \quad \quad \quad \left. + 2 \sup_{\btheta \in \bar{\bTheta}} \left| \left| \nabla_{\btheta} \log f (Y_t;\btheta) 	\right| \right|_2 \sup_{\btheta \in \bar{\bTheta}} \left| \left| \nabla_{\btheta} \log \hat{f} (Y_t;\btheta) - \nabla_{\btheta} \log f (Y_t;\btheta) \right| \right|_2 \right) \\
		&+ \sup_{\btheta \in \bar{\bTheta}} \left| \left| \frac{1}{T} \sum_{t=1}^{T} \nabla_{\btheta} \log f (Y_t;\btheta) 	\nabla_{\btheta^{\prime}} \log f (Y_t;\btheta) - \mathbb{E} \left[ \nabla_{\btheta} \log f (Y_t;\btheta) \nabla_{\btheta^{\prime}} \log f (Y_t;\btheta) \right] \right| \right|_{2,2} \\
		&+ \left| \left| \mathbb{E} \left[ \nabla_{\btheta} \log f (Y_t;\hat{\btheta}_T) \nabla_{\btheta^{\prime}} \log f 	(Y_t;\hat{\btheta}_T) \right] - \mathbb{E} \left[ \nabla_{\btheta} \log f (Y_t;\btheta_0) \nabla_{\btheta^{\prime}} \log f (Y_t;\btheta_0) \right] \right| \right|_{2,2} \\
		&= A_T + B_T + C_T
	\end{aligned} \label{eq:p1}
\end{equation}
since $\bI (\btheta_0) = \mathbb{E} [ \nabla_{\btheta} \log f (Y_t;\btheta_0) \nabla_{\btheta^{\prime}} \log f (Y_t;\btheta_0) ]$ by Lemma \ref{LemmaAsymptoticNormalityC}, where
\begin{align*}
	A_T &:= \frac{1}{T} \sum_{t=1}^{T} \left( \sup_{\btheta \in \bar{\bTheta}} \left| \left| \nabla_{\btheta} \log \hat{f} (Y_t;\btheta) - 	\nabla_{\btheta} \log f (Y_t;\btheta) \right| \right|_2^2 \right. \\
	& \quad \quad \quad \quad \quad \left. + 2 \sup_{\btheta \in \bar{\bTheta}} \left| \left| \nabla_{\btheta} \log f (Y_t;\btheta) 	\right| \right|_2 \sup_{\btheta \in \bar{\bTheta}} \left| \left| \nabla_{\btheta} \log \hat{f} (Y_t;\btheta) - \nabla_{\btheta} \log f (Y_t;\btheta) \right| \right|_2 \right), \\
	B_T &:= \sup_{\btheta \in \bar{\bTheta}} \left| \left| \frac{1}{T} \sum_{t=1}^{T} \nabla_{\btheta} \log f (Y_t;\btheta) 	\nabla_{\btheta^{\prime}} \log f (Y_t;\btheta) - \mathbb{E} \left[ \nabla_{\btheta} \log f (Y_t;\btheta) \nabla_{\btheta^{\prime}} \log f (Y_t;\btheta) \right] \right| \right|_{2,2}, \\
	C_T &:= \left| \left| \mathbb{E} \left[ \nabla_{\btheta} \log f (Y_t;\hat{\btheta}_T) \nabla_{\btheta^{\prime}} \log f 	(Y_t;\hat{\btheta}_T) \right] - \mathbb{E} \left[ \nabla_{\btheta} \log f (Y_t;\btheta_0) \nabla_{\btheta^{\prime}} \log f (Y_t;\btheta_0) \right] \right| \right|_{2,2};
\end{align*}
the second inequality in Equation \eqref{eq:p1} follows since $\hat{a}\hat{b}-ab = (\hat{a}-a)(\hat{b}-b)+(\hat{a}-a)b+a(\hat{b}-b)$. The result therefore follows if $A_T \overset{a.s.}{\rightarrow} 0$ as $T \rightarrow \infty$, $B_T \overset{a.s.}{\rightarrow} 0$ as $T \rightarrow \infty$, and $C_T \overset{a.s.}{\rightarrow} 0$ as $T \rightarrow \infty$. Note that $(\nabla_{\btheta} \log f (Y_t;\btheta))_{t \in \mathbb{Z}}$ is stationary and ergodic for all $\btheta \in \bar{\bTheta}$ and $\mathbb{E} [\sup_{\btheta \in \bar{\bTheta}} || \nabla_{\btheta} \log f (Y_t;\btheta) ||_2^2] < \infty$, which we showed in the proof of Lemma \ref{LemmaAsymptoticNormalityB}, and that $\sup_{\btheta \in \bar{\bTheta}} | | \nabla_{\btheta} \log \hat{f} (Y_t;\btheta) - \nabla_{\btheta} \log f (Y_t;\btheta) | |_2 \overset{e.a.s.}{\rightarrow} 0$ as $t \rightarrow \infty$, which we showed in the proof of Lemma \ref{LemmaAsymptoticNormalityA}. $A_T \overset{a.s.}{\rightarrow} 0$ as $T \rightarrow \infty$ therefore follows from Lemmas 2.1 and 2.2 in \cite{StraumannMikosch2006}, $B_T \overset{a.s.}{\rightarrow} 0$ as $T \rightarrow \infty$ follows from the uniform law of large numbers by \cite{Rao1962}, and $C_T \overset{a.s.}{\rightarrow} 0$ as $T \rightarrow \infty$ follows from the continuous mapping theorem since $\btheta \mapsto \mathbb{E} [ \nabla_{\btheta} \log f (Y_t;\btheta) \nabla_{\btheta^{\prime}} \log f (Y_t;\btheta) ]$ is continuous, which is a by-product of the uniform law of large numbers by \cite{Rao1962}, and $\hat{\btheta}_T \overset{a.s.}{\rightarrow} \btheta_0$ as $T \rightarrow \infty$.

\subsection{Proof of Theorem \ref{theo:c}}

In the following, $\underline{\btheta} = \inf_{\btheta \in \bTheta} \btheta$ and $\overline{\btheta} = \sup_{\btheta \in \bTheta} \btheta$. Note that $((S_t,Y_t))_{t \in \mathbb{Z}}$ is stationary and ergodic by Theorem \ref{theo:se}. We therefore only need to verify Assumptions \ref{AssumptionTheta1}-\ref{AssumptionIdentification} and the conditions in Lemmas \ref{LemmaInvertibilityX} and \ref{LemmaInvertibilityFilter}.

Let $j \in \{1,...,J\}$ be given. First, Assumption \ref{AssumptionTheta1} is true by assumption and Assumptions \ref{AssumptionF1} and \ref{AssumptionPhi1} are clearly satisfied.

We now verify the conditions in Lemmas \ref{LemmaInvertibilityX} and \ref{LemmaInvertibilityFilter}. First, by Lemma 2.2 in \cite{StraumannMikosch2006},
\begin{equation*}
	\mathbb{E} \left[ \log^{+} \sup_{\bupsilon_{j} \in \bUpsilon_{j}} \left| \phi_{j} (Y_t,x_j;\bupsilon_{j}) - x_j \right| \right] = \mathbb{E} \left[ \log^{+} \sup_{\bupsilon_{j} \in \bUpsilon_{j}} \left| \omega_{j} + \alpha_j Y_t^2 + \beta_j x_j - x_j \right| \right] \leq C_j + 2 \mathbb{E} \left[ \log^{+} \left| Y_t \right| \right]
\end{equation*}
for all $x_j \in \mathcal{X}_{j}$ where $C_j = 6 \log 2 + \log^{+} \overline{\omega}_{j} + \log^{+} \overline{\alpha}_{j} + \log^{+} \overline{\beta}_{j} + 2 \log^{+} x_j < \infty$, so Condition (i) in Lemma \ref{LemmaInvertibilityX} is satisfied since $\mathbb{E} [Y_t^2] < \infty$ implies that $\mathbb{E} [\log^{+} |Y_t|] < \infty$ by Lemma 2.2 in \cite{StraumannMikosch2006}. Moreover,
\begin{equation*}
	\mathbb{E} \left[ \log \sup_{\bupsilon_{j} \in \bUpsilon_{j}} \Lambda_{j,t} (\bupsilon_{j}) \right] = \mathbb{E} \left[ \log \sup_{\bupsilon_{j} \in \bUpsilon_{j}} \beta_j \right] = \log \overline{\beta}_j,
\end{equation*}
so Conditions (ii) and (iii) in Lemma \ref{LemmaInvertibilityX} are also satisfied. Finally,
\begin{equation*}
	\mathbb{E} \left[ \sup_{\bupsilon_j \in \bUpsilon_j} \sup_{x_j \in \mathcal{X}_{j}} \left| \nabla_{x_j} \log f_{j} (Y_t;x_j,\bupsilon_j) \right| \right] = \mathbb{E} \left[ \sup_{\bupsilon_j \in \bUpsilon_j} \sup_{x_j \in \mathcal{X}_{j}} \left| - \frac{1}{2} \frac{1}{x_j} + \frac{1}{2} \frac{Y_t^2}{x_j^2} \right| \right] \leq \frac{1}{2\underline{x}_j} + \frac{1}{2\underline{x}_j^2} \mathbb{E} \left[ Y_t^2 \right],
\end{equation*}
where $\underline{x}_j = \frac{\underline{\omega}_j}{1-\underline{\beta}_j} > 0$, so the condition in Lemma \ref{LemmaInvertibilityFilter} is also satisfied since $\mathbb{E} [Y_t^2] < \infty$.

Moreover, note that
\begin{align*}
	\mathbb{E} \left[ \sup_{\bupsilon_j \in \bUpsilon_j} \left| \log f_{j} (Y_t;X_{j,t}(\bupsilon_j),\bupsilon_j) \right| \right] &= \mathbb{E} \left[ \sup_{\bupsilon_j \in \bUpsilon_j} \left| -\frac{1}{2} \log 2\pi - \frac{1}{2} \log X_{j,t}(\bupsilon_j) - \frac{1}{2} \frac{Y_t^2}{X_{j,t}(\bupsilon_j)} \right| \right] \\
	&\leq \frac{1}{2} \log 2\pi + \frac{1}{2} \mathbb{E} \left[ \sup_{\bupsilon_j \in \bUpsilon_j} \left| \log X_{j,t}(\bupsilon_j) \right| \right]+ \frac{1}{2\underline{x}_j} \mathbb{E} \left[ Y_t^2 \right],
\end{align*}
so Assumption \ref{AssumptionF2} is also satisfied since $\mathbb{E} [ \sup_{\bupsilon_j \in \bUpsilon_j} | \log X_{j,t}(\bupsilon_j) | ] < \infty$ and $\mathbb{E} [ Y_t^2 ] < \infty$. Indeed, $\mathbb{E} [ Y_t^2 ] < \infty$ implies that $\mathbb{E} [ \sup_{\bupsilon_j \in \bUpsilon_j} | \log X_{j,t}(\bupsilon_j) | ] < \infty$, which we now show. Note that $\log x = \log^{+} x - \log^{-} x$. Thus,
\begin{equation*}
	\mathbb{E} \left[ \sup_{\bupsilon_j \in \bUpsilon_j} \left| \log X_{j,t}(\bupsilon_j) \right| \right] \leq \mathbb{E} \left[ \log^{+} \sup_{\bupsilon_j \in \bUpsilon_j} X_{j,t}(\bupsilon_j) \right] + \log^{-} \underline{x}_j.
\end{equation*}
Now, by Lemma \ref{LemmaInvertibilityX},
\begin{equation}
	X_{j,t}(\bupsilon_j) = \frac{\omega_j}{1-\beta_j} + \alpha_j \sum_{i=0}^{\infty} \beta_j^i Y_{t-1-i}^2 \quad a.s. \label{eq:FZ}
\end{equation}
Thus, 
\begin{equation*}
	\mathbb{E} \left[ \sup_{\bupsilon_j \in \bUpsilon_j} X_{j,t}(\bupsilon_j) \right] \leq \frac{\overline{\omega}_j}{1-\overline{\beta}_j} + \frac{\overline{\alpha}_j}{1-\overline{\beta}_j} \mathbb{E} \left[ Y_{t-1}^2 \right].
\end{equation*}
Hence, $\mathbb{E} [ Y_t^2 ] < \infty$ implies that $\mathbb{E} [ \sup_{\bupsilon_j \in \bUpsilon_j} | \log X_{j,t}(\bupsilon_j) | ] < \infty$ since $\mathbb{E} [ \sup_{\bupsilon_j \in \bUpsilon_j} X_{j,t}(\bupsilon_j) ] < \infty$ implies that $\mathbb{E} [ \log^{+} \sup_{\bupsilon_j \in \bUpsilon_j} X_{j,t}(\bupsilon_j) ] < \infty$ by Lemma 2.2 in \cite{StraumannMikosch2006}. Finally, Assumption \ref{AssumptionIdentification} is true by assumption.

\subsection{Proof of Theorem \ref{theo:an}}

In the following, $\underline{\btheta} = \inf_{\btheta \in \bar{\bTheta}} \btheta$ and $\overline{\btheta} = \sup_{\btheta \in \bar{\bTheta}} \btheta$. We need to verify Assumptions \ref{AssumptionTheta2}-\ref{AssumptionF4} and the conditions in Lemmas \ref{LemmaInvertibilityDX}-\ref{LemmaInvertibilityDDPrediction}.

Let, as in the proof of Theorem \ref{theo:c}, $j \in \{1,...,J\}$ be given. First, Assumption \ref{AssumptionTheta2} is true by assumption and Assumptions \ref{AssumptionF3} and \ref{AssumptionPhi2} are clearly satisfied.

We now verify the conditions in Lemma \ref{LemmaInvertibilityDX}. First,
\begin{align*}
	\mathbb{E} \left[ \log^{+} \sup_{\bupsilon_{j} \in \bar{\bUpsilon}_{j}} \left| \left| \partial_{\bupsilon_j} \phi_{j} (Y_t,X_{j,t} (\bupsilon_{j});\bupsilon_{j}) \right| \right|_{2} \right] &\leq \mathbb{E} \left[ \log^{+} \sup_{\bupsilon_{j} \in \bar{\bUpsilon}_{j}} \left( 1 + Y_t^2 + X_{j,t} (\bupsilon_{j}) \right) \right] \\
	&\leq 4 \log 2 + 2 \mathbb{E} \left[ \log^{+} \left| Y_t \right| \right] + \mathbb{E} \left[ \log^{+} \sup_{\bupsilon_{j} \in \bar{\bUpsilon}_{j}} X_{j,t} (\bupsilon_{j}) \right],
\end{align*}
so Condition (i) is satisfied since $\mathbb{E} [ Y_t^6 ] < \infty$ implies that $\mathbb{E} [ \log^{+} | Y_t | ] < \infty$ and $\mathbb{E} [ \log^{+} \sup_{\bupsilon_{j} \in \bar{\bUpsilon}_{j}} X_{j,t} (\bupsilon_{j}) ] < \infty$, see the proof of Theorem \ref{theo:c}. We now show that $\mathbb{E} [ \sup_{\bupsilon_{j} \in \bar{\bUpsilon}_{j}} | | \nabla_{\bupsilon_j} X_{j,t} (\bupsilon_j) | |_{2}^{2} ] < \infty$, which implies that $\mathbb{E} [ \log^{+} \sup_{\bupsilon_{j} \in \bar{\bUpsilon}_{j}} | | \nabla_{\bupsilon_j} X_{j,t} (\bupsilon_j) | |_{2} ] < \infty$. Note that
\begin{equation*}
	\nabla_{\bupsilon_j} X_{j,t} (\bupsilon_j) = \sum_{i=0}^{\infty} \beta_j^i \begin{bmatrix} 1 \\ Y_{t-1-i}^2 \\ X_{j,t-1-i} (\bupsilon_j) \end{bmatrix} \quad a.s.,
\end{equation*}
so, by Minkowskis inequality, 
\begin{equation*}
	\mathbb{E} \left[ \sup_{\bupsilon_{j} \in \bar{\bUpsilon}_{j}} \left| \left| \nabla_{\bupsilon_j} X_{j,t} (\bupsilon_j) \right| \right|_{2}^{2} \right]^{\frac{1}{2}} \leq \frac{1}{1-\overline{\beta}_j} + \frac{1}{1-\overline{\beta}_j} \mathbb{E} \left[ Y_{t-1}^4 \right]^{\frac{1}{2}} + \frac{1}{1-\overline{\beta}_j} \mathbb{E} \left[ \sup_{\bupsilon_{j} \in \bar{\bUpsilon}_{j}} X_{j,t-1}^{2} (\bupsilon_j) \right]^{\frac{1}{2}}.
\end{equation*}
Hence, $\mathbb{E} [ \sup_{\bupsilon_{j} \in \bar{\bUpsilon}_{j}} || \nabla_{\bupsilon_j} X_{j,t} (\bupsilon_j) ||_{2}^{2} ] < \infty$ since $\mathbb{E} [ Y_t^6] < \infty$ and $\mathbb{E}[ \sup_{\bupsilon_{j} \in \bar{\bUpsilon}_{j}} X_{j,t}^{2} (\bupsilon_j)] < \infty$, see the proof of Theorem \ref{theo:c} again. Finally,
\begin{equation*}
	\sup_{\bupsilon_{j} \in \bar{\bUpsilon}_{j}} \left| \left| \partial_{\bvartheta_j} \phi_{j} (Y_t,\hat{X}_{j,t} (\bupsilon_{j});\bupsilon_{j}) - \partial_{\bvartheta_j} \phi_{j} (Y_t,X_{j,t} (\bupsilon_{j});\bupsilon_{j}) \right| \right|_{2} = \sup_{\bupsilon_{j} \in \bar{\bUpsilon}_{j}} \left| \hat{X}_{j,t} (\bupsilon_{j}) - X_{j,t} (\bupsilon_{j}) \right|,
\end{equation*}
so Condition (ii) is also satisfied since $\sup_{\bupsilon_{j} \in \bar{\bUpsilon}_{j}} | \hat{X}_{j,t} (\bupsilon_{j}) - X_{j,t} (\bupsilon_{j}) | \overset{e.a.s.}{\rightarrow} 0$ as $t \rightarrow \infty$. The conditions in Lemma \ref{LemmaInvertibilityDDX} can be verified similarly.

Moving on to Assumption \ref{AssumptionF4}, by Hölders inequality, for all $m \in (1,\frac{6}{4}]$,
\begin{align*}
	&\mathbb{E} \left[ \sup_{\bupsilon_j \in \bar{\bUpsilon}_{j}} \left| \partial_{x_j} \log f_{j} (Y_t;X_{j,t} (\bupsilon_j),\bupsilon_j) \right|^2 \left| \left| \nabla_{\bupsilon_j} X_{j,t} (\bupsilon_j) \right| \right|_{2}^2 \right] \\
	&= \mathbb{E} \left[ \sup_{\bupsilon_j \in \bar{\bUpsilon}_{j}} \left| - \frac{1}{2} + \frac{1}{2} \frac{Y_t^2}{X_{j,t}(\bupsilon_j)} \right|^2 \left| \left| \frac{\nabla_{\bupsilon_j} X_{j,t} (\bupsilon_j)}{X_{j,t}(\bupsilon_j)} \right| \right|_{2}^2 \right] \\
	&\leq \mathbb{E} \left[ \sup_{\bupsilon_j \in \bar{\bUpsilon}_{j}} \left| - \frac{1}{2} + \frac{1}{2} \frac{Y_t^2}{X_{j,t}(\bupsilon_j)} \right|^{2m} \right]^{\frac{1}{m}} \mathbb{E} \left[ \sup_{\bupsilon_j \in \bar{\bUpsilon}_{j}} \left| \left| \frac{\nabla_{\bupsilon_j} X_{j,t} (\bupsilon_j)}{X_{j,t}(\bupsilon_j)} \right| \right|_{2}^{2m^c} \right]^{\frac{1}{m^c}} \\
	&\leq \left( 1 + \frac{1}{\underline{x}_j^{2m}} \mathbb{E} \left[ \left|Y_t\right|^{4m} \right] \right)^{\frac{1}{m}} \mathbb{E} \left[ \sup_{\bupsilon_j \in \bar{\bUpsilon}_{j}} \left| \left| \frac{\nabla_{\bupsilon_j} X_{j,t} (\bupsilon_j)}{X_{j,t}(\bupsilon_j)} \right| \right|_{2}^{2m^c} \right]^{\frac{1}{m^c}}
\end{align*}
with $\underline{x}_j = \frac{\underline{\omega}_j}{1-\underline{\beta}_j}$ and $m^c = \frac{m}{m-1}$, where the last inequality follows from the inequality $|x+y|^p \leq 2^p |x|^p + 2^p |y|^p$ for all $x,y \in (-\infty,\infty)$ and $p \in (0,\infty)$. Note that, by Equation \eqref{eq:FZ}, for all $n \in (0,\frac{1}{m^c})$,
\begin{align*}
	\frac{\nabla_{\bupsilon_j} X_{j,t} (\bupsilon_j)}{X_{j,t} (\bupsilon_j)} 
	&=
	\begin{bmatrix} 
		\frac{1}{1-\beta_j} \frac{1}{X_{j,t} (\bupsilon_j)} \\
		\left( \sum_{i=0}^{\infty} \beta_j^{i} Y_{t-1-i}^2 \right) \frac{1}{X_{j,t} (\bupsilon_j)} \\
		\left( \frac{\omega_j}{(1-\beta_j)^2} + \alpha_j \sum_{i=0}^{\infty} i \beta_j^{i-1} Y_{t-1-i}^2 \right) \frac{1}{X_{j,t} (\bupsilon_j)}
	\end{bmatrix} \\
	&\leq 
	\begin{bmatrix} 
		\frac{1}{\omega_j} \\
		\frac{1}{\alpha_j} \\
		\frac{1}{1-\beta_j} + \frac{1}{\beta_j} \sum_{i=0}^{\infty} i \frac{\alpha_j \beta_j^{i} Y_{t-1-i}^2}{\frac{\omega_j}{1-\beta_j} + \alpha_j \beta_j^{i} Y_{t-1-i}^2}
	\end{bmatrix} \\
	&\leq
	\begin{bmatrix} 
		\frac{1}{\omega_j} \\
		\frac{1}{\alpha_j} \\
		\frac{1}{1-\beta_j} + \frac{1}{\omega_j^n} \alpha_j^n \frac{(1-\beta_j)^n}{\beta_j} \sum_{i=0}^{\infty} i \beta_j^{ni} \left|Y_{t-1-i}\right|^{2n}
	\end{bmatrix} \quad a.s.,
\end{align*}
where the inequalities are elementwise and the last one follows from the inequality $\frac{x}{1+x} \leq x^p$ for all $x \in [0,\infty)$ and $p \in (0,1)$, so, by Minkowskis inequality, 
\begin{equation*}
	\mathbb{E} \left[ \sup_{\bupsilon_j \in \bar{\bUpsilon}_{j}} \left| \left| \frac{\nabla_{\bupsilon_j} X_{j,t} (\bupsilon_j)}{X_{j,t}(\bupsilon_j)} \right| \right|_{2}^{2m^c} \right]^{\frac{1}{2m^c}} \leq A_j + B_j \mathbb{E} \left[ \left|Y_{t-1} \right|^{4nm^c} \right]^{\frac{1}{2m^c}},
\end{equation*}
where $A_j = \frac{1}{\underline{\omega}_j} + \frac{1}{\underline{\alpha}_j} + \frac{1}{1-\overline{\beta}_j}$ and $B_j = \frac{1}{\underline{\omega}_j^n} \overline{\alpha}_j^n \frac{(1-\underline{\beta}_j)^n \overline{\beta}_j^n}{\underline{\beta}_j (1-\overline{\beta}_j^n)^2}$. Hence, $\mathbb{E} \left[ \sup_{\bupsilon_j \in \bar{\bUpsilon}_{j}} \left| \left| \frac{\nabla_{\bupsilon_j} X_{j,t} (\bupsilon_j)}{X_{j,t}(\bupsilon_j)} \right| \right|_{2}^{2m^c} \right] < \infty$ since $\mathbb{E} [ Y_t^6] < \infty$. Condition (i) in Assumption \ref{AssumptionF4} is thus satisfied since $\mathbb{E} [ Y_t^6] < \infty$ and $\mathbb{E} \left[ \sup_{\bupsilon_j \in \bar{\bUpsilon}_{j}} \left| \left| \frac{\nabla_{\bupsilon_j} X_{j,t} (\bupsilon_j)}{X_{j,t}(\bupsilon_j)} \right| \right|_{2}^{2m^c} \right] < \infty$. Condition (ii) is clearly satisfied. Moving on to Conditions (iii) and (iv),
\begin{align*}
	&\mathbb{E} \left[ \sup_{\bupsilon_j \in \bar{\bUpsilon}_{j}} \left| \partial_{x_j x_j} \log f_{j} (Y_t;{X}_{j,t} (\bupsilon_j),\bupsilon_j) \right| \left| \left| \nabla_{\bupsilon_j} X_{j,t} (\bupsilon_j) \right| \right|_{2}^2 \right] \\
	&= \mathbb{E} \left[ \sup_{\bupsilon_j \in \bar{\bUpsilon}_{j}} \left| \frac{1}{2} - \frac{Y_t^2}{{X}_{j,t} (\bupsilon_j)} \right| \left| \left| \frac{\nabla_{\bupsilon_j} X_{j,t} (\bupsilon_j)}{{X}_{j,t} (\bupsilon_j)} \right| \right|_{2}^2 \right] \\
	&\leq \mathbb{E} \left[ \sup_{\bupsilon_j \in \bar{\bUpsilon}_{j}} \left| \frac{1}{2} - \frac{Y_t^2}{{X}_{j,t} (\bupsilon_j)} \right|^{m} \right]^{\frac{1}{m}} \mathbb{E} \left[ \sup_{\bupsilon_j \in \bar{\bUpsilon}_{j}} \left| \left| \frac{\nabla_{\bupsilon_j} X_{j,t} (\bupsilon_j)}{{X}_{j,t} (\bupsilon_j)} \right| \right|_{2}^{2m^c} \right]^{\frac{1}{m^c}} \\
	&\leq \left( 1 + \frac{2^m}{\underline{x}_j^{m}} \mathbb{E} \left[ \left|Y_t\right|^{2m} \right] \right)^{\frac{1}{m}} \mathbb{E} \left[ \sup_{\bupsilon_j \in \bar{\bUpsilon}_{j}} \left| \left| \frac{\nabla_{\bupsilon_j} X_{j,t} (\bupsilon_j)}{{X}_{j,t} (\bupsilon_j)} \right| \right|_{2}^{2m^c} \right]^{\frac{1}{m^c}},
\end{align*}
so Condition (iii) is also satisfied. Moreover, by the Cauchy–Schwarz inequality,
\begin{align*}
	&\mathbb{E} \left[ \sup_{\bupsilon_j \in \bar{\bUpsilon}_{j}} \left| \partial_{x_j} \log f_{j} (Y_t;{X}_{j,t} (\bupsilon_j),\bupsilon_j) \right| \left| \left| \nabla_{\bupsilon_j \bupsilon_j} X_{j,t} (\bupsilon_j) \right| \right|_{2,2} \right] \\
	&= \mathbb{E} \left[ \sup_{\bupsilon_j \in \bar{\bUpsilon}_{j}} \left| - \frac{1}{2} \frac{1}{X_{j,t}(\bupsilon_j)} + \frac{1}{2} \frac{Y_t^2}{X_{j,t}^2(\bupsilon_j)} \right| \left| \left| \nabla_{\bupsilon_j \bupsilon_j} X_{j,t} (\bupsilon_j) \right| \right|_{2,2} \right] \\
	& \leq \mathbb{E} \left[ \sup_{\bupsilon_j \in \bar{\bUpsilon}_{j}} \left| - \frac{1}{2} \frac{1}{X_{j,t}(\bupsilon_j)} + \frac{1}{2} \frac{Y_t^2}{X_{j,t}^2(\bupsilon_j)} \right|^2 \right]^{\frac{1}{2}} \mathbb{E} \left[ \sup_{\bupsilon_j \in \bar{\bUpsilon}_{j}} \left| \left| \nabla_{\bupsilon_j \bupsilon_j} X_{j,t} (\bupsilon_j) \right| \right|_{2,2}^{2} \right]^{\frac{1}{2}} \\
	& \leq \left( \frac{1}{\underline{x}_j^2} + \frac{1}{\underline{x}_j^4} \mathbb{E} \left[ Y_t^4 \right] \right)^{\frac{1}{2}} \mathbb{E} \left[ \sup_{\bupsilon_j \in \bar{\bUpsilon}_{j}} \left| \left| \nabla_{\bupsilon_j \bupsilon_j} X_{j,t} (\bupsilon_j) \right| \right|_{2,2}^{2} \right]^{\frac{1}{2}},
\end{align*}
where the last inequality follows from the inequality $|x+y|^p \leq 2^p |x|^p + 2^p |y|^p$ for all $x,y \in (-\infty,\infty)$ and $p \in (0,\infty)$ as above. Note that
\begin{equation*}
	\nabla_{\bupsilon_j \bupsilon_j} X_{j,t} (\bupsilon_j) = \sum_{i=0}^{\infty} \beta_j^i \left( \nabla_{\bupsilon_j} X_{j,t} (\bupsilon_j) \begin{bmatrix} 0 & 0 & 1 \end{bmatrix} + \begin{bmatrix} 0 \\ 0 \\ 1 \end{bmatrix} \nabla_{\bupsilon_j^{\prime}} X_{j,t} (\bupsilon_j) \right) \quad a.s.,
\end{equation*}
so, by Minkowskis inequality,
\begin{equation*}
	\mathbb{E} \left[ \sup_{\bupsilon_j \in \bar{\bUpsilon}_{j}} \left| \left| \nabla_{\bupsilon_j \bupsilon_j} X_{j,t} (\bupsilon_j) \right| \right|_{2,2}^{2} \right]^{\frac{1}{2}} \leq \frac{2}{1-\overline{\beta}_j} \mathbb{E} \left[ \sup_{\bupsilon_j \in \bar{\bUpsilon}_{j}} \left| \left| \nabla_{\bupsilon_j} X_{j,t} (\bupsilon_j) \right| \right|_{2}^{2} \right]^{\frac{1}{2}}.
\end{equation*}
Hence, $\mathbb{E} [ \sup_{\bupsilon_j \in \bar{\bUpsilon}_{j}} || \nabla_{\bupsilon_j \bupsilon_j} X_{j,t} (\bupsilon_j) ||_{2,2}^{2} ] < \infty$ since $\mathbb{E} [ \sup_{\bupsilon_j \in \bar{\bUpsilon}_{j}} || \nabla_{\bupsilon_j} X_{j,t} (\bupsilon_j) ||_{2}^{2} ] < \infty$. Condition (iv) is thus also satisfied since $\mathbb{E} [ Y_t^6 ] < \infty$ and $\mathbb{E} [ \sup_{\bupsilon_j \in \bar{\bUpsilon}_{j}} || \nabla_{\bupsilon_j \bupsilon_j} X_{j,t} (\bupsilon_j) ||_{2,2}^{2} ] < \infty$. Conditions (v) and (vi) are also clearly satisfied.

Finally, we verify the condition in Lemma \ref{LemmaInvertibilityDPrediction}. Note that
\begin{align*}
	\mathbb{E} \left[ \sup_{\bupsilon_j \in \bar{\bUpsilon}_j} \sup_{x_j \in \mathcal{X}_{j}} \left| \left| \partial_{\bvartheta_j x_j} \log f_{j} (Y_t;x_j,\bupsilon_j) \right| \right|_{2} \right] &= \mathbb{E} \left[ \sup_{\bupsilon_j \in \bar{\bUpsilon}_j} \sup_{x_j \in \mathcal{X}_{j}} \left| \frac{1}{2} \frac{1}{x_j^2} - \frac{Y_t^2}{x_j^3} \right| \right] \\
	&\leq \frac{1}{2\underline{x}_j^2} + \frac{1}{\underline{x}_j^3} \mathbb{E} \left[ Y_t^2 \right],
\end{align*}
so the condition in Lemma \ref{LemmaInvertibilityDPrediction} is satisfied since $\mathbb{E} [ Y_t^6 ] < \infty$. The condition in Lemma \ref{LemmaInvertibilityDDPrediction} can be verified similarly.

\section*{Acknowledgements}

I am grateful to Leopoldo Catania, Morten Ørregaard Nielsen, Christian Francq, Jean-Michel Zakoïan, Christian Gourieroux, and Anders Rahbek, seminar participants at CREST, and conference participants at the 2024 Virtual Workshop for Junior Researchers in Time Series, the 2024 Quantitative Finance and Financial Econometrics International Conference, the 2025 Italian Congress of Econometrics and Empirical Economics, the 2025 Annual Society for Financial Econometrics Conference, and the 2026 Annual International Association for Applied Econometrics Conference for comments and discussions. This research is supported by the Danish National Research Foundation (DNRF Chair grant number DNRF154).

\bibliography{references}

\begin{thebibliography}{45}
\providecommand{\natexlab}[1]{#1}
\providecommand{\url}[1]{\texttt{#1}}
\expandafter\ifx\csname urlstyle\endcsname\relax
  \providecommand{\doi}[1]{doi: #1}\else
  \providecommand{\doi}{doi: \begingroup \urlstyle{rm}\Url}\fi

\bibitem[Aknouche and Francq(2022)]{AknoucheFrancq2022}
A.~Aknouche and C.~Francq.
\newblock Stationarity and ergodicity of markov switching positive conditional
  mean models.
\newblock \emph{Journal of Time Series Analysis}, 43\penalty0 (3):\penalty0
  436--459, 2022.

\bibitem[Ang and Timmermann(2012)]{AngTimmermann2012}
A.~Ang and A.~Timmermann.
\newblock Regime changes and financial markets.
\newblock \emph{Annual Review of Financial Economics}, 4\penalty0 (1):\penalty0
  313--337, 2012.

\bibitem[Angelini and Gorgi(2018)]{AngeliniGorgi2018}
G.~Angelini and P.~Gorgi.
\newblock Dsge models with observation-driven time-varying volatility.
\newblock \emph{Economics Letters}, 171\penalty0 (1):\penalty0 169--171, 2018.

\bibitem[Ardia(2009)]{Ardia2009}
D.~Ardia.
\newblock Bayesian estimation of a markov‐switching threshold asymmetric
  garch model with student‐t innovations.
\newblock \emph{The Econometrics Journal}, 12\penalty0 (1):\penalty0 105–126,
  2009.

\bibitem[Ardia et~al.(2018)Ardia, Bluteau, Boudt, and
  Catania]{ArdiaBluteauBoudtCatania2018}
D.~Ardia, K.~Bluteau, K.~Boudt, and L.~Catania.
\newblock Forecasting risk with markov-switching garch models:a large-scale
  performance study.
\newblock \emph{International Journal of Forecasting}, 34\penalty0
  (4):\penalty0 733--747, 2018.

\bibitem[Ardia et~al.(2019)Ardia, Bluteau, Boudt, Catania, and
  Trottier]{ArdiaBluteauBoudtCataniaTrottier2019}
D.~Ardia, K.~Bluteau, K.~Boudt, L.~Catania, and D.-A. Trottier.
\newblock Markov-switching garch models in r: The msgarch package.
\newblock \emph{Journal of Statistical Software}, 91\penalty0 (4):\penalty0
  1–38, 2019.

\bibitem[Berkes et~al.(2003)Berkes, Horváth, and
  Kokoszka]{BerkesHorvathKokoszka2003}
I.~Berkes, L.~Horváth, and P.~Kokoszka.
\newblock Garch processes: structure and estimation.
\newblock \emph{Bernoulli}, 19\penalty0 (2):\penalty0 201--227, 2003.

\bibitem[Bernardi and Catania(2019)]{BernardiCatania2019}
M.~Bernardi and L.~Catania.
\newblock Switching generalized autoregressive score copula models with
  application to systemic risk.
\newblock \emph{Journal of Applied Econometrics}, 34\penalty0 (1):\penalty0
  43--65, 2019.

\bibitem[Bickel et~al.(1998)Bickel, Ritov, and Rydén]{BickelRitovRyden1998}
P.~J. Bickel, Y.~Ritov, and T.~Rydén.
\newblock Asymptotic normality of the maximum-likelihood estimator for general
  hidden markov models.
\newblock \emph{Annals of Statistics}, 26\penalty0 (4):\penalty0 1614--1635,
  1998.

\bibitem[Billingsley(1961)]{Billingsley1961}
P.~Billingsley.
\newblock The lindeberg-lévy theorem for martingales.
\newblock \emph{Proceedings of the American Mathematical Society}, 12\penalty0
  (5):\penalty0 788--792, 1961.

\bibitem[Blasques et~al.(2018)Blasques, Gorgi, Koopman, and
  Wintenberger]{BlasquesGorgiKoopmanWintenberger2018}
F.~Blasques, P.~Gorgi, S.~J. Koopman, and O.~Wintenberger.
\newblock Feasible invertibility conditions and maximum likelihood estimation
  for observation-driven models.
\newblock \emph{Electronic Journal of Statistics}, 12\penalty0 (1):\penalty0
  1019--1052, 2018.

\bibitem[Blasques et~al.(2022)Blasques, {van Brummelen}, Koopman, and
  Lucas]{BlasquesVanBrummelenKoopmanLucas2022}
F.~Blasques, J.~{van Brummelen}, S.~J. Koopman, and A.~Lucas.
\newblock Maximum likelihood estimation for score-driven models.
\newblock \emph{Journal of Econometrics}, 227\penalty0 (2):\penalty0 325--346,
  2022.

\bibitem[Bougerol(1993)]{Bougerol1993}
P.~Bougerol.
\newblock Kalman filtering with random coefficients and contractions.
\newblock \emph{SIAM Journal on Control and Optimization}, 31\penalty0
  (4):\penalty0 942--959, 1993.

\bibitem[Broda et~al.(2013)Broda, Haas, Krause, Paolella, and
  Steude]{BrodaHaasKrausePaolellaSteude2013}
S.~A. Broda, M.~Haas, J.~Krause, M.~S. Paolella, and S.~C. Steude.
\newblock Stable mixture garch models.
\newblock \emph{Journal of Econometrics}, 172\penalty0 (2):\penalty0 292--306,
  2013.

\bibitem[Buccheri and Corsi(2021)]{BuccheriCorsi2021}
G.~Buccheri and F.~Corsi.
\newblock Hark the shark: Realized volatility modeling with measurement errors
  and nonlinear dependencies.
\newblock \emph{Journal of Financial Econometrics}, 19\penalty0 (4):\penalty0
  614–649, 2021.

\bibitem[Buccheri et~al.(2021)Buccheri, Bormetti, Corsi, and
  Lillo]{BuccheriBormettiCorsiLillo2021}
G.~Buccheri, G.~Bormetti, F.~Corsi, and F.~Lillo.
\newblock A score-driven conditional correlation model for noisy and
  asynchronous data: An application to high-frequency covariance dynamics.
\newblock \emph{Journal of Business \& Economic Statistics}, 39\penalty0
  (4):\penalty0 920--936, 2021.

\bibitem[Cai(1994)]{Cai1994}
J.~Cai.
\newblock A markov model of switching-regime arch.
\newblock \emph{Journal of Business \& Economic Statistics}, 12\penalty0
  (3):\penalty0 309--316, 1994.

\bibitem[Delle~Monache et~al.(2016)Delle~Monache, Petrella, and
  Venditti]{MonachePetrellaVenditti2016}
D.~Delle~Monache, I.~Petrella, and F.~Venditti.
\newblock \emph{Common Faith or Parting Ways? A Time Varying Parameters Factor
  Analysis of Euro-Area Inflation}.
\newblock Emerald Group Publishing Limited, 2016.

\bibitem[Delle~Monache et~al.(2021)Delle~Monache, Petrella, and
  Venditti]{MonachePetrellaVenditti2021}
D.~Delle~Monache, I.~Petrella, and F.~Venditti.
\newblock Price dividend ratio and long-run stock returns: A score-driven state
  space model.
\newblock \emph{Journal of Business \& Economic Statistics}, 39\penalty0
  (4):\penalty0 1054--1065, 2021.

\bibitem[Douc et~al.(2004)Douc, Éric Moulines, and
  Rydén]{DoucMoulinesRyden2004}
R.~Douc, Éric Moulines, and T.~Rydén.
\newblock Asymptotic properties of the maximum likelihood estimator in
  autoregressive models with markov regime.
\newblock \emph{Annals of Statistics}, 32\penalty0 (5):\penalty0 2254--2304,
  2004.

\bibitem[Douc et~al.(2011)Douc, Éric Moulines, Olsson, and {van
  Handel}]{DoucMoulinesOlssonVanHandel2011}
R.~Douc, Éric Moulines, J.~Olsson, and R.~{van Handel}.
\newblock Consistency of the maximum likelihood estimator for general hidden
  markov models.
\newblock \emph{Annals of Statistics}, 39\penalty0 (1):\penalty0 474--513,
  2011.

\bibitem[Francq and Roussignol(1998)]{FrancqRoussignol1998}
C.~Francq and M.~Roussignol.
\newblock Ergodicity of autoregressive processes with markov-switching and
  consistency of the maximum-likelihood estimator.
\newblock \emph{Statistics}, 32\penalty0 (2):\penalty0 151--173, 1998.

\bibitem[Francq and Zakoïan(2004)]{FrancqZakoian2004}
C.~Francq and J.-M. Zakoïan.
\newblock Maximum likelihood estimation of pure garch and arma-garch processes.
\newblock \emph{Bernoulli}, 10\penalty0 (4):\penalty0 605--637, 2004.

\bibitem[Francq and Zakoïan(2019)]{FrancqZakoian2019}
C.~Francq and J.-M. Zakoïan.
\newblock \emph{GARCH Models: Structure, Statistical Inference and Financial
  Applications}.
\newblock John Wiley \& Sons Limited, 2019.

\bibitem[Francq et~al.(2001)Francq, Roussignol, and
  Zakoïan]{FrancqRoussignolZakoian2001}
C.~Francq, M.~Roussignol, and J.-M. Zakoïan.
\newblock Conditional heteroskedasticity driven by hidden markov chains.
\newblock \emph{Journal of Time Series Analysis}, 22\penalty0 (2):\penalty0
  197--220, 2001.

\bibitem[Haas and Liu(2018)]{HaasLiu2018}
M.~Haas and J.-C. Liu.
\newblock A multivariate regime-switching garch model with an application to
  global stock market and real estate equity returns.
\newblock \emph{Studies in Nonlinear Dynamics \& Econometrics}, 22\penalty0
  (3):\penalty0 1--27, 2018.

\bibitem[Haas et~al.(2004{\natexlab{a}})Haas, Mittnik, and
  Paolella]{HaasMittnikPaolella2004a}
M.~Haas, S.~Mittnik, and M.~S. Paolella.
\newblock Mixed normal conditional heteroskedasticity.
\newblock \emph{Journal of Financial Econometrics}, 2\penalty0 (2):\penalty0
  211--250, 2004{\natexlab{a}}.

\bibitem[Haas et~al.(2004{\natexlab{b}})Haas, Mittnik, and
  Paolella]{HaasMittnikPaolella2004b}
M.~Haas, S.~Mittnik, and M.~S. Paolella.
\newblock A new approach to markov-switching garch models.
\newblock \emph{Journal of Financial Econometrics}, 2\penalty0 (4):\penalty0
  493--530, 2004{\natexlab{b}}.

\bibitem[Hamilton(1989)]{Hamilton1989}
J.~D. Hamilton.
\newblock A new approach to the economic analysis of nonstationary time series
  and the business cycle.
\newblock \emph{Econometrica}, 57\penalty0 (2):\penalty0 357--384, 1989.

\bibitem[Hamilton(2010)]{Hamilton2010}
J.~D. Hamilton.
\newblock \emph{Regime Switching Models}.
\newblock Palgrave Macmillan London, 2010.

\bibitem[Hamilton and Susmel(1994)]{HamiltonSusmel1994}
J.~D. Hamilton and R.~Susmel.
\newblock Autoregressive conditional heteroskedasticity and changes in regime.
\newblock \emph{Journal of Econometrics}, 64\penalty0 (1-2):\penalty0 307--333,
  1994.

\bibitem[Harvey and Palumbo(2023)]{HarveyPalumbo2023}
A.~Harvey and D.~Palumbo.
\newblock Regime switching models for circular and linear time series.
\newblock \emph{Journal of Time Series Analysis}, 44\penalty0 (4):\penalty0
  374--392, 2023.

\bibitem[Harvey et~al.(2024)Harvey, Hurn, Palumbo, and
  Thiele]{HarveyHurnPalumboThiele2024}
A.~Harvey, S.~Hurn, D.~Palumbo, and S.~Thiele.
\newblock Modelling circular time series.
\newblock \emph{Journal of Econometrics}, 239\penalty0 (1), 2024.

\bibitem[Jensen and Petersen(1999)]{JensenPetersen1999}
J.~L. Jensen and N.~V. Petersen.
\newblock Asymptotic normality of the maximum likelihood estimator in state
  space models.
\newblock \emph{Annals of Statistics}, 27\penalty0 (2):\penalty0 514--535,
  1999.

\bibitem[Kandji and Misko(2024)]{KandjiMisko2024}
B.~M. Kandji and A.~Misko.
\newblock Markov-switching normal-mixture garch.
\newblock \emph{Preprint}, 2024.

\bibitem[Kasahara and Shimotsu(2019)]{KasaharaShimotsu2019}
H.~Kasahara and K.~Shimotsu.
\newblock Asymptotic properties of the maximum likelihood estimator in regime
  switching econometric models.
\newblock \emph{Journal of Econometrics}, 208\penalty0 (2):\penalty0 442--467,
  2019.

\bibitem[Krabbe(2025)]{Krabbe2025}
F.~B. Krabbe.
\newblock A note on ''quasi-maximum-likelihood estimation in conditionally
  heteroscedastic time series: A stochastic recurrence equations approach''.
\newblock \emph{ArXiv}, 2025.
\newblock URL \url{https://arxiv.org/abs/2510.05716}.

\bibitem[Krishnamurthy and Rydén(1998)]{KrishnamurthyRyden1998}
V.~Krishnamurthy and T.~Rydén.
\newblock Consistent estimation of linear and non-linear autoregressive models
  with markov regime.
\newblock \emph{Journal of Time Series Analysis}, 19\penalty0 (3):\penalty0
  291--307, 1998.

\bibitem[Leroux(1992)]{Leroux1992}
B.~G. Leroux.
\newblock Maximum-likelihood estimation for hidden markov models.
\newblock \emph{Stochastic Processes and their Applications}, 40\penalty0
  (1):\penalty0 127--143, 1992.

\bibitem[Liu(2006)]{Liu2006}
J.-C. Liu.
\newblock Stationarity of a markov-switching garch model.
\newblock \emph{Journal of Financial Econometrics}, 4\penalty0 (4):\penalty0
  573--593, 2006.

\bibitem[P{\"o}tscher and Prucha(1997)]{PotscherPrucha1997}
B.~M. P{\"o}tscher and I.~R. Prucha.
\newblock \emph{Dynamic Nonlinear Econometric Models: Asymptotic Theory}.
\newblock Springer Berlin, Heidelberg, 1997.

\bibitem[Rao(1962)]{Rao1962}
R.~R. Rao.
\newblock Relations between weak and uniform convergence of measures with
  applications.
\newblock \emph{Annals of Mathematical Statistics}, 33\penalty0 (2):\penalty0
  659--680, 1962.

\bibitem[Straumann and Mikosch(2006)]{StraumannMikosch2006}
D.~Straumann and T.~Mikosch.
\newblock Quasi-maximum-likelihood estimation in conditionally heteroscedastic
  time series: A stochastic recurrence equations approach.
\newblock \emph{Annals of Statistics}, 34\penalty0 (5):\penalty0 2449--2495,
  2006.

\bibitem[Viterbi(1967)]{Viterbi1967}
A.~Viterbi.
\newblock Error bounds for convolutional codes and an asymptotically optimum
  decoding algorithm.
\newblock \emph{IEEE Transactions on Information Theory}, 13\penalty0
  (2):\penalty0 260--269, 1967.

\bibitem[Walden(2019)]{Walden2019}
J.~Walden.
\newblock Trading, profits, and volatility in a dynamic information network
  model.
\newblock \emph{The Review of Economic Studies}, 86\penalty0 (5):\penalty0
  2248–2283, 2019.

\end{thebibliography}

\clearpage


\appendix

\setcounter{page}{1}

\setcounter{footnote}{0}
\renewcommand{\thefootnote}{\fnsymbol{footnote}}

\begin{center}

	{\LARGE Supplementary Appendix to “Asymptotic Properties of the Maximum Likelihood Estimator for Markov-Switching Observation-Driven Models” \par}
	
	\vspace{2em}
	
	{\normalsize Frederik Bjerg Krabbe\footnote{Email: \href{mailto:frederik.krabbe@econ.au.dk}{frederik.krabbe@econ.au.dk}.} \par}
	
	\vspace{1em}
	
	{\small Department of Economics and Business Economics \par}
	{\small Aarhus University \par}
	{\small Universitetsbyen 51 \par}
	{\small 8000 Aarhus C \par}
	{\small Denmark \par}
	
\end{center}

\renewcommand{\thefootnote}{\arabic{footnote}}
\setcounter{footnote}{0}

\renewcommand{\thesection}{S.\arabic{section}}
\setcounter{equation}{0}
\renewcommand{\theequation}{S.\arabic{equation}}

\newtheorem{slemma}{Lemma}
\renewcommand{\theslemma}{S.\arabic{slemma}}

\section{Organisation} \label{sec:s0}

The Supplementary Appendix is organised as follows. Section \ref{sec:s1} collects the auxiliary lemmas and the proofs of them, and Section \ref{sec:s2} gives the explicit expression for Equation \eqref{eq:lemma s.11 1} in the proof of Lemma \ref{LemmaInvertibilityDFilter} (Section \ref{sec:s2a}) and the one for Equation \eqref{eq:lemma s.12 1} in the proof of Lemma \ref{LemmaInvertibilityDDFilter} (Section \ref{sec:s2b}).

\section{Auxiliary Lemmas} \label{sec:s1}

\begin{slemma} \label{LemmaInvertibility1}
	Under the assumptions in Lemma \ref{LemmaInvertibilityFilter}, for all $n \in \mathbb{Z}$, $(S_t)_{t \leq n}$ given $(Y_\tau)_{\tau \leq n}$ is a non-homogeneous Markov chain with transition probabilities
	\begin{equation*}
		\mathbb{P} (S_t = j \mid S_{t-1} = i, \bY_{-\infty}^{n}) = \frac{\mathbb{P} (\bY_{t}^{n} \in \textup{d}\by_{t}^{n} \mid \bY_{-\infty}^{t-1}, S_t = j) p_{ij}}{\sum_{k=1}^{J} \mathbb{P} (\bY_{t}^{n} \in \textup{d}\by_{t}^{n} \mid \bY_{-\infty}^{t-1}, S_t = k) p_{ik} }, \quad i,j \in \{1,...,J\},
	\end{equation*}
	which satisfy
	\begin{equation*}
		\mathbb{P} (S_t = j \mid S_{t-1} = i, \bY_{-\infty}^{n}) \geq p_{-} \frac{\mathbb{P} (\bY_{t}^{n} \in \textup{d}\by_{t}^{n} \mid \bY_{-\infty}^{t-1}, S_t = j)}{\sum_{k=1}^{J} \mathbb{P} (\bY_{t}^{n} \in \textup{d}\by_{t}^{n} \mid \bY_{-\infty}^{t-1}, S_t = k)}
	\end{equation*}
	for all $i,j \in \{1,...,J\}$.
\end{slemma}
\begin{proof}
	Let $n \in \mathbb{Z}$ be given. Then, for all $t \leq n$,
	\begin{align*}
		&\mathbb{P} (S_t = j \mid S_{t-1} = i,\bS_{-\infty}^{t-2},\bY_{-\infty}^{n}) \\
		&= \mathbb{P} (S_t = j \mid S_{t-1} = i, \bY_{-\infty}^{n}) \\
		&= \frac{\mathbb{P} (\bY_{t}^{n} \in \textup{d}\by_{t}^{n}, S_t = j \mid S_{t-1} = i, \bY_{-\infty}^{t-1})}{\mathbb{P} (\bY_{t}^{n} \in \textup{d}\by_{t}^{n} \mid S_{t-1} = i, \bY_{-\infty}^{t-1})} \\
		&= \frac{\mathbb{P} (\bY_{t}^{n} \in \textup{d}\by_{t}^{n}, S_t = j \mid S_{t-1} = i, \bY_{-\infty}^{t-1})}{\sum_{k=1}^{J} \mathbb{P} (\bY_{t}^{n} \in \textup{d}\by_{t}^{n}, S_t = k \mid S_{t-1} = i, \bY_{-\infty}^{t-1})} \\
		&= \frac{\mathbb{P} (\bY_{t}^{n} \in \textup{d}\by_{t}^{n} \mid \bY_{-\infty}^{t-1}, S_t = j , S_{t-1} = i) \mathbb{P} (S_t = j \mid S_{t-1} = i, \bY_{-\infty}^{t-1})}{\sum_{k=1}^{J} \mathbb{P} (\bY_{t}^{n} \in \textup{d}\by_{t}^{n} \mid \bY_{-\infty}^{t-1}, S_t = k , S_{t-1} = i) \mathbb{P} (S_t = k \mid S_{t-1} = i, \bY_{-\infty}^{t-1}) } \\
		&= \frac{\mathbb{P} (\bY_{t}^{n} \in \textup{d}\by_{t}^{n} \mid \bY_{-\infty}^{t-1}, S_t = j) \mathbb{P} (S_t = j \mid S_{t-1} = i)}{\sum_{k=1}^{J} \mathbb{P} (\bY_{t}^{n} \in \textup{d}\by_{t}^{n} \mid \bY_{-\infty}^{t-1}, S_t = k) \mathbb{P} (S_t = k \mid S_{t-1} = i) } \\
		&= \frac{\mathbb{P} (\bY_{t}^{n} \in \textup{d}\by_{t}^{n} \mid \bY_{-\infty}^{t-1}, S_t = j) p_{ij}}{\sum_{k=1}^{J} \mathbb{P} (\bY_{t}^{n} \in \textup{d}\by_{t}^{n} \mid \bY_{-\infty}^{t-1}, S_t = k) p_{ik} } \\
		&\geq p_{-} \frac{\mathbb{P} (\bY_{t}^{n} \in \textup{d}\by_{t}^{n} \mid \bY_{-\infty}^{t-1}, S_t = j)}{\sum_{k=1}^{J} \mathbb{P} (\bY_{t}^{n} \in \textup{d}\by_{t}^{n} \mid \bY_{-\infty}^{t-1}, S_t = k)}
	\end{align*}
	for all $i,j \in \{1,...,J\}$, where the first equality follows since $S_t$ is independent of $\bS_{-\infty}^{t-2}$ given $S_{t-1}$ and $\bY_{-\infty}^{n}$ – showing that $(S_t)_{t \leq n}$ given $(Y_\tau)_{\tau \leq n}$ is a Markov chain – the fifth equality follows since $\bY_{t}^{n}$ is independent of $S_{t-1}$ given $\bY_{-\infty}^{t-1}$ and $S_t$, and $S_t$ is independent of $\bY_{-\infty}^{t-1}$ given $S_{t-1}$, and the inequality follows since $p_{-} \leq p_{ij} \leq 1$ for all $i,j \in \{1,...,J\}$. This concludes the proof.
\end{proof}
\begin{slemma} \label{LemmaInvertibilityA}
	Let $\bM$ be a row stochastic matrix with generic element $m_{ij},i,j \in \{1,...,J\}$. Assume that there exist an $\varepsilon \in (0,1)$ and a stochastic vector $\bv$ with generic element $v_j,j \in \{1,...,J\}$ such that
	\begin{equation*}
		m_{ij} \geq \varepsilon v_j
	\end{equation*}
	for all $i,j \in \{1,...,J\}$. Then, there exists an $\alpha \in (0,1)$ such that 
	\begin{equation*}
		|| \bM^{\prime} \bx - \bM^{\prime} \by ||_{1} \leq \alpha || \bx - \by ||_{1}
	\end{equation*}
	for all $\bx,\by \in \mathcal{S}$.
\end{slemma}
\begin{proof}
	Let $\tilde{\bM}$ be a matrix with generic element $\tilde{m}_{ij} = (1-\varepsilon)^{-1}(m_{ij} - \varepsilon v_j),i,j \in \{1,...,J\}$. Note that $\tilde{\bM}$ is a row stochastic matrix. We thus have that
	\begin{align*}
		|| \bM^{\prime} \bx - \bM^{\prime} \by ||_{1} &= \sum_{j=1}^{J} \left| \sum_{i=1}^{J} m_{ij} (x_i-y_i) \right| \\
		&= (1-\varepsilon) \sum_{j=1}^{J} \left| \sum_{i=1}^{J} \tilde{m}_{ij} (x_i-y_i) \right| \\
		& \leq (1-\varepsilon) \sum_{j=1}^{J} \sum_{i=1}^{J} \tilde{m}_{ij} | x_i-y_i | \\
		&= (1-\varepsilon) \sum_{i=1}^{J} | x_i-y_i | = (1-\varepsilon) || \bx - \by ||_{1}
	\end{align*}
	for all $\bx,\by \in \mathcal{S}$, which concludes the proof.
\end{proof}
\begin{slemma} \label{LemmaConsistencyA}
	Under the assumptions in Theorem \ref{TheoremConsistency},
	\begin{equation*}
		\sup_{\btheta \in \bTheta} | \hat{L}_T (\btheta) - L_T (\btheta) | \overset{a.s.}{\rightarrow} 0 \quad \text{as} \quad T \rightarrow \infty.
	\end{equation*}
\end{slemma}
\begin{proof}
	We have that
	\begin{equation*}
		\sup_{\btheta \in \bTheta} | \hat{L}_T (\btheta) - L_T (\btheta) | \leq \frac{1}{T} \sum_{t=1}^{T} \sup_{\btheta \in \bTheta} | \log \hat{f} (Y_t;\btheta) - \log f (Y_t;\btheta) |.
	\end{equation*}
	The conclusion thus follows from Lemma 2.1 in \citet{StraumannMikosch2006} if $\sup_{\btheta \in \bTheta} | \log \hat{f} (Y_t;\btheta) - \log f (Y_t;\btheta) | \overset{e.a.s.}{\rightarrow} 0$ as $t \rightarrow \infty$, which we now show. We have that
	\begin{equation}
		| \log \hat{f} (Y_t;\btheta) - \log f (Y_t;\btheta) | \leq | \log \hat{f} (Y_t;\btheta) - \log \tilde{f} (Y_t;\btheta) | + | \log \tilde{f} (Y_t;\btheta) - \log f (Y_t;\btheta) |, \label{eq:TermAConsistency}
	\end{equation}
	where $\tilde{f} (Y_t;\btheta) = \sum_{j=1}^{J} \pi_{j,t \mid t-1} (\btheta) f_{j} (Y_t;\hat{X}_{j,t} (\bupsilon_j),\bupsilon_j)$. First, we consider $| \log \hat{f} (Y_t;\btheta) - \log \tilde{f} (Y_t;\btheta) |$. Let $m_{1,t}^{\btheta} : \mathcal{S}_{-} \rightarrow \mathbb{R}$ be given by
	\begin{equation*}
		m_{1,t}^{\btheta} (\bs) = \log \sum_{j=1}^{J} s_j f_{j} (Y_t;\hat{X}_{j,t}(\bupsilon_j),\bupsilon_j).
	\end{equation*}
	Note that $\log \hat{f} (Y_t;\btheta) = m_{1,t}^{\btheta} (\hat{\bpi}_{t \mid t-1} (\btheta))$ and $\log \tilde{f} (Y_t;\btheta) = m_{1,t}^{\btheta} (\bpi_{t \mid t-1} (\btheta))$. Thus, by the mean value theorem and the Cauchy-Schwarz inequality, there exists a $\bar{\bpi}_{t \mid t-1} (\btheta) \in \mathcal{S}_{-}$ such that
	\begin{equation*}
		| \log \hat{f} (Y_t;\btheta) - \log \tilde{f} (Y_t;\btheta) | \leq C || \nabla_{\bs} m_{1,t}^{\btheta} (\bar{\bpi}_{t \mid t-1} (\btheta)) ||_{2} || \hat{\bpi}_{t \mid t-1} (\btheta) - \bpi_{t \mid t-1} (\btheta) ||_{2},
	\end{equation*}
	where
	\begin{align*}
		| [\nabla_{\bs} m_{1,t}^{\btheta} (\bar{\bpi}_{t \mid t-1} (\btheta) )]_{k} | &= \frac{f_{k} (Y_t;\hat{X}_{k,t}(\bupsilon_k),\bupsilon_k)}{ \sum_{j=1}^{J} \bar{\pi}_{j,t \mid t-1} (\btheta) f_{j} (Y_t;\hat{X}_{j,t}(\bupsilon_j),\bupsilon_j)} \\
		&= \frac{1}{\bar{\pi}_{k,t \mid t-1} (\btheta) }\frac{\bar{\pi}_{k,t \mid t-1} (\btheta) f_{k} (Y_t;\hat{X}_{k,t}(\bupsilon_k),\bupsilon_k)}{ \sum_{j=1}^{J} \bar{\pi}_{j,t \mid t-1} (\btheta) f_{j} (Y_t;\hat{X}_{j,t}(\bupsilon_j),\bupsilon_j)} \\
		& \leq \frac{1}{p_{-}}
	\end{align*}
	since, by Equation \eqref{eq:remark} in Remark \ref{remark}, $\bar{\pi}_{k,t \mid t-1} (\btheta) \geq p_{-}$, so
	\begin{equation}
		| \log \hat{f} (Y_t;\btheta) - \log \tilde{f} (Y_t;\btheta) | \leq C || \hat{\bpi}_{t \mid t-1} (\btheta) - \bpi_{t \mid t-1} (\btheta) ||_{2}. \label{TermBConsistency}
	\end{equation}
	We now consider $| \log \tilde{f} (Y_t;\btheta) - \log f (Y_t;\btheta) |$. Let $m_{2,t}^{\btheta} : \mathcal{X}_{1} \times \cdots \times \mathcal{X}_{J} \rightarrow \mathbb{R}$ be given by
	\begin{equation*}
		m_{2,t}^{\btheta} (\bx) = \log \sum_{j=1}^{J} \pi_{j,t \mid t-1} (\btheta) f_{j} (Y_t;x_j,\bupsilon_j).
	\end{equation*}
	Note that $\log \tilde{f} (Y_t;\btheta) = m_{2,t}^{\btheta} (\hat{\bX}_t (\mathbf{\bupsilon}))$ and $\log f (Y_t;\btheta) = m_{2,t}^{\btheta} (\bX_t (\mathbf{\bupsilon}))$. Thus, by the mean value theorem and the Cauchy-Schwarz inequality again, there exists an $\bar{\bX}_t (\mathbf{\bupsilon}) \in \mathcal{X}_{1} \times \cdots \times \mathcal{X}_{J}$ such that
	\begin{equation*}
		| \log \tilde{f} (Y_t;\btheta) - \log f (Y_t;\btheta) | \leq || \nabla_{\bx} m_{2,t}^{\btheta} (\bar{\bX}_t (\mathbf{\bupsilon})) ||_{2} || \hat{\bX}_t (\mathbf{\bupsilon}) - \bX_t (\mathbf{\bupsilon}) ||_{2},
	\end{equation*}
	where
	\begin{align*}
		| [\nabla_{\bx} m_{2,t}^{\btheta} (\bar{\bX}_t (\mathbf{\bupsilon}))]_{k} | &= \frac{\pi_{k,t \mid t-1} (\btheta) | \nabla_{x_k} f_{k} (Y_t;\bar{X}_{k,t}(\bupsilon_k),\bupsilon_k) |}{\sum_{j=1}^{J} \pi_{j,t \mid t-1} (\btheta) f_{j} (Y_t;\bar{X}_{j,t}(\bupsilon_j),\bupsilon_j)} \\
		&= \frac{\pi_{k,t \mid t-1} (\btheta) f_{k} (Y_t;\bar{X}_{k,t}(\bupsilon_k),\bupsilon_k)}{\sum_{j=1}^{J} \pi_{j,t \mid t-1} (\btheta) f_{j} (Y_t;\bar{X}_{j,t}(\bupsilon_j),\bupsilon_j)} | \nabla_{x_k} \log f_{k} (Y_t;\bar{X}_{k,t}(\bupsilon_k),\bupsilon_k) | \\
		& \leq | \nabla_{x_k} \log f_{k} (Y_t;\bar{X}_{k,t}(\bupsilon_k),\bupsilon_k) |,
	\end{align*}
	so
	\begin{equation}
		| \log \tilde{f} (Y_t;\btheta) - \log f (Y_t;\btheta) | \leq \sum_{k=1}^{J} \left| \nabla_{x_k} \log f_{k}(Y_t;\bar{X}_{k,t}(\bupsilon_k),\bupsilon_k) \right| || \hat{\bX}_t (\mathbf{\bupsilon}) - \bX_t (\mathbf{\bupsilon}) ||_{2}. \label{TermCConsistency}
	\end{equation}
	Hence, continuing from Equation \eqref{eq:TermAConsistency} and using Equations \eqref{TermBConsistency} and \eqref{TermCConsistency}, we have that
	\begin{equation}
		\begin{aligned}
			| \log \hat{f} (Y_t;\btheta) - \log f (Y_t;\btheta) | &\leq C || \hat{\bpi}_{t \mid t-1} (\btheta) - \bpi_{t \mid t-1} (\btheta) ||_{2} \\
			&+ \sum_{k,l=1}^{J} | \nabla_{x_k} \log f_{k} (Y_t;\bar{X}_{k,t}(\bupsilon_k),\bupsilon_k) | | \hat{X}_{l,t} (\bupsilon_l) - X_{l,t} (\bupsilon_l) |.
		\end{aligned} \label{eq:TermDConsistency}
	\end{equation}
	It thus follows from Lemmas 2.1 and 2.2 in \citet{StraumannMikosch2006} that $\sup_{\btheta \in \bTheta} | \log \hat{f} (Y_t;\btheta) - \log f (Y_t;\btheta) | \overset{e.a.s.}{\rightarrow} 0$ as $t \rightarrow \infty$ since $(Y_t)_{t \in \mathbb{Z}}$ is stationary by Assumption \ref{AssumptionY}, for each $j \in \{1,...,J\}$, there exists an $m_j > 0$ such that $\mathbb{E} \left[ \sup_{\bupsilon_j \in \bUpsilon_j} \sup_{x_j \in \mathcal{X}_{j}} \left| \nabla_{x_j} \log f_{j} (Y_t;x_j,\bupsilon_j) \right|^{m_j} \right] < \infty$ by assumption, for each $j \in \{1,...,J\}$, $\sup_{\bupsilon_{j} \in \bUpsilon_{j}} | \hat{X}_{j,t} (\bupsilon_j) - X_{j,t} (\bupsilon_j) | \overset{e.a.s.}{\rightarrow} 0$ as $t \rightarrow \infty$ by Lemma \ref{LemmaInvertibilityX}, and $\sup_{\btheta \in \bTheta} || \hat{\bpi}_{t \mid t-1} (\btheta) - \bpi_{t \mid t-1} (\btheta) ||_{2} \overset{e.a.s.}{\rightarrow} 0$ as $t \rightarrow \infty$ by Corollary \ref{CorollaryInvertibilityPrediction}.
\end{proof}
\begin{slemma} \label{LemmaConsistencyB}
	Under the assumptions in Theorem \ref{TheoremConsistency},
	\begin{equation*}
		\sup_{\btheta \in \bTheta} | L_T (\btheta) - L(\btheta) | \overset{a.s.}{\rightarrow} 0 \quad \text{as} \quad T \rightarrow \infty.
	\end{equation*}
\end{slemma}
\begin{proof}
	The conclusion follows from the uniform law of large numbers by \citet{Rao1962} if $(\log f (Y_t;\btheta))_{t \in \mathbb{Z}}$ is stationary and ergodic for all $\btheta \in \bTheta$ and $\mathbb{E} \left[ \sup_{\btheta \in \bTheta} |\log f (Y_t;\btheta)| \right] < \infty$. $(\log f (Y_t;\btheta))_{t \in \mathbb{Z}}$ is stationary and ergodic for all $\btheta \in \bTheta$ since $(Y_t)_{t \in \mathbb{Z}}$ is stationary and ergodic by Assumption \ref{AssumptionY}, for each $j \in \{1,...,J\}$, $(X_{j,t}(\bupsilon_{j}))_{t \in \mathbb{Z}}$ is stationary and ergodic for all $\bupsilon_j \in \bUpsilon_j$ by Lemma \ref{LemmaInvertibilityX}, and $(\bpi_{t \mid t-1} (\btheta))_{t \in \mathbb{Z}}$ is stationary and ergodic for all $\btheta \in \bTheta$ by Corollary \ref{CorollaryInvertibilityPrediction}. We now show that $\mathbb{E} \left[ \sup_{\btheta \in \bTheta} |\log f (Y_t;\btheta)| \right] < \infty$. Note that
	\begin{equation*}
		|\log f (Y_t;\btheta)| \leq \sum_{j=1}^{J} | \log f_{j} (Y_t;X_{j,t} (\bupsilon_j),\bupsilon_j) |
	\end{equation*}
	since
	\begin{align*}
		\log f (Y_t;\btheta)
		&\geq \log \min_{j \in \{1,...,J\}} f_{j} (Y_t;X_{j,t} (\bupsilon_j),\bupsilon_j) \\
		&= \min_{j \in \{1,...,J\}} \log f_{j} (Y_t;X_{j,t} (\bupsilon_j),\bupsilon_j) \\
		& \geq \min_{j \in \{1,...,J\}} - | \log f_{j} (Y_t;X_{j,t} (\bupsilon_j),\bupsilon_j) | \geq - \sum_{j=1}^{J} | \log f_{j} (Y_t;X_{j,t} (\bupsilon_j),\bupsilon_j) |
	\end{align*}
	and
	\begin{align*}
		\log f (Y_t;\btheta)
		&\leq \log \max_{j \in \{1,...,J\}} f_{j} (Y_t;X_{j,t} (\bupsilon_j),\bupsilon_j) \\
		&= \max_{j \in \{1,...,J\}} \log f_{j} (Y_t;X_{j,t} (\bupsilon_j),\bupsilon_j) \\
		& \leq \max_{j \in \{1,...,J\}} | \log f_{j} (Y_t;X_{j,t} (\bupsilon_j),\bupsilon_j) | \leq \sum_{j=1}^{J} | \log f_{j} (Y_t;X_{j,t} (\bupsilon_j),\bupsilon_j) |.
	\end{align*}
	Hence, $\mathbb{E} \left[ \sup_{\btheta \in \bTheta} |\log f (Y_t;\btheta)| \right] < \infty$ since, for each $j \in \{1,...,J\}$, $\mathbb{E}\left[ \sup_{\bupsilon_j \in \bUpsilon_j} | \log f_{j} (Y_t;X_{j,t}(\bupsilon_j),\bupsilon_j)| \right] < \infty$ by Assumption \ref{AssumptionF2}.
\end{proof}
\begin{slemma} \label{LemmaConsistencyC}
	Under the assumptions in Theorem \ref{TheoremConsistency},
	\begin{equation*}
		L(\btheta) \leq L(\btheta_0)
	\end{equation*}
	for all $\btheta \in \bTheta$ with equality if and only if $\btheta = \btheta_0$.
\end{slemma}
\begin{proof}
	We have that
	\begin{equation}
		L(\btheta) - L(\btheta_0) = \mathbb{E} \left[ \log \frac{f (Y_t;\btheta)}{f (Y_t;\btheta_0)} \right]. \label{eq:justaterm}
	\end{equation}
	Observe that
	\begin{align*}
		\mathbb{E} \left[ \log \frac{f^{(m)} (\bY_{t-m+1}^{t};\btheta)}{f^{(m)} (\bY_{t-m+1}^{t};\btheta_0)} \right] &= \mathbb{E} \left[ \log \prod_{i=1}^{m} \frac{f (Y_{t-i+1} ; \btheta)}{f (Y_{t-i+1} ; \btheta_0)} \right] \\
		&= \sum_{i=1}^{m} \mathbb{E} \left[ \log \frac{f (Y_{t-i+1} ; \btheta)}{f (Y_{t-i+1} ; \btheta_0)} \right] = m \mathbb{E} \left[ \log \frac{f (Y_{t} ; \btheta)}{f (Y_{t} ; \btheta_0)} \right],
	\end{align*}
	so
	\begin{equation}
		\mathbb{E} \left[ \log \frac{f (Y_{t} ; \btheta)}{f (Y_{t} ; \btheta_0)} \right] = \frac{1}{m} \mathbb{E} \left[ \log \frac{f^{(m)} (\bY_{t-m+1}^{t};\btheta)}{f^{(m)} (\bY_{t-m+1}^{t};\btheta_0)} \right]. \label{eq:termterm}
	\end{equation}
	Note that $\log x \leq x-1$ for all $x > 0$ with equality if and only if $x = 1$. Thus, in Equation \eqref{eq:termterm},
	\begin{equation}
		\mathbb{E} \left[ \log \frac{f^{(m)} (\bY_{t-m+1}^{t};\btheta)}{f^{(m)} (\bY_{t-m+1}^{t};\btheta_0)} \right] \leq 0 \label{eq:termtermterm}
	\end{equation}
	for all $\btheta \in \bTheta$ with equality if and only if $f^{(m)} (\bY_{t-m+1}^{t};\btheta) = f^{(m)} (\bY_{t-m+1}^{t};\btheta_0)$ a.s. Hence, continuing from Equation \eqref{eq:justaterm} and using Equations \eqref{eq:termterm} and \eqref{eq:termtermterm}, we have that
	\begin{equation*}
		L(\btheta) - L(\btheta_0) \leq 0
	\end{equation*}
	for all $\btheta \in \bTheta$ with equality if and only if $\btheta = \btheta_0$ by Assumption \ref{AssumptionIdentification}, which concludes the proof.
\end{proof}
\begin{slemma} \label{LemmaInvertibilityDFilter}
	Assume that Assumptions \ref{AssumptionY}-\ref{AssumptionPhi1}, \ref{AssumptionF3}-\ref{AssumptionF4}, and the conditions in Lemmas \ref{LemmaInvertibilityX}-\ref{LemmaInvertibilityDX} hold. Moreover, assume that for each $j \in \{1,...,J\}$, there exists an $m_j > 0$ such that
	\begin{equation*}
		\mathbb{E} \left[ \sup_{\bupsilon_j \in \bar{\bUpsilon}_j} \sup_{x_j \in \mathcal{X}_{j}} \left| \left| \partial_{\bvartheta_j x_j} \log f_{j} (Y_t;x_j,\bupsilon_j) \right| \right|_{2}^{m_j} \right] < \infty.
	\end{equation*}
	Then, $( \nabla_{\btheta} \bpi_{t \mid t} (\btheta) )_{t \in \mathbb{Z}}$ is stationary and ergodic for all $\btheta \in \bar{\bTheta}$ with $\mathbb{E} [ \sup_{\btheta \in \bar{\bTheta}} || \nabla_{\btheta} \bpi_{t \mid t} (\btheta) ||_{2}^{2} ] < \infty$ and
	\begin{equation*}
		\sup_{\btheta \in \bar{\bTheta}} | | \nabla_{\btheta} \hat{\bpi}_{t \mid t} (\btheta) - \nabla_{\btheta} \bpi_{t \mid t} (\btheta) | |_{2} \overset{e.a.s.}{\rightarrow} 0 \quad \text{as} \quad t \rightarrow \infty
	\end{equation*}
	for any initialisation $\nabla_{\btheta} \hat{\bpi}_{0 \mid 0} (\btheta) \in \mathbb{R}^{Jd}$.
\end{slemma}
\begin{proof}
	Let $\bpi_{t \mid t}^{-} (\btheta) := (\pi_{1,t \mid t} (\btheta),...,\pi_{J-1,t \mid t} (\btheta))^{\prime}$ and $\hat{\bpi}_{t \mid t}^{-} (\btheta) := (\hat{\pi}_{1,t \mid t} (\btheta),...,\hat{\pi}_{J-1,t \mid t} (\btheta))^{\prime}$ such that $\bpi_{t \mid t} (\btheta)= ((\bpi_{t \mid t}^{-} (\btheta))^{\prime},1-(\mathbf{1}^{-})^{\prime}\bpi_{t \mid t}^{-}(\btheta))^{\prime}$ and $\hat{\bpi}_{t \mid t} (\btheta)= ((\hat{\bpi}_{t \mid t}^{-} (\btheta))^{\prime},1-(\mathbf{1}^{-})^{\prime}\hat{\bpi}_{t \mid t}^{-}(\btheta))^{\prime}$ where $\mathbf{1}^{-} := (1,...,1)^{\prime}$ is a $(J-1) \times 1$ vector of ones. The fact that $( \nabla_{\btheta} \bpi_{t \mid t} (\btheta) )_{t \in \mathbb{Z}}$ is stationary and ergodic for all $\btheta \in \bar{\bTheta}$ with $\mathbb{E} [ \sup_{\btheta \in \bar{\bTheta}} || \nabla_{\btheta} \bpi_{t \mid t} (\btheta) ||_{2}^{2} ] < \infty$ and $\sup_{\btheta \in \bar{\bTheta}} | | \nabla_{\btheta} \hat{\bpi}_{t \mid t} (\btheta) - \nabla_{\btheta} \bpi_{t \mid t} (\btheta) | |_{2} \overset{e.a.s.}{\rightarrow} 0$ as $t \rightarrow \infty$ for any initialisation $\nabla_{\btheta} \hat{\bpi}_{0 \mid 0} (\btheta) \in \mathbb{R}^{Jd}$ therefore follows if $( \nabla_{\btheta} \bpi_{t \mid t}^{-} (\btheta) )_{t \in \mathbb{Z}}$ is stationary and ergodic for all $\btheta \in \bar{\bTheta}$ with $\mathbb{E} [ \sup_{\btheta \in \bar{\bTheta}} || \nabla_{\btheta} \bpi_{t \mid t}^{-} (\btheta) ||_{2}^{2} ] < \infty$ and $\sup_{\btheta \in \bar{\bTheta}} | | \nabla_{\btheta} \hat{\bpi}_{t \mid t}^{-} (\btheta) - \nabla_{\btheta} \bpi_{t \mid t}^{-} (\btheta) | |_{2} \overset{e.a.s.}{\rightarrow} 0$ as $t \rightarrow \infty$ for any initialisation $\nabla_{\btheta} \hat{\bpi}_{0 \mid 0}^{-} (\btheta) \in \mathbb{R}^{(J-1)d}$, which we now prove. 
	
	As the proof of Lemma \ref{LemmaInvertibilityFilter}, the proof consists of two parts. The first part of the proof is about the stochastic process $(\nabla_{\btheta} \bpi_{t \mid t}^{-} (\btheta))_{t \in \mathbb{Z}}$, and the second part of the proof is about the stochastic process $(\nabla_{\btheta} \hat{\bpi}_{t \mid t}^{-} (\btheta))_{t \in \mathbb{N}_0}$. $(\nabla_{\btheta} \bpi_{t \mid t}^{-} (\btheta))_{t \in \mathbb{Z}}$ is a stochastic process taking values in $\mathbb{R}^{(J-1)d}$ given by
	\begin{equation}
		\nabla_{\btheta} \bpi_{t \mid t}^{-} (\btheta) = \bphi_{t}^{\bpi} (\nabla_{\btheta} \bpi_{t-1 \mid t-1}^{-} (\btheta) ; \btheta) \label{eq:lemma s.11 1}
	\end{equation}
	with
	\begin{align}
		&\left[ \bphi_{t}^{\bpi} (\nabla_{\btheta} \bpi_{t-1 \mid t-1}^{-} (\btheta) ; \btheta) \right]_{(n-1)d+m} \nonumber \\
		&= \sum_{j=1}^{J-1} \left( \partial_{\left[ \bpi \right]_j} \left[ \bphi_{t} (\bpi_{t-1 \mid t-1} (\btheta) ; \btheta) \right]_{n} - \partial_{\left[ \bpi \right]_J} \left[ \bphi_{t} (\bpi_{t-1 \mid t-1} (\btheta) ; \btheta) \right]_{n} \right) \label{eq:lemma s.11 2} \\
		& \quad \quad \quad \cdot \left[ \nabla_{\btheta} \bpi_{t-1 \mid t-1}^{-} (\btheta) \right]_{(j-1)d+m} \label{eq:lemma s.11 3} \\
		&+ \partial_{\left[ \btheta \right]_m} \left[ \bphi_{t} (\bpi_{t-1 \mid t-1} (\btheta) ; \btheta) \right]_{n}, \quad (n,m) \in \{1,...,J-1\} \times \{1,...,d\}. \label{eq:lemma s.11 4}
	\end{align}
	The explicit expression for Equation \eqref{eq:lemma s.11 1} is given in Section \ref{sec:s2a}, where the term in Equations \eqref{eq:lemma s.11 2} and \eqref{eq:lemma s.11 3} is given by the terms in Equations \eqref{eq:t4}, \eqref{eq:t6}, \eqref{eq:td}, and \eqref{eq:tf}, and the term in Equation \eqref{eq:lemma s.11 4} is given by the terms in Equations \eqref{eq:t3}, \eqref{eq:t5}, \eqref{eq:tc}, and \eqref{eq:te}. Let
	\begin{equation*}
		\Lambda ((\bphi_{t}^{\bpi})^{(r)};\btheta) = \sup_{\underset{\bx \neq \by}{\bx,\by \in \mathbb{R}^{(J-1)d}}} \frac{|| (\bphi_{t}^{\bpi})^{(r)} (\bx;\btheta) - (\bphi_{t}^{\bpi})^{(r)} (\by;\btheta) ||_{2}}{|| \bx - \by ||_{2}},
	\end{equation*}
	where
	\begin{equation*} 
		(\bphi_{t}^{\bpi})^{(r)} (\, \cdot \,;\btheta) = \bphi_{t}^{\bpi} \circ \cdots \circ \bphi_{t-r+1}^{\bpi} (\, \cdot \,;\btheta), \quad r \in \mathbb{N}.
	\end{equation*}
	The fact that $( \nabla_{\btheta} \bpi_{t \mid t}^{-} (\btheta) )_{t \in \mathbb{Z}}$ is stationary and ergodic for all $\btheta \in \bar{\bTheta}$ follows from Theorem 2.10 in \cite{StraumannMikosch2006} if
	\begin{enumerate} [(i)]
		\item $\mathbb{E} [\log^{+} \sup_{\btheta \in \bar{\bTheta}} ||\bphi_{t}^{\bpi} (\mathbf{0};\btheta)||_{2}] < \infty$,
		\item $\mathbb{E} [\log^{+} \sup_{\btheta \in \bar{\bTheta}} \Lambda (\bphi_{t}^{\bpi};\btheta)] < \infty$, and
		\item there exists an $r \in \mathbb{N}$ such that $- \infty \leq \mathbb{E} [\log \sup_{\btheta \in \bar{\bTheta}} \Lambda ((\bphi_{t}^{\bpi})^{(r)};\btheta)] < 0$,
	\end{enumerate}
	which we now show. Condition (i) follows from Lemma 2.2 in \cite{StraumannMikosch2006} if, for each $(n,m) \in \{1,...,J-1\} \times \{1,...,d\}$,
	\begin{equation}
		\mathbb{E} \left[\sup_{\btheta \in \bar{\bTheta}} \left| \partial_{\left[ \btheta \right]_m} \left[ \bphi_{t} (\bpi_{t-1 \mid t-1} (\btheta) ; \btheta) \right]_{n} \right|^2 \right] < \infty. \label{eq:c.i.1}
	\end{equation}
	Let $\pi_{j,t \mid t}$ denote $\pi_{j,t \mid t} (\btheta)$, $f_{j,t}$ denote $f_j (Y_t;X_{j,t},\bupsilon_j)$, and $X_{j,t}$ denote $X_{j,t} (\bupsilon_j)$. Note that the condition in Equation \eqref{eq:c.i.1} follows if, for each $j \in \{1,...,J-1\}$,
	\begin{equation}
		\mathbb{E} \left[ \sup_{\btheta \in \bar{\bTheta}} \left| \frac{f_{b,t} \pi_{a,t-1 \mid t-1}}{\mathbf{f}_t^{\prime} \mathbf{P}^{\prime} \bpi_{t-1 \mid t-1}} \right|^2 \right] < \infty \label{eq:c.a.1}
	\end{equation}
	and
	\begin{equation}
		\mathbb{E} \left[ \sup_{\btheta \in \bar{\bTheta}} \left| \frac{f_{j,t} \mathbf{p}_{\bullet j}^{\prime} \bpi_{t-1 \mid t-1} f_{b,t} \pi_{a,t-1 \mid t-1} }{( \mathbf{f}_t^{\prime} \mathbf{P}^{\prime} \bpi_{t-1 \mid t-1} )^2} \right|^2 \right] < \infty \label{eq:c.b.1}
	\end{equation}
	for all $a,b \in \{1,...,J\}$, and
	\begin{equation}
		\mathbb{E} \left[ \sup_{\btheta \in \bar{\bTheta}} \left| \frac{\nabla_{\left[\bupsilon_a \right]_{b}} f_{a,t} \mathbf{p}_{\bullet a}^{\prime} \bpi_{t-1 \mid t-1} }{\mathbf{f}_t^{\prime} \mathbf{P}^{\prime} \bpi_{t-1 \mid t-1} } \right|^2 \right] < \infty \label{eq:c.c.1}
	\end{equation}
	and
	\begin{equation}
		\mathbb{E} \left[ \sup_{\btheta \in \bar{\bTheta}} \left| \frac{( f_{j,t} \mathbf{p}_{\bullet j}^{\prime} \bpi_{t-1 \mid t-1} ) ( \nabla_{\left[\bupsilon_a \right]_{b}} f_{a,t} \mathbf{p}_{\bullet a}^{\prime} \bpi_{t-1 \mid t-1} )}{( \mathbf{f}_t^{\prime} \mathbf{P}^{\prime} \bpi_{t-1 \mid t-1}  )^2} \right|^2 \right] < \infty \label{eq:c.d.1}
	\end{equation}
	for all $a \in \{1,...,J\}$ and $b \in \{1,...,d_a\}$, where the terms in Equations \eqref{eq:c.a.1} and \eqref{eq:c.b.1} concern the terms in Equations \eqref{eq:t3} and \eqref{eq:t5}, and the terms in Equations \eqref{eq:c.c.1} and \eqref{eq:c.d.1} concern the terms in Equations \eqref{eq:tc} and \eqref{eq:te}. The conditions in Equations \eqref{eq:c.a.1} and \eqref{eq:c.b.1} follow since, for each $j \in \{1,...,J-1\}$,
	\begin{equation*}
		\mathbb{E} \left[ \sup_{\btheta \in \bar{\bTheta}} \left| \frac{f_{b,t} \pi_{a,t-1 \mid t-1}}{\mathbf{f}_t^{\prime} \mathbf{P}^{\prime} \bpi_{t-1 \mid t-1}} \right|^2 \right] \leq \frac{1}{p_{-}^2}
	\end{equation*}
	and
	\begin{equation*}
		\mathbb{E} \left[ \sup_{\btheta \in \bar{\bTheta}} \left| \frac{f_{j,t} \mathbf{p}_{\bullet j}^{\prime} \bpi_{t-1 \mid t-1} f_{b,t} \pi_{a,t-1 \mid t-1} }{( \mathbf{f}_t^{\prime} \mathbf{P}^{\prime} \bpi_{t-1 \mid t-1} )^2} \right|^2 \right] \leq \frac{1}{p_{-}^4}
	\end{equation*}
	for all $a,b \in \{1,...,J\}$ where the inequalities follow since $p_{-} \leq p_{ij} \leq 1$ for all $i,j \in \{1,...,J\}$. For the conditions in Equations \eqref{eq:c.c.1} and \eqref{eq:c.d.1}, for each $j \in \{1,...,J-1\}$, we have that
	\begin{align*}
		&\mathbb{E} \left[ \sup_{\btheta \in \bar{\bTheta}} \left| \frac{\nabla_{\left[\bupsilon_a \right]_{b}} f_{a,t} \mathbf{p}_{\bullet a}^{\prime} \bpi_{t-1 \mid t-1} }{\mathbf{f}_t^{\prime} \mathbf{P}^{\prime} \bpi_{t-1 \mid t-1} } \right|^2 \right] \\
		&\leq \frac{1}{p_{-}^2} \mathbb{E} \left[ \sup_{\bupsilon_a \in \bar{\bUpsilon}_a} \left| \partial_{x_a} \log f_{a,t} \right|^2 \left| \nabla_{\left[\bupsilon_a \right]_{b}} X_{a,t} \right|^2 \right] \\
		&+ \frac{2}{p_{-}^2} \mathbb{E} \left[ \sup_{\bupsilon_a \in \bar{\bUpsilon}_a} \left| \partial_{x_a} \log f_{a,t} \right|^2 \left| \nabla_{\left[\bupsilon_a \right]_{b}} X_{a,t} \right|^2 \right] \mathbb{E} \left[ \sup_{\bupsilon_a \in \bar{\bUpsilon}_a} \left| \partial_{\left[\bupsilon_a \right]_{b}} \log f_{a,t} \right|^2 \right] \\
		&+ \frac{1}{p_{-}^2} \mathbb{E} \left[ \sup_{\bupsilon_a \in \bar{\bUpsilon}_a} \left| \partial_{\left[\bupsilon_a \right]_{b}} \log f_{a,t} \right|^2 \right]
	\end{align*}
	and
	\begin{align*}
		&\mathbb{E} \left[ \sup_{\btheta \in \bar{\bTheta}} \left| \frac{( f_{j,t} \mathbf{p}_{\bullet j}^{\prime} \bpi_{t-1 \mid t-1} ) ( \nabla_{\left[\bupsilon_a \right]_{b}} f_{a,t} \mathbf{p}_{\bullet a}^{\prime} \bpi_{t-1 \mid t-1} )}{( \mathbf{f}_t^{\prime} \mathbf{P}^{\prime} \bpi_{t-1 \mid t-1}  )^2} \right|^2 \right] \\
		&\leq \frac{1}{p_{-}^4} \mathbb{E} \left[ \sup_{\bupsilon_a \in \bar{\bUpsilon}_a} \left| \partial_{x_a} \log f_{a,t} \right|^2 \left| \nabla_{\left[\bupsilon_a \right]_{b}} X_{a,t} \right|^2 \right] \\
		&+ \frac{2}{p_{-}^4} \mathbb{E} \left[ \sup_{\bupsilon_a \in \bar{\bUpsilon}_a} \left| \partial_{x_a} \log f_{a,t} \right|^2 \left| \nabla_{\left[\bupsilon_a \right]_{b}} X_{a,t} \right|^2 \right] \mathbb{E} \left[ \sup_{\bupsilon_a \in \bar{\bUpsilon}_a} \left| \partial_{\left[\bupsilon_a \right]_{b}} \log f_{a,t} \right|^2 \right] \\
		&+ \frac{1}{p_{-}^4} \mathbb{E} \left[ \sup_{\bupsilon_a \in \bar{\bUpsilon}_a} \left| \partial_{\left[\bupsilon_a \right]_{b}} \log f_{a,t} \right|^2 \right]
	\end{align*}
	for all $a \in \{1,...,J\}$ and $b \in \{1,...,d_a\}$ where the inequalities follow since $p_{-} \leq p_{ij} \leq 1$ for all $i,j \in \{1,...,J\}$. The conditions in Equations \eqref{eq:c.c.1} and \eqref{eq:c.d.1} thus follow since, for each $j \in \{1,...,J\}$, $\mathbb{E} [ \sup_{\bupsilon_j \in \bar{\bUpsilon}_{j}} | \partial_{x_j} \log f_{j,t} |^{2} | | \nabla_{\bupsilon_j} X_{j,t} ||_{2}^{2} ] < \infty$ and $\mathbb{E} [ \sup_{\bupsilon_j \in \bar{\bUpsilon}_{j}} || \partial_{\bupsilon_j} \log f_{j,t} ||_{2}^{2} ] < \infty$ by assumption. For Conditions (ii) and (iii), note that, for each $(n,m) \in \{1,...,J-1\} \times \{1,...,d\}$,
	\begin{align*}
		&\left[ (\bphi_{t}^{\bpi})^{(r)} (\bx ; \btheta) \right]_{(n-1)d+m} \\
		&= \sum_{j=1}^{J-1} \left( \partial_{\left[ \bpi \right]_j} \left[ \bphi_{t}^{(r)} (\bpi_{t-r \mid t-r} (\btheta) ; \btheta) \right]_{n} - \partial_{\left[ \bpi \right]_J} \left[ \bphi_{t}^{(r)} (\bpi_{t-r \mid t-r} (\btheta) ; \btheta) \right]_{n} \right) \left[ \bx \right]_{(j-1)d+m} \\
		&+ \partial_{\left[ \btheta \right]_m} \left[ \bphi_{t}^{(r)} (\bpi_{t-r \mid t-r} (\btheta) ; \btheta) \right]_{n}
	\end{align*}
	for all $r \in \mathbb{N}$ and $\bx \in \mathbb{R}^{(J-1)d}$. Note that
	\begin{equation*}
		|| \bphi_{t}^{(r)} (\bx;\btheta) - \bphi_{t}^{(r)} (\by;\btheta) ||_{2} \leq \Lambda (\bphi_{t}^{(r)};\btheta) || \bx - \by ||_{2}
	\end{equation*}
	for all $r \in \mathbb{N}$ and $\bx, \by \in \mathcal{S}$, so, for each $j,n \in \{1,...,J-1\}$, 
	\begin{equation*}
		\left| \partial_{\left[ \bpi \right]_j} \left[ \bphi_{t}^{(r)} (\bpi_{t-r \mid t-r} (\btheta) ; \btheta) \right]_{n} - \partial_{\left[ \bpi \right]_J} \left[ \bphi_{t}^{(r)} (\bpi_{t-r \mid t-r} (\btheta) ; \btheta) \right]_{n} \right| \leq \Lambda (\bphi_{t}^{(r)};\btheta)
	\end{equation*}
	for all $r \in \mathbb{N}$, see Section 6 in \citet{StraumannMikosch2006} for a similar argument. Hence, for each $(n,m) \in \{1,...,J-1\} \times \{1,...,d\}$,
	\begin{equation*}
		\left| \left[ (\bphi_{t}^{\bpi})^{(r)} (\bx ; \btheta) \right]_{(n-1)d+m} - \left[ (\bphi_{t}^{\bpi})^{(r)} (\by ; \btheta) \right]_{(n-1)d+m} \right| \leq \Lambda (\bphi_{t}^{(r)};\btheta) \sum_{j=1}^{J-1} \left| \left[ \bx \right]_{(j-1)d+m} - \left[ \by \right]_{(j-1)d+m} \right|
	\end{equation*}
	for all $r \in \mathbb{N}$ and $\bx, \by \in \mathbb{R}^{(J-1)d}$. We thus have that
	\begin{equation}
		\Lambda ((\bphi_{t}^{\bpi})^{(r)};\btheta) \leq C \Lambda (\bphi_{t}^{(r)};\btheta) \label{eq:VIENNA}
	\end{equation}
	for all $r \in \mathbb{N}$. Conditions (ii) and (iii) thus follow by using the same arguments as the ones used to show that $\mathbb{E} [\log^{+} \sup_{\btheta \in \bTheta} \Lambda (\bphi_t;\btheta)] < \infty$ and there exists an $r \in \mathbb{N}$ such that $-\infty \leq \mathbb{E} [\log \sup_{\btheta \in \bTheta} \Lambda (\bphi_t^{(r)};\btheta)] < 0$ (see the proof of Lemma \ref{LemmaInvertibilityFilter}).
	
	We now show that $\mathbb{E} [ \sup_{\btheta \in \bar{\bTheta}} || \nabla_{\btheta} \bpi_{t \mid t}^{-} (\btheta) ||_{2}^{2} ] < \infty$. Under Conditions (i)-(iii), we have that
	\begin{equation*}
		\nabla_{\btheta} \bpi_{t \mid t}^{-} (\btheta) = \sum_{k=0}^{\infty} ( (\bphi_{t}^{\bpi})^{(k)} ( \bphi_{t-k}^{\bpi} (\mathbf{0};\btheta) ; \btheta) - (\bphi_{t}^{\bpi})^{(k)} ( \mathbf{0} ; \btheta) ),
	\end{equation*}
	where $(\bphi_{t}^{\bpi})^{(0)} ( \bx ; \btheta) = \bx$ by convention. Thus, 
	\begin{equation}
		\begin{aligned}
			|| \nabla_{\btheta} \bpi_{t \mid t}^{-} (\btheta) ||_2^2 &\leq \sum_{k=0}^{\infty} \sum_{l=0}^{\infty} || (\bphi_{t}^{\bpi})^{(k)} ( \bphi_{t-k}^{\bpi} (\mathbf{0};\btheta) ; \btheta) - (\bphi_{t}^{\bpi})^{(k)} ( \mathbf{0} ; \btheta) ||_2 \\
			& \quad \quad \quad \quad \ \cdot || (\bphi_{t}^{\bpi})^{(l)} ( \bphi_{t-l}^{\bpi} (\mathbf{0};\btheta) ; \btheta) - (\bphi_{t}^{\bpi})^{(l)} ( \mathbf{0} ; \btheta) ||_2.
		\end{aligned} \label{eq:momentmoment}
	\end{equation}
	Note that, by Equations \eqref{eq:cc8} and \eqref{eq:VIENNA}, there exists an $\alpha \in (0,1)$ such that
	\begin{equation}
		\sup_{\underset{\bx \neq \by}{\bx,\by \in \mathbb{R}^{(J-1)d}}} \frac{|| (\bphi_{t}^{\bpi})^{(k)} (\bx;\btheta) - (\bphi_{t}^{\bpi})^{(k)} (\by;\btheta) ||_{2}}{|| \bx - \by ||_{2}} \leq C \alpha^{k} \label{eq:anothermoment}
	\end{equation}
	for all $k \in \mathbb{N}_{0}$. Thus, continuing from Equation \eqref{eq:momentmoment} and using Equation \eqref{eq:anothermoment},
	\begin{equation*}
		|| \nabla_{\btheta} \bpi_{t \mid t}^{-} (\btheta) ||_2^2 \leq C \sum_{k=0}^{\infty} \sum_{l=0}^{\infty} \alpha^k || \bphi_{t-k}^{\bpi} (\mathbf{0};\btheta) ||_2 \alpha^l || \bphi_{t-l}^{\bpi} (\mathbf{0};\btheta) ||_2.
	\end{equation*}
	Hence, we have that
	\begin{equation*}
		\mathbb{E} \left[\sup_{\btheta \in \bar{\bTheta}} || \nabla_{\btheta} \bpi_{t \mid t}^{-} (\btheta) ||_{2}^{2} \right] \leq C \mathbb{E} \left[\sup_{\btheta \in \bar{\bTheta}} ||\bphi_{t}^{\bpi} (\mathbf{0};\btheta)||_2^2 \right].
	\end{equation*}
	Hence, $\mathbb{E} [ \sup_{\btheta \in \bar{\bTheta}} || \nabla_{\btheta} \bpi_{t \mid t}^{-} (\btheta) ||_{2}^{2} ] < \infty$ since $\mathbb{E} [\sup_{\btheta \in \bar{\bTheta}} ||\bphi_{t}^{\bpi} (\mathbf{0};\btheta)||_2^2] < \infty$, which follows from the condition in Equation \eqref{eq:c.i.1}.
	
	Moving on to the second part of the proof, $(\nabla_{\btheta} \hat{\bpi}_{t \mid t}^{-} (\btheta))_{t \in \mathbb{N}_0}$ is a stochastic process taking values in $\mathbb{R}^{(J-1)d}$ given by
	\begin{equation*}
		\nabla_{\btheta} \hat{\bpi}_{t \mid t}^{-} (\btheta) = \hat{\bphi}_{t}^{\bpi} (\nabla_{\btheta} \hat{\bpi}_{t-1 \mid t-1}^{-} (\btheta) ; \btheta)
	\end{equation*}
	with
	\begin{align*}
		&\left[ \hat{\bphi}_{t}^{\bpi} (\nabla_{\btheta} \hat{\bpi}_{t-1 \mid t-1}^{-} (\btheta) ; \btheta) \right]_{(n-1)d+m} \\
		&= \sum_{j=1}^{J-1} \left( \partial_{\left[ \bpi \right]_j} \left[ \hat{\bphi}_{t} (\hat{\bpi}_{t-1 \mid t-1} (\btheta) ; \btheta) \right]_{n} - \partial_{\left[ \bpi \right]_J} \left[ \hat{\bphi}_{t} (\hat{\bpi}_{t-1 \mid t-1} (\btheta) ; \btheta) \right]_{n} \right) \\
		& \quad \quad \quad \cdot \left[ \nabla_{\btheta} \hat{\bpi}_{t-1 \mid t-1}^{-} (\btheta) \right]_{(j-1)d+m} \\
		&+ \partial_{\left[ \btheta \right]_m} \left[ \hat{\bphi}_{t} (\hat{\bpi}_{t-1 \mid t-1} (\btheta) ; \btheta) \right]_{n}, \quad (n,m) \in \{1,...,J-1\} \times \{1,...,d\},
	\end{align*}
	where $\nabla_{\btheta} \hat{\bpi}_{0 \mid 0}^{-} (\btheta)$ is given. The fact that $\sup_{\btheta \in \bar{\bTheta}} || \nabla_{\btheta} \hat{\bpi}_{t \mid t}^{-} (\btheta) - \nabla_{\btheta} \bpi_{t \mid t}^{-} (\btheta) ||_{2} \overset{e.a.s.}{\rightarrow} 0$ as $t \rightarrow \infty$ follows also from Theorem 2.10 in \cite{StraumannMikosch2006} if
	\begin{enumerate} [(i)]
		\item $\mathbb{E} [\log^{+} \sup_{\btheta \in \bar{\bTheta}} || \nabla_{\btheta} \bpi_{t \mid t}^{-} (\btheta) ||_{2} ] < \infty$,
		\item $\sup_{\btheta \in \bar{\bTheta}} || \hat{\bphi}_t^{\bpi} (\mathbf{0};\btheta) - \bphi_t^{\bpi} (\mathbf{0};\btheta) ||_{2} \overset{e.a.s.}{\rightarrow} 0$ as $t \rightarrow \infty$, and
		\item $\sup_{\btheta \in \bar{\bTheta}} \Lambda (\hat{\bphi}_t^{\bpi} - \bphi_t^{\bpi};\btheta) \overset{e.a.s.}{\rightarrow} 0$ as $t \rightarrow \infty$,
	\end{enumerate}
	which we now show. Condition (i) follows from Lemma 2.2 in \cite{StraumannMikosch2006} since $\mathbb{E} [ \sup_{\btheta \in \bar{\bTheta}} || \nabla_{\btheta} \bpi_{t \mid t}^{-} (\btheta) ||_{2}^{2} ] < \infty$. Condition (ii) follows if, for each $(n,m) \in \{1,...,J-1\} \times \{1,...,d\}$,
	\begin{equation}
		\sup_{\btheta \in \bar{\bTheta}} \left| \partial_{\left[ \btheta \right]_m} \left[ \hat{\bphi}_{t} (\hat{\bpi}_{t-1 \mid t-1} (\btheta) ; \btheta) \right]_{n} - \partial_{\left[ \btheta \right]_m} \left[ \bphi_{t} (\bpi_{t-1 \mid t-1} (\btheta) ; \btheta) \right]_{n} \right| \overset{e.a.s.}{\rightarrow} 0 \quad \text{as} \quad t \rightarrow \infty. \label{eq:c.ii.2}
	\end{equation}
	Let $\hat{\pi}_{j,t \mid t}$ denote $\hat{\pi}_{j,t \mid t} (\btheta)$, $\hat{f}_{j,t}$ denote $f_j (Y_t;\hat{X}_{j,t},\bupsilon_j)$, and $\hat{X}_{j,t}$ denote $\hat{X}_{j,t} (\bupsilon_j)$. Note that the condition in Equation \eqref{eq:c.ii.2} follows if, for each $j \in \{1,...,J-1\}$,
	\begin{equation}
		\sup_{\btheta \in \bar{\bTheta}} \left| \frac{\hat{f}_{b,t} \hat{\pi}_{a,t-1 \mid t-1}}{\hat{\mathbf{f}}_t^{\prime} \mathbf{P}^{\prime} \hat{\bpi}_{t-1 \mid t-1}} - \frac{f_{b,t} \pi_{a,t-1 \mid t-1}}{\mathbf{f}_t^{\prime} \mathbf{P}^{\prime} \bpi_{t-1 \mid t-1}} \right| \overset{e.a.s.}{\rightarrow} 0 \quad \text{as} \quad t \rightarrow \infty \label{eq:c.a.2}
	\end{equation}
	and
	\begin{equation}
		\sup_{\btheta \in \bar{\bTheta}} \left| \frac{\hat{f}_{j,t} \mathbf{p}_{\bullet j}^{\prime} \hat{\bpi}_{t-1 \mid t-1} \hat{f}_{b,t} \hat{\pi}_{a,t-1 \mid t-1} }{( \hat{\mathbf{f}}_t^{\prime} \mathbf{P}^{\prime} \hat{\bpi}_{t-1 \mid t-1} )^2} - \frac{f_{j,t} \mathbf{p}_{\bullet j}^{\prime} \bpi_{t-1 \mid t-1} f_{b,t} \pi_{a,t-1 \mid t-1} }{( \mathbf{f}_t^{\prime} \mathbf{P}^{\prime} \bpi_{t-1 \mid t-1} )^2} \right| \overset{e.a.s.}{\rightarrow} 0 \quad \text{as} \quad t \rightarrow \infty \label{eq:c.b.2}
	\end{equation}
	for all $a,b \in \{1,...,J\}$, and
	\begin{equation}
		\sup_{\btheta \in \bar{\bTheta}} \left| \frac{\nabla_{\left[\bupsilon_a \right]_{b}} \hat{f}_{a,t} \mathbf{p}_{\bullet a}^{\prime} \hat{\bpi}_{t-1 \mid t-1} }{\hat{\mathbf{f}}_t^{\prime} \mathbf{P}^{\prime} \hat{\bpi}_{t-1 \mid t-1} } - \frac{\nabla_{\left[\bupsilon_a \right]_{b}} f_{a,t} \mathbf{p}_{\bullet a}^{\prime} \bpi_{t-1 \mid t-1} }{\mathbf{f}_t^{\prime} \mathbf{P}^{\prime} \bpi_{t-1 \mid t-1} } \right| \overset{e.a.s.}{\rightarrow} 0 \quad \text{as} \quad t \rightarrow \infty \label{eq:c.c.2}
	\end{equation}
	and
	\begin{equation}
		\begin{aligned}
			&\sup_{\btheta \in \bar{\bTheta}} \left| \frac{( \hat{f}_{j,t} \mathbf{p}_{\bullet j}^{\prime} \hat{\bpi}_{t-1 \mid t-1} ) ( \nabla_{\left[\bupsilon_a \right]_{b}} \hat{f}_{a,t} \mathbf{p}_{\bullet a}^{\prime} \hat{\bpi}_{t-1 \mid t-1} )}{( \hat{\mathbf{f}}_t^{\prime} \mathbf{P}^{\prime} \hat{\bpi}_{t-1 \mid t-1}  )^2} \right. \\
			& \left. \quad \quad - \frac{( f_{j,t} \mathbf{p}_{\bullet j}^{\prime} \bpi_{t-1 \mid t-1} ) ( \nabla_{\left[\bupsilon_a \right]_{b}} f_{a,t} \mathbf{p}_{\bullet a}^{\prime} \bpi_{t-1 \mid t-1} )}{( \mathbf{f}_t^{\prime} \mathbf{P}^{\prime} \bpi_{t-1 \mid t-1}  )^2} \right| \overset{e.a.s.}{\rightarrow} 0 \quad \text{as} \quad t \rightarrow \infty
		\end{aligned} \label{eq:c.d.2}
	\end{equation}
	for all $a \in \{1,...,J\}$ and $b \in \{1,...,d_a\}$, where the terms in Equations \eqref{eq:c.a.2} and \eqref{eq:c.b.2} concern the terms in Equations \eqref{eq:t3} and \eqref{eq:t5}, and the terms in Equations \eqref{eq:c.c.2} and \eqref{eq:c.d.2} concern the terms in Equations \eqref{eq:tc} and \eqref{eq:te}. For the condition in Equation \eqref{eq:c.a.2}, we have that
	\begin{equation}
		\begin{aligned}
			&\sup_{\btheta \in \bar{\bTheta}} \left| \frac{\hat{f}_{b,t} \hat{\pi}_{a,t-1 \mid t-1}}{\hat{\mathbf{f}}_t^{\prime} \mathbf{P}^{\prime} \hat{\bpi}_{t-1 \mid t-1}} - \frac{f_{b,t} \pi_{a,t-1 \mid t-1}}{\mathbf{f}_t^{\prime} \mathbf{P}^{\prime} \bpi_{t-1 \mid t-1}} \right| \\
			&\leq \sup_{\btheta \in \bar{\bTheta}} \left| \frac{\hat{f}_{b,t} \hat{\pi}_{a,t-1 \mid t-1}}{\hat{\mathbf{f}}_t^{\prime} \mathbf{P}^{\prime} \hat{\bpi}_{t-1 \mid t-1}} - \frac{\hat{f}_{b,t} \pi_{a,t-1 \mid t-1}}{\hat{\mathbf{f}}_t^{\prime} \mathbf{P}^{\prime} \bpi_{t-1 \mid t-1}} \right| + \sup_{\btheta \in \bar{\bTheta}} \left| \frac{\hat{f}_{b,t} \pi_{a,t-1 \mid t-1}}{\hat{\mathbf{f}}_t^{\prime} \mathbf{P}^{\prime} \bpi_{t-1 \mid t-1}} - \frac{f_{b,t} \pi_{a,t-1 \mid t-1}}{\mathbf{f}_t^{\prime} \mathbf{P}^{\prime} \bpi_{t-1 \mid t-1}} \right|
		\end{aligned} \label{eq:boo}
	\end{equation}
	for all $a,b \in \{1,...,J\}$. First, we consider $\left| \frac{\hat{f}_{b,t} \hat{\pi}_{a,t-1 \mid t-1}}{\hat{\mathbf{f}}_t^{\prime} \mathbf{P}^{\prime} \hat{\bpi}_{t-1 \mid t-1}} - \frac{\hat{f}_{b,t} \pi_{a,t-1 \mid t-1}}{\hat{\mathbf{f}}_t^{\prime} \mathbf{P}^{\prime} \bpi_{t-1 \mid t-1}} \right|$. Let $m_{1,t}^{\btheta} : \mathcal{S} \rightarrow \mathbb{R}$ be given by
	\begin{equation*}
		m_{1,t}^{\btheta} (\bs) = \frac{f_{b} (Y_t;\hat{X}_{b,t},\bupsilon_b) s_a}{(\mathbf{f} (Y_t;\hat{\bX}_t,\bupsilon))^{\prime} \mathbf{P}^{\prime} \bs}.
	\end{equation*}
	Note that $m_{1,t}^{\btheta} (\hat{\bpi}_{t-1 \mid t-1}) = \frac{\hat{f}_{b,t} \hat{\pi}_{a,t-1 \mid t-1}}{\hat{\mathbf{f}}_t^{\prime} \mathbf{P}^{\prime} \hat{\bpi}_{t-1 \mid t-1}}$ and $m_{1,t}^{\btheta} (\bpi_{t-1 \mid t-1}) = \frac{\hat{f}_{b,t} \pi_{a,t-1 \mid t-1}}{\hat{\mathbf{f}}_t^{\prime} \mathbf{P}^{\prime} \bpi_{t-1 \mid t-1}}$. Thus, by the mean value theorem and the Cauchy-Schwarz inequality, there exists a $\bar{\bpi}_{t-1 \mid t-1} \in \mathcal{S}$ such that
	\begin{equation}
		\left| \frac{\hat{f}_{b,t} \hat{\pi}_{a,t-1 \mid t-1}}{\hat{\mathbf{f}}_t^{\prime} \mathbf{P}^{\prime} \hat{\bpi}_{t-1 \mid t-1}} - \frac{\hat{f}_{b,t} \pi_{a,t-1 \mid t-1}}{\hat{\mathbf{f}}_t^{\prime} \mathbf{P}^{\prime} \bpi_{t-1 \mid t-1}} \right| \leq C || \nabla_{\bs} m_{1,t}^{\btheta} (\bar{\bpi}_{t-1 \mid t-1}) ||_{2} || \hat{\bpi}_{t-1 \mid t-1} - \bpi_{t-1 \mid t-1}||_{2}. \label{eq:boo1}
	\end{equation}
	If $m \neq a$, then
	\begin{equation}
		\begin{aligned}
			| [ \nabla_{\bs} m_{1,t}^{\btheta} (\bar{\bpi}_{t-1 \mid t-1}) ]_m | 
			&\leq \frac{\hat{f}_{b,t} \bar{\pi}_{a,t-1 \mid t-1} \hat{\mathbf{f}}_t^{\prime} \mathbf{p}_{m \bullet}}{(\hat{\mathbf{f}}_t^{\prime} \mathbf{P}^{\prime} \bar{\bpi}_{t-1 \mid t-1})^2} \\
			&\leq \frac{1}{p_{-}^2} \frac{\hat{f}_{b,t} \hat{\mathbf{f}}_t^{\prime} \mathbf{1}}{(\hat{\mathbf{f}}_t^{\prime} \mathbf{1})^2} \\ &\leq \frac{1}{p_{-}^2},
		\end{aligned} \label{eq:boo2}
	\end{equation}
	and if $m = a$, then
	\begin{equation}
		\begin{aligned}
			| [ \nabla_{\bs} m_{1,t}^{\btheta} (\bar{\bpi}_{t-1 \mid t-1}) ]_m | 
			&\leq \frac{\hat{f}_{b,t}}{\hat{\mathbf{f}}_t^{\prime} \mathbf{P}^{\prime} \bar{\bpi}_{t-1 \mid t-1}} + \frac{\hat{f}_{b,t} \bar{\pi}_{m,t-1 \mid t-1} \hat{\mathbf{f}}_t^{\prime} \mathbf{p}_{m \bullet}}{(\hat{\mathbf{f}}_t^{\prime} \mathbf{P}^{\prime} \bar{\bpi}_{t-1 \mid t-1})^2} \\
			&\leq \frac{1}{p_{-}} \frac{\hat{f}_{b,t}}{\hat{\mathbf{f}}_t^{\prime} \mathbf{1}} + \frac{1}{p_{-}^2} \frac{\hat{f}_{b,t} \hat{\mathbf{f}}_t^{\prime} \mathbf{1}}{(\hat{\mathbf{f}}_t^{\prime} \mathbf{1})^2} \\
			&\leq \frac{1}{p_{-}} + \frac{1}{p_{-}^2},
		\end{aligned} \label{eq:boo3}
	\end{equation}
	where the second inequality in both Equations \eqref{eq:boo2} and \eqref{eq:boo3} follows since $p_{-} \leq p_{ij} \leq 1$ for all $i,j \in \{1,...,J\}$. Thus, continuing from Equation \eqref{eq:boo1},
	\begin{equation}
		\left| \frac{\hat{f}_{b,t} \hat{\pi}_{a,t-1 \mid t-1}}{\hat{\mathbf{f}}_t^{\prime} \mathbf{P}^{\prime} \hat{\bpi}_{t-1 \mid t-1}} - \frac{\hat{f}_{b,t} \pi_{a,t-1 \mid t-1}}{\hat{\mathbf{f}}_t^{\prime} \mathbf{P}^{\prime} \bpi_{t-1 \mid t-1}} \right| \leq C || \hat{\bpi}_{t-1 \mid t-1} - \bpi_{t-1 \mid t-1}||_{2}. \label{eq:boo4}
	\end{equation}
	We now consider $\left| \frac{\hat{f}_{b,t} \pi_{a,t-1 \mid t-1}}{\hat{\mathbf{f}}_t^{\prime} \mathbf{P}^{\prime} \bpi_{t-1 \mid t-1}} - \frac{f_{b,t} \pi_{a,t-1 \mid t-1}}{\mathbf{f}_t^{\prime} \mathbf{P}^{\prime} \bpi_{t-1 \mid t-1}} \right|$. Let $m_{2,t}^{\btheta} : \mathcal{X}_1 \times \cdots \times \mathcal{X}_J \rightarrow \mathbb{R}$ be given by
	\begin{equation*}
		m_{2,t}^{\btheta} (\bx) = \frac{f_{b} (Y_t;x_b,\bupsilon_b) \pi_{a,t-1 \mid t-1}}{(\mathbf{f}(Y_t;\bx,\bupsilon))^{\prime} \mathbf{P}^{\prime} \bpi_{t-1 \mid t-1}}.
	\end{equation*}
	Note that $m_{2,t}^{\btheta} (\hat{\bX}_t) = \frac{\hat{f}_{b,t} \pi_{a,t-1 \mid t-1}}{\hat{\mathbf{f}}_t^{\prime} \mathbf{P}^{\prime} \bpi_{t-1 \mid t-1}}$ and $m_{2,t}^{\btheta} (\bX_t) = \frac{f_{b,t} \pi_{a,t-1 \mid t-1}}{\mathbf{f}_t^{\prime} \mathbf{P}^{\prime} \bpi_{t-1 \mid t-1}}$. Thus, by the mean value theorem and the Cauchy-Schwarz inequality again, there exists an $\bar{\bX}_t \in \mathcal{X}_1 \times \cdots \times \mathcal{X}_J$ such that
	\begin{equation}
		\left| \frac{\hat{f}_{b,t} \pi_{a,t-1 \mid t-1}}{\hat{\mathbf{f}}_t^{\prime} \mathbf{P}^{\prime} \bpi_{t-1 \mid t-1}} - \frac{f_{b,t} \pi_{a,t-1 \mid t-1}}{\mathbf{f}_t^{\prime} \mathbf{P}^{\prime} \bpi_{t-1 \mid t-1}} \right| \leq || \nabla_{\bx} m_{2,t}^{\btheta} (\bar{\bX}_t) ||_{2} || \hat{\bX}_t - \bX_t ||_{2}. \label{eq:booa}
	\end{equation}
	If $m \neq b$, then
	\begin{equation}
		\begin{aligned}
			| [ \nabla_{\bx} m_{2,t}^{\btheta} (\bar{\bX}_t) ]_m | &\leq \frac{\bar{f}_{b,t} \pi_{a,t-1 \mid t-1} | \partial_{x_m} \bar{f}_{m,t} | \mathbf{p}_{\bullet m}^{\prime} \bpi_{t-1 \mid t-1}}{(\bar{\mathbf{f}}_t^{\prime} \mathbf{P}^{\prime} \bpi_{t-1 \mid t-1})^2} \\
			&\leq \frac{1}{p_{-}^2} \frac{\bar{f}_{b,t} | \partial_{x_m} \bar{f}_{m,t} |}{(\bar{\mathbf{f}}_t^{\prime} \mathbf{1})^2} \\
			&\leq \frac{1}{p_{-}^2} | \partial_{x_m} \log \bar{f}_{m,t} |,
		\end{aligned} \label{eq:boob}
	\end{equation}
	and if $m = b$, then
	\begin{equation}
		\begin{aligned}
			| [ \nabla_{\bx} m_{2,t}^{\btheta} (\bar{\bX}_t) ]_m | &\leq \frac{|\partial_{x_m} \bar{f}_{m,t} | \pi_{a,t-1 \mid t-1}}{\bar{\mathbf{f}}_t^{\prime} \mathbf{P}^{\prime} \bpi_{t-1 \mid t-1}} + \frac{\bar{f}_{m,t} \pi_{a,t-1 \mid t-1} | \partial_{x_m} \bar{f}_{m,t} | \mathbf{p}_{\bullet m}^{\prime} \bpi_{t-1 \mid t-1}}{(\bar{\mathbf{f}}_t^{\prime} \mathbf{P}^{\prime} \bpi_{t-1 \mid t-1})^2} \\
			&\leq \frac{1}{p_{-}} \frac{|\partial_{x_m} \bar{f}_{m,t} |}{\bar{\mathbf{f}}_t^{\prime} \mathbf{1}} + \frac{1}{p_{-}^2} \frac{\bar{f}_{m,t} | \partial_{x_m} \bar{f}_{m,t} |}{(\bar{\mathbf{f}}_t^{\prime} \mathbf{1})^2} \\
			&\leq \left(\frac{1}{p_{-}}+\frac{1}{p_{-}^2}\right) | \partial_{x_m} \log \bar{f}_{m,t} |,
		\end{aligned} \label{eq:booc}
	\end{equation}
	where the second inequality in both Equations \eqref{eq:boob} and \eqref{eq:booc} follows since $p_{-} \leq p_{ij} \leq 1$ for all $i,j \in \{1,...,J\}$. Thus, continuing from Equation \eqref{eq:booa},
	\begin{equation}
		\left| \frac{\hat{f}_{b,t} \pi_{a,t-1 \mid t-1}}{\hat{\mathbf{f}}_t^{\prime} \mathbf{P}^{\prime} \bpi_{t-1 \mid t-1}} - \frac{f_{b,t} \pi_{a,t-1 \mid t-1}}{\mathbf{f}_t^{\prime} \mathbf{P}^{\prime} \bpi_{t-1 \mid t-1}} \right| \leq C \sum_{m,n = 1}^{J} | \partial_{x_m} \log \bar{f}_{m,t} | | \hat{X}_{n,t} - X_{n,t} |. \label{eq:bood}
	\end{equation}
	Hence, continuing from Equation \eqref{eq:boo} and using Equations \eqref{eq:boo4} and \eqref{eq:bood}, we have that
	\begin{equation}
		\begin{aligned}
			&\sup_{\btheta \in \bar{\bTheta}} \left| \frac{\hat{f}_{b,t} \hat{\pi}_{a,t-1 \mid t-1}}{\hat{\mathbf{f}}_t^{\prime} \mathbf{P}^{\prime} \hat{\bpi}_{t-1 \mid t-1}} - \frac{f_{b,t} \pi_{a,t-1 \mid t-1}}{\mathbf{f}_t^{\prime} \mathbf{P}^{\prime} \bpi_{t-1 \mid t-1}} \right| \\
			&\leq C \sup_{\btheta \in \bar{\bTheta}} || \hat{\bpi}_{t-1 \mid t-1} - \bpi_{t-1 \mid t-1}||_{2} \\
			&+ C \sum_{m,n = 1}^{J} \sup_{\bupsilon_m \in \bar{\bUpsilon}_m} | \partial_{x_m} \log \bar{f}_{m,t} | \sup_{\bupsilon_n \in \bar{\bUpsilon}_n} | \hat{X}_{n,t} - X_{n,t} |
		\end{aligned} \label{eq:booboo}
	\end{equation} 
	for all $a,b \in \{1,...,J\}$. The condition in Equation \eqref{eq:c.a.2} thus follows from Lemmas 2.1 and 2.2 in \cite{StraumannMikosch2006} since $(Y_t)_{t \in \mathbb{Z}}$ is stationary by Assumption \ref{AssumptionY}, for each $j \in \{1,...,J\}$, there exists an $m_j > 0$ such that $\mathbb{E} [ \sup_{\bupsilon_j \in \bar{\bUpsilon}_{j}} \sup_{x_j \in \mathcal{X}_{j}} | \partial_{x_j} \log f_{j,t} |^{m_j} ] < \infty$ by assumption, for each $j \in \{1,...,J\}$, $\sup_{\bupsilon_{j} \in \bar{\bUpsilon}_{j}} | \hat{X}_{j,t} - X_{j,t} | \overset{e.a.s.}{\rightarrow} 0$ as $t \rightarrow \infty$ by Lemma \ref{LemmaInvertibilityX}, and $\sup_{\btheta \in \bar{\bTheta}} || \hat{\bpi}_{t \mid t} - \bpi_{t \mid t} ||_{2} \overset{e.a.s.}{\rightarrow} 0$ as $t \rightarrow \infty$ by Lemma \ref{LemmaInvertibilityFilter}. For the condition in Equation \eqref{eq:c.b.2}, for each $j \in \{1,...,J-1\}$, we have that
	\begin{equation}
		\begin{aligned}
			&\sup_{\btheta \in \bar{\bTheta}} \left| \frac{\hat{f}_{j,t} \mathbf{p}_{\bullet j}^{\prime} \hat{\bpi}_{t-1 \mid t-1} \hat{f}_{b,t} \hat{\pi}_{a,t-1 \mid t-1} }{( \hat{\mathbf{f}}_t^{\prime} \mathbf{P}^{\prime} \hat{\bpi}_{t-1 \mid t-1} )^2} - \frac{f_{j,t} \mathbf{p}_{\bullet j}^{\prime} \bpi_{t-1 \mid t-1} f_{b,t} \pi_{a,t-1 \mid t-1} }{( \mathbf{f}_t^{\prime} \mathbf{P}^{\prime} \bpi_{t-1 \mid t-1} )^2} \right| \\
			&\leq \frac{1}{p_{-}} \sup_{\btheta \in \bar{\bTheta}} \left| \frac{\hat{f}_{b,t} \hat{\pi}_{a,t-1 \mid t-1} }{\hat{\mathbf{f}}_t^{\prime} \mathbf{P}^{\prime} \hat{\bpi}_{t-1 \mid t-1}} - \frac{f_{b,t} \pi_{a,t-1 \mid t-1} }{\mathbf{f}_t^{\prime} \mathbf{P}^{\prime} \bpi_{t-1 \mid t-1}} \right| \\
			&+ \frac{1}{p_{-}} \sup_{\btheta \in \bar{\bTheta}} \left| \frac{\hat{f}_{j,t} \mathbf{p}_{\bullet j}^{\prime} \hat{\bpi}_{t-1 \mid t-1} }{\hat{\mathbf{f}}_t^{\prime} \mathbf{P}^{\prime} \hat{\bpi}_{t-1 \mid t-1}} - \frac{f_{j,t} \mathbf{p}_{\bullet j}^{\prime} \bpi_{t-1 \mid t-1} }{\mathbf{f}_t^{\prime} \mathbf{P}^{\prime} \bpi_{t-1 \mid t-1}} \right| \\
			&\leq C \sum_{\alpha,\beta=1}^{J} \sup_{\btheta \in \bar{\bTheta}} \left| \frac{\hat{f}_{\beta,t} \hat{\pi}_{\alpha,t-1 \mid t-1} }{\hat{\mathbf{f}}_t^{\prime} \mathbf{P}^{\prime} \hat{\bpi}_{t-1 \mid t-1}} - \frac{f_{\beta,t} \pi_{\alpha,t-1 \mid t-1} }{\mathbf{f}_t^{\prime} \mathbf{P}^{\prime} \bpi_{t-1 \mid t-1}} \right|
		\end{aligned} \label{eq:SMUK1}
	\end{equation}
	for all $a,b \in \{1,...,J\}$ where the first inequality follows since $\hat{x}\hat{y} - xy = (\hat{x}-x)\hat{y} + x(\hat{y}-y)$ and $p_{-} \leq p_{ij} \leq 1$ for all $i,j \in \{1,...,J\}$, and the second inequality follows since $p_{ij} \leq 1$ for all $i,j \in \{1,...,J\}$. Note that the term on the right hand side of Equation \eqref{eq:SMUK1} is equal to a constant times the sum of the terms on the left hand side of Equation \eqref{eq:boo}. The condition in Equation \eqref{eq:c.b.2} thus follows by using the same arguments as the ones used to show the condition in Equation \eqref{eq:c.a.2}. For the condition in Equation \eqref{eq:c.c.2}, we have that
	\begin{equation}
		\begin{aligned}
			&\sup_{\btheta \in \bar{\bTheta}} \left| \frac{\nabla_{\left[\bupsilon_a \right]_{b}} \hat{f}_{a,t} \mathbf{p}_{\bullet a}^{\prime} \hat{\bpi}_{t-1 \mid t-1} }{\hat{\mathbf{f}}_t^{\prime} \mathbf{P}^{\prime} \hat{\bpi}_{t-1 \mid t-1} } - \frac{\nabla_{\left[\bupsilon_a \right]_{b}} f_{a,t} \mathbf{p}_{\bullet a}^{\prime} \bpi_{t-1 \mid t-1} }{\mathbf{f}_t^{\prime} \mathbf{P}^{\prime} \bpi_{t-1 \mid t-1} } \right| \\
			&\leq \sup_{\btheta \in \bar{\bTheta}} \left| \frac{\partial_{x_a} \hat{f}_{a,t} \nabla_{\left[\bupsilon_a \right]_{b}} \hat{X}_{a,t} \mathbf{p}_{\bullet a}^{\prime} \hat{\bpi}_{t-1 \mid t-1} }{\hat{\mathbf{f}}_t^{\prime} \mathbf{P}^{\prime} \hat{\bpi}_{t-1 \mid t-1} } - \frac{\partial_{x_a} f_{a,t} \nabla_{\left[\bupsilon_a \right]_{b}} X_{a,t} \mathbf{p}_{\bullet a}^{\prime} \bpi_{t-1 \mid t-1} }{\mathbf{f}_t^{\prime} \mathbf{P}^{\prime} \bpi_{t-1 \mid t-1} } \right| \\
			&+ \sup_{\btheta \in \bar{\bTheta}} \left| \frac{\partial_{\left[\bupsilon_a \right]_{b}} \hat{f}_{a,t} \mathbf{p}_{\bullet a}^{\prime} \hat{\bpi}_{t-1 \mid t-1} }{\hat{\mathbf{f}}_t^{\prime} \mathbf{P}^{\prime} \hat{\bpi}_{t-1 \mid t-1} } - \frac{\partial_{\left[\bupsilon_a \right]_{b}} f_{a,t} \mathbf{p}_{\bullet a}^{\prime} \bpi_{t-1 \mid t-1} }{\mathbf{f}_t^{\prime} \mathbf{P}^{\prime} \bpi_{t-1 \mid t-1} } \right| \\
			&\leq \sup_{\btheta \in \bar{\bTheta}} \left| \frac{\partial_{x_a} \hat{f}_{a,t} \mathbf{p}_{\bullet a}^{\prime} \hat{\bpi}_{t-1 \mid t-1} }{\hat{\mathbf{f}}_t^{\prime} \mathbf{P}^{\prime} \hat{\bpi}_{t-1 \mid t-1} } - \frac{\partial_{x_a} f_{a,t} \mathbf{p}_{\bullet a}^{\prime} \bpi_{t-1 \mid t-1} }{\mathbf{f}_t^{\prime} \mathbf{P}^{\prime} \bpi_{t-1 \mid t-1} } \right| \sup_{\bupsilon_a \in \bar{\bUpsilon}_a} \left| \nabla_{\left[\bupsilon_a \right]_{b}} X_{a,t} \right| \\
			&+ \frac{1}{p_{-}} \sup_{\bupsilon_a \in \bar{\bUpsilon}_a} \left| \partial_{x_a} \log \hat{f}_{a,t} \right| \sup_{\bupsilon_a \in \bar{\bUpsilon}_a} \left| \nabla_{\left[\bupsilon_a \right]_{b}} \hat{X}_{a,t} - \nabla_{\left[\bupsilon_a \right]_{b}} X_{a,t} \right| \\
			&+ \sup_{\btheta \in \bar{\bTheta}} \left| \frac{\partial_{\left[\bupsilon_a \right]_{b}} \hat{f}_{a,t} \mathbf{p}_{\bullet a}^{\prime} \hat{\bpi}_{t-1 \mid t-1} }{\hat{\mathbf{f}}_t^{\prime} \mathbf{P}^{\prime} \hat{\bpi}_{t-1 \mid t-1} } - \frac{\partial_{\left[\bupsilon_a \right]_{b}} f_{a,t} \mathbf{p}_{\bullet a}^{\prime} \bpi_{t-1 \mid t-1} }{\mathbf{f}_t^{\prime} \mathbf{P}^{\prime} \bpi_{t-1 \mid t-1} } \right| \\
			&\leq \sup_{\btheta \in \bar{\bTheta}} \left| \frac{\partial_{x_a} \hat{f}_{a,t} \mathbf{p}_{\bullet a}^{\prime} \hat{\bpi}_{t-1 \mid t-1} }{\hat{\mathbf{f}}_t^{\prime} \mathbf{P}^{\prime} \hat{\bpi}_{t-1 \mid t-1} } - \frac{\partial_{x_a} \hat{f}_{a,t} \mathbf{p}_{\bullet a}^{\prime} \bpi_{t-1 \mid t-1} }{\hat{\mathbf{f}}_t^{\prime} \mathbf{P}^{\prime} \bpi_{t-1 \mid t-1} } \right| \sup_{\bupsilon_a \in \bar{\bUpsilon}_a} \left| \nabla_{\left[\bupsilon_a \right]_{b}} X_{a,t} \right| \\
			&+ \sup_{\btheta \in \bar{\bTheta}} \left| \frac{\partial_{x_a} \hat{f}_{a,t} \mathbf{p}_{\bullet a}^{\prime} \bpi_{t-1 \mid t-1} }{\hat{\mathbf{f}}_t^{\prime} \mathbf{P}^{\prime} \bpi_{t-1 \mid t-1} } - \frac{\partial_{x_a} f_{a,t} \mathbf{p}_{\bullet a}^{\prime} \bpi_{t-1 \mid t-1} }{\mathbf{f}_t^{\prime} \mathbf{P}^{\prime} \bpi_{t-1 \mid t-1} } \right| \sup_{\bupsilon_a \in \bar{\bUpsilon}_a} \left| \nabla_{\left[\bupsilon_a \right]_{b}} X_{a,t} \right| \\
			&+ \frac{1}{p_{-}} \sup_{\bupsilon_a \in \bar{\bUpsilon}_a} \left| \partial_{x_a} \log \hat{f}_{a,t} \right| \sup_{\bupsilon_a \in \bar{\bUpsilon}_a} \left| \nabla_{\left[\bupsilon_a \right]_{b}} \hat{X}_{a,t} - \nabla_{\left[\bupsilon_a \right]_{b}} X_{a,t} \right| \\
			&+ \sup_{\btheta \in \bar{\bTheta}} \left| \frac{\partial_{\left[\bupsilon_a \right]_{b}} \hat{f}_{a,t} \mathbf{p}_{\bullet a}^{\prime} \hat{\bpi}_{t-1 \mid t-1} }{\hat{\mathbf{f}}_t^{\prime} \mathbf{P}^{\prime} \hat{\bpi}_{t-1 \mid t-1} } - \frac{\partial_{\left[\bupsilon_a \right]_{b}} \hat{f}_{a,t} \mathbf{p}_{\bullet a}^{\prime} \bpi_{t-1 \mid t-1} }{\hat{\mathbf{f}}_t^{\prime} \mathbf{P}^{\prime} \bpi_{t-1 \mid t-1} } \right| \\
			&+ \sup_{\btheta \in \bar{\bTheta}} \left| \frac{\partial_{\left[\bupsilon_a \right]_{b}} \hat{f}_{a,t} \mathbf{p}_{\bullet a}^{\prime} \bpi_{t-1 \mid t-1} }{\hat{\mathbf{f}}_t^{\prime} \mathbf{P}^{\prime} \bpi_{t-1 \mid t-1} } - \frac{\partial_{\left[\bupsilon_a \right]_{b}} f_{a,t} \mathbf{p}_{\bullet a}^{\prime} \bpi_{t-1 \mid t-1} }{\mathbf{f}_t^{\prime} \mathbf{P}^{\prime} \bpi_{t-1 \mid t-1} } \right|
		\end{aligned} \label{eq:SMUK2}
	\end{equation}
	for all $a \in \{1,...,J\}$ and $b \in \{1,...,d_a\}$ where the second inequality follows since $\hat{x}\hat{y} - xy = (\hat{x}-x)\hat{y} + x(\hat{y}-y)$ and $p_{-} \leq p_{ij} \leq 1$ for all $i,j \in \{1,...,J\}$. Hence, by using similar arguments to the ones used to go from Equation \eqref{eq:boo} to Equation \eqref{eq:booboo}, we have that
	\begin{equation}
		\begin{aligned}
			&\sup_{\btheta \in \bar{\bTheta}} \left| \frac{\nabla_{\left[\bupsilon_a \right]_{b}} \hat{f}_{a,t} \mathbf{p}_{\bullet a}^{\prime} \hat{\bpi}_{t-1 \mid t-1} }{\hat{\mathbf{f}}_t^{\prime} \mathbf{P}^{\prime} \hat{\bpi}_{t-1 \mid t-1} } - \frac{\nabla_{\left[\bupsilon_a \right]_{b}} f_{a,t} \mathbf{p}_{\bullet a}^{\prime} \bpi_{t-1 \mid t-1} }{\mathbf{f}_t^{\prime} \mathbf{P}^{\prime} \bpi_{t-1 \mid t-1} } \right| \\
			&\leq C \sup_{\bupsilon_a \in \bar{\bUpsilon}_a} \left| \partial_{x_a} \log \hat{f}_{a,t} \right| \sup_{\btheta \in \bar{\bTheta}} || \hat{\bpi}_{t-1 \mid t-1} - \bpi_{t-1 \mid t-1}||_{2} \sup_{\bupsilon_a \in \bar{\bUpsilon}_a} \left| \nabla_{\left[\bupsilon_a \right]_{b}} X_{a,t} \right| \\
			&+ C \sup_{\bupsilon_a \in \bar{\bUpsilon}_a} \left| \partial_{x_a x_a} \log \bar{f}_{a,t} \right| \sum_{n = 1}^{J} \sup_{\bupsilon_n \in \bar{\bUpsilon}_n} | \hat{X}_{n,t} - X_{n,t} | \sup_{\bupsilon_a \in \bar{\bUpsilon}_a} \left| \nabla_{\left[\bupsilon_a \right]_{b}} X_{a,t} \right| \\
			&+ C \sup_{\bupsilon_a \in \bar{\bUpsilon}_a} \left| \partial_{x_a} \log \bar{f}_{a,t} \right| \sum_{m,n = 1}^{J} \sup_{\bupsilon_m \in \bar{\bUpsilon}_m} | \partial_{x_m} \log \bar{f}_{m,t} | \sup_{\bupsilon_n \in \bar{\bUpsilon}_n} | \hat{X}_{n,t} - X_{n,t} | \sup_{\bupsilon_a \in \bar{\bUpsilon}_a} \left| \nabla_{\left[\bupsilon_a \right]_{b}} X_{a,t} \right| \\
			&+ C \sup_{\bupsilon_a \in \bar{\bUpsilon}_a} \left| \partial_{x_a} \log \hat{f}_{a,t} \right| \sup_{\bupsilon_a \in \bar{\bUpsilon}_a} \left| \nabla_{\left[\bupsilon_a \right]_{b}} \hat{X}_{a,t} - \nabla_{\left[\bupsilon_a \right]_{b}} X_{a,t} \right| \\
			&+ C \sup_{\bupsilon_a \in \bar{\bUpsilon}_a} \left| \partial_{\left[\bupsilon_a \right]_{b}} \log \hat{f}_{a,t} \right| \sup_{\btheta \in \bar{\bTheta}} || \hat{\bpi}_{t-1 \mid t-1} - \bpi_{t-1 \mid t-1}||_{2} \\
			&+ C \sup_{\bupsilon_a \in \bar{\bUpsilon}_a} \left| \partial_{\left[\bupsilon_a \right]_{b} x_a} \log \bar{f}_{a,t} \right| \sum_{n = 1}^{J} \sup_{\bupsilon_n \in \bar{\bUpsilon}_n} | \hat{X}_{n,t} - X_{n,t} | \\
			&+ C \sup_{\bupsilon_a \in \bar{\bUpsilon}_a} \left| \partial_{\left[\bupsilon_a \right]_{b}} \log \bar{f}_{a,t} \right| \sum_{m,n = 1}^{J} \sup_{\bupsilon_m \in \bar{\bUpsilon}_m} | \partial_{x_m} \log \bar{f}_{m,t} | \sup_{\bupsilon_n \in \bar{\bUpsilon}_n} | \hat{X}_{n,t} - X_{n,t} | \\
		\end{aligned}
	\end{equation}
	for all $a \in \{1,...,J\}$ and $b \in \{1,...,d_a\}$ where, by the triangle inequality and the mean value theorem,
	\begin{equation}
		\begin{aligned}
			&\sup_{\bupsilon_a \in \bar{\bUpsilon}_a} \left| \partial_{[\bupsilon_a]_b} \log \hat{f}_{a,t} \right| \\
			&\leq \sup_{\bupsilon_a \in \bar{\bUpsilon}_a} \left| \partial_{[\bupsilon_a]_b} \log \hat{f}_{a,t} - \partial_{[\bupsilon_a]_b} \log f_{a,t} \right| + \sup_{\bupsilon_a \in \bar{\bUpsilon}_a} \left| \partial_{[\bupsilon_a]_b} \log f_{a,t} \right| \\
			&\leq \sup_{\bupsilon_a \in \bar{\bUpsilon}_a} \sup_{x_a \in \mathcal{X}_{a}} \left| \partial_{[\bupsilon_a]_b x_a} \log f_{a,t} \right| \sup_{\bupsilon_a \in \bar{\bUpsilon}_a} \left| \hat{X}_{a,t} - X_{a,t} \right|+ \sup_{\bupsilon_a \in \bar{\bUpsilon}_a} \left| \partial_{[\bupsilon_a]_b} \log f_{a,t} \right|
		\end{aligned} \label{eq:obs1}
	\end{equation}
	and
	\begin{equation}
		\begin{aligned}
			&\sup_{\bupsilon_a \in \bar{\bUpsilon}_a} \left| \partial_{[\bupsilon_a]_b} \log \bar{f}_{a,t} \right| \\
			&\leq \sup_{\bupsilon_a \in \bar{\bUpsilon}_a} \left| \partial_{[\bupsilon_a]_b} \log \bar{f}_{a,t} - \partial_{[\bupsilon_a]_b} \log f_{a,t} \right| + \sup_{\bupsilon_a \in \bar{\bUpsilon}_a} \left| \partial_{[\bupsilon_a]_b} \log f_{a,t} \right| \\
			&\leq \sup_{\bupsilon_a \in \bar{\bUpsilon}_a} \sup_{x_a \in \mathcal{X}_{a}} \left| \partial_{[\bupsilon_a]_b x_a} \log f_{a,t} \right| \sup_{\bupsilon_a \in \bar{\bUpsilon}_a} \left| \bar{X}_{a,t} - X_{a,t} \right|+ \sup_{\bupsilon_a \in \bar{\bUpsilon}_a} \left| \partial_{[\bupsilon_a]_b} \log f_{a,t} \right| \\
			&\leq \sup_{\bupsilon_a \in \bar{\bUpsilon}_a} \sup_{x_a \in \mathcal{X}_{a}} \left| \partial_{[\bupsilon_a]_b x_a} \log f_{a,t} \right| \sup_{\bupsilon_a \in \bar{\bUpsilon}_a} \left| \hat{X}_{a,t} - X_{a,t} \right|+ \sup_{\bupsilon_a \in \bar{\bUpsilon}_a} \left| \partial_{[\bupsilon_a]_b} \log f_{a,t} \right|.
		\end{aligned} \label{eq:obs2}
	\end{equation}
	The condition in Equation \eqref{eq:c.c.2} thus follows from Lemmas 2.1 and 2.2 in \cite{StraumannMikosch2006} since – in addition to the arguments used to show the condition in Equation \eqref{eq:c.a.2} – for each $j \in \{1,...,J\}$, $\mathbb{E} [ \sup_{\bupsilon_j \in \bar{\bUpsilon}_j} | | \partial_{\bupsilon_j} \log f_{j,t} | |_{2}^{2} ] < \infty$ by Assumption \ref{AssumptionF4}, for each $j \in \{1,...,J\}$, there exists an $m_j > 0$ such that $\mathbb{E} [ \sup_{\bupsilon_j \in \bar{\bUpsilon}_j} \sup_{x_j \in \mathcal{X}_{j}} | | \partial_{\bvartheta_j x_j} \log f_{j,t} | |_{2}^{m_j} ] < \infty$ by assumption, and, for each $j \in \{1,...,J\}$, $\sup_{\bupsilon_j \in \bar{\bUpsilon}_{j}} || \nabla_{\bupsilon_{j}} \hat{X}_{j,t} - \nabla_{\bupsilon_{j}} X_{j,t} ||_{2}  \overset{e.a.s.}{\rightarrow} 0$ as $t \rightarrow \infty$ where $(\nabla_{\bupsilon_{j}} X_{j,t})_{t \in \mathbb{Z}}$ is stationary with $\mathbb{E} [ \log^{+}  \sup_{\bupsilon_{j} \in \bar{\bUpsilon}_{j}} | | \nabla_{\bupsilon_{j}} X_{j,t} | |_{2} ] < \infty$ by Lemma \ref{LemmaInvertibilityDX}. For the condition in Equation \eqref{eq:c.d.2}, for each $j \in \{1,...,J-1\}$, we have that
	\begin{equation}
		\begin{aligned}
			&\sup_{\btheta \in \bar{\bTheta}} \left| \frac{( \hat{f}_{j,t} \mathbf{p}_{\bullet j}^{\prime} \hat{\bpi}_{t-1 \mid t-1} ) ( \nabla_{\left[\bupsilon_a \right]_{b}} \hat{f}_{a,t} \mathbf{p}_{\bullet a}^{\prime} \hat{\bpi}_{t-1 \mid t-1} )}{( \hat{\mathbf{f}}_t^{\prime} \mathbf{P}^{\prime} \hat{\bpi}_{t-1 \mid t-1}  )^2} - \frac{( f_{j,t} \mathbf{p}_{\bullet j}^{\prime} \bpi_{t-1 \mid t-1} ) ( \nabla_{\left[\bupsilon_a \right]_{b}} f_{a,t} \mathbf{p}_{\bullet a}^{\prime} \bpi_{t-1 \mid t-1} )}{( \mathbf{f}_t^{\prime} \mathbf{P}^{\prime} \bpi_{t-1 \mid t-1}  )^2} \right| \\
			&\leq \frac{1}{p_{-}} \sup_{\btheta \in \bar{\bTheta}} \left| \frac{\nabla_{\left[\bupsilon_a \right]_{b}} \hat{f}_{a,t} \mathbf{p}_{\bullet a}^{\prime} \hat{\bpi}_{t-1 \mid t-1}}{\hat{\mathbf{f}}_t^{\prime} \mathbf{P}^{\prime} \hat{\bpi}_{t-1 \mid t-1} } - \frac{ \nabla_{\left[\bupsilon_a \right]_{b}} f_{a,t} \mathbf{p}_{\bullet a}^{\prime} \bpi_{t-1 \mid t-1}}{\mathbf{f}_t^{\prime} \mathbf{P}^{\prime} \bpi_{t-1 \mid t-1}} \right| \\
			&+ \frac{1}{p_{-}} \sup_{\btheta \in \bar{\bTheta}} \left| \nabla_{\left[\bupsilon_a \right]_{b}} \log f_{a,t} \right| \sup_{\btheta \in \bar{\bTheta}} \left| \frac{\hat{f}_{j,t} \mathbf{p}_{\bullet j}^{\prime} \hat{\bpi}_{t-1 \mid t-1}}{\hat{\mathbf{f}}_t^{\prime} \mathbf{P}^{\prime} \hat{\bpi}_{t-1 \mid t-1}} - \frac{f_{j,t} \mathbf{p}_{\bullet j}^{\prime} \bpi_{t-1 \mid t-1}}{\mathbf{f}_t^{\prime} \mathbf{P}^{\prime} \bpi_{t-1 \mid t-1}} \right| \\
			&\leq C \sup_{\btheta \in \bar{\bTheta}} \left| \frac{\nabla_{\left[\bupsilon_a \right]_{b}} \hat{f}_{a,t} \mathbf{p}_{\bullet a}^{\prime} \hat{\bpi}_{t-1 \mid t-1}}{\hat{\mathbf{f}}_t^{\prime} \mathbf{P}^{\prime} \hat{\bpi}_{t-1 \mid t-1} } - \frac{ \nabla_{\left[\bupsilon_a \right]_{b}} f_{a,t} \mathbf{p}_{\bullet a}^{\prime} \bpi_{t-1 \mid t-1}}{\mathbf{f}_t^{\prime} \mathbf{P}^{\prime} \bpi_{t-1 \mid t-1}} \right| \\
			&+ C \sup_{\btheta \in \bar{\bTheta}} \left| \nabla_{\left[\bupsilon_a \right]_{b}} \log f_{a,t} \right| \sum_{\alpha,\beta=1}^{J} \sup_{\btheta \in \bar{\bTheta}} \left| \frac{\hat{f}_{\beta,t} \hat{\pi}_{\alpha,t-1 \mid t-1} }{\hat{\mathbf{f}}_t^{\prime} \mathbf{P}^{\prime} \hat{\bpi}_{t-1 \mid t-1}} - \frac{f_{\beta,t} \pi_{\alpha,t-1 \mid t-1} }{\mathbf{f}_t^{\prime} \mathbf{P}^{\prime} \bpi_{t-1 \mid t-1}} \right|
		\end{aligned} \label{eq:SMUK100}
	\end{equation}
	for all $a \in \{1,...,J\}$ and $b \in \{1,...,d_a\}$ where the first inequality follows since $\hat{x}\hat{y} - xy = (\hat{x}-x)\hat{y} + x(\hat{y}-y)$ and $p_{-} \leq p_{ij} \leq 1$ for all $i,j \in \{1,...,J\}$, and the second inequality follows since $p_{ij} \leq 1$ for all $i,j \in \{1,...,J\}$. Note that the first term on the right hand side of Equation \eqref{eq:SMUK100} is equal to a constant times the term on the left hand side of Equation \eqref{eq:SMUK2}, and that the second term on the right hand side of Equation \eqref{eq:SMUK100} is equal to a constant times $\sup_{\btheta \in \bar{\bTheta}} | \nabla_{[\bupsilon_a ]_{b}} \log f_{a,t} |$ times the sum of the terms on the left hand side of Equation \eqref{eq:boo}. The condition in Equation \eqref{eq:c.d.2} thus follows by using the same arguments as the ones used to show the conditions in Equations \eqref{eq:c.a.2} and \eqref{eq:c.c.2}. For Condition (iii), by using similar arguments to the ones used to arrive at Equation \eqref{eq:VIENNA}, we have that
	\begin{equation*}
		\Lambda (\hat{\bphi}_t^{\bpi} - \bphi_t^{\bpi};\btheta) \leq C \Lambda (\hat{\bphi}_{t} - \bphi_{t};\btheta).
	\end{equation*}
	Condition (iii) thus follows by using the same arguments as the ones used to show that $\sup_{\btheta \in \bTheta} \Lambda (\hat{\bphi}_t - \bphi_t;\btheta) \overset{e.a.s.}{\rightarrow} 0$ as $t \rightarrow \infty$ (see the proof of Lemma \ref{LemmaInvertibilityFilter}).
\end{proof}
\begin{slemma} \label{LemmaInvertibilityDDFilter}
	Assume that Assumptions \ref{AssumptionY}-\ref{AssumptionPhi1}, \ref{AssumptionF3}-\ref{AssumptionF4}, and the conditions in Lemmas \ref{LemmaInvertibilityX}-\ref{LemmaInvertibilityDPrediction} hold. Moreover, assume that for each $j \in \{1,...,J\}$, there exists an $m_j > 0$ such that
	\begin{equation*}
		\mathbb{E} \left[ \sup_{\bupsilon_j \in \bar{\bUpsilon}_j} \sup_{x_j \in \mathcal{X}_{j}} \left| \left| \partial_{\bvartheta_j \bvartheta_j x_j} \log f_{j} (Y_t;x_j,\bupsilon_j) \right| \right|_{2,2}^{m_j} \right] < \infty.
	\end{equation*}
	Then, $( \nabla_{\btheta \btheta} \bpi_{t \mid t} (\btheta) )_{t \in \mathbb{Z}}$ is stationary and ergodic for all $\btheta \in \bar{\bTheta}$ with $\mathbb{E} [ \sup_{\btheta \in \bar{\bTheta}} || \nabla_{\btheta \btheta} \bpi_{t \mid t} (\btheta) ||_{2} ] < \infty$ and
	\begin{equation*}
		\sup_{\btheta \in \bar{\bTheta}} | | \nabla_{\btheta \btheta} \hat{\bpi}_{t \mid t} (\btheta) - \nabla_{\btheta \btheta} \bpi_{t \mid t} (\btheta) | |_{2} \overset{e.a.s.}{\rightarrow} 0 \quad \text{as} \quad t \rightarrow \infty
	\end{equation*}
	for any initialisation $\nabla_{\btheta \btheta} \hat{\bpi}_{0 \mid 0} (\btheta) \in \mathbb{R}^{Jd^2}$.
\end{slemma}
\begin{proof}
	As in the proof of Lemma \ref{LemmaInvertibilityDFilter}, let $\bpi_{t \mid t}^{-} (\btheta) := (\pi_{1,t \mid t} (\btheta),...,\pi_{J-1,t \mid t} (\btheta))^{\prime}$ and $\hat{\bpi}_{t \mid t}^{-} (\btheta) := (\hat{\pi}_{1,t \mid t} (\btheta),...,\hat{\pi}_{J-1,t \mid t} (\btheta))^{\prime}$ such that $\bpi_{t \mid t} (\btheta)= ((\bpi_{t \mid t}^{-} (\btheta))^{\prime},1-(\mathbf{1}^{-})^{\prime}\bpi_{t \mid t}^{-}(\btheta))^{\prime}$ and $\hat{\bpi}_{t \mid t} (\btheta)= ((\hat{\bpi}_{t \mid t}^{-} (\btheta))^{\prime},1-(\mathbf{1}^{-})^{\prime}\hat{\bpi}_{t \mid t}^{-}(\btheta))^{\prime}$ where $\mathbf{1}^{-} := (1,...,1)^{\prime}$ is a $(J-1) \times 1$ vector of ones. The fact that $( \nabla_{\btheta \btheta} \bpi_{t \mid t} (\btheta) )_{t \in \mathbb{Z}}$ is stationary and ergodic for all $\btheta \in \bar{\bTheta}$ with $\mathbb{E} [ \sup_{\btheta \in \bar{\bTheta}} || \nabla_{\btheta \btheta} \bpi_{t \mid t} (\btheta) ||_{2} ] < \infty$ and $\sup_{\btheta \in \bar{\bTheta}} | | \nabla_{\btheta \btheta} \hat{\bpi}_{t \mid t} (\btheta) - \nabla_{\btheta \btheta} \bpi_{t \mid t} (\btheta) | |_{2} \overset{e.a.s.}{\rightarrow} 0$ as $t \rightarrow \infty$ for any initialisation $\nabla_{\btheta \btheta} \hat{\bpi}_{0 \mid 0} (\btheta) \in \mathbb{R}^{Jd^2}$ therefore follows if $( \nabla_{\btheta \btheta} \bpi_{t \mid t}^{-} (\btheta) )_{t \in \mathbb{Z}}$ is stationary and ergodic for all $\btheta \in \bar{\bTheta}$ with $\mathbb{E} [ \sup_{\btheta \in \bar{\bTheta}} || \nabla_{\btheta \btheta} \bpi_{t \mid t}^{-} (\btheta) ||_{2} ] < \infty$ and $\sup_{\btheta \in \bar{\bTheta}} | | \nabla_{\btheta \btheta} \hat{\bpi}_{t \mid t}^{-} (\btheta) - \nabla_{\btheta \btheta} \bpi_{t \mid t}^{-} (\btheta) | |_{2} \overset{e.a.s.}{\rightarrow} 0$ as $t \rightarrow \infty$ for any initialisation $\nabla_{\btheta \btheta} \hat{\bpi}_{0 \mid 0}^{-} (\btheta) \in \mathbb{R}^{(J-1)d^2}$, which we now prove. 
	
	The proof consists of two parts like the one of Lemma \ref{LemmaInvertibilityDFilter}. The first part of the proof is about the stochastic process $(\nabla_{\btheta \btheta} \bpi_{t \mid t}^{-} (\btheta))_{t \in \mathbb{Z}}$, and the second part of the proof is about the stochastic process $(\nabla_{\btheta \btheta} \hat{\bpi}_{t \mid t}^{-} (\btheta))_{t \in \mathbb{N}_0}$. $(\nabla_{\btheta \btheta} \bpi_{t \mid t}^{-} (\btheta))_{t \in \mathbb{Z}}$ is a stochastic process taking values in $\mathbb{R}^{(J-1)d^2}$ given by
	\begin{equation}
		\nabla_{\btheta \btheta} \bpi_{t \mid t}^{-} (\btheta) = \bphi_{t}^{\bpi \bpi} (\nabla_{\btheta \btheta} \bpi_{t-1 \mid t-1}^{-} (\btheta) ; \btheta) \label{eq:lemma s.12 1}
	\end{equation}
	with
	\begin{align}
		&\left[ \bphi_{t}^{\bpi \bpi} (\nabla_{\btheta \btheta} \bpi_{t-1 \mid t-1}^{-} (\btheta) ; \btheta) \right]_{(n-1)d^2+(m_1-1)d+m_2} \nonumber \\
		&= \sum_{j=1}^{J-1} \left( \partial_{\left[ \bpi \right]_j} \left[ \bphi_{t} (\bpi_{t-1 \mid t-1} (\btheta) ; \btheta) \right]_{n} - \partial_{\left[ \bpi \right]_J} \left[ \bphi_{t} (\bpi_{t-1 \mid t-1} (\btheta) ; \btheta) \right]_{n} \right) \label{eq:lemma s.12 2} \\
		& \quad \quad \quad \cdot \left[ \nabla_{\btheta \btheta} \bpi_{t-1 \mid t-1}^{-} (\btheta) \right]_{(j-1)d^2+(m_1-1)d+m_2} \label{eq:lemma s.12 3} \\
		& + \sum_{j=1}^{J-1} \left( \partial_{\left[ \bpi \right]_j \left[ \btheta \right]_{m_2}} \left[ \bphi_{t} (\bpi_{t-1 \mid t-1} (\btheta) ; \btheta) \right]_{n} - \partial_{\left[ \bpi \right]_J \left[ \btheta \right]_{m_2}} \left[ \bphi_{t} (\bpi_{t-1 \mid t-1} (\btheta) ; \btheta) \right]_{n} \right) \label{eq:lemma s.12 4} \\
		& \quad \quad \quad \cdot \nabla_{\left[ \btheta \right]_{m_1}} \left[ \bpi_{t-1 \mid t-1}^{-} (\btheta) \right]_{j} \\
		& + \sum_{j=1}^{J-1} \sum_{i=1}^{J-1} \left( \left( \partial_{\left[ \bpi \right]_j \left[ \bpi \right]_i} \left[ \bphi_{t} (\bpi_{t-1 \mid t-1} (\btheta) ; \btheta) \right]_{n} - \partial_{\left[ \bpi \right]_j \left[ \bpi \right]_J} \left[ \bphi_{t} (\bpi_{t-1 \mid t-1} (\btheta) ; \btheta) \right]_{n} \right) \right. \\
		& \quad \quad \quad \quad \quad \left. - \left( \partial_{\left[ \bpi \right]_J \left[ \bpi \right]_i} \left[ \bphi_{t} (\bpi_{t-1 \mid t-1} (\btheta) ; \btheta) \right]_{n} - \partial_{\left[ \bpi \right]_J \left[ \bpi \right]_J} \left[ \bphi_{t} (\bpi_{t-1 \mid t-1} (\btheta) ; \btheta) \right]_{n} \right) \right) \\
		& \quad \quad \quad \quad \quad \cdot \nabla_{\left[ \btheta \right]_{m_1}} \left[ \bpi_{t-1 \mid t-1}^{-} (\btheta) \right]_{j} \nabla_{\left[ \btheta \right]_{m_2}} \left[ \bpi_{t-1 \mid t-1}^{-} (\btheta) \right]_{i} \\
		& + \sum_{i=1}^{J-1} \left( \partial_{\left[ \btheta \right]_{m_1} \left[ \bpi \right]_i} \left[ \bphi_{t} (\bpi_{t-1 \mid t-1} (\btheta) ; \btheta) \right]_{n} - \partial_{\left[ \btheta \right]_{m_1} \left[ \bpi \right]_J} \left[ \bphi_{t} (\bpi_{t-1 \mid t-1} (\btheta) ; \btheta) \right]_{n}  \right) \\ 
		& \quad \quad \quad \cdot \nabla_{\left[ \btheta \right]_{m_2}} \left[ \bpi_{t-1 \mid t-1}^{-} (\btheta) \right]_{i} \\
		&+ \partial_{\left[ \btheta \right]_{m_1} \left[ \btheta \right]_{m_2}} \left[ \bphi_{t} (\bpi_{t-1 \mid t-1} (\btheta) ; \btheta) \right]_{n}, \ (n,m_1,m_2) \in \{1,...,J-1\} \times \{1,...,d\} \times \{1,...,d\}. \label{eq:lemma s.12 5}
	\end{align}
	The explicit expression for Equation \eqref{eq:lemma s.12 1} is given in Section \ref{sec:s2b}, where the term in Equations \eqref{eq:lemma s.12 2}-\eqref{eq:lemma s.12 3} is given by the terms in Equations \eqref{eq:5}, \eqref{eq:16}, \eqref{eq:6.6}, \eqref{eq:18.18}, \eqref{eq:6.6.6}, and \eqref{eq:18.18.18}, and the terms in Equations \eqref{eq:lemma s.12 4}-\eqref{eq:lemma s.12 5} are given by the remaining terms in Equations \eqref{eq:1}-\eqref{eq:20.20.20}. Let
	\begin{equation*}
		\Lambda ((\bphi_{t}^{\bpi \bpi})^{(r)};\btheta) = \sup_{\underset{\bx \neq \by}{\bx,\by \in \mathbb{R}^{(J-1)d^2}}} \frac{|| (\bphi_{t}^{\bpi \bpi})^{(r)} (\bx;\btheta) - (\bphi_{t}^{\bpi \bpi})^{(r)} (\by;\btheta) ||_{2}}{|| \bx - \by ||_{2}},
	\end{equation*}
	where
	\begin{equation*} 
		(\bphi_{t}^{\bpi \bpi})^{(r)} (\, \cdot \,;\btheta) = \bphi_{t}^{\bpi \bpi} \circ \cdots \circ \bphi_{t-r+1}^{\bpi \bpi} (\, \cdot \,;\btheta), \quad r \in \mathbb{N}.
	\end{equation*}
	The fact that $( \nabla_{\btheta \btheta} \bpi_{t \mid t}^{-} (\btheta) )_{t \in \mathbb{Z}}$ is stationary and ergodic for all $\btheta \in \bar{\bTheta}$ follows from Theorem 2.10 in \cite{StraumannMikosch2006} if
	\begin{enumerate} [(i)]
		\item $\mathbb{E} [\log^{+} \sup_{\btheta \in \bar{\bTheta}} ||\bphi_{t}^{\bpi \bpi} (\mathbf{0};\btheta)||_{2}] < \infty$,
		\item $\mathbb{E} [\log^{+} \sup_{\btheta \in \bar{\bTheta}} \Lambda (\bphi_{t}^{\bpi \bpi};\btheta)] < \infty$, and
		\item there exists an $r \in \mathbb{N}$ such that $- \infty \leq \mathbb{E} [\log \sup_{\btheta \in \bar{\bTheta}} \Lambda ((\bphi_{t}^{\bpi \bpi})^{(r)};\btheta)] < 0$.
	\end{enumerate}
	Let
	\begin{align*}
		A_t^{(j,i,n,m_1,m_2)} (\btheta) &:= \left( \left( \partial_{\left[ \bpi \right]_j \left[ \bpi \right]_i} \left[ \bphi_{t} (\bpi_{t-1 \mid t-1} (\btheta) ; \btheta) \right]_{n} - \partial_{\left[ \bpi \right]_j \left[ \bpi \right]_J} \left[ \bphi_{t} (\bpi_{t-1 \mid t-1} (\btheta) ; \btheta) \right]_{n} \right) \right. \\
		& \quad \quad \left. - \left( \partial_{\left[ \bpi \right]_J \left[ \bpi \right]_i} \left[ \bphi_{t} (\bpi_{t-1 \mid t-1} (\btheta) ; \btheta) \right]_{n} - \partial_{\left[ \bpi \right]_J \left[ \bpi \right]_J} \left[ \bphi_{t} (\bpi_{t-1 \mid t-1} (\btheta) ; \btheta) \right]_{n} \right) \right) \\
		& \quad \quad \cdot \nabla_{\left[ \btheta \right]_{m_1}} \left[ \bpi_{t-1 \mid t-1}^{-} (\btheta) \right]_{j} \nabla_{\left[ \btheta \right]_{m_2}} \left[ \bpi_{t-1 \mid t-1}^{-} (\btheta) \right]_{i}, \\
		B_t^{(j,n,m_1,m_2)} (\btheta) &:= \left( \partial_{\left[ \bpi \right]_j \left[ \btheta \right]_{m_1}} \left[ \bphi_{t} (\bpi_{t-1 \mid t-1} (\btheta) ; \btheta) \right]_{n} - \partial_{\left[ \bpi \right]_J \left[ \btheta \right]_{m_1}} \left[ \bphi_{t} (\bpi_{t-1 \mid t-1} (\btheta) ; \btheta) \right]_{n} \right) \\
		& \quad \quad \cdot \nabla_{\left[ \btheta \right]_{m_2}} \left[ \bpi_{t-1 \mid t-1}^{-} (\btheta) \right]_{j}, \\
		C_t^{(n,m_1,m_2)} (\btheta) &:= \partial_{\left[ \btheta \right]_{m_1} \left[ \btheta \right]_{m_2}} \left[ \bphi_{t} (\bpi_{t-1 \mid t-1} (\btheta) ; \btheta) \right]_{n}.
	\end{align*}
	Then, Condition (i) follows from Lemma 2.2 in \cite{StraumannMikosch2006} if, for each $(j,i) \in \{1,...,J-1\} \times \{1,...,J-1\}$ and $(n,m_1,m_2) \in \{1,...,J-1\} \times \{1,...,d\} \times \{1,...,d\}$,
	\begin{align}
		\mathbb{E} \left[ \sup_{\btheta \in \bar{\bTheta}} \left| A_t^{(j,i,n,m_1,m_2)} (\btheta) \right| \right] < \infty, \label{eq:a} \\
		\mathbb{E} \left[ \sup_{\btheta \in \bar{\bTheta}} \left| B_t^{(j,n,m_1,m_2)} (\btheta) \right| \right] < \infty, \\
		\mathbb{E} \left[ \sup_{\btheta \in \bar{\bTheta}} \left| C_t^{(n,m_1,m_2)} (\btheta) \right| \right] < \infty. \label{eq:b}
	\end{align}
	The proof of the conditions in Equations \eqref{eq:a}-\eqref{eq:b} follows the same lines as the proof of the condition in Equation \eqref{eq:c.i.1}. For instance, let $\pi_{j,t \mid t}$ denote $\pi_{j,t \mid t} (\btheta)$ and $f_{j,t}$ denote $f_j (Y_t;X_{j,t}(\bupsilon_j),\bupsilon_j)$, and let
	\small \begin{equation*}
		\mathbf{P}^{-}
		:=
		\begin{bmatrix}
			p_{11}-p_{J1} & \cdots & p_{1(J-1)}-p_{J(J-1)} \\
			\vdots & \ddots & \vdots \\
			p_{(J-1)1}-p_{J1} & \cdots & p_{(J-1)(J-1)}-p_{J(J-1)} \\
		\end{bmatrix},
	\end{equation*} \normalsize
	$\mathbf{p}_{\bullet j}^{-} := (p_{1j}-p_{Jj},...,p_{(J-1)j}-p_{Jj})^{\prime}$, and $\mathbf{f}_t^{-} := (f_{1,t}-f_{J,t},...,f_{J-1,t}-f_{J,t})^{\prime}$. Then, the condition in Equation \eqref{eq:a} follows if, for each $j \in \{1,...,J-1\}$,
	\small \begin{align}
		\mathbb{E} \left[ \sup_{\btheta \in \bar{\bTheta}} \left| \frac{( f_{j,t} (\mathbf{p}_{\bullet j}^{-})^{\prime} \nabla_{p_{ab}} \bpi_{t-1 \mid t-1}^{-} ) ( (\mathbf{f}_t^{-})^{\prime} (\mathbf{P}^{-})^{\prime} \nabla_{p_{\alpha \beta}} \bpi^{-}_{t-1 \mid t-1} ) }{(\mathbf{f}_t^{\prime} \mathbf{P}^{\prime} \bpi_{t-1 \mid t-1})^2} \right| \right] < \infty, \label{eq:a1} \\
		\mathbb{E} \left[ \sup_{\btheta \in \bar{\bTheta}} \left| \frac{( f_{j,t} \mathbf{p}_{\bullet j}^{\prime} \bpi_{t-1 \mid t-1} ) ( (\mathbf{f}_t^{-})^{\prime} (\mathbf{P}^{-})^{\prime} \nabla_{p_{ab}} \bpi^{-}_{t-1 \mid t-1} ) ( (\mathbf{f}_t^{-})^{\prime} (\mathbf{P}^{-})^{\prime} \nabla_{p_{\alpha \beta}} \bpi^{-}_{t-1 \mid t-1} )}{( \mathbf{f}_t^{\prime} \mathbf{P}^{\prime} \bpi_{t-1 \mid t-1} )^3} \right| \right] < \infty, \label{eq:a2}
	\end{align} \normalsize
	for all $a \in \{1,...,J\}$, $b \in \{1,...,J-1\}$, $\alpha \in \{1,...,J\}$, and $\beta \in \{1,...,J-1\}$,
	\small \begin{align}
		\mathbb{E} \left[ \sup_{\btheta \in \bar{\bTheta}} \left| \frac{(f_{j,t} (\mathbf{p}_{\bullet j}^{-})^{\prime} \nabla_{p_{ab}} \bpi_{t-1 \mid t-1}^{-}) ( (\mathbf{f}_t^{-})^{\prime} (\mathbf{P}^{-})^{\prime} \nabla_{[\bupsilon_\alpha]_{\beta}} \bpi^{-}_{t-1 \mid t-1} ) }{( \mathbf{f}_t^{\prime} \mathbf{P}^{\prime} \bpi_{t-1 \mid t-1} )^2} \right| \right] < \infty, \label{eq:a3} \\
		\mathbb{E} \left[ \sup_{\btheta \in \bar{\bTheta}} \left| \frac{(f_{j,t} (\mathbf{p}_{\bullet j}^{-})^{\prime} \nabla_{[\bupsilon_{\alpha}]_{\beta}} \bpi_{t-1 \mid t-1}^{-}) ( (\mathbf{f}_t^{-})^{\prime} (\mathbf{P}^{-})^{\prime} \nabla_{p_{ab}} \bpi^{-}_{t-1 \mid t-1} ) }{( \mathbf{f}_t^{\prime} \mathbf{P}^{\prime} \bpi_{t-1 \mid t-1} )^2} \right| \right] < \infty, \label{eq:a4} \\
		\mathbb{E} \left[ \sup_{\btheta \in \bar{\bTheta}} \left| \frac{( f_{j,t} \mathbf{p}_{\bullet j}^{\prime} \bpi_{t-1 \mid t-1} ) ( (\mathbf{f}_t^{-})^{\prime} (\mathbf{P}^{-})^{\prime} \nabla_{p_{ab}} \bpi^{-}_{t-1 \mid t-1} ) ( (\mathbf{f}_t^{-})^{\prime} (\mathbf{P}^{-})^{\prime} \nabla_{[\bupsilon_\alpha]_{\beta}} \bpi^{-}_{t-1 \mid t-1} ) }{( \mathbf{f}_t^{\prime} \mathbf{P}^{\prime} \bpi_{t-1 \mid t-1} )^3} \right| \right] < \infty, \label{eq:a5}
	\end{align} \normalsize
	for all $a \in \{1,...,J\}$, $b \in \{1,...,J-1\}$, $\alpha \in \{1,...,J\}$, and $\beta \in \{1,...,d_{\alpha}\}$, and
	\small \begin{align}
		\mathbb{E} \left[ \sup_{\btheta \in \bar{\bTheta}} \left| \frac{( f_{j,t} (\mathbf{p}_{\bullet j}^{-})^{\prime} \nabla_{[\bupsilon_a]_{b}} \bpi_{t-1 \mid t-1}^{-} ) ( (\mathbf{f}_t^{-})^{\prime} (\mathbf{P}^{-})^{\prime} \nabla_{[\bupsilon_\alpha]_{\beta}} \bpi^{-}_{t-1 \mid t-1} ) }{( \mathbf{f}_t^{\prime} \mathbf{P}^{\prime} \bpi_{t-1 \mid t-1} )^2} \right| \right] < \infty, \label{eq:a6} \\
		\mathbb{E} \left[ \sup_{\btheta \in \bar{\bTheta}} \left| \frac{( f_{j,t} \mathbf{p}_{\bullet j}^{\prime} \bpi_{t-1 \mid t-1} ) ( (\mathbf{f}_t^{-})^{\prime} (\mathbf{P}^{-})^{\prime} \nabla_{[\bupsilon_a]_{b}} \bpi^{-}_{t-1 \mid t-1} ) ( (\mathbf{f}_t^{-})^{\prime} (\mathbf{P}^{-})^{\prime} \nabla_{[\bupsilon_\alpha]_{\beta}} \bpi^{-}_{t-1 \mid t-1} )}{( \mathbf{f}_t^{\prime} \mathbf{P}^{\prime} \bpi_{t-1 \mid t-1}  )^3} \right| \right] < \infty, \label{eq:a7}
	\end{align} \normalsize
	for all $a \in \{1,...,J\}$, $b \in \{1,...,d_a\}$, $\alpha \in \{1,...,J\}$, and $\beta \in \{1,...,d_{\alpha}\}$, where the terms in Equations \eqref{eq:a1}-\eqref{eq:a2} concern the terms in Equations \eqref{eq:7}, \eqref{eq:14}, and \eqref{eq:18}, the ones in Equations \eqref{eq:a3}-\eqref{eq:a5} concern the ones in Equations \eqref{eq:8.8}, \eqref{eq:16.16}, and \eqref{eq:20.20}, and the terms in Equations \eqref{eq:a6}-\eqref{eq:a7} concern the terms in Equations \eqref{eq:8.8.8}, \eqref{eq:16.16.16}, and \eqref{eq:20.20.20}. The conditions in Equations \eqref{eq:a1}-\eqref{eq:a7} can then be shown by using similar arguments to the ones used to show the conditions in Equations \eqref{eq:c.a.1}-\eqref{eq:c.d.1}. The details of the proof of the conditions in Equations \eqref{eq:a}-\eqref{eq:b} are therefore omitted. For Conditions (ii) and (iii), by using similar arguments to the ones used to arrive at Equation \eqref{eq:VIENNA}, we have that
	\begin{equation}
		\Lambda ((\bphi_{t}^{\bpi \bpi})^{(r)};\btheta) \leq C \Lambda (\bphi_{t}^{(r)};\btheta) \label{eq:AARHUS}
	\end{equation}
	for all $r \in \mathbb{N}$. Conditions (ii) and (iii) thus follow by using the same arguments as the ones used to show that $\mathbb{E} [\log^{+} \sup_{\btheta \in \bTheta} \Lambda (\bphi_t;\btheta)] < \infty$ and there exists an $r \in \mathbb{N}$ such that $-\infty \leq \mathbb{E} [\log \sup_{\btheta \in \bTheta} \Lambda (\bphi_t^{(r)};\btheta)] < 0$ (see the proof of Lemma \ref{LemmaInvertibilityFilter}).
	
	Under Conditions (i)-(iii), we have that
	\begin{equation*}
		\nabla_{\btheta \btheta} \bpi_{t \mid t}^{-} (\btheta) = \sum_{k=0}^{\infty} ( (\bphi_{t}^{\bpi \bpi})^{(k)} ( \bphi_{t-k}^{\bpi \bpi} (\mathbf{0};\btheta) ; \btheta) - (\bphi_{t}^{\bpi \bpi})^{(k)} ( \mathbf{0} ; \btheta) ),
	\end{equation*}
	where $(\bphi_{t}^{\bpi \bpi})^{(0)} ( \bx ; \btheta) = \bx$ by convention. Thus, 
	\begin{equation}
		|| \nabla_{\btheta \btheta} \bpi_{t \mid t}^{-} (\btheta) ||_2 \leq \sum_{k=0}^{\infty} || (\bphi_{t}^{\bpi \bpi})^{(k)} ( \bphi_{t-k}^{\bpi \bpi} (\mathbf{0};\btheta) ; \btheta) - (\bphi_{t}^{\bpi \bpi})^{(k)} ( \mathbf{0} ; \btheta) ||_2. \label{eq:momentmoment.}
	\end{equation}
	Note that, by Equations \eqref{eq:cc8} and \eqref{eq:AARHUS}, there exists an $\alpha \in (0,1)$ such that
	\begin{equation}
		\sup_{\underset{\bx \neq \by}{\bx,\by \in \mathbb{R}^{(J-1)d^2}}} \frac{|| (\bphi_{t}^{\bpi \bpi})^{(k)} (\bx;\btheta) - (\bphi_{t}^{\bpi \bpi})^{(k)} (\by;\btheta) ||_{2}}{|| \bx - \by ||_{2}} \leq C \alpha^{k} \label{eq:anothermoment.}
	\end{equation}
	for all $k \in \mathbb{N}_{0}$. Thus, continuing from Equation \eqref{eq:momentmoment.} and using Equation \eqref{eq:anothermoment.},
	\begin{equation*}
		|| \nabla_{\btheta \btheta} \bpi_{t \mid t}^{-} (\btheta) ||_2 \leq C \sum_{k=0}^{\infty} \alpha^k || \bphi_{t-k}^{\bpi \bpi} (\mathbf{0};\btheta) ||_2.
	\end{equation*}
	Hence, we have that
	\begin{equation*}
		\mathbb{E} \left[\sup_{\btheta \in \bar{\bTheta}} || \nabla_{\btheta \btheta} \bpi_{t \mid t}^{-} (\btheta) ||_{2} \right] \leq C \mathbb{E} \left[\sup_{\btheta \in \bar{\bTheta}} ||\bphi_{t}^{\bpi \bpi} (\mathbf{0};\btheta)||_2 \right].
	\end{equation*}
	The fact that $\mathbb{E} [\sup_{\btheta \in \bar{\bTheta}} || \nabla_{\btheta \btheta} \bpi_{t \mid t}^{-} (\btheta) ||_{2}] < \infty$ thus follows since $\mathbb{E} [\sup_{\btheta \in \bar{\bTheta}} ||\bphi_{t}^{\bpi \bpi} (\mathbf{0};\btheta)||_2] < \infty$, which follows since, for each $(j,i) \in \{1,...,J-1\} \times \{1,...,J-1\}$ and $(n,m_1,m_2) \in \{1,...,J-1\} \times \{1,...,d\} \times \{1,...,d\}$, $\mathbb{E} [ \sup_{\btheta \in \bar{\bTheta}} | A_t^{(j,i,n,m_1,m_2)} (\btheta) | ] < \infty$, $\mathbb{E} [ \sup_{\btheta \in \bar{\bTheta}} | B_t^{(j,n,m_1,m_2)} (\btheta) | ] < \infty$, and $\mathbb{E} [ \sup_{\btheta \in \bar{\bTheta}} | C_t^{(n,m_1,m_2)} (\btheta) | ] < \infty$.
	
	Moving on to the second part of the proof, $(\nabla_{\btheta \btheta} \hat{\bpi}_{t \mid t}^{-} (\btheta))_{t \in \mathbb{N}_0}$ is a stochastic process taking values in $\mathbb{R}^{(J-1)d^2}$ given by
	\begin{equation*}
		\nabla_{\btheta \btheta} \hat{\bpi}_{t \mid t}^{-} (\btheta) = \hat{\bphi}_{t}^{\bpi \bpi} (\nabla_{\btheta \btheta} \hat{\bpi}_{t-1 \mid t-1}^{-} (\btheta) ; \btheta)
	\end{equation*}
	with
	\begin{align*}
		&\left[ \hat{\bphi}_{t}^{\bpi \bpi} (\nabla_{\btheta \btheta} \hat{\bpi}_{t-1 \mid t-1}^{-} (\btheta) ; \btheta) \right]_{(n-1)d^2+(m_1-1)d+m_2} \\
		&= \sum_{j=1}^{J-1} \left( \partial_{\left[ \bpi \right]_j} \left[ \hat{\bphi}_{t} (\hat{\bpi}_{t-1 \mid t-1} (\btheta) ; \btheta) \right]_{n} - \partial_{\left[ \bpi \right]_J} \left[ \hat{\bphi}_{t} (\hat{\bpi}_{t-1 \mid t-1} (\btheta) ; \btheta) \right]_{n} \right) \\
		& \quad \quad \quad \cdot \left[ \nabla_{\btheta \btheta} \hat{\bpi}_{t-1 \mid t-1}^{-} (\btheta) \right]_{(j-1)d^2+(m_1-1)d+m_2} \\
		& + \sum_{j=1}^{J-1} \left( \partial_{\left[ \bpi \right]_j \left[ \btheta \right]_{m_2}} \left[ \hat{\bphi}_{t} (\hat{\bpi}_{t-1 \mid t-1} (\btheta) ; \btheta) \right]_{n} - \partial_{\left[ \bpi \right]_J \left[ \btheta \right]_{m_2}} \left[ \hat{\bphi}_{t} (\hat{\bpi}_{t-1 \mid t-1} (\btheta) ; \btheta) \right]_{n} \right) \\
		& \quad \quad \quad \cdot \nabla_{\left[ \btheta \right]_{m_1}} \left[ \hat{\bpi}_{t-1 \mid t-1}^{-} (\btheta) \right]_{j} \\
		& + \sum_{j=1}^{J-1} \sum_{i=1}^{J-1} \left( \left( \partial_{\left[ \bpi \right]_j \left[ \bpi \right]_i} \left[ \hat{\bphi}_{t} (\hat{\bpi}_{t-1 \mid t-1} (\btheta) ; \btheta) \right]_{n} - \partial_{\left[ \bpi \right]_j \left[ \bpi \right]_J} \left[ \hat{\bphi}_{t} (\hat{\bpi}_{t-1 \mid t-1} (\btheta) ; \btheta) \right]_{n} \right) \right. \\
		& \quad \quad \quad \quad \quad \left. - \left( \partial_{\left[ \bpi \right]_J \left[ \bpi \right]_i} \left[ \hat{\bphi}_{t} (\hat{\bpi}_{t-1 \mid t-1} (\btheta) ; \btheta) \right]_{n} - \partial_{\left[ \bpi \right]_J \left[ \bpi \right]_J} \left[ \hat{\bphi}_{t} (\hat{\bpi}_{t-1 \mid t-1} (\btheta) ; \btheta) \right]_{n} \right) \right) \\
		& \quad \quad \quad \quad \quad \cdot \nabla_{\left[ \btheta \right]_{m_1}} \left[ \hat{\bpi}_{t-1 \mid t-1}^{-} (\btheta) \right]_{j} \nabla_{\left[ \btheta \right]_{m_2}} \left[ \hat{\bpi}_{t-1 \mid t-1}^{-} (\btheta) \right]_{i} \\
		& + \sum_{i=1}^{J-1} \left( \partial_{\left[ \btheta \right]_{m_1} \left[ \bpi \right]_i} \left[ \hat{\bphi}_{t} (\hat{\bpi}_{t-1 \mid t-1} (\btheta) ; \btheta) \right]_{n} - \partial_{\left[ \btheta \right]_{m_1} \left[ \bpi \right]_J} \left[ \hat{\bphi}_{t} (\hat{\bpi}_{t-1 \mid t-1} (\btheta) ; \btheta) \right]_{n}  \right) \\
		& \quad \quad \quad \cdot \nabla_{\left[ \btheta \right]_{m_2}} \left[ \hat{\bpi}_{t-1 \mid t-1}^{-} (\btheta) \right]_{i} \\
		&+ \partial_{\left[ \btheta \right]_{m_1} \left[ \btheta \right]_{m_2}} \left[ \hat{\bphi}_{t} (\hat{\bpi}_{t-1 \mid t-1} (\btheta) ; \btheta) \right]_{n}, \ (n,m_1,m_2) \in \{1,...,J-1\} \times \{1,...,d\} \times \{1,...,d\},
	\end{align*}
	where $\nabla_{\btheta \btheta} \hat{\bpi}_{0 \mid 0}^{-} (\btheta)$ is given. The fact that $\sup_{\btheta \in \bar{\bTheta}} || \nabla_{\btheta \btheta} \hat{\bpi}_{t \mid t}^{-} (\btheta) - \nabla_{\btheta \btheta} \bpi_{t \mid t}^{-} (\btheta) ||_{2} \overset{e.a.s.}{\rightarrow} 0$ as $t \rightarrow \infty$ follows also from Theorem 2.10 in \cite{StraumannMikosch2006} if
	\begin{enumerate} [(i)]
		\item $\mathbb{E} [\log^{+} \sup_{\btheta \in \bar{\bTheta}} || \nabla_{\btheta \btheta} \bpi_{t \mid t}^{-} (\btheta) ||_{2} ] < \infty$,
		\item $\sup_{\btheta \in \bar{\bTheta}} || \hat{\bphi}_t^{\bpi \bpi} (\mathbf{0};\btheta) - \bphi_t^{\bpi \bpi} (\mathbf{0};\btheta) ||_{2} \overset{e.a.s.}{\rightarrow} 0$ as $t \rightarrow \infty$, and
		\item $\sup_{\btheta \in \bar{\bTheta}} \Lambda (\hat{\bphi}_t^{\bpi \bpi} - \bphi_t^{\bpi \bpi};\btheta) \overset{e.a.s.}{\rightarrow} 0$ as $t \rightarrow \infty$.
	\end{enumerate}
	Condition (i) follows from Lemma 2.2 in \cite{StraumannMikosch2006} since $\mathbb{E} [\sup_{\btheta \in \bar{\bTheta}} || \nabla_{\btheta \btheta} \bpi_{t \mid t}^{-} (\btheta) ||_{2}] < \infty$. Let
	\begin{align*}
		\hat{A}_t^{(j,i,n,m_1,m_2)} (\btheta) &:= \left( \left( \partial_{\left[ \bpi \right]_j \left[ \bpi \right]_i} \left[ \hat{\bphi}_{t} (\hat{\bpi}_{t-1 \mid t-1} (\btheta) ; \btheta) \right]_{n} - \partial_{\left[ \bpi \right]_j \left[ \bpi \right]_J} \left[ \hat{\bphi}_{t} (\hat{\bpi}_{t-1 \mid t-1} (\btheta) ; \btheta) \right]_{n} \right) \right. \\
		& \quad \quad \left. - \left( \partial_{\left[ \bpi \right]_J \left[ \bpi \right]_i} \left[ \hat{\bphi}_{t} (\hat{\bpi}_{t-1 \mid t-1} (\btheta) ; \btheta) \right]_{n} - \partial_{\left[ \bpi \right]_J \left[ \bpi \right]_J} \left[ \hat{\bphi}_{t} (\hat{\bpi}_{t-1 \mid t-1} (\btheta) ; \btheta) \right]_{n} \right) \right) \\
		& \quad \quad \cdot \nabla_{\left[ \btheta \right]_{m_1}} \left[ \hat{\bpi}_{t-1 \mid t-1}^{-} (\btheta) \right]_{j} \nabla_{\left[ \btheta \right]_{m_2}} \left[ \hat{\bpi}_{t-1 \mid t-1}^{-} (\btheta) \right]_{i}, \\
		\hat{B}_t^{(j,n,m_1,m_2)} (\btheta) &:= \left( \partial_{\left[ \bpi \right]_j \left[ \btheta \right]_{m_1}} \left[ \hat{\bphi}_{t} (\hat{\bpi}_{t-1 \mid t-1} (\btheta) ; \btheta) \right]_{n} - \partial_{\left[ \bpi \right]_J \left[ \btheta \right]_{m_1}} \left[ \hat{\bphi}_{t} (\hat{\bpi}_{t-1 \mid t-1} (\btheta) ; \btheta) \right]_{n} \right) \\
		& \quad \quad \cdot \nabla_{\left[ \btheta \right]_{m_2}} \left[ \hat{\bpi}_{t-1 \mid t-1}^{-} (\btheta) \right]_{j}, \\
		\hat{C}_t^{(n,m_1,m_2)} (\btheta) &:= \partial_{\left[ \btheta \right]_{m_1} \left[ \btheta \right]_{m_2}} \left[ \hat{\bphi}_{t} (\hat{\bpi}_{t-1 \mid t-1} (\btheta) ; \btheta) \right]_{n}.
	\end{align*}
	Then, Condition (ii) follows if, for each $(j,i) \in \{1,...,J-1\} \times \{1,...,J-1\}$ and $(n,m_1,m_2) \in \{1,...,J-1\} \times \{1,...,d\} \times \{1,...,d\}$,
	\begin{align}
		\sup_{\btheta \in \bar{\bTheta}} \left| \hat{A}_t^{(j,i,n,m_1,m_2)} (\btheta)-A_t^{(j,i,n,m_1,m_2)} (\btheta) \right| \overset{e.a.s.}{\rightarrow} 0 \quad \textup{as} \quad t \rightarrow \infty, \label{eq:alpha} \\
		\sup_{\btheta \in \bar{\bTheta}} \left| \hat{B}_t^{(j,n,m_1,m_2)} (\btheta) - B_t^{(j,n,m_1,m_2)} (\btheta) \right| \overset{e.a.s.}{\rightarrow} 0 \quad \textup{as} \quad t \rightarrow \infty, \\
		\sup_{\btheta \in \bar{\bTheta}} \left| \hat{C}_t^{(n,m_1,m_2)} (\btheta) - C_t^{(n,m_1,m_2)} (\btheta) \right| \overset{e.a.s.}{\rightarrow} 0 \quad \textup{as} \quad t \rightarrow \infty. \label{eq:beta}
	\end{align}
	The proof of the conditions in Equations \eqref{eq:alpha}-\eqref{eq:beta} follows the same lines as the proof of the condition in Equation \eqref{eq:c.ii.2}. For instance, let $\hat{\pi}_{j,t \mid t}$ denote $\hat{\pi}_{j,t \mid t} (\btheta)$ and $\hat{f}_{j,t}$ denote $f_j (Y_t;\hat{X}_{j,t}(\bupsilon_j),\bupsilon_j)$, and let $\hat{\mathbf{f}}_t^{-} := (\hat{f}_{1,t}-\hat{f}_{J,t},...,\hat{f}_{J-1,t}-\hat{f}_{J,t})^{\prime}$. Then, the condition in Equation \eqref{eq:alpha} follows if, for each $j \in \{1,...,J-1\}$,
	\small \begin{align}
		&\begin{aligned}
			&\sup_{\btheta \in \bar{\bTheta}} \left| \frac{( \hat{f}_{j,t} (\mathbf{p}_{\bullet j}^{-})^{\prime} \nabla_{p_{ab}} \hat{\bpi}_{t-1 \mid t-1}^{-} ) ( (\hat{\mathbf{f}}_t^{-})^{\prime} (\mathbf{P}^{-})^{\prime} \nabla_{p_{\alpha \beta}} \hat{\bpi}^{-}_{t-1 \mid t-1} ) }{(\hat{\mathbf{f}}_t^{\prime} \mathbf{P}^{\prime} \hat{\bpi}_{t-1 \mid t-1})^2} \right. \\
			&\left. \quad - \frac{( f_{j,t} (\mathbf{p}_{\bullet j}^{-})^{\prime} \nabla_{p_{ab}} \bpi_{t-1 \mid t-1}^{-} ) ( (\mathbf{f}_t^{-})^{\prime} (\mathbf{P}^{-})^{\prime} \nabla_{p_{\alpha \beta}} \bpi^{-}_{t-1 \mid t-1} ) }{(\mathbf{f}_t^{\prime} \mathbf{P}^{\prime} \bpi_{t-1 \mid t-1})^2} \right| \overset{e.a.s.}{\rightarrow} 0 \ \textup{as} \ t \rightarrow \infty,
		\end{aligned} \label{eq:b1} \\
		&\begin{aligned}
			&\sup_{\btheta \in \bar{\bTheta}} \left| \frac{( \hat{f}_{j,t} \mathbf{p}_{\bullet j}^{\prime} \hat{\bpi}_{t-1 \mid t-1} ) ( (\hat{\mathbf{f}}_t^{-})^{\prime} (\mathbf{P}^{-})^{\prime} \nabla_{p_{ab}} \hat{\bpi}^{-}_{t-1 \mid t-1} ) ( (\hat{\mathbf{f}}_t^{-})^{\prime} (\mathbf{P}^{-})^{\prime} \nabla_{p_{\alpha \beta}} \hat{\bpi}^{-}_{t-1 \mid t-1} )}{( \hat{\mathbf{f}}_t^{\prime} \mathbf{P}^{\prime} \hat{\bpi}_{t-1 \mid t-1} )^3} \right. \\
			&\left. \quad - \frac{( f_{j,t} \mathbf{p}_{\bullet j}^{\prime} \bpi_{t-1 \mid t-1} ) ( (\mathbf{f}_t^{-})^{\prime} (\mathbf{P}^{-})^{\prime} \nabla_{p_{ab}} \bpi^{-}_{t-1 \mid t-1} ) ( (\mathbf{f}_t^{-})^{\prime} (\mathbf{P}^{-})^{\prime} \nabla_{p_{\alpha \beta}} \bpi^{-}_{t-1 \mid t-1} )}{( \mathbf{f}_t^{\prime} \mathbf{P}^{\prime} \bpi_{t-1 \mid t-1} )^3} \right| \overset{e.a.s.}{\rightarrow} 0 \ \textup{as} \ t \rightarrow \infty,
		\end{aligned} \label{eq:b2}
	\end{align} \normalsize
	for all $a \in \{1,...,J\}$, $b \in \{1,...,J-1\}$, $\alpha \in \{1,...,J\}$, and $\beta \in \{1,...,J-1\}$,
	\small \begin{align}
		&\begin{aligned}
			&\sup_{\btheta \in \bar{\bTheta}} \left| \frac{(\hat{f}_{j,t} (\mathbf{p}_{\bullet j}^{-})^{\prime} \nabla_{p_{ab}} \hat{\bpi}_{t-1 \mid t-1}^{-}) ( (\hat{\mathbf{f}}_t^{-})^{\prime} (\mathbf{P}^{-})^{\prime} \nabla_{[\bupsilon_\alpha]_{\beta}} \hat{\bpi}^{-}_{t-1 \mid t-1} ) }{( \hat{\mathbf{f}}_t^{\prime} \mathbf{P}^{\prime} \hat{\bpi}_{t-1 \mid t-1} )^2} \right. \\
			&\left. \quad - \frac{(f_{j,t} (\mathbf{p}_{\bullet j}^{-})^{\prime} \nabla_{p_{ab}} \bpi_{t-1 \mid t-1}^{-}) ( (\mathbf{f}_t^{-})^{\prime} (\mathbf{P}^{-})^{\prime} \nabla_{[\bupsilon_\alpha]_{\beta}} \bpi^{-}_{t-1 \mid t-1} ) }{( \mathbf{f}_t^{\prime} \mathbf{P}^{\prime} \bpi_{t-1 \mid t-1} )^2} \right| \overset{e.a.s.}{\rightarrow} 0 \ \textup{as} \ t \rightarrow \infty,
		\end{aligned} \label{eq:b3} \\
		&\begin{aligned}
			&\sup_{\btheta \in \bar{\bTheta}} \left| \frac{(\hat{f}_{j,t} (\mathbf{p}_{\bullet j}^{-})^{\prime} \nabla_{[\bupsilon_{\alpha}]_{\beta}} \hat{\bpi}_{t-1 \mid t-1}^{-}) ( (\hat{\mathbf{f}}_t^{-})^{\prime} (\mathbf{P}^{-})^{\prime} \nabla_{p_{ab}} \hat{\bpi}^{-}_{t-1 \mid t-1} ) }{( \hat{\mathbf{f}}_t^{\prime} \mathbf{P}^{\prime} \hat{\bpi}_{t-1 \mid t-1} )^2} \right. \\
			&\left. \quad - \frac{(f_{j,t} (\mathbf{p}_{\bullet j}^{-})^{\prime} \nabla_{[\bupsilon_{\alpha}]_{\beta}} \bpi_{t-1 \mid t-1}^{-}) ( (\mathbf{f}_t^{-})^{\prime} (\mathbf{P}^{-})^{\prime} \nabla_{p_{ab}} \bpi^{-}_{t-1 \mid t-1} ) }{( \mathbf{f}_t^{\prime} \mathbf{P}^{\prime} \bpi_{t-1 \mid t-1} )^2} \right| \overset{e.a.s.}{\rightarrow} 0 \ \textup{as} \ t \rightarrow \infty,
		\end{aligned} \label{eq:b4} \\
		&\begin{aligned}
			&\sup_{\btheta \in \bar{\bTheta}} \left| \frac{( \hat{f}_{j,t} \mathbf{p}_{\bullet j}^{\prime} \hat{\bpi}_{t-1 \mid t-1} ) ( (\hat{\mathbf{f}}_t^{-})^{\prime} (\mathbf{P}^{-})^{\prime} \nabla_{p_{ab}} \hat{\bpi}^{-}_{t-1 \mid t-1} ) ( (\hat{\mathbf{f}}_t^{-})^{\prime} (\mathbf{P}^{-})^{\prime} \nabla_{[\bupsilon_\alpha]_{\beta}} \hat{\bpi}^{-}_{t-1 \mid t-1} ) }{( \hat{\mathbf{f}}_t^{\prime} \mathbf{P}^{\prime} \hat{\bpi}_{t-1 \mid t-1} )^3} \right. \\
			&\left. \quad - \frac{( f_{j,t} \mathbf{p}_{\bullet j}^{\prime} \bpi_{t-1 \mid t-1} ) ( (\mathbf{f}_t^{-})^{\prime} (\mathbf{P}^{-})^{\prime} \nabla_{p_{ab}} \bpi^{-}_{t-1 \mid t-1} ) ( (\mathbf{f}_t^{-})^{\prime} (\mathbf{P}^{-})^{\prime} \nabla_{[\bupsilon_\alpha]_{\beta}} \bpi^{-}_{t-1 \mid t-1} ) }{( \mathbf{f}_t^{\prime} \mathbf{P}^{\prime} \bpi_{t-1 \mid t-1} )^3} \right| \overset{e.a.s.}{\rightarrow} 0 \ \textup{as} \ t \rightarrow \infty,
		\end{aligned} \label{eq:b5}
	\end{align} \normalsize
	for all $a \in \{1,...,J\}$, $b \in \{1,...,J-1\}$, $\alpha \in \{1,...,J\}$, and $\beta \in \{1,...,d_{\alpha}\}$, and
	\small \begin{align}
		&\begin{aligned}
			&\sup_{\btheta \in \bar{\bTheta}} \left| \frac{( \hat{f}_{j,t} (\mathbf{p}_{\bullet j}^{-})^{\prime} \nabla_{[\bupsilon_a]_{b}} \hat{\bpi}_{t-1 \mid t-1}^{-} ) ( (\hat{\mathbf{f}}_t^{-})^{\prime} (\mathbf{P}^{-})^{\prime} \nabla_{[\bupsilon_\alpha]_{\beta}} \hat{\bpi}^{-}_{t-1 \mid t-1} ) }{( \hat{\mathbf{f}}_t^{\prime} \mathbf{P}^{\prime} \hat{\bpi}_{t-1 \mid t-1} )^2} \right. \\
			&\left. \quad - \frac{( f_{j,t} (\mathbf{p}_{\bullet j}^{-})^{\prime} \nabla_{[\bupsilon_a]_{b}} \bpi_{t-1 \mid t-1}^{-} ) ( (\mathbf{f}_t^{-})^{\prime} (\mathbf{P}^{-})^{\prime} \nabla_{[\bupsilon_\alpha]_{\beta}} \bpi^{-}_{t-1 \mid t-1} ) }{( \mathbf{f}_t^{\prime} \mathbf{P}^{\prime} \bpi_{t-1 \mid t-1} )^2} \right| \overset{e.a.s.}{\rightarrow} 0 \ \textup{as} \ t \rightarrow \infty,
		\end{aligned} \label{eq:b6} \\
		&\begin{aligned}
			&\sup_{\btheta \in \bar{\bTheta}} \left| \frac{( \hat{f}_{j,t} \mathbf{p}_{\bullet j}^{\prime} \hat{\bpi}_{t-1 \mid t-1} ) ( (\hat{\mathbf{f}}_t^{-})^{\prime} (\mathbf{P}^{-})^{\prime} \nabla_{[\bupsilon_a]_{b}} \hat{\bpi}^{-}_{t-1 \mid t-1} ) ( (\hat{\mathbf{f}}_t^{-})^{\prime} (\mathbf{P}^{-})^{\prime} \nabla_{[\bupsilon_\alpha]_{\beta}} \hat{\bpi}^{-}_{t-1 \mid t-1} )}{( \hat{\mathbf{f}}_t^{\prime} \mathbf{P}^{\prime} \hat{\bpi}_{t-1 \mid t-1}  )^3} \right. \\
			&\left. \quad - \frac{( f_{j,t} \mathbf{p}_{\bullet j}^{\prime} \bpi_{t-1 \mid t-1} ) ( (\mathbf{f}_t^{-})^{\prime} (\mathbf{P}^{-})^{\prime} \nabla_{[\bupsilon_a]_{b}} \bpi^{-}_{t-1 \mid t-1} ) ( (\mathbf{f}_t^{-})^{\prime} (\mathbf{P}^{-})^{\prime} \nabla_{[\bupsilon_\alpha]_{\beta}} \bpi^{-}_{t-1 \mid t-1} )}{( \mathbf{f}_t^{\prime} \mathbf{P}^{\prime} \bpi_{t-1 \mid t-1}  )^3} \right| \overset{e.a.s.}{\rightarrow} 0 \ \textup{as} \ t \rightarrow \infty,
		\end{aligned} \label{eq:b7}
	\end{align} \normalsize
	for all $a \in \{1,...,J\}$, $b \in \{1,...,d_a\}$, $\alpha \in \{1,...,J\}$, and $\beta \in \{1,...,d_{\alpha}\}$, where the terms in Equations \eqref{eq:b1}-\eqref{eq:b2} concern the terms in Equations \eqref{eq:7}, \eqref{eq:14}, and \eqref{eq:18}, the ones in Equations \eqref{eq:b3}-\eqref{eq:b5} concern the ones in Equations \eqref{eq:8.8}, \eqref{eq:16.16}, and \eqref{eq:20.20}, and the terms in Equations \eqref{eq:b6}-\eqref{eq:b7} concern the terms in Equations \eqref{eq:8.8.8}, \eqref{eq:16.16.16}, and \eqref{eq:20.20.20}. The conditions in Equations \eqref{eq:b1}-\eqref{eq:b7} can then be shown by using similar arguments to the ones used to show the conditions in Equations \eqref{eq:c.a.2}-\eqref{eq:c.d.2}. The details of the proof of the conditions in Equations \eqref{eq:alpha}-\eqref{eq:beta} are therefore also omitted. For Condition (iii), by using similar arguments to the ones used to arrive at Equation \eqref{eq:VIENNA} again, we have that
	\begin{equation*}
		\Lambda (\hat{\bphi}_t^{\bpi \bpi} - \bphi_t^{\bpi \bpi};\btheta) \leq C \Lambda (\hat{\bphi}_{t} - \bphi_{t};\btheta).
	\end{equation*}
	Condition (iii) thus follows by using the same arguments as the ones used to show that $\sup_{\btheta \in \bTheta} \Lambda (\hat{\bphi}_t - \bphi_t;\btheta) \overset{e.a.s.}{\rightarrow} 0$ as $t \rightarrow \infty$ (see the proof of Lemma \ref{LemmaInvertibilityFilter} again).
\end{proof}
\begin{slemma} \label{LemmaAsymptoticNormalityA}
	Under the assumptions in Theorem \ref{TheoremAsymptoticNormality},
	\begin{equation*}
		\sup_{\btheta \in \bar{\bTheta}} || \sqrt{T} \nabla_{\btheta} \hat{L}_T (\btheta) - \sqrt{T} \nabla_{\btheta} L_T (\btheta) ||_{2} \overset{a.s.}{\rightarrow} 0 \quad \text{as} \quad T \rightarrow \infty.
	\end{equation*}
\end{slemma}
\begin{proof}
	We have that
	\begin{equation*}
		\sup_{\btheta \in \bar{\bTheta}} || \sqrt{T} \nabla_{\btheta} \hat{L}_T (\btheta) - \sqrt{T} \nabla_{\btheta} L_T (\btheta) ||_{2} \leq \frac{1}{\sqrt{T}} \sum_{t=1}^{T} \sup_{\btheta \in \bar{\bTheta}} \left| \left| \frac{\nabla_{\btheta} \hat{f}_{t}}{\hat{f}_{t}} - \frac{\nabla_{\btheta} f_{t}}{f_{t} } \right| \right|_{2},
	\end{equation*}
	where $\hat{f}_{t}$ denotes $\hat{f} (Y_t;\btheta)$ and $f_{t}$ denotes $f (Y_t;\btheta)$ and that
	\begin{align*}
		\sup_{\btheta \in \bar{\bTheta}} \left| \left| \frac{\nabla_{\btheta} \hat{f}_{t}}{\hat{f}_{t}} - \frac{\nabla_{\btheta} {f}_{t}}{{f}_{t}} \right| \right|_{2} &\leq \sum_{j=1}^{J} \sup_{\btheta \in \bar{\bTheta}} \left| \left| \frac{\nabla_{\btheta} \hat{\pi}_{j,t \mid t-1} \hat{f}_{j,t}}{\hat{f}_{t}} - \frac{\nabla_{\btheta} \pi_{j,t \mid t-1} f_{j,t}}{f_{t}} \right| \right|_{2} \\
		&+ \sum_{j=1}^{J} \sup_{\btheta \in \bar{\bTheta}} \left| \left| \frac{\hat{\pi}_{j,t \mid t-1} \nabla_{\btheta} \hat{f}_{j,t}}{\hat{f}_{t}} - \frac{\pi_{j,t \mid t-1} \nabla_{\btheta} f_{j,t}}{f_{t}} \right| \right|_{2},
	\end{align*}
	where $\hat{\pi}_{j,t \mid t-1}$ denotes $\hat{\pi}_{j,t \mid t-1} (\btheta)$, $\hat{f}_{j,t}$ denotes $f_j (Y_t;\hat{X}_{j,t},\bupsilon_j)$, $\hat{X}_{j,t}$ denotes $\hat{X}_{j,t} (\bupsilon_j)$, ${\pi}_{j,t \mid t-1}$ denotes ${\pi}_{j,t \mid t-1} (\btheta)$, ${f}_{j,t}$ denotes $f_j (Y_t;{X}_{j,t},\bupsilon_j)$, and ${X}_{j,t}$ denotes ${X}_{j,t} (\bupsilon_j)$.
	
	As in the proof of Lemma \ref{LemmaConsistencyA}, the conclusion thus follows from Lemma 2.1 in \citet{StraumannMikosch2006} if $\sup_{\btheta \in \bar{\bTheta}} \left| \left| \frac{\nabla_{\btheta} \hat{f}_{t}}{\hat{f}_{t}} - \frac{\nabla_{\btheta} f_{t}}{f_{t} } \right| \right|_{2} \overset{e.a.s.}{\rightarrow} 0$ as $t \rightarrow \infty$. This follows if, for each $j \in \{1,...,J\}$, $\sup_{\btheta \in \bar{\bTheta}} \left| \left| \frac{\nabla_{\btheta} \hat{\pi}_{j,t \mid t-1} \hat{f}_{j,t}}{\hat{f}_{t}} - \frac{\nabla_{\btheta} \pi_{j,t \mid t-1} f_{j,t}}{f_{t}} \right| \right|_{2} \overset{e.a.s.}{\rightarrow} 0$ as $t \rightarrow \infty$ and $\sup_{\btheta \in \bar{\bTheta}} \left| \left| \frac{\hat{\pi}_{j,t \mid t-1} \nabla_{\btheta} \hat{f}_{j,t}}{\hat{f}_{t}} - \frac{\pi_{j,t \mid t-1} \nabla_{\btheta} f_{j,t}}{f_{t}} \right| \right|_{2} \overset{e.a.s.}{\rightarrow} 0$ as $t \rightarrow \infty$.
	
	First, we show that, for each $j \in \{1,...,J\}$, $\sup_{\btheta \in \bar{\bTheta}} \left| \left| \frac{\nabla_{\btheta} \hat{\pi}_{j,t \mid t-1} \hat{f}_{j,t}}{\hat{f}_{t}} - \frac{\nabla_{\btheta} \pi_{j,t \mid t-1} f_{j,t}}{f_{t}} \right| \right|_{2} \overset{e.a.s.}{\rightarrow} 0$ as $t \rightarrow \infty$. Let $j \in \{1,...,J\}$ be given. We have that
	\begin{equation}
		\begin{aligned}
			&\left| \left| \frac{\nabla_{\btheta} \hat{\pi}_{j,t \mid t-1} \hat{f}_{j,t}}{\hat{f}_{t}} - \frac{\nabla_{\btheta} \pi_{j,t \mid t-1} f_{j,t}}{f_{t}} \right| \right|_{2} \\
			&= \left| \left| \frac{\nabla_{\btheta} \hat{\pi}_{j,t \mid t-1} \hat{f}_{j,t}}{\hat{f}_{t}} - \frac{\nabla_{\btheta} \pi_{j,t \mid t-1} \hat{f}_{j,t}}{\hat{f}_{t}} + \frac{\nabla_{\btheta} \pi_{j,t \mid t-1} \hat{f}_{j,t}}{\hat{f}_{t}} - \frac{\nabla_{\btheta} \pi_{j,t \mid t-1} f_{j,t}}{f_{t}} \right| \right|_{2} \\		
			&\leq \left| \left| \nabla_{\btheta} \hat{\pi}_{j,t \mid t-1} - \nabla_{\btheta} \pi_{j,t \mid t-1} \right| \right|_{2} \frac{\hat{f}_{j,t}}{\hat{f}_{t}} + \left| \left| \nabla_{\btheta} \pi_{j,t \mid t-1} \right| \right|_{2} \left| \frac{\hat{f}_{j,t}}{\hat{f}_{t}} - \frac{{f}_{j,t}}{f_{t}} \right| \\
			&\leq \frac{1}{p_{-}} \left| \left| \nabla_{\btheta} \hat{\pi}_{j,t \mid t-1} - \nabla_{\btheta} \pi_{j,t \mid t-1} \right| \right|_{2} + \left| \left| \nabla_{\btheta} \pi_{j,t \mid t-1} \right| \right|_{2} \left| \frac{\hat{f}_{j,t}}{\hat{f}_{t}} - \frac{{f}_{j,t}}{f_{t}} \right|,
		\end{aligned} \label{eq:ref1}
	\end{equation}
	where the second inequality follows since, by Equation \eqref{eq:remark} in Remark \ref{remark}, $\hat{\pi}_{j,t \mid t-1} \geq p_{-}$. Note that $\sup_{\btheta \in \bar{\bTheta}} || \nabla_{\btheta} \hat{\pi}_{j,t \mid t-1} - \nabla_{\btheta} \pi_{j,t \mid t-1} ||_{2} \overset{e.a.s.}{\rightarrow} 0$ as $t \rightarrow \infty$ where $(\nabla_{\btheta} \pi_{j,t \mid t-1} )_{t \in \mathbb{Z}}$ is stationary for all $\btheta \in \bar{\bTheta}$ and $\mathbb{E} \left[ \sup_{\btheta \in \bar{\bTheta}} || \nabla_{\btheta} \pi_{j,t \mid t-1} ||_2^2 \right] < \infty$ by Lemma \ref{LemmaInvertibilityDPrediction}. Thus, by Lemmas 2.1 and 2.2 in \cite{StraumannMikosch2006}, $\sup_{\btheta \in \bar{\bTheta}} \left| \left| \frac{\nabla_{\btheta} \hat{\pi}_{j,t \mid t-1} \hat{f}_{j,t}}{\hat{f}_{t}} - \frac{\nabla_{\btheta} \pi_{j,t \mid t-1} f_{j,t}}{f_{t}} \right| \right|_{2} \overset{e.a.s.}{\rightarrow} 0$ as $t \rightarrow \infty$ if $\sup_{\btheta \in \bar{\bTheta}} \left| \frac{\hat{f}_{j,t}}{\hat{f}_{t}} - \frac{{f}_{j,t}}{f_{t}} \right| \overset{e.a.s.}{\rightarrow} 0$ as $t \rightarrow \infty$, which we now show. We have that
	\begin{equation}
		\left| \frac{\hat{f}_{j,t}}{\hat{f}_{t}} - \frac{{f}_{j,t}}{f_{t}} \right| \leq \left| \frac{\hat{f}_{j,t}}{\hat{f}_{t}} - \frac{\hat{f}_{j,t}}{\tilde{f}_{t}} \right| + \left| \frac{\hat{f}_{j,t}}{\tilde{f}_{t}} - \frac{{f}_{j,t}}{f_{t}} \right|,
	\end{equation}
	where $\tilde{f}_{t}$ denotes $\tilde{f} (Y_t;\btheta) = \sum_{j=1}^{J} \pi_{j,t \mid t-1} (\btheta) f_{j} (Y_t;\hat{X}_{j,t} (\bupsilon_j),\bupsilon_j)$, so, by using similar arguments to the ones used to go from Equation \eqref{eq:TermAConsistency} to Equation \eqref{eq:TermDConsistency}, we have that
	\begin{equation}
		\left| \frac{\hat{f}_{j,t}}{\hat{f}_{t}} - \frac{{f}_{j,t}}{f_{t}} \right| \leq C \left| \left| \hat{\bpi}_{t \mid t-1} - \bpi_{t \mid t-1} \right| \right|_{2} + C \sum_{k,l=1}^{J} \left| \partial_{x_k} \log \bar{f}_{k,t} \right| | \hat{X}_{l,t} - X_{l,t} |, \label{eq:ref2}
	\end{equation}
	where $\bar{f}_{j,t}$ denotes $f_{j}(Y_t;\bar{X}_{j,t} (\bupsilon_j),\bupsilon_j)$. It thus follows from Lemmas 2.1 and 2.2 in \citet{StraumannMikosch2006} that $\sup_{\btheta \in \bar{\bTheta}} \left| \frac{\hat{f}_{j,t}}{\hat{f}_{t}} - \frac{{f}_{j,t}}{f_{t}} \right| \overset{e.a.s.}{\rightarrow} 0$ as $t \rightarrow \infty$ since $(Y_t)_{t \in \mathbb{Z}}$ is stationary by Assumption \ref{AssumptionY}, for each $j \in \{1,...,J\}$, there exists an $m_j > 0$ such that $\mathbb{E} \left[ \sup_{\bupsilon_j \in \bUpsilon_j} \sup_{x_j \in \mathcal{X}_{j}} \left| \partial_{x_j} \log f_{j,t} \right|^{m_j} \right] < \infty$ by assumption, for each $j \in \{1,...,J\}$, $\sup_{\bupsilon_{j} \in \bUpsilon_{j}} | \hat{X}_{j,t} - X_{j,t} | \overset{e.a.s.}{\rightarrow} 0$ as $t \rightarrow \infty$ by Lemma \ref{LemmaInvertibilityX}, and $\sup_{\btheta \in \bTheta} || \hat{\bpi}_{t \mid t-1} - \bpi_{t \mid t-1} ||_{2} \overset{e.a.s.}{\rightarrow} 0$ as $t \rightarrow \infty$ by Corollary \ref{CorollaryInvertibilityPrediction}.
	
	We now show that, for each $j \in \{1,...,J\}$, also $\sup_{\btheta \in \bar{\bTheta}} \left| \left| \frac{\hat{\pi}_{j,t \mid t-1} \nabla_{\btheta} \hat{f}_{j,t}}{\hat{f}_{t}} - \frac{\pi_{j,t \mid t-1} \nabla_{\btheta} f_{j,t}}{f_{t}} \right| \right|_{2} \overset{e.a.s.}{\rightarrow} 0$ as $t \rightarrow \infty$. Let $j \in \{1,...,J\}$ be given again. We have that
	\begin{equation}
		\begin{aligned}
			& \left| \left| \frac{\hat{\pi}_{j,t \mid t-1} \nabla_{\btheta} \hat{f}_{j,t}}{\hat{f}_{t}} - \frac{\pi_{j,t \mid t-1} \nabla_{\btheta} f_{j,t}}{f_{t}} \right| \right|_{2} \\
			&= \left| \left| \frac{\hat{\pi}_{j,t \mid t-1} \nabla_{\btheta} \hat{f}_{j,t}}{\hat{f}_{t}} - \frac{\hat{\pi}_{j,t \mid t-1} \nabla_{\btheta} f_{j,t}}{f_{t}} + \frac{\hat{\pi}_{j,t \mid t-1} \nabla_{\btheta} f_{j,t}}{f_{t}} - \frac{\pi_{j,t \mid t-1} \nabla_{\btheta} f_{j,t}}{f_{t}} \right| \right|_{2} \\
			&\leq \hat{\pi}_{j,t \mid t-1} \left| \left| \frac{\nabla_{\btheta} \hat{f}_{j,t}}{\hat{f}_{t}} - \frac{\nabla_{\btheta} f_{j,t}}{f_{t}} \right| \right|_{2} + \left| \hat{\pi}_{j,t \mid t-1} - {\pi}_{j,t \mid t-1} \right| \left| \left| \frac{ \nabla_{\btheta} f_{j,t}}{f_{t}} \right| \right|_{2} \\
			&\leq \left| \left| \frac{\nabla_{\btheta} \hat{f}_{j,t}}{\hat{f}_{t}} - \frac{\nabla_{\btheta} f_{j,t}}{f_{t}} \right| \right|_{2} + \left| \hat{\pi}_{j,t \mid t-1} - {\pi}_{j,t \mid t-1} \right| \left| \left| \frac{ \nabla_{\btheta} f_{j,t}}{f_{t}} \right| \right|_{2}.
		\end{aligned} \label{eq:refa}
	\end{equation}
	In Equation \eqref{eq:refa},
	\begin{equation}
		\begin{aligned}
			&\left| \left| \frac{ \nabla_{\btheta} \hat{f}_{j,t}}{\hat{f}_{t}} - \frac{ \nabla_{\btheta} f_{j,t}}{f_{t}} \right| \right|_{2} \\
			&\leq \left| \left| \frac{ \partial_{x_j} \hat{f}_{j,t} \nabla_{\btheta} \hat{X}_{j,t}}{\hat{f}_{t}} - \frac{ \partial_{x_j} f_{j,t} \nabla_{\btheta} {X}_{j,t}}{f_{t}} \right| \right|_{2} + \left| \left| \frac{ \partial_{\btheta} \hat{f}_{j,t}}{\hat{f}_{t}} - \frac{ \partial_{\btheta} f_{j,t}}{f_{t}} \right| \right|_{2} \\
			&\leq \left| \frac{ \partial_{x_j} \hat{f}_{j,t}}{\hat{f}_{t}} \right| \left| \left| \nabla_{\btheta} \hat{X}_{j,t} - \nabla_{\btheta} X_{j,t} \right| \right|_{2} + \left| \frac{ \partial_{x_j} \hat{f}_{j,t}}{\hat{f}_{t}} - \frac{ \partial_{x_j} f_{j,t}}{f_{t}} \right| \left| \left| \nabla_{\btheta} X_{j,t} \right| \right|_2 \\
			&+ \left| \left| \frac{ \partial_{\btheta} \hat{f}_{j,t}}{\hat{f}_{t}} - \frac{ \partial_{\btheta} f_{j,t}}{f_{t}} \right| \right|_{2} \\
			&\leq \frac{1}{p_{-}} \left| \partial_{x_j} \log \hat{f}_{j,t} \right| \left| \left| \nabla_{\btheta} \hat{X}_{j,t} - \nabla_{\btheta} X_{j,t} \right| \right|_{2} + \left| \frac{ \partial_{x_j} \hat{f}_{j,t}}{\hat{f}_{t}} - \frac{ \partial_{x_j} f_{j,t}}{f_{t}} \right| \left| \left| \nabla_{\btheta} X_{j,t} \right| \right|_2 \\
			&+ \left| \left| \frac{ \partial_{\btheta} \hat{f}_{j,t}}{\hat{f}_{t}} - \frac{ \partial_{\btheta} f_{j,t}}{f_{t}} \right| \right|_{2},
		\end{aligned}
	\end{equation}
	where the second inequality follows since $\hat{a}\hat{b}-ab = (\hat{a}-a)\hat{b}+a(\hat{b}-b)$ and the third inequality follows since $\hat{\pi}_{j,t \mid t-1} \geq p_{-}$. Moreover,
	\begin{equation}
		\begin{aligned}
			\left| \left| \frac{ \nabla_{\btheta} f_{j,t}}{f_{t}} \right| \right|_{2} 
			&\leq \left| \frac{ \partial_{x_j} f_{j,t}}{f_{t}} \right| \left| \left| \nabla_{\btheta} {X}_{j,t} \right| \right|_{2} + \left| \left| \frac{ \partial_{\btheta} f_{j,t}}{f_{t}} \right| \right|_{2} \\
			&\leq \frac{1}{p_{-}} \left| \partial_{x_j} \log f_{j,t} \right| \left| \left| \nabla_{\btheta} {X}_{j,t} \right| \right|_{2} + \frac{1}{p_{-}} \left| \left| \partial_{\btheta} \log f_{j,t} \right| \right|_{2},
		\end{aligned}
	\end{equation}
	where the second inequality follows since, by Equation \eqref{eq:remark} in Remark \ref{remark}, $\pi_{j,t \mid t-1} \geq p_{-}$. Note that $(Y_t)_{t \in \mathbb{Z}}$ is stationary by Assumption \ref{AssumptionY}, there exists an $m_j > 0$ such that $\mathbb{E} \left[ \sup_{\bupsilon_j \in \bUpsilon_j} \sup_{x_j \in \mathcal{X}_{j}} \left| \partial_{x_j} \log f_{j,t} \right|^{m_j} \right] < \infty$ by assumption, $\mathbb{E} \left[ \sup_{\bupsilon_j \in \bar{\bUpsilon}_j} \left| \left| \partial_{\bupsilon_{j}} \log f_{j,t} \right| \right|_{2}^{2} \right] < \infty$ by Assumption \ref{AssumptionF4}, $\sup_{\btheta \in \bTheta} || \hat{\bpi}_{t \mid t-1} - \bpi_{t \mid t-1} ||_{2} \overset{e.a.s.}{\rightarrow} 0$ as $t \rightarrow \infty$ by Corollary \ref{CorollaryInvertibilityPrediction}, and $\sup_{\bupsilon_{j} \in \bar{\bUpsilon}_{j}} || \nabla_{\bupsilon_{j}}  \hat{X}_{j,t} - \nabla_{\bupsilon_{j}} X_{j,t} ||_{2} \overset{e.a.s.}{\rightarrow} 0$ as $t \rightarrow \infty$ where $(\nabla_{\bupsilon_{j}}  {X}_{j,t})_{t \in \mathbb{Z}}$ is stationary for all $\bupsilon_j \in \bar{\bUpsilon}_{j}$ and $\mathbb{E} \left[ \log^{+} \sup_{\bupsilon_{j} \in \bar{\bUpsilon}_{j}} || \nabla_{\bupsilon_{j}}  {X}_{j,t} ||_2 \right] < \infty$ by Lemma \ref{LemmaInvertibilityDX}. Thus, $\sup_{\btheta \in \bar{\bTheta}} \left| \left| \frac{\hat{\pi}_{j,t \mid t-1} \nabla_{\btheta} \hat{f}_{j,t}}{\hat{f}_{t}} - \frac{\pi_{j,t \mid t-1} \nabla_{\btheta} f_{j,t}}{f_{t}} \right| \right|_{2} \overset{e.a.s.}{\rightarrow} 0$ as $t \rightarrow \infty$ if $\sup_{\btheta \in \bar{\bTheta}} \left| \frac{ \partial_{x_j} \hat{f}_{j,t}}{\hat{f}_{t}} - \frac{ \partial_{x_j} f_{j,t}}{f_{t}} \right| \overset{e.a.s.}{\rightarrow} 0$ as $t \rightarrow \infty$ and $\sup_{\btheta \in \bar{\bTheta}} \left| \left| \frac{ \partial_{\btheta} \hat{f}_{j,t}}{\hat{f}_{t}} - \frac{ \partial_{\btheta} f_{j,t}}{f_{t}} \right| \right|_{2} \overset{e.a.s.}{\rightarrow} 0$ as $t \rightarrow \infty$, which we now show. We have that
	\begin{equation}
		\left| \frac{ \partial_{x_j} \hat{f}_{j,t}}{\hat{f}_{t}} - \frac{ \partial_{x_j} f_{j,t} }{f_{t}} \right| \leq \left| \frac{ \partial_{x_j} \hat{f}_{j,t}}{\hat{f}_{t}} - \frac{ \partial_{x_j} \hat{f}_{j,t} }{\tilde{f}_{t}} \right| + \left| \frac{ \partial_{x_j} \hat{f}_{j,t}}{\tilde{f}_{t}} - \frac{ \partial_{x_j} f_{j,t} }{f_{t}} \right|,
	\end{equation}
	so, by using similar arguments to the ones used to go from Equation \eqref{eq:TermAConsistency} to Equation \eqref{eq:TermDConsistency} again, we have that
	\begin{equation}
		\begin{aligned}
			\left| \frac{ \partial_{x_j} \hat{f}_{j,t}}{\hat{f}_{t}} - \frac{ \partial_{x_j} f_{j,t} }{f_{t}} \right| &\leq C \left| \partial_{x_j} \log \hat{f}_{j,t} \right| \left| \left| \hat{\bpi}_{t \mid t-1} - \bpi_{t \mid t-1} \right| \right|_{2} \\
			&+ C \sum_{l = 1}^{J} \left| \partial_{x_j x_j} \log \bar{f}_{j,t} \right| | \hat{X}_{l,t} - X_{l,t} | \\
			&+ C \sum_{k,l=1}^{J} \left| \partial_{x_j} \log \bar{f}_{j,t} \right|  \left| \partial_{x_k} \log \bar{f}_{k,t} \right| | \hat{X}_{l,t} - X_{l,t} |.
		\end{aligned}
	\end{equation}
	Moreover, we have that
	\begin{equation}
		\left| \left| \frac{ \partial_{\btheta} \hat{f}_{j,t}}{\hat{f}_{t}} - \frac{ \partial_{\btheta} f_{j,t}}{f_{t}} \right| \right|_{2} \leq \left| \left| \frac{ \partial_{\btheta} \hat{f}_{j,t}}{\hat{f}_{t}} - \frac{ \partial_{\btheta} \hat{f}_{j,t}}{\tilde{f}_{t}} \right| \right|_{2} + \left| \left| \frac{ \partial_{\btheta} \hat{f}_{j,t}}{\tilde{f}_{t}} - \frac{ \partial_{\btheta} f_{j,t}}{f_{t}} \right| \right|_{2},
	\end{equation}
	so, by using similar arguments to the ones used to go from Equation \eqref{eq:TermAConsistency} to Equation \eqref{eq:TermDConsistency} once again, we have that
	\begin{equation}
		\begin{aligned}
			\left| \left| \frac{ \partial_{\btheta} \hat{f}_{j,t}}{\hat{f}_{t}} - \frac{ \partial_{\btheta} f_{j,t}}{f_{t}} \right| \right|_{2} &\leq C \sum_{h=1}^{d_j} \left| \partial_{[\bupsilon_j]_h} \log \hat{f}_{j,t} \right| \left| \left| \hat{\bpi}_{t \mid t-1} - \bpi_{t \mid t-1} \right| \right|_{2} \\
			&+ C \sum_{h = 1}^{d_j} \sum_{l = 1}^{J} \left| \partial_{[\bupsilon_j]_h x_j} \log \bar{f}_{j,t} \right| | \hat{X}_{l,t} - X_{l,t} | \\
			&+ C \sum_{h = 1}^{d_j} \sum_{k,l=1}^{J} \left| \partial_{[\bupsilon_j]_h} \log \bar{f}_{j,t} \right| \left| \partial_{x_k} \log \bar{f}_{k,t} \right| | \hat{X}_{l,t} - X_{l,t} |,
		\end{aligned} \label{eq:refb}
	\end{equation}
	where
	\begin{equation*}
		\left| \partial_{[\bupsilon_j]_h} \log \hat{f}_{j,t} \right| \leq \sup_{x_j \in \mathcal{X}_{j}} \left| \partial_{[\bupsilon_j]_h x_j} \log f_{j,t} \right| \left| \hat{X}_{j,t} - X_{j,t} \right|+ \left| \partial_{[\bupsilon_j]_h} \log f_{j,t} \right|
	\end{equation*}
	and
	\begin{equation*}
		\left| \partial_{[\bupsilon_j]_h} \log \bar{f}_{j,t} \right| \leq \sup_{x_j \in \mathcal{X}_{j}} \left| \partial_{[\bupsilon_j]_h x_j} \log f_{j,t} \right| \left| \hat{X}_{j,t} - X_{j,t} \right|+ \left| \partial_{[\bupsilon_j]_h} \log f_{j,t} \right|,
	\end{equation*}
	see Equations \eqref{eq:obs1} and \eqref{eq:obs2}. It thus follows that $\sup_{\btheta \in \bar{\bTheta}} \left| \frac{ \partial_{x_j} \hat{f}_{j,t}}{\hat{f}_{t}} - \frac{ \partial_{x_j} f_{j,t}}{f_{t}} \right| \overset{e.a.s.}{\rightarrow} 0$ as $t \rightarrow \infty$ and $\sup_{\btheta \in \bar{\bTheta}} \left| \left| \frac{ \partial_{\btheta} \hat{f}_{j,t}}{\hat{f}_{t}} - \frac{ \partial_{\btheta} f_{j,t}}{f_{t}} \right| \right|_{2} \overset{e.a.s.}{\rightarrow} 0$ as $t \rightarrow \infty$ since – in addition to the arguments used to show that $\sup_{\btheta \in \bar{\bTheta}} \left| \frac{\hat{f}_{j,t}}{\hat{f}_{t}} - \frac{{f}_{j,t}}{f_{t}} \right| \overset{e.a.s.}{\rightarrow} 0$ as $t \rightarrow \infty$ – for each $j \in \{1,...,J\}$, $\mathbb{E} \left[ \sup_{\bupsilon_j \in \bar{\bUpsilon}_j} \left| \left| \partial_{\bupsilon_{j}} \log f_{j,t} \right| \right|_{2}^{2} \right] < \infty$ by Assumption \ref{AssumptionF4} and, for each $j \in \{1,...,J\}$, there exists an $m_j > 0$ such that $\mathbb{E} \left[ \sup_{\bupsilon_j \in \bar{\bUpsilon}_j} \sup_{x_j \in \mathcal{X}_{j}} \left| \left| \partial_{\bvartheta_j x_j} \log f_{j,t} \right| \right|_{2}^{m_j} \right] < \infty$ by assumption.
\end{proof}
\begin{slemma} \label{LemmaAsymptoticNormalityB}
	Under the assumptions in Theorem \ref{TheoremAsymptoticNormality},
	\begin{equation*}
		\sqrt{T} \nabla_{\btheta} L_T (\btheta_0) \overset{d}{\rightarrow} \mathcal{N} \left(\mathbf{0},\mathbb{E} \left[ \nabla_{\btheta} \log f (Y_t;\btheta_0) \nabla_{\btheta^{\prime}} \log f (Y_t;\btheta_0) \right] \right) \quad \text{as} \quad T \rightarrow \infty.
	\end{equation*}
\end{slemma}
\begin{proof}
	The conclusion follows from the central limit theorem by \citet{Billingsley1961} together with the Cramér–Wold theorem if $(\nabla_{\btheta} \log f (Y_t;\btheta_0))_{t \in \mathbb{Z}}$ is a stationary and ergodic martingale difference sequence and $\mathbb{E} \left[ || \nabla_{\btheta} \log f (Y_t;\btheta_0) ||_{2}^{2} \right] < \infty$. First, $(\nabla_{\btheta} \log f (Y_t;\btheta))_{t \in \mathbb{Z}}$ is stationary and ergodic for all $\btheta \in \bar{\bTheta}$ since – in addition to the arguments used to show that $(\log f (Y_t;\btheta))_{t \in \mathbb{Z}}$ is stationary and ergodic for all $\btheta \in \bTheta$ (see the proof of Lemma \ref{LemmaConsistencyB}) – for each $j \in \{1,...,J\}$, $(\nabla_{\bupsilon_{j}} X_{j,t}(\bupsilon_{j}))_{t \in \mathbb{Z}}$ is stationary and ergodic for all $\bupsilon_j \in \bar{\bUpsilon}_{j}$ by Lemma \ref{LemmaInvertibilityDX} and, for each $j \in \{1,...,J\}$, $(\nabla_{\btheta} \pi_{j,t \mid t-1} (\btheta))_{t \in \mathbb{Z}}$ is stationary and ergodic for all $\btheta \in \bar{\bTheta}$ by Lemma \ref{LemmaInvertibilityDPrediction}. Thus, $(\nabla_{\btheta} \log f (Y_t;\btheta_0))_{t \in \mathbb{Z}}$ is stationary and ergodic. Moreover, $(\nabla_{\btheta} \log f (Y_t;\btheta_0))_{t \in \mathbb{Z}}$ is a martingale difference sequence since
	\begin{equation*}
		\mathbb{E} \left[ \nabla_{\btheta} \log f (Y_t;\btheta_0) \mid \bY_{-\infty}^{t-1} \right] = \int_{\mathcal{Y}} \left. \nabla_{\btheta} f (y;\btheta) \right|_{\btheta = \btheta_0} \textup{d}y = \left. \nabla_{\btheta} \int_{\mathcal{Y}} f (y;\btheta) \textup{d}y \right|_{\btheta = \btheta_0} = \mathbf{0}.
	\end{equation*}
	We now show that also $\mathbb{E} \left[ || \nabla_{\btheta} \log f (Y_t;\btheta_0) ||_{2}^{2} \right] < \infty$. We have that
	\begin{align*}
		&|| \nabla_{\btheta} \log f (Y_t;\btheta) ||_{2}^{2} \\
		&= \frac{|| \nabla_{\btheta} f (Y_t;\btheta) ||_{2}^{2}}{f^2 (Y_t;\btheta)} \\
		&\leq \sum_{i,j=1}^{J} || \nabla_{\btheta} \pi_{i,t \mid t-1} (\btheta) ||_2 \frac{f_{i} (Y_t;X_{i,t} (\bupsilon_{i}),\bupsilon_{i})}{f (Y_t;\btheta)} || \nabla_{\btheta} \pi_{j,t \mid t-1} (\btheta) ||_2 \frac{f_{j} (Y_t;X_{j,t} (\bupsilon_{j}),\bupsilon_{j})}{f (Y_t;\btheta)} \\
		&+ 2 \sum_{i,j=1}^{J} || \nabla_{\btheta} \pi_{i,t \mid t-1} (\btheta) ||_2 \frac{f_{i} (Y_t;X_{i,t} (\bupsilon_{i}),\bupsilon_{i})}{f (Y_t;\btheta)} \pi_{j,t \mid t-1} (\btheta) \frac{|| \nabla_{\btheta} f_{j} (Y_t;X_{j,t} (\bupsilon_{j}),\bupsilon_{j}) ||_{2}}{f (Y_t;\btheta)} \\
		&+ \sum_{i,j=1}^{J} \pi_{i,t \mid t-1} (\btheta) \frac{|| \nabla_{\btheta} f_{i} (Y_t;X_{i,t} (\bupsilon_{i}),\bupsilon_{i}) ||_{2}}{f (Y_t;\btheta)} \pi_{j,t \mid t-1} (\btheta) \frac{|| \nabla_{\btheta} f_{j} (Y_t;X_{j,t} (\bupsilon_{j}),\bupsilon_{j}) ||_{2}}{f (Y_t;\btheta)} \\
		&\leq \frac{1}{p_{-}^2} \sum_{i,j=1}^{J} || \nabla_{\btheta} \pi_{i,t \mid t-1} (\btheta) ||_2 || \nabla_{\btheta} \pi_{j,t \mid t-1} (\btheta) ||_2 \\
		&+ \frac{2}{p_{-}} \sum_{i,j=1}^{J} || \nabla_{\btheta} \pi_{i,t \mid t-1} (\btheta) ||_2 || \nabla_{\btheta} \log f_{j} (Y_t;X_{j,t} (\bupsilon_{j}),\bupsilon_{j}) ||_{2} \\
		&+ \sum_{i,j=1}^{J} || \nabla_{\btheta} \log f_{i} (Y_t;X_{i,t} (\bupsilon_{i}),\bupsilon_{i}) ||_{2} || \nabla_{\btheta} \log f_{j} (Y_t;X_{j,t} (\bupsilon_{j}),\bupsilon_{j}) ||_{2}
	\end{align*}
	since, by Equation \eqref{eq:remark} in Remark \ref{remark}, $\pi_{j,t \mid t-1} (\btheta) \geq p_{-}$ for all $j \in \{1,...,J\}$. Thus, by the Cauchy-Schwarz inequality,
	\begin{align*}
		& \mathbb{E} \left[ \sup_{\btheta \in \bar{\bTheta}} || \nabla_{\btheta} \log f (Y_t;\btheta) ||_{2}^{2} \right] \\ &\leq \frac{1}{p_{-}^2} \left( \sum_{j=1}^{J} \left( \mathbb{E} \left[ \sup_{\btheta \in \bar{\bTheta}} || \nabla_{\btheta} \pi_{j,t \mid t-1} (\btheta) ||_{2}^{2} \right] \right)^{1/2} \right)^{2} \\
		&+ \frac{2}{p_{-}} \sum_{i,j=1}^{J} \left( \mathbb{E} \left[ \sup_{\btheta \in \bar{\bTheta}} || \nabla_{\btheta} \pi_{i,t \mid t-1} (\btheta) ||_{2}^{2} \right] \right)^{1/2} \left( \mathbb{E} \left[ \sup_{\btheta \in \bar{\bTheta}} || \nabla_{\btheta} \log f_{j} (Y_t;X_{j,t} (\bupsilon_{j}),\bupsilon_{j}) ||_{2}^{2} \right] \right)^{1/2} \\
		&+ \left( \sum_{j=1}^{J} \left( \mathbb{E} \left[ \sup_{\btheta \in \bar{\bTheta}} || \nabla_{\btheta} \log f_{j} (Y_t;X_{j,t} (\bupsilon_{j}),\bupsilon_{j}) ||_{2}^{2} \right] \right)^{1/2} \right)^{2}.
	\end{align*}
	Note that, for each $j \in \{1,...,J\}$, $\mathbb{E} \left[ \sup_{\btheta \in \bar{\bTheta}} || \nabla_{\btheta} \pi_{j,t \mid t-1} (\btheta) ||_{2}^{2} \right] < \infty$ by Lemma \ref{LemmaInvertibilityDPrediction}. Hence, $\mathbb{E} \left[ \sup_{\btheta \in \bar{\bTheta}} || \nabla_{\btheta} \log f (Y_t;\btheta) ||_{2}^{2} \right] < \infty$ and thus $\mathbb{E} \left[ || \nabla_{\btheta} \log f (Y_t;\btheta_0) ||_{2}^{2} \right] < \infty$ if, for each $j \in \{1,...,J\}$, $\mathbb{E} \left[ \sup_{\btheta \in \bar{\bTheta}} || \nabla_{\btheta} \log f_{j} (Y_t;X_{j,t} (\bupsilon_{j}),\bupsilon_{j}) ||_{2}^{2} \right] < \infty$; this, we now show. We have that
	\begin{align*}
		& || \nabla_{\btheta} \log f_{j} (Y_t;X_{j,t} (\bupsilon_{j}),\bupsilon_{j}) ||_{2}^{2} \\
		&= \frac{|| \nabla_{\btheta} f_{j} (Y_t;X_{j,t} (\bupsilon_{j}),\bupsilon_{j}) ||_{2}^{2}}{f_{j}^{2} (Y_t;X_{j,t} (\bupsilon_{j}),\bupsilon_{j})} \\
		&\leq \frac{\left| \partial_{x_j} f_{j} (Y_t;X_{j,t} (\bupsilon_{j}),\bupsilon_{j}) \right|^2}{f_{j}^{2} (Y_t;X_{j,t} (\bupsilon_{j}),\bupsilon_{j})} || \nabla_{\btheta} X_{j,t} (\bupsilon_{j}) ||_{2}^{2} \\
		&+ 2 \frac{\left| \partial_{x_j} f_{j} (Y_t;X_{j,t} (\bupsilon_{j}),\bupsilon_{j}) \right|}{f_{j} (Y_t;X_{j,t} (\bupsilon_{j}),\bupsilon_{j})} || \nabla_{\btheta} X_{j,t} (\bupsilon_{j}) ||_{2} \frac{|| \partial_{\btheta} f_{j} (Y_t;X_{j,t} (\bupsilon_{j}),\bupsilon_{j}) ||_{2}}{f_{j} (Y_t;X_{j,t} (\bupsilon_{j}),\bupsilon_{j})} \\
		&+ \frac{|| \partial_{\btheta} f_{j} (Y_t;X_{j,t} (\bupsilon_{j}),\bupsilon_{j}) ||_{2}^{2}}{f_{j}^{2} (Y_t;X_{j,t} (\bupsilon_{j}),\bupsilon_{j})} \\
		&\leq \left| \partial_{x_j} \log f_{j} (Y_t;X_{j,t} (\bupsilon_{j}),\bupsilon_{j}) \right|^2 || \nabla_{\btheta} X_{j,t} (\bupsilon_{j}) ||_{2}^{2} \\
		&+ 2 \left| \partial_{x_j} \log f_{j} (Y_t;X_{j,t} (\bupsilon_{j}),\bupsilon_{j}) \right| || \nabla_{\btheta} X_{j,t} (\bupsilon_{j}) ||_{2} || \partial_{\btheta} \log f_{j} (Y_t;X_{j,t} (\bupsilon_{j}),\bupsilon_{j}) ||_{2} \\
		&+ || \partial_{\btheta} \log f_{j} (Y_t;X_{j,t} (\bupsilon_{j}),\bupsilon_{j}) ||_{2}^{2}.
	\end{align*}
	Thus, by the Cauchy-Schwarz inequality again,
	\begin{align*}
		&\mathbb{E} \left[ \sup_{\btheta \in \bar{\bTheta}} || \nabla_{\btheta} \log f_{j} (Y_t;X_{j,t} (\bupsilon_{j}),\bupsilon_{j}) ||_{2}^{2} \right] \\
		&\leq \mathbb{E} \left[ \sup_{\btheta \in \bar{\bTheta}} \left| \partial_{x_j} \log f_{j} (Y_t;X_{j,t} (\bupsilon_{j}),\bupsilon_{j}) \right|^2 ||  \nabla_{\btheta} X_{j,t} (\bupsilon_{j}) ||_{2}^{2} \right] \\
		&+ 2 \left( \mathbb{E} \left[ \sup_{\btheta \in \bar{\bTheta}} \left| \partial_{x_j} \log f_{j} (Y_t;X_{j,t} (\bupsilon_{j}),\bupsilon_{j}) \right|^2 ||  \nabla_{\btheta} X_{j,t} (\bupsilon_{j}) ||_{2}^{2} \right] \right)^{1/2} \\
		& \quad \cdot \left( \mathbb{E} \left[ \sup_{\btheta \in \bar{\bTheta}} || \partial_{\btheta} \log f_{j} (Y_t;X_{j,t} (\bupsilon_{j}),\bupsilon_{j}) ||_{2}^{2} \right] \right)^{1/2} \\
		&+ \mathbb{E} \left[ \sup_{\btheta \in \bar{\bTheta}} || \partial_{\btheta} \log f_{j} (Y_t;X_{j,t} (\bupsilon_{j}),\bupsilon_{j}) ||_{2}^{2} \right].
	\end{align*}
	Hence, for each $j \in \{1,...,J\}$, $\mathbb{E} \left[ \sup_{\btheta \in \bar{\bTheta}} || \nabla_{\btheta} \log f_{j} (Y_t;X_{j,t} (\bupsilon_{j}),\bupsilon_{j}) ||_{2}^{2} \right] < \infty$ since, for each $j \in \{1,...,J\}$, $\mathbb{E} \left[ \sup_{\bupsilon_j \in \bar{\bUpsilon}_{j}} \left| \partial_{x_j} \log f_{j} (Y_t;X_{j,t} (\bupsilon_j),\bupsilon_j) \right|^2 ||  \nabla_{\bupsilon_{j}} X_{j,t} (\bupsilon_{j}) ||_{2}^{2}  \right] < \infty$ and $\mathbb{E} \left[ \sup_{\bupsilon_j \in \bar{\bUpsilon}_{j}} || \partial_{\bupsilon_j} \log f_{j} (Y_t;X_{j,t} (\bupsilon_j),\bupsilon_j) ||_{2}^{2} \right] < \infty$ by Assumption \ref{AssumptionF4}. Thus, $\mathbb{E} \left[ || \nabla_{\btheta} \log f (Y_t;\btheta_0) ||_{2}^{2} \right] < \infty$. This concludes the proof.
\end{proof}
\begin{slemma} \label{LemmaAsymptoticNormalityC}
	Under the assumptions in Theorem \ref{TheoremAsymptoticNormality},
	\begin{equation*}
		\bI(\btheta_0) = \mathbb{E} \left[ \nabla_{\btheta} \log f (Y_t;\btheta_0) \nabla_{\btheta^{\prime}} \log f (Y_t;\btheta_0) \right].
	\end{equation*}
\end{slemma}
\begin{proof}
	We have that
	\begin{align*}
		\mathbb{E} \left[ \nabla_{\btheta \btheta} \log f (Y_t;\btheta_0) \mid \bY_{-\infty}^{t-1} \right] &= - \mathbb{E} \left[ \nabla_{\btheta} \log f (Y_t;\btheta_0) \nabla_{\btheta^{\prime}} \log f (Y_t;\btheta_0) \mid \bY_{-\infty}^{t-1} \right] \\
		&+ \int_{\mathcal{Y}} \left. \nabla_{\btheta \btheta} f (y;\btheta) \right|_{\btheta = \btheta_0} \textup{d}y \\
		&= - \mathbb{E} \left[ \nabla_{\btheta} \log f (Y_t;\btheta_0) \nabla_{\btheta^{\prime}} \log f (Y_t;\btheta_0) \mid \bY_{-\infty}^{t-1} \right] \\
		&+ \left. \nabla_{\btheta \btheta} \int_{\mathcal{Y}} f (y;\btheta) \textup{d}y \right|_{\btheta = \btheta_0} \\
		&= - \mathbb{E} \left[ \nabla_{\btheta} \log f (Y_t;\btheta_0) \nabla_{\btheta^{\prime}} \log f (Y_t;\btheta_0) \mid \bY_{-\infty}^{t-1} \right]
	\end{align*}
	and thus, by the law of total expectation, that
	\begin{equation*}
		\mathbb{E} \left[ \nabla_{\btheta \btheta} \log f (Y_t;\btheta_0) \right] = - \mathbb{E} \left[ \nabla_{\btheta} \log f (Y_t;\btheta_0) \nabla_{\btheta^{\prime}} \log f (Y_t;\btheta_0) \right].
	\end{equation*}
\end{proof}
\begin{slemma} \label{LemmaAsymptoticNormalityD}
	Under the assumptions in Theorem \ref{TheoremAsymptoticNormality},
	\begin{equation*}
		\sup_{\btheta \in \bar{\bTheta}} || \nabla_{\btheta \btheta} \hat{L}_T (\btheta) - \nabla_{\btheta \btheta} L_T (\btheta) ||_{2,2} \overset{a.s.}{\rightarrow} 0 \quad \text{as} \quad T \rightarrow \infty.
	\end{equation*}
\end{slemma}
\begin{proof}
	We have that
	\begin{align*}
		\sup_{\btheta \in \bar{\bTheta}} || \nabla_{\btheta \btheta} \hat{L}_T (\btheta) - \nabla_{\btheta \btheta} L_T (\btheta) ||_{2,2} &\leq \frac{1}{T} \sum_{t=1}^{T} \sup_{\btheta \in \bar{\bTheta}} \left| \left| \frac{\nabla_{\btheta} \hat{f}_{t}}{\hat{f}_{t}} \frac{\nabla_{\btheta^{\prime}} \hat{f}_{t}}{\hat{f}_{t}} - \frac{\nabla_{\btheta} f_{t}}{f_{t} } \frac{\nabla_{\btheta^{\prime}} f_{t}}{f_{t} } \right| \right|_{2,2} \\
		&+ \frac{1}{T} \sum_{t=1}^{T} \sup_{\btheta \in \bar{\bTheta}} \left| \left| \frac{\nabla_{\btheta \btheta} \hat{f}_{t}}{\hat{f}_{t}} - \frac{\nabla_{\btheta \btheta} f_{t}}{f_{t} } \right| \right|_{2,2},
	\end{align*}
	where $\hat{f}_{t}$ denotes $\hat{f}(Y_t;\btheta)$ and $f_{t}$ denotes $f(Y_t;\btheta)$, that
	\begin{align*}
		&\sup_{\btheta \in \bar{\bTheta}} \left| \left| \frac{\nabla_{\btheta} \hat{f}_{t}}{\hat{f}_{t}} \frac{\nabla_{\btheta^{\prime}} \hat{f}_{t}}{\hat{f}_{t}} - \frac{\nabla_{\btheta} f_{t}}{f_{t} } \frac{\nabla_{\btheta^{\prime}} f_{t}}{f_{t} } \right| \right|_{2,2} \\
		&= \sup_{\btheta \in \bar{\bTheta}} \left| \left| \left(\frac{\nabla_{\btheta} \hat{f}_{t}}{\hat{f}_{t}} - \frac{\nabla_{\btheta} f_{t}}{f_{t} } \right) \left( \frac{\nabla_{\btheta^{\prime}} \hat{f}_{t}}{\hat{f}_{t}} - \frac{\nabla_{\btheta^{\prime}} f_{t}}{f_{t}} \right) + \left(\frac{\nabla_{\btheta} \hat{f}_{t}}{\hat{f}_{t}} - \frac{\nabla_{\btheta} f_{t}}{f_{t} } \right) \frac{\nabla_{\btheta^{\prime}} f_{t}}{f_{t}} \right. \right. \\
		& \quad \quad \quad \quad \left. \left. + \frac{\nabla_{\btheta} f_{t}}{f_{t} } \left( \frac{\nabla_{\btheta^{\prime}} \hat{f}_{t}}{\hat{f}_{t}} - \frac{\nabla_{\btheta^{\prime}} f_{t}}{f_{t}} \right) \right| \right|_{2,2} \\
		&\leq \sup_{\btheta \in \bar{\bTheta}} \left| \left| \frac{\nabla_{\btheta} \hat{f}_{t}}{\hat{f}_{t}} - \frac{\nabla_{\btheta} f_{t}}{f_{t} } \right| \right|_{2}^{2} + 2 \sup_{\btheta \in \bar{\bTheta}} \left| \left| \frac{\nabla_{\btheta} f_{t}}{f_{t} } \right| \right|_{2} \sup_{\btheta \in \bar{\bTheta}} \left| \left| \frac{\nabla_{\btheta} \hat{f}_{t}}{\hat{f}_{t}} - \frac{\nabla_{\btheta} f_{t}}{f_{t} } \right| \right|_{2},
	\end{align*}
	and that
	\begin{align*}
		\sup_{\btheta \in \bar{\bTheta}} \left| \left| \frac{\nabla_{\btheta \btheta} \hat{f}_{t}}{\hat{f}_{t}} - \frac{\nabla_{\btheta \btheta} f_{t}}{f_{t} } \right| \right|_{2,2} &\leq \sum_{j=1}^{J} \sup_{\btheta \in \bar{\bTheta}} \left| \left| \frac{\nabla_{\btheta \btheta} \hat{\pi}_{j,t \mid t-1} \hat{f}_{j,t}}{\hat{f}_{t}} - \frac{\nabla_{\btheta \btheta} \pi_{j,t \mid t-1} f_{j,t}}{f_{t}} \right| \right|_{2,2} \\
		&+ 2 \sum_{j=1}^{J} \sup_{\btheta \in \bar{\bTheta}} \left| \left| \frac{\nabla_{\btheta} \hat{\pi}_{j,t \mid t-1} \nabla_{\btheta^{\prime}} \hat{f}_{j,t}}{\hat{f}_{t}} - \frac{\nabla_{\btheta} \pi_{j,t \mid t-1} \nabla_{\btheta^{\prime}} f_{j,t}}{f_{t}} \right| \right|_{2,2}  \\
		&+ \sum_{j=1}^{J} \sup_{\btheta \in \bar{\bTheta}} \left| \left| \frac{\hat{\pi}_{j,t \mid t-1} \nabla_{\btheta \btheta} \hat{f}_{j,t}}{\hat{f}_{t}} - \frac{\pi_{j,t \mid t-1} \nabla_{\btheta \btheta} f_{j,t}}{f_{t}} \right| \right|_{2,2},
	\end{align*}
	where $\hat{\pi}_{j,t \mid t-1}$ denotes $\hat{\pi}_{j,t \mid t-1} (\btheta)$, $\hat{f}_{j,t}$ denotes $f_{j} (Y_t;\hat{X}_{j,t},\bupsilon_j)$, $\hat{X}_{j,t}$ denotes $\hat{X}_{j,t} (\bupsilon_j)$, ${\pi}_{j,t \mid t-1}$ denotes ${\pi}_{j,t \mid t-1} (\btheta)$, ${f}_{j,t}$ denotes $f_{j} (Y_t;{X}_{j,t},\bupsilon_j)$, and ${X}_{j,t}$ denotes ${X}_{j,t} (\bupsilon_j)$.
	
	The conclusion thus follows from Lemma 2.1 in \citet{StraumannMikosch2006} if $\sup_{\btheta \in \bar{\bTheta}} \left| \left| \frac{\nabla_{\btheta} \hat{f}_{t}}{\hat{f}_{t}} \frac{\nabla_{\btheta^{\prime}} \hat{f}_{t}}{\hat{f}_{t}} - \frac{\nabla_{\btheta} f_{t}}{f_{t} } \frac{\nabla_{\btheta^{\prime}} f_{t}}{f_{t} } \right| \right|_{2,2} \overset{e.a.s.}{\rightarrow} 0$ as $t \rightarrow \infty$ and $\sup_{\btheta \in \bar{\bTheta}} \left| \left| \frac{\nabla_{\btheta \btheta} \hat{f}_{t}}{\hat{f}_{t}} - \frac{\nabla_{\btheta \btheta} f_{t}}{f_{t} } \right| \right|_{2,2} \overset{e.a.s.}{\rightarrow} 0$ as $t \rightarrow \infty$. The former follows since $\left(\frac{\nabla_{\btheta} f_{t}}{f_{t}}\right)_{t \in \mathbb{Z}}$ is stationary for all $\btheta \in \bar{\bTheta}$ with $\mathbb{E} \left[ \sup_{\btheta \in \bar{\bTheta}} \left| \left| \frac{\nabla_{\btheta} f_{t}}{f_{t}} \right| \right|_{2}^{2} \right] < \infty$, which we showed in the proof of Lemma \ref{LemmaAsymptoticNormalityB}, and $\sup_{\btheta \in \bar{\bTheta}} \left| \left| \frac{\nabla_{\btheta} \hat{f}_{t}}{\hat{f}_{t}} - \frac{\nabla_{\btheta} f_{t}}{f_{t} } \right| \right|_{2} \overset{e.a.s.}{\rightarrow} 0$ as $t \rightarrow \infty$, which we showed in the proof of Lemma \ref{LemmaAsymptoticNormalityA}. The latter follows if, for each $j \in \{1,...,J\}$, $\sup_{\btheta \in \bar{\bTheta}} \left| \left| \frac{\nabla_{\btheta \btheta} \hat{\pi}_{j,t \mid t-1} \hat{f}_{j,t}}{\hat{f}_{t}} - \frac{\nabla_{\btheta \btheta} \pi_{j,t \mid t-1} f_{j,t}}{f_{t}} \right| \right|_{2,2} \overset{e.a.s.}{\rightarrow} 0$ as $t \rightarrow \infty$, $\sup_{\btheta \in \bar{\bTheta}} \left| \left| \frac{\nabla_{\btheta} \hat{\pi}_{j,t \mid t-1} \nabla_{\btheta^{\prime}} \hat{f}_{j,t}}{\hat{f}_{t}} - \frac{\nabla_{\btheta} \pi_{j,t \mid t-1} \nabla_{\btheta^{\prime}} f_{j,t}}{f_{t}} \right| \right|_{2,2} \overset{e.a.s.}{\rightarrow} 0$ as $t \rightarrow \infty$, and $\sup_{\btheta \in \bar{\bTheta}} \left| \left| \frac{\hat{\pi}_{j,t \mid t-1} \nabla_{\btheta \btheta} \hat{f}_{j,t}}{\hat{f}_{t}} - \frac{\pi_{j,t \mid t-1} \nabla_{\btheta \btheta} f_{j,t}}{f_{t}} \right| \right|_{2,2} \overset{e.a.s.}{\rightarrow} 0$ as $t \rightarrow \infty$.
	
	First, we show that, for each $j \in \{1,...,J\}$, $\sup_{\btheta \in \bar{\bTheta}} \left| \left| \frac{\nabla_{\btheta \btheta} \hat{\pi}_{j,t \mid t-1} \hat{f}_{j,t}}{\hat{f}_{t}} - \frac{\nabla_{\btheta \btheta} \pi_{j,t \mid t-1} f_{j,t}}{f_{t}} \right| \right|_{2,2} \overset{e.a.s.}{\rightarrow} 0$ as $t \rightarrow \infty$. Let $j \in \{1,...,J\}$ be given. We have that
	\begin{align*}
		& \left| \left| \frac{\nabla_{\btheta \btheta} \hat{\pi}_{j,t \mid t-1} \hat{f}_{j,t}}{\hat{f}_{t}} - \frac{\nabla_{\btheta \btheta} \pi_{j,t \mid t-1} f_{j,t}}{f_{t}} \right| \right|_{2,2} \\
		&= \left| \left| \frac{\nabla_{\btheta \btheta} \hat{\pi}_{j,t \mid t-1} \hat{f}_{j,t}}{\hat{f}_{t}} - \frac{\nabla_{\btheta \btheta} \pi_{j,t \mid t-1} \hat{f}_{j,t}}{\hat{f}_{t}} + \frac{\nabla_{\btheta \btheta} \pi_{j,t \mid t-1} \hat{f}_{j,t}}{\hat{f}_{t}} - \frac{\nabla_{\btheta \btheta} \pi_{j,t \mid t-1} f_{j,t}}{f_{t}} \right| \right|_{2,2} \\
		&\leq \left| \left| \nabla_{\btheta \btheta} \hat{\pi}_{j,t \mid t-1} - \nabla_{\btheta \btheta} \pi_{j,t \mid t-1} \right| \right|_{2,2} \frac{\hat{f}_{j,t}}{\hat{f}_{t}} + \left| \left| \nabla_{\btheta \btheta} \pi_{j,t \mid t-1} \right| \right|_{2,2} \left| \frac{\hat{f}_{j,t}}{\hat{f}_{t}} - \frac{f_{j,t}}{f_{t}} \right| \\
		&\leq \frac{1}{p_{-}} \left| \left| \nabla_{\btheta \btheta} \hat{\pi}_{j,t \mid t-1} - \nabla_{\btheta \btheta} \pi_{j,t \mid t-1} \right| \right|_{2,2} + \left| \left| \nabla_{\btheta \btheta} \pi_{j,t \mid t-1} \right| \right|_{2,2} \left| \frac{\hat{f}_{j,t}}{\hat{f}_{t}} - \frac{f_{j,t}}{f_{t}} \right|,
	\end{align*}
	where the second inequality follows since, by Equation \eqref{eq:remark} in Remark \ref{remark}, $\hat{\pi}_{j,t \mid t-1} \geq p_{-}$. It thus follows from Lemmas 2.1 and 2.2 in \citet{StraumannMikosch2006} that $\sup_{\btheta \in \bar{\bTheta}} \left| \left| \frac{\nabla_{\btheta \btheta} \hat{\pi}_{j,t \mid t-1} \hat{f}_{j,t}}{\hat{f}_{t}} - \frac{\nabla_{\btheta \btheta} \pi_{j,t \mid t-1} f_{j,t}}{f_{t}} \right| \right|_{2,2} \overset{e.a.s.}{\rightarrow} 0$ as $t \rightarrow \infty$ since $\sup_{\btheta \in \bar{\bTheta}} || \nabla_{\btheta \btheta} \hat{\pi}_{j,t \mid t-1} - \nabla_{\btheta \btheta} \pi_{j,t \mid t-1}  ||_{2,2} \overset{e.a.s.}{\rightarrow} 0$ as $t \rightarrow \infty$ where $(\nabla_{\btheta \btheta} \pi_{j,t \mid t-1} )_{t \in \mathbb{Z}}$ is stationary for all $\btheta \in \bar{\bTheta}$ and $\mathbb{E} \left[ \sup_{\btheta \in \bar{\bTheta}} || \nabla_{\btheta \btheta} \pi_{j,t \mid t-1} ||_{2,2} \right] < \infty$ by Lemma \ref{LemmaInvertibilityDDPrediction} and $\sup_{\btheta \in \bar{\bTheta}} \left| \frac{\hat{f}_{j,t}}{\hat{f}_{t}} - \frac{f_{j,t}}{f_{t}} \right| \overset{e.a.s.}{\rightarrow} 0$ as $t \rightarrow \infty$, which we showed in the proof of Lemma \ref{LemmaAsymptoticNormalityA} (see Equations \eqref{eq:ref1} to \eqref{eq:ref2}). 
	
	We now show that, for each $j \in \{1,...,J\}$, also $\sup_{\btheta \in \bar{\bTheta}} \left| \left| \frac{\nabla_{\btheta} \hat{\pi}_{j,t \mid t-1} \nabla_{\btheta^{\prime}} \hat{f}_{j,t}}{\hat{f}_{t}} - \frac{\nabla_{\btheta} \pi_{j,t \mid t-1} \nabla_{\btheta^{\prime}} f_{j,t}}{f_{t}} \right| \right|_{2,2} \overset{e.a.s.}{\rightarrow} 0$ as $t \rightarrow \infty$. Again, let $j \in \{1,...,J\}$ be given. We have that
	\begin{align*}
		&\left| \left| \frac{\nabla_{\btheta} \hat{\pi}_{j,t \mid t-1} \nabla_{\btheta^{\prime}} \hat{f}_{j,t}}{\hat{f}_{t}} - \frac{\nabla_{\btheta} \pi_{j,t \mid t-1} \nabla_{\btheta^{\prime}} f_{j,t}}{f_{t}} \right| \right|_{2,2} \\
		&= \left| \left| \left(\nabla_{\btheta} \hat{\pi}_{j,t \mid t-1}-\nabla_{\btheta} \pi_{j,t \mid t-1}\right)\left(\frac{\nabla_{\btheta^{\prime}} \hat{f}_{j,t}}{\hat{f}_{t}}-\frac{\nabla_{\btheta^{\prime}} f_{j,t}}{f_{t}}\right) + \left(\nabla_{\btheta} \hat{\pi}_{j,t \mid t-1}-\nabla_{\btheta} \pi_{j,t \mid t-1}\right)\frac{\nabla_{\btheta^{\prime}} f_{j,t}}{f_{t}} \right. \right. \\
		& \quad \quad \left. \left. + \nabla_{\btheta} \pi_{j,t \mid t-1}\left(\frac{\nabla_{\btheta^{\prime}} \hat{f}_{j,t}}{\hat{f}_{t}}-\frac{\nabla_{\btheta^{\prime}} f_{j,t}}{f_{t}}\right) \right| \right|_{2,2} \\
		&\leq \left| \left| \nabla_{\btheta} \hat{\pi}_{j,t \mid t-1} - \nabla_{\btheta} \pi_{j,t \mid t-1}\right| \right|_{2} \left| \left| \frac{\nabla_{\btheta} \hat{f}_{j,t}}{\hat{f}_{t}} - \frac{\nabla_{\btheta} f_{j,t}}{f_{t}} \right| \right|_{2} + \left| \left| \nabla_{\btheta} \hat{\pi}_{j,t \mid t-1} - \nabla_{\btheta} \pi_{j,t \mid t-1}\right| \right|_{2} \left| \left| \frac{\nabla_{\btheta} f_{j,t}}{f_{t}} \right| \right|_{2} \\
		&+ \left| \left| \nabla_{\btheta} \pi_{j,t \mid t-1}\right| \right|_{2} \left| \left| \frac{\nabla_{\btheta} \hat{f}_{j,t}}{\hat{f}_{t}} - \frac{\nabla_{\btheta} f_{j,t}}{f_{t}} \right| \right|_{2} \\
		&\leq \left| \left| \nabla_{\btheta} \hat{\pi}_{j,t \mid t-1} - \nabla_{\btheta} \pi_{j,t \mid t-1}\right| \right|_{2} \left| \left| \frac{\nabla_{\btheta} \hat{f}_{j,t}}{\hat{f}_{t}} - \frac{\nabla_{\btheta} f_{j,t}}{f_{t}} \right| \right|_{2} + \frac{1}{p_{-}} \left| \left| \nabla_{\btheta} \hat{\pi}_{j,t \mid t-1} - \nabla_{\btheta} \pi_{j,t \mid t-1}\right| \right|_{2} \left| \left| \frac{\nabla_{\btheta} f_{j,t}}{f_{j,t}} \right| \right|_{2} \\
		&+ \left| \left| \nabla_{\btheta} \pi_{j,t \mid t-1}\right| \right|_{2} \left| \left| \frac{\nabla_{\btheta} \hat{f}_{j,t}}{\hat{f}_{t}} - \frac{\nabla_{\btheta} f_{j,t}}{f_{t}} \right| \right|_{2}.
	\end{align*}
	It thus follows that also $\sup_{\btheta \in \bar{\bTheta}} \left| \left| \frac{\nabla_{\btheta} \hat{\pi}_{j,t \mid t-1} \nabla_{\btheta^{\prime}} \hat{f}_{j,t}}{\hat{f}_{t}} - \frac{\nabla_{\btheta} \pi_{j,t \mid t-1} \nabla_{\btheta^{\prime}} f_{j,t}}{f_{t}} \right| \right|_{2,2} \overset{e.a.s.}{\rightarrow} 0$ since $\sup_{\btheta \in \bar{\bTheta}} || \nabla_{\btheta} \hat{\pi}_{j,t \mid t-1} - \nabla_{\btheta} \pi_{j,t \mid t-1} ||_{2} \overset{e.a.s.}{\rightarrow} 0$ as $t \rightarrow \infty$ where $(\nabla_{\btheta} \pi_{j,t \mid t-1} )_{t \in \mathbb{Z}}$ is stationary for all $\btheta \in \bar{\bTheta}$ and $\mathbb{E} \left[ \sup_{\btheta \in \bar{\bTheta}} || \nabla_{\btheta} \pi_{j,t \mid t-1} ||_2^2 \right] < \infty$ by Lemma \ref{LemmaInvertibilityDPrediction}, $\left(\frac{\nabla_{\btheta} f_{j,t}}{f_{j,t}}\right)_{t \in \mathbb{Z}}$ is stationary for all $\btheta \in \bar{\bTheta}$ with $\mathbb{E} \left[ \sup_{\btheta \in \bar{\bTheta}} \left| \left| \frac{\nabla_{\btheta} f_{j,t}}{f_{j,t}} \right| \right|_{2}^{2} \right] < \infty$, which we showed in the proof of Lemma \ref{LemmaAsymptoticNormalityB}, and $\sup_{\btheta \in \bar{\bTheta}} \left| \left| \frac{\nabla_{\btheta} \hat{f}_{j,t}}{\hat{f}_{t}} - \frac{\nabla_{\btheta} f_{j,t}}{f_{t}} \right| \right|_{2} \overset{e.a.s.}{\rightarrow} 0$ as $t \rightarrow \infty$, which we showed in the proof of Lemma \ref{LemmaAsymptoticNormalityA} (see Equations \eqref{eq:refa} to \eqref{eq:refb}).
	
	Finally, we show that, for each $j \in \{1,...,J\}$, $\sup_{\btheta \in \bar{\bTheta}} \left| \left| \frac{\hat{\pi}_{j,t \mid t-1} \nabla_{\btheta \btheta} \hat{f}_{j,t}}{\hat{f}_{t}} - \frac{\pi_{j,t \mid t-1} \nabla_{\btheta \btheta} f_{j,t}}{f_{t}} \right| \right|_{2,2} \overset{e.a.s.}{\rightarrow} 0$ as $t \rightarrow \infty$. Let $j \in \{1,...,J\}$ be given once again. We have that
	\begin{equation}
		\begin{aligned}
			& \left| \left| \frac{\hat{\pi}_{j,t \mid t-1} \nabla_{\btheta \btheta} \hat{f}_{j,t}}{\hat{f}_{t}} - \frac{\pi_{j,t \mid t-1} \nabla_{\btheta \btheta} f_{j,t}}{f_{t}} \right| \right|_{2,2} \\
			&= \left| \left| \frac{\hat{\pi}_{j,t \mid t-1} \nabla_{\btheta \btheta} \hat{f}_{j,t}}{\hat{f}_{t}} - \frac{\hat{\pi}_{j,t \mid t-1} \nabla_{\btheta \btheta} {f}_{j,t}}{{f}_{t}} + \frac{\hat{\pi}_{j,t \mid t-1} \nabla_{\btheta \btheta} {f}_{j,t}}{{f}_{t}} - \frac{\pi_{j,t \mid t-1} \nabla_{\btheta \btheta} f_{j,t}}{f_{t}} \right| \right|_{2,2} \\
			&\leq \hat{\pi}_{j,t \mid t-1} \left| \left| \frac{\nabla_{\btheta \btheta} \hat{f}_{j,t}}{\hat{f}_{t}} - \frac{\nabla_{\btheta \btheta} {f}_{j,t}}{{f}_{t}} \right| \right|_{2,2} + \left| \hat{\pi}_{j,t \mid t-1} - \pi_{j,t \mid t-1} \right| \left| \left| \frac{\nabla_{\btheta \btheta} {f}_{j,t}}{{f}_{t}} \right| \right|_{2,2} \\
			&\leq \left| \left| \frac{\nabla_{\btheta \btheta} \hat{f}_{j,t}}{\hat{f}_{t}} - \frac{\nabla_{\btheta \btheta} {f}_{j,t}}{{f}_{t}} \right| \right|_{2,2} + \left| \hat{\pi}_{j,t \mid t-1} - \pi_{j,t \mid t-1} \right| \left| \left| \frac{\nabla_{\btheta \btheta} {f}_{j,t}}{{f}_{t}} \right| \right|_{2,2}.
		\end{aligned} \label{eq:lasteq}
	\end{equation}
	In Equation \eqref{eq:lasteq},
	\begin{align*}
		&\left| \left| \frac{ \nabla_{\btheta \btheta} \hat{f}_{j,t}}{\hat{f}_{t}} - \frac{\nabla_{\btheta \btheta} f_{j,t}}{f_{t}} \right| \right|_{2,2} \\
		&\leq \left| \left| \frac{ \partial_{x_j x_j} \hat{f}_{j,t} \nabla_{\btheta} \hat{X}_{j,t} \nabla_{\btheta^{\prime}} \hat{X}_{j,t}}{\hat{f}_{t}} - \frac{ \partial_{x_j x_j} f_{j,t} \nabla_{\btheta} {X}_{j,t} \nabla_{\btheta^{\prime}} {X}_{j,t}}{f_{t}} \right| \right|_{2,2} \\
		&+ \left| \left| \frac{ \partial_{x_j} \hat{f}_{j,t} \nabla_{\btheta \btheta} \hat{X}_{j,t}}{\hat{f}_{t}} - \frac{ \partial_{x_j} f_{j,t} \nabla_{\btheta \btheta} {X}_{j,t}}{f_{t}} \right| \right|_{2,2} + 2 \left| \left| \frac{ \partial_{\btheta x_j} \hat{f}_{j,t} \nabla_{\btheta^{\prime}} \hat{X}_{j,t} }{\hat{f}_{t}} - \frac{ \partial_{\btheta x_j} f_{j,t} \nabla_{\btheta^{\prime}} X_{j,t}}{f_{t}} \right| \right|_{2,2} \\
		&+ \left| \left| \frac{ \partial_{\btheta \btheta} \hat{f}_{j,t}}{\hat{f}_{t}} - \frac{ \partial_{\btheta \btheta} f_{j,t}}{f_{t}} \right| \right|_{2,2} \\
		&\leq \left| \frac{ \partial_{x_j x_j} \hat{f}_{j,t}}{\hat{f}_{t}} \right| \left| \left| \nabla_{\btheta} \hat{X}_{j,t} \nabla_{\btheta^{\prime}} \hat{X}_{j,t} - \nabla_{\btheta} {X}_{j,t} \nabla_{\btheta^{\prime}} {X}_{j,t} \right| \right|_{2,2} + \left| \frac{ \partial_{x_j x_j} \hat{f}_{j,t}}{\hat{f}_{t}} - \frac{ \partial_{x_j x_j} f_{j,t}}{f_{t}} \right| \left| \left| \nabla_{\btheta} {X}_{j,t} \right| \right|_{2}^{2} \\
		&+ \left| \frac{ \partial_{x_j} \hat{f}_{j,t}}{\hat{f}_{t}} \right| \left| \left| \nabla_{\btheta \btheta} \hat{X}_{j,t} - \nabla_{\btheta \btheta} X_{j,t} \right| \right|_{2,2} + \left| \frac{ \partial_{x_j} \hat{f}_{j,t}}{\hat{f}_{t}} - \frac{ \partial_{x_j} f_{j,t}}{f_{t}} \right| \left| \left| \nabla_{\btheta \btheta} X_{j,t} \right| \right|_{2,2} \\
		&+ 2 \left| \left| \frac{ \partial_{\btheta x_j} \hat{f}_{j,t}}{\hat{f}_{t}} \right| \right|_2 \left| \left| \nabla_{\btheta} \hat{X}_{j,t} - \nabla_{\btheta} X_{j,t} \right| \right|_{2} + 2 \left| \left| \frac{ \partial_{\btheta x_j} \hat{f}_{j,t} }{\hat{f}_{t}} - \frac{ \partial_{\btheta x_j} f_{j,t}}{f_{t}} \right| \right|_{2} \left| \left| \nabla_{\btheta} X_{j,t} \right| \right|_2 \\
		&+ \left| \left| \frac{ \partial_{\btheta \btheta} \hat{f}_{j,t}}{\hat{f}_{t}} - \frac{ \partial_{\btheta \btheta} f_{j,t}}{f_{t}} \right| \right|_{2,2} \\
		&\leq \left| \frac{ \partial_{x_j x_j} \hat{f}_{j,t}}{\hat{f}_{t}} \right| \left( \left| \left| \nabla_{\btheta} \hat{X}_{j,t} - \nabla_{\btheta} {X}_{j,t} \right| \right|_{2}^{2} + 2 \left| \left| \nabla_{\btheta} {X}_{j,t} \right| \right|_{2} \left| \left| \nabla_{\btheta} \hat{X}_{j,t} - \nabla_{\btheta} {X}_{j,t} \right| \right|_{2} \right) \\
		&+\left| \frac{ \partial_{x_j x_j} \hat{f}_{j,t}}{\hat{f}_{t}} - \frac{ \partial_{x_j x_j} f_{j,t}}{f_{t}} \right| \left| \left| \nabla_{\btheta} {X}_{j,t} \right| \right|_{2}^{2} + \left| \frac{ \partial_{x_j} \hat{f}_{j,t}}{\hat{f}_{t}} \right| \left| \left| \nabla_{\btheta \btheta} \hat{X}_{j,t} - \nabla_{\btheta \btheta} X_{j,t} \right| \right|_{2,2} \\
		&+ \left| \frac{ \partial_{x_j} \hat{f}_{j,t}}{\hat{f}_{t}} - \frac{ \partial_{x_j} f_{j,t}}{f_{t}} \right| \left| \left| \nabla_{\btheta \btheta} X_{j,t} \right| \right|_{2,2} + 2 \left| \left| \frac{ \partial_{\btheta x_j} \hat{f}_{j,t}}{\hat{f}_{t}} \right| \right|_2 \left| \left| \nabla_{\btheta} \hat{X}_{j,t} - \nabla_{\btheta} X_{j,t} \right| \right|_{2} \\
		&+ 2 \left| \left| \frac{ \partial_{\btheta x_j} \hat{f}_{j,t} }{\hat{f}_{t}} - \frac{ \partial_{\btheta x_j} f_{j,t}}{f_{t}} \right| \right|_{2} \left| \left| \nabla_{\btheta} X_{j,t} \right| \right|_2 + \left| \left| \frac{ \partial_{\btheta \btheta} \hat{f}_{j,t}}{\hat{f}_{t}} - \frac{ \partial_{\btheta \btheta} f_{j,t}}{f_{t}} \right| \right|_{2,2} \\
		&\leq \frac{1}{p_{-}} \left( \left| \partial_{x_j} \log \hat{f}_{j,t} \right|^2 + \left| \partial_{x_j x_j} \log \hat{f}_{j,t} \right| \right) \\
		& \quad \cdot \left( \left| \left| \nabla_{\btheta} \hat{X}_{j,t} - \nabla_{\btheta} {X}_{j,t} \right| \right|_{2}^{2} + 2 \left| \left| \nabla_{\btheta} {X}_{j,t} \right| \right|_{2} \left| \left| \nabla_{\btheta} \hat{X}_{j,t} - \nabla_{\btheta} {X}_{j,t} \right| \right|_{2} \right) \\
		&+ \left| \frac{ \partial_{x_j x_j} \hat{f}_{j,t}}{\hat{f}_{t}} - \frac{ \partial_{x_j x_j} f_{j,t}}{f_{t}} \right| \left| \left| \nabla_{\btheta} {X}_{j,t} \right| \right|_{2}^{2} + \frac{1}{p_{-}} \left| \partial_{x_j} \log \hat{f}_{j,t} \right| \left| \left| \nabla_{\btheta \btheta} \hat{X}_{j,t} - \nabla_{\btheta \btheta} X_{j,t} \right| \right|_{2,2} \\
		&+ \left| \frac{ \partial_{x_j} \hat{f}_{j,t}}{\hat{f}_{t}} - \frac{ \partial_{x_j} f_{j,t}}{f_{t}} \right| \left| \left| \nabla_{\btheta \btheta} X_{j,t} \right| \right|_{2,2} \\
		&+ \frac{2}{p_{-}} \left( \left| \partial_{x_j} \log \hat{f}_{j,t} \right| \left| \left| \partial_{\btheta} \log \hat{f}_{j,t} \right| \right|_2 + \left| \left| \partial_{\btheta x_j} \log \hat{f}_{j,t} \right| \right|_2 \right) \left| \left| \nabla_{\btheta} \hat{X}_{j,t} - \nabla_{\btheta} X_{j,t} \right| \right|_{2} \\
		&+ 2 \left| \left| \frac{ \partial_{\btheta x_j} \hat{f}_{j,t} }{\hat{f}_{t}} - \frac{ \partial_{\btheta x_j} f_{j,t}}{f_{t}} \right| \right|_{2} \left| \left| \nabla_{\btheta} X_{j,t} \right| \right|_2 + \left| \left| \frac{ \partial_{\btheta \btheta} \hat{f}_{j,t}}{\hat{f}_{t}} - \frac{ \partial_{\btheta \btheta} f_{j,t}}{f_{t}} \right| \right|_{2,2},
	\end{align*}
	where the second inequality follows since $\hat{a}\hat{b}-ab = (\hat{a}-a)\hat{b}+a(\hat{b}-b)$, the third inequality follows since $\hat{a}\hat{b}-ab = (\hat{a}-a)(\hat{b}-b)+(\hat{a}-a)b+a(\hat{b}-b)$, and the fourth inequality follows since $\hat{\pi}_{j,t \mid t-1} \geq p_{-}$. Moreover,
	\begin{align*}
		\left| \left| \frac{\nabla_{\btheta \btheta} f_{j,t}}{f_{t}} \right| \right|_{2,2} &\leq \left| \left| \frac{ \partial_{x_j x_j} {f}_{j,t} \nabla_{\btheta} {X}_{j,t} \nabla_{\btheta^{\prime}} {X}_{j,t}}{{f}_{t}} \right| \right|_{2,2} + \left| \left| \frac{ \partial_{x_j} f_{j,t} \nabla_{\btheta \btheta} {X}_{j,t}}{f_{t}} \right| \right|_{2,2} + 2 \left| \left| \frac{ \partial_{\btheta x_j} f_{j,t} \nabla_{\btheta^{\prime}} X_{j,t}}{f_{t}} \right| \right|_{2,2} \\
		&+ \left| \left| \frac{ \partial_{\btheta \btheta} f_{j,t}}{f_{t}} \right| \right|_{2,2} \\
		&\leq \frac{1}{p_{-}} \left( \left| \partial_{x_j} \log {f}_{j,t} \right|^2 + \left| \partial_{x_j x_j} \log {f}_{j,t} \right| \right) \left| \left| \nabla_{\btheta} {X}_{j,t} \right| \right|_{2}^{2} + \frac{1}{p_{-}} \left| \partial_{x_j} \log f_{j,t} \right| \left| \left| \nabla_{\btheta \btheta} {X}_{j,t} \right| \right|_{2,2} \\
		&+ \frac{2}{p_{-}} \left( \left| \partial_{x_j} \log f_{j,t} \right| \left| \left| \partial_{\btheta} \log f_{j,t} \right| \right|_{2} + \left| \left| \partial_{\btheta x_j} \log f_{j,t} \right| \right|_{2} \right) \left| \left| \nabla_{\btheta} X_{j,t} \right| \right|_{2} \\
		&+ \frac{1}{p_{-}} \left( \left| \left| \partial_{\btheta} \log f_{j,t} \right| \right|_{2}^{2} + \left| \left| \partial_{\btheta \btheta} \log f_{j,t} \right| \right|_{2,2} \right),
	\end{align*}
	where the second inequality follows since, by Equation \eqref{eq:remark} in Remark \ref{remark}, $\pi_{j,t \mid t-1} \geq p_{-}$. Note that – in addition to the arguments used to show that $\sup_{\btheta \in \bar{\bTheta}} \left| \left| \frac{\hat{\pi}_{j,t \mid t-1} \nabla_{\btheta} \hat{f}_{j,t}}{\hat{f}_{t}} - \frac{\pi_{j,t \mid t-1} \nabla_{\btheta} f_{j,t}}{f_{t}} \right| \right|_{2} \overset{e.a.s.}{\rightarrow} 0$ as $t \rightarrow \infty$ if $\sup_{\btheta \in \bar{\bTheta}} \left| \frac{ \partial_{x_j} \hat{f}_{j,t}}{\hat{f}_{t}} - \frac{ \partial_{x_j} f_{j,t}}{f_{t}} \right| \overset{e.a.s.}{\rightarrow} 0$ as $t \rightarrow \infty$ and $\sup_{\btheta \in \bar{\bTheta}} \left| \left| \frac{ \partial_{\btheta} \hat{f}_{j,t}}{\hat{f}_{t}} - \frac{ \partial_{\btheta} f_{j,t}}{f_{t}} \right| \right|_{2} \overset{e.a.s.}{\rightarrow} 0$ as $t \rightarrow \infty$ (see the proof of Lemma \ref{LemmaAsymptoticNormalityA}) – $\mathbb{E} \left[ \sup_{\bupsilon_j \in \bar{\bUpsilon}_j} \left| \left| \partial_{\bupsilon_{j} \bupsilon_{j}} \log f_{j,t} \right| \right|_{2,2} \right] < \infty$ by Assumption \ref{AssumptionF4}, there exists an $m_j > 0$ such that $\mathbb{E} \left[ \sup_{\bupsilon_j \in \bar{\bUpsilon}_j} \sup_{x_j \in \mathcal{X}_{j}} \left| \left| \partial_{\bvartheta_j x_j} \log f_{j,t} \right| \right|_{2}^{m_j} \right] < \infty$ by assumption, and $\sup_{\bupsilon_{j} \in \bar{\bUpsilon}_{j}} || \nabla_{\bupsilon_{j} \bupsilon_{j}}  \hat{X}_{j,t} - \nabla_{\bupsilon_{j} \bupsilon_{j}} X_{j,t} ||_{2,2} \overset{e.a.s.}{\rightarrow} 0$ as $t \rightarrow \infty$ where $(\nabla_{\bupsilon_{j} \bupsilon_{j}}  {X}_{j,t})_{t \in \mathbb{Z}}$ is stationary for all $\bupsilon_j \in \bar{\bUpsilon}_{j}$ and $\mathbb{E} \left[ \log^{+} \sup_{\bupsilon_{j} \in \bar{\bUpsilon}_{j}} || \nabla_{\bupsilon_{j} \bupsilon_{j}}  {X}_{j,t} ||_{2,2} \right] < \infty$ by Lemma \ref{LemmaInvertibilityDDX}. Thus, by Lemmas 2.1 and 2.2 in \cite{StraumannMikosch2006}, $\sup_{\btheta \in \bar{\bTheta}} \left| \left| \frac{\hat{\pi}_{j,t \mid t-1} \nabla_{\btheta \btheta} \hat{f}_{j,t}}{\hat{f}_{t}} - \frac{\pi_{j,t \mid t-1} \nabla_{\btheta \btheta} f_{j,t}}{f_{t}} \right| \right|_{2,2} \overset{e.a.s.}{\rightarrow} 0$ as $t \rightarrow \infty$ if $\sup_{\btheta \in \bar{\bTheta}} \left| \frac{ \partial_{x_j x_j} \hat{f}_{j,t}}{\hat{f}_{t}} - \frac{ \partial_{x_j x_j} f_{j,t}}{f_{t}} \right| \overset{e.a.s.}{\rightarrow} 0$ as $t \rightarrow \infty$, $\sup_{\btheta \in \bar{\bTheta}} \left| \left| \frac{ \partial_{\btheta x_j} \hat{f}_{j,t} }{\hat{f}_{t}} - \frac{ \partial_{\btheta x_j} f_{j,t}}{f_{t}} \right| \right|_{2} \overset{e.a.s.}{\rightarrow} 0$ as $t \rightarrow \infty$, and $\sup_{\btheta \in \bar{\bTheta}} \left| \left| \frac{ \partial_{\btheta \btheta} \hat{f}_{j,t}}{\hat{f}_{t}} - \frac{ \partial_{\btheta \btheta} f_{j,t}}{f_{t}} \right| \right|_{2,2} \overset{e.a.s.}{\rightarrow} 0$ as $t \rightarrow \infty$, which we now show. We have that
	\begin{equation*}
		\left| \frac{ \partial_{x_j x_j} \hat{f}_{j,t} }{\hat{f}_{t}} - \frac{ \partial_{x_j x_j} f_{j,t}}{f_{t}} \right| \leq \left| \frac{ \partial_{x_j x_j} \hat{f}_{j,t} }{\hat{f}_{t}} - \frac{ \partial_{x_j x_j} \hat{f}_{j,t}}{\tilde{f}_{t}} \right| + \left| \frac{ \partial_{x_j x_j} \hat{f}_{j,t} }{\tilde{f}_{t}} - \frac{ \partial_{x_j x_j} f_{j,t}}{f_{t}} \right|,
	\end{equation*}
	where $\tilde{f}_{t}$ denotes $\tilde{f} (Y_t;\btheta) = \sum_{j=1}^{J} \pi_{j,t \mid t-1} (\btheta) f_{j} (Y_t;\hat{X}_{j,t} (\bupsilon_j),\bupsilon_j)$, so, by using similar arguments to the ones used to go from Equation \eqref{eq:TermAConsistency} to Equation \eqref{eq:TermDConsistency}, we have that
	\begin{align*}
		\left| \frac{ \partial_{x_j x_j} \hat{f}_{j,t} }{\hat{f}_{t}} - \frac{ \partial_{x_j x_j} f_{j,t}}{f_{t}} \right| &\leq C \left| \partial_{x_j x_j} \log \hat{f}_{j,t} \right| \left| \left| \hat{\bpi}_{t \mid t-1} - \bpi_{t \mid t-1} \right| \right|_{2} \\
		&+ C \left| \partial_{x_j} \log \hat{f}_{j,t} \right|^{2} \left| \left| \hat{\bpi}_{t \mid t-1} - \bpi_{t \mid t-1} \right| \right|_{2} \\
		&+ C \sum_{l = 1}^{J} \left| \partial_{x_j x_j x_j} \log \bar{f}_{j,t} \right| | \hat{X}_{l,t} - X_{l,t} | \\
		&+ C \sum_{k,l=1}^{J}  \left| \partial_{x_j x_j} \log \bar{f}_{j,t} \right| \left| \partial_{x_k} \log \bar{f}_{k,t} \right| | \hat{X}_{l,t} - X_{l,t} | \\
		&+ C \sum_{k,l=1}^{J}  \left| \partial_{x_j} \log \bar{f}_{j,t} \right|^2 \left| \partial_{x_k} \log \bar{f}_{k,t} \right| | \hat{X}_{l,t} - X_{l,t} |,
	\end{align*}
	where $\bar{f}_{j,t}$ denotes $f_{j} (Y_t;\bar{X}_{j,t}(\bupsilon_j),\bupsilon_j)$. Moreover,
	\begin{equation*}
		\left| \left| \frac{ \partial_{\btheta x_j} \hat{f}_{j,t} }{\hat{f}_{t}} - \frac{ \partial_{\btheta x_j} f_{j,t} }{f_{t}} \right| \right|_{2} \leq \left| \left| \frac{ \partial_{\btheta x_j} \hat{f}_{j,t} }{\hat{f}_{t}} - \frac{ \partial_{\btheta x_j} \hat{f}_{j,t} }{\tilde{f}_{t}} \right| \right|_{2} + \left| \left| \frac{ \partial_{\btheta x_j} \hat{f}_{j,t} }{\tilde{f}_{t}} - \frac{ \partial_{\btheta x_j} f_{j,t} }{f_{t}} \right| \right|_{2}.
	\end{equation*}
	Thus,
	\begin{align*}
		\left| \left| \frac{ \partial_{\btheta x_j} \hat{f}_{j,t} }{\hat{f}_{t}} - \frac{ \partial_{\btheta x_j} f_{j,t} }{f_{t}} \right| \right|_{2} &\leq C \sum_{h=1}^{d_j} \left| \partial_{[\bupsilon_j]_h x_j} \log \hat{f}_{j,t} \right| \left| \left| \hat{\bpi}_{t \mid t-1} - \bpi_{t \mid t-1} \right| \right|_{2} \\
		&+ C \sum_{h=1}^{d_j} \left| \partial_{[\bupsilon_j]_h} \log \hat{f}_{j,t} \right| \left| \partial_{x_j} \log \hat{f}_{j,t} \right| \left| \left| \hat{\bpi}_{t \mid t-1} - \bpi_{t \mid t-1} \right| \right|_{2} \\
		&+ C \sum_{h = 1}^{d_j} \sum_{l = 1}^{J} \left| \partial_{[\bupsilon_j]_h x_j x_j} \log \bar{f}_{j,t} \right| | \hat{X}_{l,t} - X_{l,t} | \\
		&+ C \sum_{h = 1}^{d_j} \sum_{k,l=1}^{J} \left| \partial_{[\bupsilon_j]_h x_j} \log \bar{f}_{j,t} \right| \left| \partial_{x_k} \log \bar{f}_{k,t} \right| | \hat{X}_{l,t} - X_{l,t} | \\
		&+ C \sum_{h = 1}^{d_j} \sum_{l = 1}^{J} \left| \partial_{x_j x_j} \log \bar{f}_{j,t} \right| \left| \partial_{[\bupsilon_j]_h} \log \bar{f}_{j,t} \right| | \hat{X}_{l,t} - X_{l,t} | \\
		&+ C \sum_{h = 1}^{d_j} \sum_{k,l=1}^{J} \left| \partial_{[\bupsilon_j]_h} \log \bar{f}_{j,t} \right| \left| \partial_{x_j} \log \bar{f}_{j,t} \right| \left| \partial_{x_k} \log \bar{f}_{k,t} \right| | \hat{X}_{l,t} - X_{l,t} |
	\end{align*}
	by using similar arguments to the ones used to go from Equation \eqref{eq:TermAConsistency} to Equation \eqref{eq:TermDConsistency} again where
	\begin{equation*}
		\left| \partial_{[\bupsilon_j]_h} \log \hat{f}_{j,t} \right| \leq \sup_{x_j \in \mathcal{X}_{j}} \left| \partial_{[\bupsilon_j]_h x_j} \log f_{j,t} \right| \left| \hat{X}_{j,t} - X_{j,t} \right|+ \left| \partial_{[\bupsilon_j]_h} \log f_{j,t} \right|
	\end{equation*}
	and
	\begin{equation*}
		\left| \partial_{[\bupsilon_j]_h} \log \bar{f}_{j,t} \right| \leq \sup_{x_j \in \mathcal{X}_{j}} \left| \partial_{[\bupsilon_j]_h x_j} \log f_{j,t} \right| \left| \hat{X}_{j,t} - X_{j,t} \right|+ \left| \partial_{[\bupsilon_j]_h} \log f_{j,t} \right|,
	\end{equation*}
	see Equations \eqref{eq:obs1} and \eqref{eq:obs2}. Finally, we have that
	\begin{equation*}
		\left| \left| \frac{ \partial_{\btheta \btheta} \hat{f}_{j,t}}{\hat{f}_{t}} - \frac{ \partial_{\btheta \btheta} f_{j,t}}{f_{t}} \right| \right|_{2,2} \leq \left| \left| \frac{ \partial_{\btheta \btheta} \hat{f}_{j,t}}{\hat{f}_{t}} - \frac{ \partial_{\btheta \btheta} \hat{f}_{j,t}}{\tilde{f}_{t}} \right| \right|_{2,2} + \left| \left| \frac{ \partial_{\btheta \btheta} \hat{f}_{j,t}}{\tilde{f}_{t}} - \frac{ \partial_{\btheta \btheta} f_{j,t}}{f_{t}} \right| \right|_{2,2},
	\end{equation*}
	so, by using similar arguments to the ones used to go from Equation \eqref{eq:TermAConsistency} to Equation \eqref{eq:TermDConsistency} once again, we have that
	\begin{align*}
		\left| \left| \frac{ \partial_{\btheta \btheta} \hat{f}_{j,t}}{\hat{f}_{t}} - \frac{ \partial_{\btheta \btheta} f_{j,t}}{f_{t}} \right| \right|_{2,2}
		&\leq C \sum_{h,i=1}^{d_j} \left| \partial_{[\bupsilon_j]_h [\bupsilon_j]_i} \log \hat{f}_{j,t} \right| \left| \left| \hat{\bpi}_{t \mid t-1} - \bpi_{t \mid t-1} \right| \right|_{2} \\
		&+ C \sum_{h,i=1}^{d_j} \left| \partial_{[\bupsilon_j]_h} \log \hat{f}_{j,t} \right| \left| \partial_{[\bupsilon_j]_i} \log \hat{f}_{j,t} \right| \left| \left| \hat{\bpi}_{t \mid t-1} - \bpi_{t \mid t-1} \right| \right|_{2} \\
		&+ C \sum_{h,i = 1}^{d_j} \sum_{l = 1}^{J} \left| \partial_{[\bupsilon_j]_h [\bupsilon_j]_i x_j} \log \bar{f}_{j,t} \right| | \hat{X}_{l,t} - X_{l,t} | \\
		&+ C \sum_{h,i = 1}^{d_j} \sum_{k,l=1}^{J} \left| \partial_{[\bupsilon_j]_h [\bupsilon_j]_i} \log \bar{f}_{j,t} \right| \left| \partial_{x_k} \log \bar{f}_{k,t} \right| | \hat{X}_{l,t} - X_{l,t} | \\
		&+ C \sum_{h,i = 1}^{d_j} \sum_{l = 1}^{J} \left| \partial_{[\bupsilon_j]_h x_j} \log \bar{f}_{j,t} \right| \left| \partial_{[\bupsilon_j]_i} \log \bar{f}_{j,t} \right| | \hat{X}_{l,t} - X_{l,t} | \\
		&+ C \sum_{h,i = 1}^{d_j} \sum_{l = 1}^{J} \left| \partial_{[\bupsilon_j]_i x_j} \log \bar{f}_{j,t} \right| \left| \partial_{[\bupsilon_j]_h} \log \bar{f}_{j,t} \right| | \hat{X}_{l,t} - X_{l,t} | \\
		&+ C \sum_{h,i = 1}^{d_j} \sum_{k,l=1}^{J} \left| \partial_{[\bupsilon_j]_h} \log \bar{f}_{j,t} \right|  \left| \partial_{[\bupsilon_j]_i} \log \bar{f}_{j,t} \right|  \left| \partial_{x_k} \log \bar{f}_{k,t} \right| | \hat{X}_{l,t} - X_{l,t} |,
	\end{align*}
	where, by the triangle inequality and the mean value theorem,
	\begin{align*}
		\left| \partial_{[\bupsilon_j]_h [\bupsilon_j]_i} \log \hat{f}_{j,t} \right| &\leq \left| \partial_{[\bupsilon_j]_h [\bupsilon_j]_i} \log \hat{f}_{j,t} - \partial_{[\bupsilon_j]_h [\bupsilon_j]_i} \log f_{j,t} \right| + \left| \partial_{[\bupsilon_j]_h [\bupsilon_j]_i} \log f_{j,t} \right| \\
		&\leq \sup_{x_j \in \mathcal{X}_{j}} \left| \partial_{[\bupsilon_j]_h [\bupsilon_j]_i x_j} \log f_{j,t} \right| \left| \hat{X}_{j,t} - X_{j,t} \right|+ \left| \partial_{[\bupsilon_j]_h [\bupsilon_j]_i} \log f_{j,t} \right|
	\end{align*}
	and
	\begin{align*}
		\left| \partial_{[\bupsilon_j]_h [\bupsilon_j]_i} \log \bar{f}_{j,t} \right| &\leq \left| \partial_{[\bupsilon_j]_h [\bupsilon_j]_i} \log \bar{f}_{j,t} - \partial_{[\bupsilon_j]_h [\bupsilon_j]_i} \log f_{j,t} \right| + \left| \partial_{[\bupsilon_j]_h [\bupsilon_j]_i} \log f_{j,t} \right| \\
		&\leq \sup_{x_j \in \mathcal{X}_{j}} \left| \partial_{[\bupsilon_j]_h [\bupsilon_j]_i x_j} \log f_{j,t} \right| \left| \bar{X}_{j,t} - X_{j,t} \right|+ \left| \partial_{[\bupsilon_j]_h [\bupsilon_j]_i} \log f_{j,t} \right| \\
		&\leq \sup_{x_j \in \mathcal{X}_{j}} \left| \partial_{[\bupsilon_j]_h [\bupsilon_j]_i x_j} \log f_{j,t} \right| \left| \hat{X}_{j,t} - X_{j,t} \right|+ \left| \partial_{[\bupsilon_j]_h [\bupsilon_j]_i} \log f_{j,t} \right|.
	\end{align*}
	It thus follows from Lemmas 2.1 and 2.2 in \cite{StraumannMikosch2006} that $\sup_{\btheta \in \bar{\bTheta}} \left| \frac{ \partial_{x_j x_j} \hat{f}_{j,t}}{\hat{f}_{t}} - \frac{ \partial_{x_j x_j} f_{j,t}}{f_{t}} \right| \overset{e.a.s.}{\rightarrow} 0$ as $t \rightarrow \infty$, $\sup_{\btheta \in \bar{\bTheta}} \left| \left| \frac{ \partial_{\btheta x_j} \hat{f}_{j,t} }{\hat{f}_{t}} - \frac{ \partial_{\btheta x_j} f_{j,t}}{f_{t}} \right| \right|_{2} \overset{e.a.s.}{\rightarrow} 0$ as $t \rightarrow \infty$, and $\sup_{\btheta \in \bar{\bTheta}} \left| \left| \frac{ \partial_{\btheta \btheta} \hat{f}_{j,t}}{\hat{f}_{t}} - \frac{ \partial_{\btheta \btheta} f_{j,t}}{f_{t}} \right| \right|_{2,2} \overset{e.a.s.}{\rightarrow} 0$ as $t \rightarrow \infty$ since – in addition to the arguments used to show that $\sup_{\btheta \in \bar{\bTheta}} \left| \frac{ \partial_{x_j} \hat{f}_{j,t}}{\hat{f}_{t}} - \frac{ \partial_{x_j} f_{j,t}}{f_{t}} \right| \overset{e.a.s.}{\rightarrow} 0$ as $t \rightarrow \infty$ and $\sup_{\btheta \in \bar{\bTheta}} \left| \left| \frac{ \partial_{\btheta} \hat{f}_{j,t}}{\hat{f}_{t}} - \frac{ \partial_{\btheta} f_{j,t}}{f_{t}} \right| \right|_{2} \overset{e.a.s.}{\rightarrow} 0$ as $t \rightarrow \infty$ (see the proof of Lemma \ref{LemmaAsymptoticNormalityA}) – for each $j \in \{1,...,J\}$, $\mathbb{E} \left[ \sup_{\bupsilon_j \in \bar{\bUpsilon}_j} \left| \left| \partial_{\bupsilon_j \bupsilon_j} \log f_{j,t} \right| \right|_{2,2} \right] < \infty$ by Assumption \ref{AssumptionF4} and, for each $j \in \{1,...,J\}$, there exists an $m_j > 0$ such that $\mathbb{E} \left[ \sup_{\bupsilon_j \in \bar{\bUpsilon}_j} \sup_{x_j \in \mathcal{X}_{j}} \left| \left| \partial_{\bvartheta_j \bvartheta_j x_j} \log f_{j,t} \right| \right|_{2,2}^{m_j} \right] < \infty$ by assumption.
\end{proof}
\begin{slemma} \label{LemmaAsymptoticNormalityE}
	Under the assumptions in Theorem \ref{TheoremAsymptoticNormality},
	\begin{equation*}
		\sup_{\btheta \in \bar{\bTheta}} || \nabla_{\btheta \btheta} L_T (\btheta) - (-\bI(\btheta)) ||_{2,2} \overset{a.s.}{\rightarrow} 0 \quad \text{as} \quad T \rightarrow \infty.
	\end{equation*}
\end{slemma}
\begin{proof}
	As in the proof of Lemma \ref{LemmaConsistencyB}, the conclusion follows from the uniform law of large numbers by \citet{Rao1962} if $(\nabla_{\btheta \btheta} \log f (Y_t;\btheta))_{t \in \mathbb{Z}}$ is stationary and ergodic for all $\btheta \in \bar{\bTheta}$ and $\mathbb{E} \left[ \sup_{\btheta \in \bar{\bTheta}} || \nabla_{\btheta \btheta} \log f (Y_t;\btheta) ||_{2,2} \right] < \infty$. $(\nabla_{\btheta \btheta} \log f (Y_t;\btheta))_{t \in \mathbb{Z}}$ is stationary and ergodic for all $\btheta \in \bar{\bTheta}$ since – in addition to the arguments used to show that $(\log f (Y_t;\btheta))_{t \in \mathbb{Z}}$ is stationary and ergodic for all $\btheta \in \bTheta$ (see the proof of Lemma \ref{LemmaConsistencyB}) and $(\nabla_{\btheta} \log f (Y_t;\btheta_0))_{t \in \mathbb{Z}}$ is a stationary and ergodic martingale difference sequence (see the proof of Lemma \ref{LemmaAsymptoticNormalityB}) – for each $j \in \{1,...,J\}$, $(\nabla_{\bupsilon_{j} \bupsilon_{j}} X_{j,t}(\bupsilon_{j}))_{t \in \mathbb{Z}}$ is stationary and ergodic for all $\bupsilon_j \in \bar{\bUpsilon}_j$ by Lemma \ref{LemmaInvertibilityDDX} and, for each $j \in \{1,...,J\}$, $(\nabla_{\btheta \btheta} \pi_{j,t \mid t-1} (\btheta))_{t \in \mathbb{Z}}$ is stationary and ergodic for all $\btheta \in \bar{\bTheta}$ by Lemma \ref{LemmaInvertibilityDDPrediction}. We now show that $\mathbb{E} \left[ \sup_{\btheta \in \bar{\bTheta}} || \nabla_{\btheta \btheta} \log f (Y_t;\btheta) ||_{2,2} \right] < \infty$. We have that
	\begin{align*}
		&|| \nabla_{\btheta \btheta} \log f (Y_t;\btheta) ||_{2,2} \\
		&\leq || \nabla_{\btheta} \log {f} (Y_t;\btheta) ||_{2}^{2} + \frac{|| \nabla_{\btheta \btheta} {f} (Y_t;\btheta) ||_{2,2}}{{f} (Y_t;\btheta)} \\
		&\leq || \nabla_{\btheta} \log {f} (Y_t;\btheta) ||_{2}^{2} + \sum_{j=1}^{J} || \nabla_{\btheta \btheta} {\pi}_{j,t \mid t-1} (\btheta) ||_{2,2} \frac{f_{j} (Y_t;{X}_{j,t} (\bupsilon_j),\bupsilon_j)}{{f} (Y_t;\btheta)} \\
		&+ 2 \sum_{j=1}^{J} || \nabla_{\btheta} {\pi}_{j,t \mid t-1} (\btheta) ||_{2} \frac{|| \nabla_{\btheta} f_{j} (Y_t;{X}_{j,t} (\bupsilon_j),\bupsilon_j) ||_{2}}{{f} (Y_t;\btheta)} + \sum_{j=1}^{J} {\pi}_{j,t \mid t-1} (\btheta) \frac{|| \nabla_{\btheta \btheta} f_{j} (Y_t;{X}_{j,t} (\bupsilon_j),\bupsilon_j) ||_{2,2}}{{f} (Y_t;\btheta)} \\
		&\leq || \nabla_{\btheta} \log {f} (Y_t;\btheta) ||_{2}^{2} + \frac{1}{p_{-}} \sum_{j=1}^{J} || \nabla_{\btheta \btheta} {\pi}_{j,t \mid t-1} (\btheta) ||_{2,2} \\
		&+ \frac{2}{p_{-}} \sum_{j=1}^{J} || \nabla_{\btheta} {\pi}_{j,t \mid t-1} (\btheta) ||_{2} || \nabla_{\btheta} \log f_{j} (Y_t;{X}_{j,t} (\bupsilon_j),\bupsilon_j) ||_{2} \\
		&+ \sum_{j=1}^{J} || \nabla_{\btheta} \log f_{j} (Y_t;{X}_{j,t} (\bupsilon_j),\bupsilon_j) ||_{2}^{2} + \sum_{j=1}^{J} || \nabla_{\btheta \btheta} \log f_{j} (Y_t;{X}_{j,t} (\bupsilon_j),\bupsilon_j) ||_{2,2}
	\end{align*}
	since, by Equation \eqref{eq:remark} in Remark \ref{remark}, $\pi_{j,t \mid t-1} (\btheta) \geq p_{-}$ for all $j \in \{1,...,J\}$, so, by the Cauchy-Schwarz inequality,
	\begin{align*}
		&\mathbb{E} \left[ \sup_{\btheta \in \bar{\bTheta}} || \nabla_{\btheta \btheta} \log f (Y_t;\btheta) ||_{2,2} \right] \\
		&\leq \mathbb{E} \left[ \sup_{\btheta \in \bar{\bTheta}} || \nabla_{\btheta} \log {f} (Y_t;\btheta) ||_{2}^{2} \right] + \frac{1}{p_{-}} \sum_{j=1}^{J} \mathbb{E} \left[ \sup_{\btheta \in \bar{\bTheta}} || \nabla_{\btheta \btheta} {\pi}_{j,t \mid t-1} (\btheta) ||_{2,2} \right] \\
		&+ \frac{2}{p_{-}} \sum_{j=1}^{J} \left( \mathbb{E} \left[ \sup_{\btheta \in \bar{\bTheta}} || \nabla_{\btheta} {\pi}_{j,t \mid t-1} (\btheta) ||_{2}^{2} \right] \right)^{1/2} \left( \mathbb{E} \left[ \sup_{\btheta \in \bar{\bTheta}} || \nabla_{\btheta} \log f_{j} (Y_t;{X}_{j,t} (\bupsilon_j),\bupsilon_j) ||_{2}^{2} \right] \right)^{1/2} \\
		&+ \sum_{j=1}^{J} \mathbb{E} \left[ \sup_{\btheta \in \bar{\bTheta}} || \nabla_{\btheta} \log f_{j} (Y_t;{X}_{j,t} (\bupsilon_j),\bupsilon_j) ||_{2}^{2} \right] + \sum_{j=1}^{J} \mathbb{E} \left[ \sup_{\btheta \in \bar{\bTheta}} || \nabla_{\btheta \btheta} \log f_{j} (Y_t;{X}_{j,t} (\bupsilon_j),\bupsilon_j) ||_{2,2} \right].
	\end{align*}
	Note that – in addition to the arguments used to show that $\mathbb{E} \left[ || \nabla_{\btheta} \log {f} (Y_t;\btheta_0) ||_{2}^{2} \right] < \infty$ (see the proof of Lemma \ref{LemmaAsymptoticNormalityB}) – for each $j \in \{1,...,J\}$, $\mathbb{E} \left[ \sup_{\btheta \in \bar{\bTheta}} || \nabla_{\btheta \btheta} {\pi}_{j,t \mid t-1} (\btheta) ||_{2,2} \right] < \infty$ by Lemma \ref{LemmaInvertibilityDDPrediction}. Hence, $\mathbb{E} \left[ \sup_{\btheta \in \bar{\bTheta}} || \nabla_{\btheta \btheta} \log f (Y_t;\btheta) ||_{2,2} \right] < \infty$ if, for each $j \in \{1,...,J\}$, $\mathbb{E} \left[ \sup_{\btheta \in \bar{\bTheta}} || \nabla_{\btheta \btheta} \log f_{j} (Y_t;{X}_{j,t} (\bupsilon_j),\bupsilon_j) ||_{2,2} \right] < \infty$, which we now show. We have that
	\begin{align*}
		&|| \nabla_{\btheta \btheta} \log f_{j} (Y_t;{X}_{j,t} (\bupsilon_j),\bupsilon_j) ||_{2,2} \\
		&\leq || \nabla_{\btheta} \log f_{j} (Y_t;{X}_{j,t} (\bupsilon_j),\bupsilon_j) ||_{2}^{2} + \frac{|| \nabla_{\btheta \btheta} f_{j} (Y_t;{X}_{j,t} (\bupsilon_j),\bupsilon_j) ||_{2,2}}{f_{j} (Y_t;{X}_{j,t} (\bupsilon_j),\bupsilon_j)} \\
		&\leq || \nabla_{\btheta} \log f_{j} (Y_t;{X}_{j,t} (\bupsilon_j),\bupsilon_j) ||_{2}^{2} + \frac{| \partial_{x_j x_j} f_{j} (Y_t;{X}_{j,t} (\bupsilon_j),\bupsilon_j) |}{f_{j} (Y_t;{X}_{j,t} (\bupsilon_j),\bupsilon_j)} || \nabla_{\btheta} {X}_{j,t} (\bupsilon_j) ||_{2}^{2} \\
		&+ \frac{| \partial_{x_j} f_{j} (Y_t;{X}_{j,t} (\bupsilon_j),\bupsilon_j) |}{f_{j} (Y_t;{X}_{j,t} (\bupsilon_j),\bupsilon_j)} || \nabla_{\btheta \btheta} {X}_{j,t}  (\bupsilon_j) ||_{2,2} + 2 \frac{|| \partial_{\btheta x_j} f_{j} (Y_t;{X}_{j,t} (\bupsilon_j),\bupsilon_j) ||_{2}}{f_{j} (Y_t;{X}_{j,t} (\bupsilon_j),\bupsilon_j)} || \nabla_{\btheta} {X}_{j,t} (\bupsilon_j) ||_{2} \\
		&+ \frac{|| \partial_{\btheta \btheta} f_{j} (Y_t;{X}_{j,t} (\bupsilon_j),\bupsilon_j) ||_{2,2}}{f_{j} (Y_t;{X}_{j,t} (\bupsilon_j),\bupsilon_j)} \\
		&\leq || \nabla_{\btheta} \log f_{j} (Y_t;{X}_{j,t} (\bupsilon_j),\bupsilon_j) ||_{2}^{2} + | \partial_{x_j} \log f_{j} (Y_t;{X}_{j,t} (\bupsilon_j),\bupsilon_j) |^2 || \nabla_{\btheta} {X}_{j,t} (\bupsilon_j) ||_{2}^{2} \\
		&+ | \partial_{x_j x_j} \log f_{j} (Y_t;{X}_{j,t} (\bupsilon_j),\bupsilon_j) | || \nabla_{\btheta} {X}_{j,t} (\bupsilon_j) ||_{2}^{2} + | \partial_{x_j} \log f_{j} (Y_t;{X}_{j,t} (\bupsilon_j),\bupsilon_j) | || \nabla_{\btheta \btheta} {X}_{j,t}  (\bupsilon_j) ||_{2,2} \\
		&+ 2 || \partial_{\btheta} \log f_{j} (Y_t;{X}_{j,t} (\bupsilon_j),\bupsilon_j) ||_{2} | \partial_{x_j} \log f_{j} (Y_t;{X}_{j,t} (\bupsilon_j),\bupsilon_j) | || \nabla_{\btheta} {X}_{j,t} (\bupsilon_j) ||_{2} \\
		&+ 2 || \partial_{\btheta x_j} \log f_{j} (Y_t;{X}_{j,t} (\bupsilon_j),\bupsilon_j) ||_{2} || \nabla_{\btheta} {X}_{j,t} (\bupsilon_j) ||_{2} + || \partial_{\btheta} \log f_{j} (Y_t;{X}_{j,t} (\bupsilon_j),\bupsilon_j) ||_{2}^{2} \\
		&+ || \partial_{\btheta \btheta} \log f_{j} (Y_t;{X}_{j,t} (\bupsilon_j),\bupsilon_j) ||_{2,2}.
	\end{align*}
	Thus, by the Cauchy-Schwarz inequality,
	\begin{align*}
		&\mathbb{E} \left[ \sup_{\btheta \in \bar{\bTheta}} || \nabla_{\btheta \btheta} \log f_{j} (Y_t;{X}_{j,t} (\bupsilon_j),\bupsilon_j) ||_{2,2} \right] \\
		&\leq \mathbb{E} \left[ \sup_{\btheta \in \bar{\bTheta}} || \nabla_{\btheta} \log f_{j} (Y_t;{X}_{j,t} (\bupsilon_j),\bupsilon_j) ||_{2}^{2} \right] + \mathbb{E} \left[ \sup_{\btheta \in \bar{\bTheta}} | \partial_{x_j} \log f_{j} (Y_t;{X}_{j,t} (\bupsilon_j),\bupsilon_j) |^{2} || \nabla_{\btheta} {X}_{j,t} (\bupsilon_j) ||_{2}^{2} \right] \\
		&+ \mathbb{E} \left[ \sup_{\btheta \in \bar{\bTheta}} | \partial_{x_j x_j} \log f_{j} (Y_t;{X}_{j,t} (\bupsilon_j),\bupsilon_j) | || \nabla_{\btheta} {X}_{j,t} (\bupsilon_j) ||_{2}^{2} \right] \\
		&+ \mathbb{E} \left[ \sup_{\btheta \in \bar{\bTheta}} | \partial_{x_j} \log f_{j} (Y_t;{X}_{j,t} (\bupsilon_j),\bupsilon_j) | || \nabla_{\btheta \btheta} {X}_{j,t}  (\bupsilon_j) ||_{2,2} \right] \\
		&+ 2 \left( \mathbb{E} \left[ \sup_{\btheta \in \bar{\bTheta}} | \partial_{x_j} \log f_{j} (Y_t;{X}_{j,t} (\bupsilon_j),\bupsilon_j) |^{2} || \nabla_{\btheta} {X}_{j,t} (\bupsilon_j) ||_{2}^{2} \right] \right)^{1/2} \\
		&\quad \cdot \left( \mathbb{E} \left[ \sup_{\btheta \in \bar{\bTheta}} || \partial_{\btheta} \log f_{j} (Y_t;{X}_{j,t} (\bupsilon_j),\bupsilon_j) ||_{2}^{2} \right] \right)^{1/2} \\
		&+ 2 \mathbb{E} \left[ \sup_{\btheta \in \bar{\bTheta}} || \partial_{\btheta x_j} \log f_{j} (Y_t;{X}_{j,t} (\bupsilon_j),\bupsilon_j) ||_{2} || \nabla_{\btheta} {X}_{j,t} (\bupsilon_j) ||_{2} \right] \\
		&+ \mathbb{E} \left[ \sup_{\btheta \in \bar{\bTheta}} || \partial_{\btheta} \log f_{j} (Y_t;{X}_{j,t} (\bupsilon_j),\bupsilon_j) ||_{2}^{2} \right] + \mathbb{E} \left[ \sup_{\btheta \in \bar{\bTheta}} || \partial_{\btheta \btheta} \log f_{j} (Y_t;{X}_{j,t} (\bupsilon_j),\bupsilon_j) ||_{2,2} \right].
	\end{align*}
	Hence, for each $j \in \{1,...,J\}$, $\mathbb{E} \left[ \sup_{\btheta \in \bar{\bTheta}} || \nabla_{\btheta \btheta} \log f_{j} (Y_t;{X}_{j,t} (\bupsilon_j),\bupsilon_j) ||_{2,2} \right] < \infty$ since – in addition to the arguments used to show that, for each $j \in \{1,...,J\}$, $\mathbb{E} \left[ \sup_{\btheta \in \bar{\bTheta}} || \nabla_{\btheta} \log f_{j} (Y_t;{X}_{j,t} (\bupsilon_{j}),\bupsilon_{j}) ||_{2}^{2} \right] < \infty$ (see the proof of Lemma \ref{LemmaAsymptoticNormalityB}) – for each $j \in \{1,...,J\}$, $\mathbb{E} \left[ \sup_{\bupsilon_j \in \bar{\bUpsilon}_{j}} | \partial_{x_j x_j} \log f_{j} (Y_t;{X}_{j,t} (\bupsilon_j),\bupsilon_j) | || \nabla_{\bupsilon_j} {X}_{j,t} (\bupsilon_j) ||_{2}^{2} \right] < \infty$, $\mathbb{E} \left[ \sup_{\bupsilon_j \in \bar{\bUpsilon}_{j}} | \partial_{x_j} \log f_{j} (Y_t;{X}_{j,t} (\bupsilon_j),\bupsilon_j) | || \nabla_{\bupsilon_j \bupsilon_j} {X}_{j,t}  (\bupsilon_j) ||_{2,2} \right] < \infty$, $\mathbb{E} \left[ \sup_{\bupsilon_j \in \bar{\bUpsilon}_{j}} || \partial_{\bupsilon_j x_j} \log f_{j} (Y_t;{X}_{j,t} (\bupsilon_j),\bupsilon_j) ||_{2} || \nabla_{\bupsilon_j} {X}_{j,t} (\bupsilon_j) ||_{2} \right] < \infty$, and $\mathbb{E} \left[ \sup_{\bupsilon_j \in \bar{\bUpsilon}_{j}} || \partial_{\bupsilon_j \bupsilon_j} \log f_{j} (Y_t;{X}_{j,t} (\bupsilon_j),\bupsilon_j) ||_{2,2} \right] < \infty$ by Assumption \ref{AssumptionF4}.
\end{proof}

\section{Explicit Expressions for Equations \eqref{eq:lemma s.11 1} and \eqref{eq:lemma s.12 1}} \label{sec:s2}

In this section, the explicit expressions for Equations \eqref{eq:lemma s.11 1} and \eqref{eq:lemma s.12 1} are given. The explicit expression for Equation \eqref{eq:lemma s.11 1} is given in Section \ref{sec:s2a}, and the explicit expression for Equation \eqref{eq:lemma s.12 1} is given in Section \ref{sec:s2b}.

Let $\pi_{j,t \mid t}$ denote $\pi_{j,t \mid t} (\btheta)$ and $f_{j,t}$ denote $f_j (Y_t;X_{j,t}(\bupsilon_j),\bupsilon_j)$. Recall that, for each $j \in \{1,...,J\}$, $\pi_{j,t \mid t}$ is given by
\begin{equation*}
	\pi_{j,t \mid t} = \frac{f_{j,t} \mathbf{p}_{\bullet j}^{\prime} \bpi_{t-1 \mid t-1}}{\mathbf{f}_t^{\prime} \mathbf{P}^{\prime} \bpi_{t-1 \mid t-1}},
\end{equation*}
where $\mathbf{p}_{\bullet j} = (p_{1j},...,p_{Jj})^{\prime}$, $\bpi_{t-1 \mid t-1} = (\pi_{1,t-1 \mid t-1},...,\pi_{J,t-1 \mid t-1})^{\prime}$, and $\mathbf{f}_t = (f_{1,t},...,f_{J,t})^{\prime}$, and note that
\begin{equation}
	p_{iJ} = 1-\sum_{j=1}^{J-1} p_{ij} \label{eq:key}
\end{equation}
for all $i \in \{1,...,J\}$.

\subsection{Explicit Expression for Equation \eqref{eq:lemma s.11 1}} \label{sec:s2a}

Let
\begin{equation*}
	\mathbf{P}^{-}
	=
	\begin{bmatrix}
		p_{11}-p_{J1} & \cdots & p_{1(J-1)}-p_{J(J-1)} \\
		\vdots & \ddots & \vdots \\
		p_{(J-1)1}-p_{J1} & \cdots & p_{(J-1)(J-1)}-p_{J(J-1)} \\
	\end{bmatrix},
\end{equation*}
$\mathbf{p}_{\bullet j}^{-} = (p_{1j}-p_{Jj},...,p_{(J-1)j}-p_{Jj})^{\prime}$, $\bpi_{t-1 \mid t-1}^{-} = (\pi_{1,t-1 \mid t-1},...,\pi_{J-1,t-1 \mid t-1})^{\prime}$, and $\mathbf{f}_t^{-} = (f_{1,t}-f_{J,t},...,f_{J-1,t}-f_{J,t})^{\prime}$. Then, for each $j \in \{1,...,J-1\}$, $\nabla_{p_{ab}} \pi_{j,t \mid t}$ is given by
\begin{align}
	\nabla_{p_{ab}} \pi_{j,t \mid t}
	&= \frac{1_{\{b = j\}} f_{j,t} \pi_{a,t-1 \mid t-1} + f_{j,t} (\mathbf{p}_{\bullet j}^{-})^{\prime} \nabla_{p_{ab}} \bpi_{t-1 \mid t-1}^{-}}{\mathbf{f}_t^{\prime} \mathbf{P}^{\prime} \bpi_{t-1 \mid t-1}} \label{eq:t1} \\
	&- \frac{(f_{j,t} \mathbf{p}_{\bullet j}^{\prime} \bpi_{t-1 \mid t-1})((f_{b,t}-f_{J,t}) \pi_{a,t-1 \mid t-1} + (\mathbf{f}_t^{-})^{\prime} (\mathbf{P}^{-})^{\prime} \nabla_{p_{ab}} \bpi^{-}_{t-1 \mid t-1} )}{( \mathbf{f}_t^{\prime} \mathbf{P}^{\prime} \bpi_{t-1 \mid t-1} )^2} \label{eq:t2} \\
	&= \frac{1_{\{b = j\}} f_{j,t} \pi_{a,t-1 \mid t-1}}{\mathbf{f}_t^{\prime} \mathbf{P}^{\prime} \bpi_{t-1 \mid t-1}} \label{eq:t3}\\
	&+ \frac{f_{j,t} (\mathbf{p}_{\bullet j}^{-})^{\prime} \nabla_{p_{ab}} \bpi_{t-1 \mid t-1}^{-} }{\mathbf{f}_t^{\prime} \mathbf{P}^{\prime} \bpi_{t-1 \mid t-1}} \label{eq:t4} \\
	&- \frac{f_{j,t} \mathbf{p}_{\bullet j}^{\prime} \bpi_{t-1 \mid t-1} (f_{b,t}-f_{J,t}) \pi_{a,t-1 \mid t-1} }{( \mathbf{f}_t^{\prime} \mathbf{P}^{\prime} \bpi_{t-1 \mid t-1} )^2} \label{eq:t5} \\
	&- \frac{( f_{j,t} \mathbf{p}_{\bullet j}^{\prime} \bpi_{t-1 \mid t-1} ) ( (\mathbf{f}_t^{-})^{\prime} (\mathbf{P}^{-})^{\prime} \nabla_{p_{ab}} \bpi^{-}_{t-1 \mid t-1} ) }{( \mathbf{f}_t^{\prime} \mathbf{P}^{\prime} \bpi_{t-1 \mid t-1} )^2} \label{eq:t6}
\end{align}
for all $a \in \{1,...,J\}$ and $b \in \{1,...,J-1\}$ because of Equation \eqref{eq:key}, where the term in Equation \eqref{eq:t1} is the derivative of $f_{j,t} \mathbf{p}_{\bullet j}^{\prime} \bpi_{t-1 \mid t-1}$ w.r.t. $p_{ab}$ times $\frac{1}{\mathbf{f}_t^{\prime} \mathbf{P}^{\prime} \bpi_{t-1 \mid t-1}}$, the term in Equation \eqref{eq:t2} is $f_{j,t} \mathbf{p}_{\bullet j}^{\prime} \bpi_{t-1 \mid t-1}$ times the derivative of $\frac{1}{\mathbf{f}_t^{\prime} \mathbf{P}^{\prime} \bpi_{t-1 \mid t-1}}$ w.r.t. $p_{ab}$, and the fact that $\pi_{J,t-1 \mid t-1} = 1-\sum_{j=1}^{J-1} \pi_{j,t-1 \mid t-1}$ has been used to write $\nabla_{p_{ab}} \pi_{j,t \mid t}$ as a function of only $\nabla_{p_{ab}} \pi_{1,t-1 \mid t-1},...,\nabla_{p_{ab}} \pi_{J-1,t-1 \mid t-1}$. Moreover, for each $j \in \{1,...,J-1\}$, $\nabla_{[\bupsilon_a]_{b}} \pi_{j,t \mid t}$ is given by
\begin{align}
	\nabla_{[\bupsilon_a]_{b}} \pi_{j,t \mid t}
	&= \frac{1_{\{a=j,a \in \{1,...,J-1\}\}}  \nabla_{[\bupsilon_a]_{b}} f_{j,t} \mathbf{p}_{\bullet j}^{\prime} \bpi_{t-1 \mid t-1} + f_{j,t} (\mathbf{p}_{\bullet j}^{-})^{\prime} \nabla_{[\bupsilon_a]_{b}} \bpi_{t-1 \mid t-1}^{-} }{\mathbf{f}_t^{\prime} \mathbf{P}^{\prime} \bpi_{t-1 \mid t-1} } \label{eq:ta} \\
	&- \frac{( f_{j,t} \mathbf{p}_{\bullet j}^{\prime} \bpi_{t-1 \mid t-1} ) ( \nabla_{[\bupsilon_a]_{b}} f_{a,t} \mathbf{p}_{\bullet a}^{\prime} \bpi_{t-1 \mid t-1} + (\mathbf{f}_t^{-})^{\prime} (\mathbf{P}^{-})^{\prime} \nabla_{[\bupsilon_a]_{b}} \bpi^{-}_{t-1 \mid t-1})}{( \mathbf{f}_t^{\prime} \mathbf{P}^{\prime} \bpi_{t-1 \mid t-1}  )^2} \label{eq:tb} \\
	&= \frac{1_{\{a=j,a \in \{1,...,J-1\}\}}  \nabla_{[\bupsilon_a]_{b}} f_{j,t} \mathbf{p}_{\bullet j}^{\prime} \bpi_{t-1 \mid t-1} }{\mathbf{f}_t^{\prime} \mathbf{P}^{\prime} \bpi_{t-1 \mid t-1} } \label{eq:tc} \\
	&+ \frac{f_{j,t} (\mathbf{p}_{\bullet j}^{-})^{\prime} \nabla_{[\bupsilon_a]_{b}} \bpi_{t-1 \mid t-1}^{-} }{\mathbf{f}_t^{\prime} \mathbf{P}^{\prime} \bpi_{t-1 \mid t-1} } \label{eq:td} \\
	&- \frac{( f_{j,t} \mathbf{p}_{\bullet j}^{\prime} \bpi_{t-1 \mid t-1} ) ( \nabla_{[\bupsilon_a]_{b}} f_{a,t} \mathbf{p}_{\bullet a}^{\prime} \bpi_{t-1 \mid t-1} )}{( \mathbf{f}_t^{\prime} \mathbf{P}^{\prime} \bpi_{t-1 \mid t-1}  )^2} \label{eq:te} \\
	&- \frac{( f_{j,t} \mathbf{p}_{\bullet j}^{\prime} \bpi_{t-1 \mid t-1} ) ( (\mathbf{f}_t^{-})^{\prime} (\mathbf{P}^{-})^{\prime} \nabla_{[\bupsilon_a]_{b}} \bpi^{-}_{t-1 \mid t-1} )  }{( \mathbf{f}_t^{\prime} \mathbf{P}^{\prime} \bpi_{t-1 \mid t-1} )^2} \label{eq:tf}
\end{align}
for all $a \in \{1,...,J\}$ and $b \in \{1,...,d_a\}$ because of Equation \eqref{eq:key} again, where the term in Equation \eqref{eq:ta} is the derivative of $f_{j,t} \mathbf{p}_{\bullet j}^{\prime} \bpi_{t-1 \mid t-1}$ w.r.t. $[\bupsilon_a]_{b}$ times $\frac{1}{\mathbf{f}_t^{\prime} \mathbf{P}^{\prime} \bpi_{t-1 \mid t-1}}$, the term in Equation \eqref{eq:tb} is $f_{j,t} \mathbf{p}_{\bullet j}^{\prime} \bpi_{t-1 \mid t-1}$ times the derivative of $\frac{1}{\mathbf{f}_t^{\prime} \mathbf{P}^{\prime} \bpi_{t-1 \mid t-1}}$ w.r.t. $[\bupsilon_a]_{b}$, and the fact that $\pi_{J,t-1 \mid t-1} = 1-\sum_{j=1}^{J-1} \pi_{j,t-1 \mid t-1}$ has been used to write $\nabla_{[\bupsilon_a]_{b}} \pi_{j,t \mid t}$ as a function of only $\nabla_{[\bupsilon_a]_{b}} \pi_{1,t-1 \mid t-1},...,\nabla_{[\bupsilon_a]_{b}} \pi_{J-1,t-1 \mid t-1}$.

\subsection{Explicit Expression for Equation \eqref{eq:lemma s.12 1}} \label{sec:s2b}

Let $\mathbf{1}^{-} = (1,...,1)^{\prime}$ be a $(J-1) \times 1$ vector of ones. Then, for each $j \in \{1,...,J-1\}$, $\nabla_{p_{ab} p_{\alpha \beta}} \pi_{j,t \mid t}$ is given by
\begin{align}
	&\nabla_{p_{ab} p_{\alpha \beta}} \pi_{j,t \mid t} \nonumber \\
	&= \frac{1_{\{b=j\}} f_{j,t} (1_{\{a \in \{1,...,J-1\}\}} \nabla_{p_{\alpha \beta}} \pi_{a,t-1 \mid t-1} - 1_{\{a = J\}} (\mathbf{1}^{-})^{\prime} \nabla_{p_{\alpha \beta}} \bpi_{t-1 \mid t-1}^{-} )}{\mathbf{f}_t^{\prime} \mathbf{P}^{\prime} \bpi_{t-1 \mid t-1}} \label{eq:1} \\
	&- \frac{1_{\{b=j\}} f_{j,t} \pi_{a,t-1 \mid t-1} (f_{\beta,t}-f_{J,t}) \pi_{\alpha,t-1 \mid t-1} }{(\mathbf{f}_t^{\prime} \mathbf{P}^{\prime} \bpi_{t-1 \mid t-1})^2} \label{eq:2} \\
	&- \frac{1_{\{b=j\}} f_{j,t} \pi_{a,t-1 \mid t-1} (\mathbf{f}_t^{-})^{\prime} (\mathbf{P}^{-})^{\prime} \nabla_{p_{\alpha \beta}} \bpi^{-}_{t-1 \mid t-1} }{(\mathbf{f}_t^{\prime} \mathbf{P}^{\prime} \bpi_{t-1 \mid t-1})^2} \label{eq:3} \\
	&+ \frac{1_{\{\beta=j\}} f_{j,t} (1_{\{\alpha \in \{1,...,J-1\}\}} \nabla_{p_{a b}} \pi_{\alpha,t-1 \mid t-1} - 1_{\{\alpha = J\}} (\mathbf{1}^{-})^{\prime} \nabla_{p_{a b}} \bpi_{t-1 \mid t-1}^{-} )}{\mathbf{f}_t^{\prime} \mathbf{P}^{\prime} \bpi_{t-1 \mid t-1} } \label{eq:4} \\
	&+ \frac{ f_{j,t} (\mathbf{p}_{\bullet j}^{-})^{\prime} \nabla_{p_{ab} p_{\alpha \beta}} \bpi_{t-1 \mid t-1}^{-}}{\mathbf{f}_t^{\prime} \mathbf{P}^{\prime} \bpi_{t-1 \mid t-1} } \label{eq:5} \\
	&- \frac{f_{j,t} (\mathbf{p}_{\bullet j}^{-})^{\prime} \nabla_{p_{ab}} \bpi_{t-1 \mid t-1}^{-} (f_{\beta,t}-f_{J,t}) \pi_{\alpha,t-1 \mid t-1} }{(\mathbf{f}_t^{\prime} \mathbf{P}^{\prime} \bpi_{t-1 \mid t-1})^2} \label{eq:6} \\
	&- \frac{( f_{j,t} (\mathbf{p}_{\bullet j}^{-})^{\prime} \nabla_{p_{ab}} \bpi_{t-1 \mid t-1}^{-} ) ( (\mathbf{f}_t^{-})^{\prime} (\mathbf{P}^{-})^{\prime} \nabla_{p_{\alpha \beta}} \bpi^{-}_{t-1 \mid t-1} ) }{(\mathbf{f}_t^{\prime} \mathbf{P}^{\prime} \bpi_{t-1 \mid t-1})^2} \label{eq:7} \\
	&- \frac{1_{\{\beta = j\}} f_{j,t} \pi_{\alpha,t-1 \mid t-1} (f_{b,t}-f_{J,t}) \pi_{a,t-1 \mid t-1} }{( \mathbf{f}_t^{\prime} \mathbf{P}^{\prime} \bpi_{t-1 \mid t-1} )^2} \label{eq:8} \\
	&- \frac{ f_{j,t} (\mathbf{p}_{\bullet j}^{-})^{\prime} \nabla_{p_{\alpha \beta}} \bpi_{t-1 \mid t-1}^{-} (f_{b,t}-f_{J,t}) \pi_{a,t-1 \mid t-1} }{( \mathbf{f}_t^{\prime} \mathbf{P}^{\prime} \bpi_{t-1 \mid t-1} )^2} \label{eq:9} \\ 
	&- \frac{ f_{j,t} \mathbf{p}_{\bullet j}^{\prime} \bpi_{t-1 \mid t-1}  (f_{b,t}-f_{J,t}) (1_{\{a \in \{1,...,J-1\}\}} \nabla_{p_{\alpha \beta}} \pi_{a,t-1 \mid t-1} - 1_{\{a = J\}} (\mathbf{1}^{-})^{\prime} \nabla_{p_{\alpha \beta}} \bpi_{t-1 \mid t-1}^{-} ) }{( \mathbf{f}_t^{\prime} \mathbf{P}^{\prime} \bpi_{t-1 \mid t-1} )^2} \label{eq:10} \\
	&+ 2 \frac{ f_{j,t} \mathbf{p}_{\bullet j}^{\prime} \bpi_{t-1 \mid t-1}  (f_{b,t}-f_{J,t}) \pi_{a,t-1 \mid t-1}  (f_{\beta,t}-f_{J,t}) \pi_{\alpha,t-1 \mid t-1}}{( \mathbf{f}_t^{\prime} \mathbf{P}^{\prime} \bpi_{t-1 \mid t-1} )^3} \label{eq:11} \\
	&+ 2 \frac{( f_{j,t} \mathbf{p}_{\bullet j}^{\prime} \bpi_{t-1 \mid t-1} ) (f_{b,t}-f_{J,t}) \pi_{a,t-1 \mid t-1} ( (\mathbf{f}_t^{-})^{\prime} (\mathbf{P}^{-})^{\prime} \nabla_{p_{\alpha \beta}} \bpi^{-}_{t-1 \mid t-1} ) }{( \mathbf{f}_t^{\prime} \mathbf{P}^{\prime} \bpi_{t-1 \mid t-1} )^3} \label{eq:12} \\
	&- \frac{1_{\{\beta = j\}} f_{j,t} \pi_{\alpha,t-1 \mid t-1}  (\mathbf{f}_t^{-})^{\prime} (\mathbf{P}^{-})^{\prime} \nabla_{p_{ab}} \bpi^{-}_{t-1 \mid t-1}  }{( \mathbf{f}_t^{\prime} \mathbf{P}^{\prime} \bpi_{t-1 \mid t-1} )^2} \label{eq:13} \\
	&- \frac{( f_{j,t} (\mathbf{p}_{\bullet j}^{-})^{\prime} \nabla_{p_{\alpha \beta}} \bpi_{t-1 \mid t-1}^{-} ) ( (\mathbf{f}_t^{-})^{\prime} (\mathbf{P}^{-})^{\prime} \nabla_{p_{ab}} \bpi^{-}_{t-1 \mid t-1} ) }{( \mathbf{f}_t^{\prime} \mathbf{P}^{\prime} \bpi_{t-1 \mid t-1} )^2} \label{eq:14} \\ 
	&- \frac{f_{j,t} \mathbf{p}_{\bullet j}^{\prime} \bpi_{t-1 \mid t-1}  (f_{\beta,t}-f_{J,t}) (1_{\{\alpha \in \{1,...,J-1\}\}} \nabla_{p_{a b}} \pi_{\alpha,t-1 \mid t-1} - 1_{\{\alpha = J\}} (\mathbf{1}^{-})^{\prime} \nabla_{p_{a b}} \bpi_{t-1 \mid t-1}^{-} )}{( \mathbf{f}_t^{\prime} \mathbf{P}^{\prime} \bpi_{t-1 \mid t-1} )^2} \label{eq:15} \\
	&- \frac{( f_{j,t} \mathbf{p}_{\bullet j}^{\prime} \bpi_{t-1 \mid t-1} ) ( (\mathbf{f}_t^{-})^{\prime} (\mathbf{P}^{-})^{\prime} \nabla_{p_{ab} p_{\alpha \beta}} \bpi^{-}_{t-1 \mid t-1} )}{( \mathbf{f}_t^{\prime} \mathbf{P}^{\prime} \bpi_{t-1 \mid t-1} )^2} \label{eq:16} \\
	&+ 2 \frac{( f_{j,t} \mathbf{p}_{\bullet j}^{\prime} \bpi_{t-1 \mid t-1} ) ( (\mathbf{f}_t^{-})^{\prime} (\mathbf{P}^{-})^{\prime} \nabla_{p_{ab}} \bpi^{-}_{t-1 \mid t-1} ) (f_{\beta,t}-f_{J,t}) \pi_{\alpha,t-1 \mid t-1} }{( \mathbf{f}_t^{\prime} \mathbf{P}^{\prime} \bpi_{t-1 \mid t-1} )^3} \label{eq:17} \\
	&+ 2 \frac{( f_{j,t} \mathbf{p}_{\bullet j}^{\prime} \bpi_{t-1 \mid t-1} ) ( (\mathbf{f}_t^{-})^{\prime} (\mathbf{P}^{-})^{\prime} \nabla_{p_{ab}} \bpi^{-}_{t-1 \mid t-1} ) ( (\mathbf{f}_t^{-})^{\prime} (\mathbf{P}^{-})^{\prime} \nabla_{p_{\alpha \beta}} \bpi^{-}_{t-1 \mid t-1} )}{( \mathbf{f}_t^{\prime} \mathbf{P}^{\prime} \bpi_{t-1 \mid t-1} )^3} \label{eq:18} 
\end{align}
for all $a \in \{1,...,J\}$, $b \in \{1,...,J-1\}$, $\alpha \in \{1,...,J\}$, and $\beta \in \{1,...,J-1\}$ because of Equation \eqref{eq:key}, where the terms in Equations \eqref{eq:1}-\eqref{eq:3} are the derivative of the term in Equation \eqref{eq:t3} w.r.t. $p_{\alpha \beta}$, the terms in Equations \eqref{eq:4}-\eqref{eq:7} are the derivative of the term in Equation \eqref{eq:t4} w.r.t. $p_{\alpha \beta}$, the terms in Equations \eqref{eq:8}-\eqref{eq:12} are the derivative of the term in Equation \eqref{eq:t5} w.r.t. $p_{\alpha \beta}$, the terms in Equations \eqref{eq:13}-\eqref{eq:18} are the derivative of the term in Equation \eqref{eq:t6} w.r.t. $p_{\alpha \beta}$, and the fact that $\pi_{J,t-1 \mid t-1} = 1-\sum_{j=1}^{J-1} \pi_{j,t-1 \mid t-1}$ has been used to write $\nabla_{p_{ab} p_{\alpha \beta}} \pi_{j,t \mid t}$ as a function of only $\nabla_{p_{ab}} \pi_{1,t-1 \mid t-1},...,\nabla_{p_{ab}} \pi_{J-1,t-1 \mid t-1}$, $\nabla_{p_{\alpha \beta}} \pi_{1,t-1 \mid t-1},...,\nabla_{p_{\alpha \beta}} \pi_{J-1,t-1 \mid t-1}$, and $\nabla_{p_{ab} p_{\alpha \beta}} \pi_{1,t-1 \mid t-1},...,\nabla_{p_{ab} p_{\alpha \beta}} \pi_{J-1,t-1 \mid t-1}$. Moreover, for each $j \in \{1,...,J-1\}$, $\nabla_{p_{ab} [\bupsilon_{\alpha}]_{\beta}} \pi_{j,t \mid t}$ is given by
\begin{align}
	&\nabla_{p_{ab} [\bupsilon_{\alpha}]_{\beta}} \pi_{j,t \mid t} \nonumber \\
	&= \frac{1_{\{b=j\}} 1_{\{\alpha=j,\alpha \in \{1,...,J-1\}\}} \nabla_{[\bupsilon_{\alpha}]_{\beta}} f_{j,t} \pi_{a,t-1 \mid t-1}}{\mathbf{f}_t^{\prime} \mathbf{P}^{\prime} \bpi_{t-1 \mid t-1}} \label{eq:1.1} \\
	&+ \frac{1_{\{b=j\}} f_{j,t} (1_{\{a \in \{1,...,J-1\}\}} \nabla_{[\bupsilon_{\alpha}]_{\beta}} \pi_{a,t-1 \mid t-1}-1_{\{a = J\}} (\mathbf{1}^{-})^{\prime} \nabla_{[\bupsilon_{\alpha}]_{\beta}} \bpi_{t-1 \mid t-1}^{-})}{\mathbf{f}_t^{\prime} \mathbf{P}^{\prime} \bpi_{t-1 \mid t-1}} \label{eq:2.2} \\
	&- \frac{1_{\{b=j\}} f_{j,t} \pi_{a,t-1 \mid t-1}  \nabla_{[\bupsilon_\alpha]_{\beta}} f_{\alpha,t} \mathbf{p}_{\bullet \alpha}^{\prime} \bpi_{t-1 \mid t-1}  }{(\mathbf{f}_t^{\prime} \mathbf{P}^{\prime} \bpi_{t-1 \mid t-1} )^2} \label{eq:3.3} \\
	&- \frac{1_{\{b=j\}} f_{j,t} \pi_{a,t-1 \mid t-1}  (\mathbf{f}_t^{-})^{\prime} (\mathbf{P}^{-})^{\prime} \nabla_{[\bupsilon_\alpha]_{\beta}} \bpi^{-}_{t-1 \mid t-1}  }{(\mathbf{f}_t^{\prime} \mathbf{P}^{\prime} \bpi_{t-1 \mid t-1} )^2} \label{eq:4.4} \\
	&+ \frac{1_{\{\alpha=j,\alpha \in \{1,...,J-1\}\}} \nabla_{[\bupsilon_{\alpha}]_{\beta}} f_{j,t} (\mathbf{p}_{\bullet j}^{-})^{\prime} \nabla_{p_{ab}} \bpi_{t-1 \mid t-1}^{-}}{\mathbf{f}_t^{\prime} \mathbf{P}^{\prime} \bpi_{t-1 \mid t-1} } \label{eq:5.5} \\
	&+ \frac{ f_{j,t} (\mathbf{p}_{\bullet j}^{-})^{\prime} \nabla_{p_{ab} [\bupsilon_{\alpha}]_{\beta}} \bpi_{t-1 \mid t-1}^{-}}{\mathbf{f}_t^{\prime} \mathbf{P}^{\prime} \bpi_{t-1 \mid t-1} } \label{eq:6.6} \\
	&- \frac{(f_{j,t} (\mathbf{p}_{\bullet j}^{-})^{\prime} \nabla_{p_{ab}} \bpi_{t-1 \mid t-1}^{-}) ( \nabla_{[\bupsilon_\alpha]_{\beta}} f_{\alpha,t} \mathbf{p}_{\bullet \alpha}^{\prime} \bpi_{t-1 \mid t-1} ) }{(\mathbf{f}_t^{\prime} \mathbf{P}^{\prime} \bpi_{t-1 \mid t-1} )^2} \label{eq:7.7} \\
	&- \frac{(f_{j,t} (\mathbf{p}_{\bullet j}^{-})^{\prime} \nabla_{p_{ab}} \bpi_{t-1 \mid t-1}^{-}) ( (\mathbf{f}_t^{-})^{\prime} (\mathbf{P}^{-})^{\prime} \nabla_{[\bupsilon_\alpha]_{\beta}} \bpi^{-}_{t-1 \mid t-1} ) }{( \mathbf{f}_t^{\prime} \mathbf{P}^{\prime} \bpi_{t-1 \mid t-1} )^2} \label{eq:8.8} \\
	&- \frac{ 1_{\{\alpha = j,\alpha \in \{1,...,J-1\}\}} \nabla_{[\bupsilon_{\alpha}]_{\beta}} f_{j,t} \mathbf{p}_{\bullet j}^{\prime} \bpi_{t-1 \mid t-1}  (f_{b,t}-f_{J,t}) \pi_{a,t-1 \mid t-1} }{( \mathbf{f}_t^{\prime} \mathbf{P}^{\prime} \bpi_{t-1 \mid t-1} )^2} \label{eq:9.9} \\
	&- \frac{f_{j,t} (\mathbf{p}_{\bullet j}^{-})^{\prime} \nabla_{[\bupsilon_{\alpha}]_{\beta}} \bpi_{t-1 \mid t-1}^{-} (f_{b,t}-f_{J,t}) \pi_{a,t-1 \mid t-1} }{( \mathbf{f}_t^{\prime} \mathbf{P}^{\prime} \bpi_{t-1 \mid t-1} )^2} \label{eq:10.10} \\
	&- \frac{ f_{j,t} \mathbf{p}_{\bullet j}^{\prime} \bpi_{t-1 \mid t-1}  (1_{\{\alpha = b\}} \nabla_{[\bupsilon_{\alpha}]_{\beta}} f_{b,t}-1_{\{\alpha = J\}} \nabla_{[\bupsilon_{\alpha}]_{\beta}} f_{J,t}) \pi_{a,t-1 \mid t-1} }{( \mathbf{f}_t^{\prime} \mathbf{P}^{\prime} \bpi_{t-1 \mid t-1} )^2} \label{eq:11.11} \\
	&- \frac{ f_{j,t} \mathbf{p}_{\bullet j}^{\prime} \bpi_{t-1 \mid t-1}  (f_{b,t}-f_{J,t}) (1_{\{a \in \{1,...,J-1\}\}} \nabla_{[\bupsilon_{\alpha}]_{\beta}} \pi_{a,t-1 \mid t-1}-1_{\{a = J\}} (\mathbf{1}^{-})^{\prime} \nabla_{[\bupsilon_{\alpha}]_{\beta}} \bpi_{t-1 \mid t-1}^{-}) }{( \mathbf{f}_t^{\prime} \mathbf{P}^{\prime} \bpi_{t-1 \mid t-1} )^2} \label{eq:12.12} \\
	&+ 2 \frac{( f_{j,t} \mathbf{p}_{\bullet j}^{\prime} \bpi_{t-1 \mid t-1} ) (f_{b,t}-f_{J,t}) \pi_{a,t-1 \mid t-1} ( \nabla_{[\bupsilon_\alpha]_{\beta}} f_{\alpha,t} \mathbf{p}_{\bullet \alpha}^{\prime} \bpi_{t-1 \mid t-1}  )}{( \mathbf{f}_t^{\prime} \mathbf{P}^{\prime} \bpi_{t-1 \mid t-1} )^3} \label{eq:13.13} \\
	&+ 2 \frac{( f_{j,t} \mathbf{p}_{\bullet j}^{\prime} \bpi_{t-1 \mid t-1} ) (f_{b,t}-f_{J,t}) \pi_{a,t-1 \mid t-1} ((\mathbf{f}_t^{-})^{\prime} (\mathbf{P}^{-})^{\prime} \nabla_{[\bupsilon_\alpha]_{\beta}} \bpi^{-}_{t-1 \mid t-1}) }{( \mathbf{f}_t^{\prime} \mathbf{P}^{\prime} \bpi_{t-1 \mid t-1} )^3} \label{eq:14.14} \\
	&- \frac{( 1_{\{\alpha = j,\alpha \in \{1,...,J-1\}\}} \nabla_{[\bupsilon_{\alpha}]_{\beta}} f_{j,t} \mathbf{p}_{\bullet j}^{\prime} \bpi_{t-1 \mid t-1} ) ( (\mathbf{f}_t^{-})^{\prime} (\mathbf{P}^{-})^{\prime} \nabla_{p_{ab}} \bpi^{-}_{t-1 \mid t-1} ) }{( \mathbf{f}_t^{\prime} \mathbf{P}^{\prime} \bpi_{t-1 \mid t-1} )^2} \label{eq:15.15} \\
	&- \frac{(f_{j,t} (\mathbf{p}_{\bullet j}^{-})^{\prime} \nabla_{[\bupsilon_{\alpha}]_{\beta}} \bpi_{t-1 \mid t-1}^{-}) ( (\mathbf{f}_t^{-})^{\prime} (\mathbf{P}^{-})^{\prime} \nabla_{p_{ab}} \bpi^{-}_{t-1 \mid t-1} ) }{( \mathbf{f}_t^{\prime} \mathbf{P}^{\prime} \bpi_{t-1 \mid t-1} )^2} \label{eq:16.16} \\
	&- \frac{( f_{j,t} \mathbf{p}_{\bullet j}^{\prime} \bpi_{t-1 \mid t-1} ) ( \nabla_{[\bupsilon_{\alpha}]_{\beta}} f_{\alpha,t} (\mathbf{p}_{\bullet \alpha}^{-})^{\prime} \nabla_{p_{ab}} \bpi_{t-1 \mid t-1}^{-} )}{( \mathbf{f}_t^{\prime} \mathbf{P}^{\prime} \bpi_{t-1 \mid t-1} )^2} \label{eq:17.17} \\
	&- \frac{( f_{j,t} \mathbf{p}_{\bullet j}^{\prime} \bpi_{t-1 \mid t-1} ) ( (\mathbf{f}_t^{-})^{\prime} (\mathbf{P}^{-})^{\prime} \nabla_{p_{ab} [\bupsilon_{\alpha}]_{\beta}} \bpi^{-}_{t-1 \mid t-1} ) }{( \mathbf{f}_t^{\prime} \mathbf{P}^{\prime} \bpi_{t-1 \mid t-1} )^2} \label{eq:18.18} \\
	&+ 2 \frac{( f_{j,t} \mathbf{p}_{\bullet j}^{\prime} \bpi_{t-1 \mid t-1} ) ( (\mathbf{f}_t^{-})^{\prime} (\mathbf{P}^{-})^{\prime} \nabla_{p_{ab}} \bpi^{-}_{t-1 \mid t-1} ) ( \nabla_{[\bupsilon_\alpha]_{\beta}} f_{\alpha,t} \mathbf{p}_{\bullet \alpha}^{\prime} \bpi_{t-1 \mid t-1} )}{( \mathbf{f}_t^{\prime} \mathbf{P}^{\prime} \bpi_{t-1 \mid t-1} )^3} \label{eq:19.19} \\
	&+ 2 \frac{( f_{j,t} \mathbf{p}_{\bullet j}^{\prime} \bpi_{t-1 \mid t-1} ) ( (\mathbf{f}_t^{-})^{\prime} (\mathbf{P}^{-})^{\prime} \nabla_{p_{ab}} \bpi^{-}_{t-1 \mid t-1} ) ( (\mathbf{f}_t^{-})^{\prime} (\mathbf{P}^{-})^{\prime} \nabla_{[\bupsilon_\alpha]_{\beta}} \bpi^{-}_{t-1 \mid t-1} ) }{( \mathbf{f}_t^{\prime} \mathbf{P}^{\prime} \bpi_{t-1 \mid t-1} )^3} \label{eq:20.20}
\end{align}
for all $a \in \{1,...,J\}$, $b \in \{1,...,J-1\}$, $\alpha \in \{1,...,J\}$, and $\beta \in \{1,...,d_{\alpha}\}$ because of Equation \eqref{eq:key} again, where the terms in Equations \eqref{eq:1.1}-\eqref{eq:4.4} are the derivative of the term in Equation \eqref{eq:t3} w.r.t. $[\bupsilon_{\alpha}]_{\beta}$, the terms in Equations \eqref{eq:5.5}-\eqref{eq:8.8} are the derivative of the term in Equation \eqref{eq:t4} w.r.t. $[\bupsilon_{\alpha}]_{\beta}$, the terms in Equations \eqref{eq:9.9}-\eqref{eq:14.14} are the derivative of the term in Equation \eqref{eq:t5} w.r.t. $[\bupsilon_{\alpha}]_{\beta}$, the terms in Equations \eqref{eq:15.15}-\eqref{eq:20.20} are the derivative of the term in Equation \eqref{eq:t6} w.r.t. $[\bupsilon_{\alpha}]_{\beta}$, and the fact that $\pi_{J,t-1 \mid t-1} = 1-\sum_{j=1}^{J-1} \pi_{j,t-1 \mid t-1}$ has been used to write $\nabla_{p_{ab} [\bupsilon_{\alpha}]_{\beta}} \pi_{j,t \mid t}$ as a function of only $\nabla_{p_{ab}} \pi_{1,t-1 \mid t-1},...,\nabla_{p_{ab}} \pi_{J-1,t-1 \mid t-1}$, $\nabla_{[\bupsilon_\alpha]_{\beta}} \pi_{1,t-1 \mid t-1},...,\nabla_{[\bupsilon_\alpha]_{\beta}} \pi_{J-1,t-1 \mid t-1}$, and $\nabla_{p_{ab} [\bupsilon_\alpha]_{\beta}} \pi_{1,t-1 \mid t-1},...,\nabla_{p_{ab} [\bupsilon_\alpha]_{\beta}} \pi_{J-1,t-1 \mid t-1}$. Finally, for each $j \in \{1,...,J-1\}$, $\nabla_{[\bupsilon_a]_{b} [\bupsilon_{\alpha}]_{\beta}} \pi_{j,t \mid t}$ is given by
\begin{align}
	&\nabla_{[\bupsilon_a]_{b} [\bupsilon_{\alpha}]_{\beta}} \pi_{j,t \mid t} \nonumber \\
	&= \frac{1_{\{a = j,a \in \{1,...,J-1\}\}} 1_{\{\alpha = j,\alpha \in \{1,...,J-1\}\}} \nabla_{[\bupsilon_a]_{b} [\bupsilon_{\alpha}]_{\beta}} f_{j,t} \mathbf{p}_{\bullet j}^{\prime} \bpi_{t-1 \mid t-1}}{\mathbf{f}_t^{\prime} \mathbf{P}^{\prime} \bpi_{t-1 \mid t-1} } \label{eq:1.1.1} \\
	&+ \frac{1_{\{a = j,a \in \{1,...,J-1\}\}} \nabla_{[\bupsilon_a]_{b}} f_{j,t} (\mathbf{p}_{\bullet j}^{-})^{\prime} \nabla_{[\bupsilon_\alpha]_{\beta}} \bpi_{t-1 \mid t-1}^{-}}{\mathbf{f}_t^{\prime} \mathbf{P}^{\prime} \bpi_{t-1 \mid t-1} } \label{eq:2.2.2} \\
	&- \frac{( 1_{\{a=j,a \in \{1,...,J-1\}\}}  \nabla_{[\bupsilon_a]_{b}} f_{j,t} \mathbf{p}_{\bullet j}^{\prime} \bpi_{t-1 \mid t-1} ) ( \nabla_{[\bupsilon_\alpha]_{\beta}} f_{\alpha,t} \mathbf{p}_{\bullet \alpha}^{\prime} \bpi_{t-1 \mid t-1} )}{( \mathbf{f}_t^{\prime} \mathbf{P}^{\prime} \bpi_{t-1 \mid t-1} )^2} \label{eq:3.3.3} \\
	&- \frac{( 1_{\{a=j,a \in \{1,...,J-1\}\}}  \nabla_{[\bupsilon_a]_{b}} f_{j,t} \mathbf{p}_{\bullet j}^{\prime} \bpi_{t-1 \mid t-1} ) ( (\mathbf{f}_t^{-})^{\prime} (\mathbf{P}^{-})^{\prime} \nabla_{[\bupsilon_\alpha]_{\beta}} \bpi^{-}_{t-1 \mid t-1} )}{( \mathbf{f}_t^{\prime} \mathbf{P}^{\prime} \bpi_{t-1 \mid t-1} )^2} \label{eq:4.4.4} \\
	&+ \frac{1_{\{\alpha = j,\alpha \in \{1,...,J-1\}\}} \nabla_{[\bupsilon_{\alpha}]_{\beta}} f_{j,t} (\mathbf{p}_{\bullet j}^{-})^{\prime} \nabla_{[\bupsilon_a]_{b}} \bpi_{t-1 \mid t-1}^{-} }{\mathbf{f}_t^{\prime} \mathbf{P}^{\prime} \bpi_{t-1 \mid t-1} } \label{eq:5.5.5} \\
	&+ \frac{f_{j,t} (\mathbf{p}_{\bullet j}^{-})^{\prime} \nabla_{[\bupsilon_a]_{b} [\bupsilon_{\alpha}]_{\beta}} \bpi_{t-1 \mid t-1}^{-} }{\mathbf{f}_t^{\prime} \mathbf{P}^{\prime} \bpi_{t-1 \mid t-1} } \label{eq:6.6.6} \\
	&- \frac{( f_{j,t} (\mathbf{p}_{\bullet j}^{-})^{\prime} \nabla_{[\bupsilon_a]_{b}} \bpi_{t-1 \mid t-1}^{-} ) ( \nabla_{[\bupsilon_\alpha]_{\beta}} f_{\alpha,t} \mathbf{p}_{\bullet \alpha}^{\prime} \bpi_{t-1 \mid t-1} ) }{( \mathbf{f}_t^{\prime} \mathbf{P}^{\prime} \bpi_{t-1 \mid t-1} )^2} \label{eq:7.7.7} \\
	&- \frac{( f_{j,t} (\mathbf{p}_{\bullet j}^{-})^{\prime} \nabla_{[\bupsilon_a]_{b}} \bpi_{t-1 \mid t-1}^{-} ) ( (\mathbf{f}_t^{-})^{\prime} (\mathbf{P}^{-})^{\prime} \nabla_{[\bupsilon_\alpha]_{\beta}} \bpi^{-}_{t-1 \mid t-1} ) }{( \mathbf{f}_t^{\prime} \mathbf{P}^{\prime} \bpi_{t-1 \mid t-1} )^2} \label{eq:8.8.8} \\
	&- \frac{(1_{\{\alpha = j,\alpha \in \{1,...,J-1\}\}} \nabla_{[\bupsilon_{\alpha}]_{\beta}} f_{j,t} \mathbf{p}_{\bullet j}^{\prime} \bpi_{t-1 \mid t-1} ) ( \nabla_{[\bupsilon_a]_{b}} f_{a,t} \mathbf{p}_{\bullet a}^{\prime} \bpi_{t-1 \mid t-1} )}{( \mathbf{f}_t^{\prime} \mathbf{P}^{\prime} \bpi_{t-1 \mid t-1}  )^2} \label{eq:9.9.9} \\
	&- \frac{( f_{j,t} (\mathbf{p}_{\bullet j}^{-})^{\prime} \nabla_{[\bupsilon_\alpha]_{\beta}} \bpi_{t-1 \mid t-1}^{-} ) ( \nabla_{[\bupsilon_a]_{b}} f_{a,t} \mathbf{p}_{\bullet a}^{\prime} \bpi_{t-1 \mid t-1} )}{( \mathbf{f}_t^{\prime} \mathbf{P}^{\prime} \bpi_{t-1 \mid t-1}  )^2} \label{eq:10.10.10} \\
	&- \frac{( f_{j,t} \mathbf{p}_{\bullet j}^{\prime} \bpi_{t-1 \mid t-1} ) ( 1_{\{\alpha = a\}} \nabla_{[\bupsilon_a]_{b} [\bupsilon_{\alpha}]_{\beta}} f_{a,t} \mathbf{p}_{\bullet a}^{\prime} \bpi_{t-1 \mid t-1} )}{( \mathbf{f}_t^{\prime} \mathbf{P}^{\prime} \bpi_{t-1 \mid t-1} )^2} \label{eq:11.11.11} \\
	&- \frac{( f_{j,t} \mathbf{p}_{\bullet j}^{\prime} \bpi_{t-1 \mid t-1} ) ( \nabla_{[\bupsilon_a]_{b}} f_{a,t} (\mathbf{p}_{\bullet a}^{-})^{\prime} \nabla_{[\bupsilon_{\alpha}]_{\beta}} \bpi_{t-1 \mid t-1}^{-} )}{( \mathbf{f}_t^{\prime} \mathbf{P}^{\prime} \bpi_{t-1 \mid t-1}  )^2} \label{eq:12.12.12} \\
	&+ 2 \frac{( f_{j,t} \mathbf{p}_{\bullet j}^{\prime} \bpi_{t-1 \mid t-1} ) ( \nabla_{[\bupsilon_a]_{b}} f_{a,t} \mathbf{p}_{\bullet a}^{\prime} \bpi_{t-1 \mid t-1} ) ( \nabla_{[\bupsilon_\alpha]_{\beta}} f_{\alpha,t} \mathbf{p}_{\bullet \alpha}^{\prime} \bpi_{t-1 \mid t-1} ) }{( \mathbf{f}_t^{\prime} \mathbf{P}^{\prime} \bpi_{t-1 \mid t-1}  )^3} \label{eq:13.13.13} \\
	&+ 2 \frac{( f_{j,t} \mathbf{p}_{\bullet j}^{\prime} \bpi_{t-1 \mid t-1} ) ( \nabla_{[\bupsilon_a]_{b}} f_{a,t} \mathbf{p}_{\bullet a}^{\prime} \bpi_{t-1 \mid t-1} ) ( (\mathbf{f}_t^{-})^{\prime} (\mathbf{P}^{-})^{\prime} \nabla_{[\bupsilon_\alpha]_{\beta}} \bpi^{-}_{t-1 \mid t-1} )}{( \mathbf{f}_t^{\prime} \mathbf{P}^{\prime} \bpi_{t-1 \mid t-1}  )^3} \label{eq:14.14.14} \\
	&- \frac{(1_{\{\alpha = j,\alpha \in \{1,...,J-1\}\}} \nabla_{[\bupsilon_{\alpha}]_{\beta}} f_{j,t} \mathbf{p}_{\bullet j}^{\prime} \bpi_{t-1 \mid t-1} ) ( (\mathbf{f}_t^{-})^{\prime} (\mathbf{P}^{-})^{\prime} \nabla_{[\bupsilon_a]_{b}} \bpi^{-}_{t-1 \mid t-1} ) }{( \mathbf{f}_t^{\prime} \mathbf{P}^{\prime} \bpi_{t-1 \mid t-1}  )^2} \label{eq:15.15.15} \\
	&- \frac{( f_{j,t} (\mathbf{p}_{\bullet j}^{-})^{\prime} \nabla_{[\bupsilon_\alpha]_{\beta}} \bpi_{t-1 \mid t-1}^{-}  ) ( (\mathbf{f}_t^{-})^{\prime} (\mathbf{P}^{-})^{\prime} \nabla_{[\bupsilon_a]_{b}} \bpi^{-}_{t-1 \mid t-1} )}{( \mathbf{f}_t^{\prime} \mathbf{P}^{\prime} \bpi_{t-1 \mid t-1}  )^2} \label{eq:16.16.16} \\
	&- \frac{( f_{j,t} \mathbf{p}_{\bullet j}^{\prime} \bpi_{t-1 \mid t-1} ) ( \nabla_{[\bupsilon_{\alpha}]_{\beta}} f_{\alpha,t} (\mathbf{p}_{\bullet \alpha}^{-})^{\prime} \nabla_{[\bupsilon_{a}]_{b}} \bpi_{t-1 \mid t-1}^{-} )}{( \mathbf{f}_t^{\prime} \mathbf{P}^{\prime} \bpi_{t-1 \mid t-1} )^2} \label{eq:17.17.17} \\
	&- \frac{( f_{j,t} \mathbf{p}_{\bullet j}^{\prime} \bpi_{t-1 \mid t-1} ) ( (\mathbf{f}_t^{-})^{\prime} (\mathbf{P}^{-})^{\prime} \nabla_{[\bupsilon_a]_{b} [\bupsilon_{\alpha}]_{\beta}} \bpi^{-}_{t-1 \mid t-1} )}{( \mathbf{f}_t^{\prime} \mathbf{P}^{\prime} \bpi_{t-1 \mid t-1}  )^2} \label{eq:18.18.18} \\
	&+ 2 \frac{( f_{j,t} \mathbf{p}_{\bullet j}^{\prime} \bpi_{t-1 \mid t-1} ) ( (\mathbf{f}_t^{-})^{\prime} (\mathbf{P}^{-})^{\prime} \nabla_{[\bupsilon_a]_{b}} \bpi^{-}_{t-1 \mid t-1} ) ( \nabla_{[\bupsilon_\alpha]_{\beta}} f_{\alpha,t} \mathbf{p}_{\bullet \alpha}^{\prime} \bpi_{t-1 \mid t-1} )}{( \mathbf{f}_t^{\prime} \mathbf{P}^{\prime} \bpi_{t-1 \mid t-1}  )^3} \label{eq:19.19.19} \\
	&+ 2 \frac{( f_{j,t} \mathbf{p}_{\bullet j}^{\prime} \bpi_{t-1 \mid t-1} ) ( (\mathbf{f}_t^{-})^{\prime} (\mathbf{P}^{-})^{\prime} \nabla_{[\bupsilon_a]_{b}} \bpi^{-}_{t-1 \mid t-1} ) ( (\mathbf{f}_t^{-})^{\prime} (\mathbf{P}^{-})^{\prime} \nabla_{[\bupsilon_\alpha]_{\beta}} \bpi^{-}_{t-1 \mid t-1} )}{( \mathbf{f}_t^{\prime} \mathbf{P}^{\prime} \bpi_{t-1 \mid t-1}  )^3} \label{eq:20.20.20}
\end{align}
for all $a \in \{1,...,J\}$, $b \in \{1,...,d_a\}$, $\alpha \in \{1,...,J\}$, and $\beta \in \{1,...,d_{\alpha}\}$ because of Equation \eqref{eq:key} once again, where the terms in Equations \eqref{eq:1.1.1}-\eqref{eq:4.4.4} are the derivative of the term in Equation \eqref{eq:tc} w.r.t. $[\bupsilon_{\alpha}]_{\beta}$, the terms in Equations \eqref{eq:5.5.5}-\eqref{eq:8.8.8} are the derivative of the term in Equation \eqref{eq:td} w.r.t. $[\bupsilon_{\alpha}]_{\beta}$, the terms in Equations \eqref{eq:9.9.9}-\eqref{eq:14.14.14} are the derivative of the term in Equation \eqref{eq:te} w.r.t. $[\bupsilon_{\alpha}]_{\beta}$, the terms in Equations \eqref{eq:15.15.15}-\eqref{eq:20.20.20} are the derivative of the term in Equation \eqref{eq:tf} w.r.t. $[\bupsilon_{\alpha}]_{\beta}$, and the fact that $\pi_{J,t-1 \mid t-1} = 1-\sum_{j=1}^{J-1} \pi_{j,t-1 \mid t-1}$ has been used to write $\nabla_{[\bupsilon_a]_{b} [\bupsilon_{\alpha}]_{\beta}} \pi_{j,t \mid t}$ as a function of only $\nabla_{[\bupsilon_a]_{b}} \pi_{1,t-1 \mid t-1},...,\nabla_{[\bupsilon_a]_{b}} \pi_{J-1,t-1 \mid t-1}$, $\nabla_{[\bupsilon_\alpha]_{\beta}} \pi_{1,t-1 \mid t-1},...,\nabla_{[\bupsilon_\alpha]_{\beta}} \pi_{J-1,t-1 \mid t-1}$, and $\nabla_{[\bupsilon_a]_{b} [\bupsilon_\alpha]_{\beta}} \pi_{1,t-1 \mid t-1},...,\nabla_{[\bupsilon_a]_{b} [\bupsilon_\alpha]_{\beta}} \pi_{J-1,t-1 \mid t-1}$.

\end{document}